\documentclass[acmsmall,nonacm]{acmart}

\acmConference{}{}{}
\acmBooktitle{}
\acmPrice{}
\acmISBN{} 
\startPage{1}

\setcopyright{none}

\usepackage{booktabs}   
\usepackage{subcaption} 

\usepackage{xspace}
\usepackage{listings}
\usepackage{ebproof}
\usepackage{dsfont}
\usepackage{listings-zk-secrec}
\usepackage{todonotes}
\usepackage{ifthen}

\newtheorem{theorem}{Theorem}[section]
\newtheorem{lemma}[theorem]{Lemma}

\makeatletter
\newenvironment{thmlist}{
  \begin{enumerate}
  \renewcommand{\theenumi}{(\@arabic\c@enumi)}
  
  \renewcommand{\p@enumi}{\thetheorem\ }
}{\end{enumerate}}
\makeatother

\renewcommand{\paragraph}[1]{\textbf{\textit{#1.}} }
\newcommand{\ZKSC}{\textsc{ZK-SecreC}\xspace}

\newcommand{\ZZ}{\mathbb{Z}}
\newcommand{\NN}{\mathbb{N}}
\newcommand{\FF}{\mathbb{F}}
\newcommand{\BB}{\mathbb{B}}

\newcommand{\GG}{\mathbb{G}}

\newcommand{\subtype}{\mathrel{\mathchar"13C\mathchar"3A}}
\newcommand{\set}[1]{\left\{#1\right\}}
\newcommand{\uopl}[1]{\mathop{\mathgroup\symoperators #1}\nolimits}

\SetSymbolFont{operators}{normal}{OT1}{cmss}{m}{n}
\DeclareMathSymbol{!}{\mathbin}{operators}{"21}
\DeclareMathSymbol{\simop}{\mathop}{symbols}{"18}

\let\figsize=\footnotesize

\begin{document}

\title{\ZKSC: a Domain-Specific Language for Zero-Knowledge Proofs}


\author{Dan Bogdanov}
\affiliation{
  \department{Information Security Research Institute}              
  \institution{Cybernetica AS}            
  \streetaddress{Narva Rd. 20}
  \city{Tartu}
  \postcode{51009}
  \country{Estonia}                    
}
\email{dan.bogdanov@cyber.ee}          

\author{Joosep Jääger}
\affiliation{
  \department{Information Security Research Institute}              
  \institution{Cybernetica AS}            
  \streetaddress{Narva Rd. 20}
  \city{Tartu}
  \postcode{51009}
  \country{Estonia}                    
}

\author{Peeter Laud}
\orcid{0000-0002-9030-8142}             
\affiliation{
  \department{Information Security Research Institute}             
  \institution{Cybernetica AS}           
  \streetaddress{Narva Rd. 20}
  \city{Tartu}
  \postcode{51009}
  \country{Estonia}                   
}
\email{peeter.laud@cyber.ee}         

\author{Härmel Nestra}
\orcid{0000-0001-7050-7171}             
\affiliation{
  \department{Information Security Research Institute}              
  \institution{Cybernetica AS}            
  \streetaddress{Narva Rd. 20}
  \city{Tartu}
  \postcode{51009}
  \country{Estonia}                    
}
\email{harmel.nestra@cyber.ee}          

\author{Martin Pettai}
\affiliation{
  \department{Information Security Research Institute}              
  \institution{Cybernetica AS}            
  \streetaddress{Narva Rd. 20}
  \city{Tartu}
  \postcode{51009}
  \country{Estonia}                    
}
\email{martin.pettai@cyber.ee}          

\author{Jaak Randmets}
\affiliation{
  \department{Information Security Research Institute}              
  \institution{Cybernetica AS}            
  \streetaddress{Narva Rd. 20}
  \city{Tartu}
  \postcode{51009}
  \country{Estonia}                    
}
\email{jaak.randmets@cyber.ee}          

\author{Ville Sokk}
\affiliation{
  \department{Information Security Research Institute}              
  \institution{Cybernetica AS}            
  \streetaddress{Narva Rd. 20}
  \city{Tartu}
  \postcode{51009}
  \country{Estonia}                    
}

\author{Kert Tali}
\affiliation{
  \department{Information Security Research Institute}              
  \institution{Cybernetica AS}            
  \streetaddress{Narva Rd. 20}
  \city{Tartu}
  \postcode{51009}
  \country{Estonia}                    
}
\email{kert.tali@cyber.ee}          

\author{Sandhra-Mirella Valdma}
\affiliation{
  \department{Information Security Research Institute}              
  \institution{Cybernetica AS}            
  \streetaddress{Narva Rd. 20}
  \city{Tartu}
  \postcode{51009}
  \country{Estonia}                    
}
\email{sandhra-mirella.valdma@cyber.ee}          

\lstset{language=ZK-SecreC,basicstyle=\figsize\sffamily,columns=fullflexible,tabsize=4,breaklines=true,breakatwhitespace=true}
\ebproofset{rule margin=0.25ex,separation=1em}

\begin{abstract}
We present \ZKSC, a domain-specific language for zero-knowledge proofs. We present the rationale for its design, its syntax and semantics, and demonstrate its usefulness on the basis of a number of non-trivial examples. The design features a type system, where each piece of data is assigned both a confidentiality and an integrity type, which are not orthogonal to each other. We perform an empiric evaluation of the statements produced by its compiler in terms of their size. We also show the integration of the compiler with the implementation of a zero-knowledge proof technique, and evaluate the running time of both Prover and Verifier.
\end{abstract}




\maketitle

\section{Introduction}\label{intro}


Zero-knowledge proofs (ZKP)~\cite{DBLP:conf/stoc/GoldwasserMR85} are two-party protocols between Prover and Verifier, where the former attempts to convince the latter that he has a piece of knowledge that validates a statement, while not revealing anything about this knowledge. Here this statement is seen as a binary relation $R$ that takes as inputs the \emph{instance} --- common knowledge of Prover and Verifier ---, and the \emph{witness} --- Prover's private knowledge ---, and decides whether the latter validates the former. Among the first practical instances of ZKP was privacy-preserving identification~\cite{DBLP:journals/joc/Schnorr91}, where both the client and the server knew client's public key $h$ --- an element in a cyclic group $\GG$ of size $p$ with a hard discrete logarithm problem ---, and the client proved to the server that he knew the discrete logarithm $x$ of $h$ (to the basis of a generator $g$). Considering $g$ a public parameter, the relation $R$ here is a subset of $\GG\times\ZZ_p$, where $\ZZ_p=\{0,1,\ldots,p-1\}$ is the ring of integers modulo $p$. We have $(h,x)\in R$ iff $g^x=h$. Both the client and the server knew the first component, but only the client knew the second one. The protocol convinced the server of having established communication with someone who knows $x$ such that $(h,x)\in R$, but gave no further information about $x$.

Later advances in ZKP~\cite{DBLP:conf/stoc/IshaiKOS07,ZKBoo,DBLP:conf/eurocrypt/Groth16,bunz2018bulletproofs} allow the creation of ZKP for statements that are significantly bigger and more complex than the described identification protocol, with privacy-preserving distributed ledgers~\cite{DBLP:conf/sp/MiersG0R13,narula2018zkledger} showing a large variety. More applications from a heterogeneous variety of areas are expected in the future~\cite{ZKProofApplicationsTrackProceeding}. All these applications need tools for expressing the binary relation.

All cryptographic techniques for ZKP expect the relation to be expressed as an arithmetic circuit over some finite ring (for most techniques, a field, often with additional constraints), handling only a very limited number of operations. For example, even if there existed a field $\FF$ that suitably embedded $\GG$ and $\ZZ_p$ from our first example, and there were ZKP techniques for circuits over $\FF$, the relation $R$ probably would not be directly representable as an arithmetic circuit over $\FF$, because exponentiation is not among the supported operations. Additions and multiplications are supported, hence we could express $R$ if we represented $x$ as a sequence of bits. These bits can be computed by Prover, and added to the second component of $R$ (changing its type in the process). If $R$ is but a part of a larger relation, then we may need $x$ being represented in various ways. Prover can add all these representations to the inputs of that relation. But then Verifier needs to be convinced that all these representations correspond to the same value of $x$, hence there need to be checks for that. When expressing a relation for the subsequent use of ZKP techniques, we want to state what extra values should be added to the witness and the instance in order to make the circuit simpler (or computable at all), and which checks should be included.

Describing and encoding a circuit and the \emph{expansion} of its inputs directly is error-prone; it is difficult to specify the circuit, as well as to understand what it does. Such state of affairs may be acceptable if ZK proofs are expected to be given only for a small number of relations $R$, but a high-level domain-specific language and associated compilation tools for specifying the relations $R$ are desirable for wide-spread adoption. The specification should mostly be in terms used by common programming languages, exposing only those ZKP-specific details that are highly significant for obtaining a circuit that is handled efficiently by the cryptographic technique and does not reveal Prover's inputs. The specified relation should be automatically translated into the arithmetic circuit, while being optimized for the performance profiles of ZKP techniques. Besides the construction of the arithmetic circuit, the language and compilation tools must help in the preparation of common and Prover's inputs for it. The features of the language have to support the execution by the two parties and a common ZKP protocol between them, with these components not fully trusting each other.

In this paper, we propose \ZKSC --- a programming language for specifying relations between instances and witnesses, together with a toolchain that produces circuits suitable as inputs for ZKP techniques. In their design, we have aimed to tackle the following issues.

\paragraph{Execution at multiple locations} A \ZKSC program specifies, and the \ZKSC compiler produces the description of a circuit, which both Prover and Verifier use as one of the inputs to a cryptographic technique for ZK proofs. Besides running the cryptographic protocol, both parties may need to run some parts of the program locally for the purpose of increasing the efficiency of computations. A \ZKSC program can thus specify local computations, and the compiler can produce code that Verifier will and Prover should execute with their inputs.

\paragraph{Compilation into a circuit} Having the arithmetic circuit as an intermediate representation makes our compiler agnostic towards the used cryptographic techniques, and allows the compilation result to be retargeted easily, so that it may be created once and then used multiple times by Prover to convince Verifier that it knows the witnesses for several different instances. Even though the circuit may only contain a very restricted set of operations, targeting only a particular ZKP technology will not significantly increase the supported set. Our type system for \ZKSC makes sure that non-supported operations cannot be added to the arithmetic circuit, and that the shape of the circuit and the operations in its nodes are public.

\paragraph{Witness and instance expansion} Adding more inputs to the circuit and verifying that they are correctly related to previous inputs is a pervasive technique for improving the efficiency of computations of the relation. \ZKSC allows to freely mix computations on circuit and off circuit. The results of the latter become additional inputs of the circuit. Tackling these issues allows us to support the deployment model that we consider likely for many use-cases of ZKP; the model is depicted in Fig.~\ref{fig:deploymentmodel}.

\begin{figure}
\includegraphics[width=\textwidth]{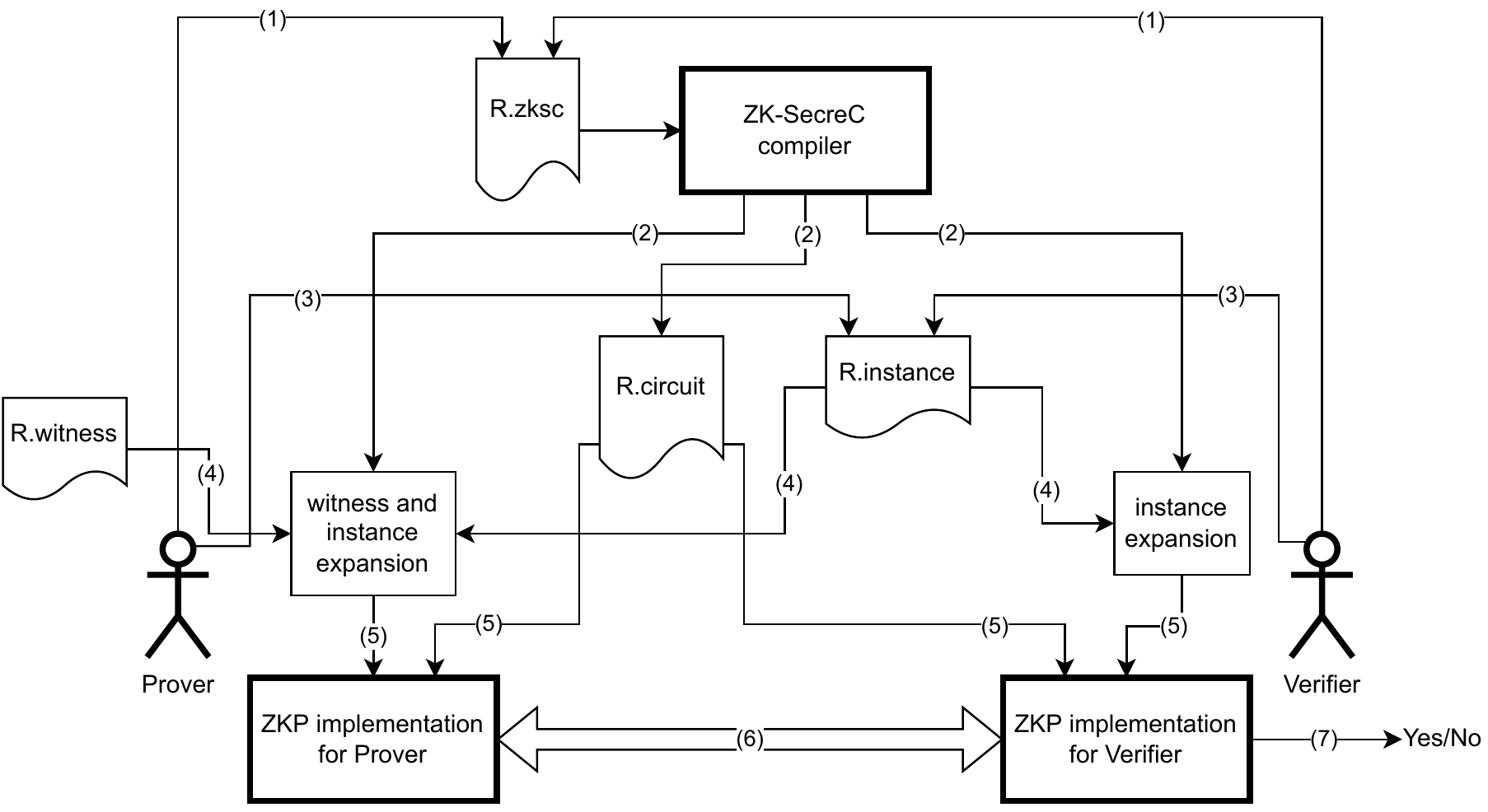}

\caption{Deploying ZK proofs with \ZKSC tools. (1) Prover and Verifier agree on the description of the relation $R$. (2) The relation $R$ is compiled, perhaps by a service provider, perhaps by Prover and Verifier themselves, producing the description of the circuit, and executables for processing the instance and witness. (3) Prover and Verifier agree on the instance for which the proof will be made; the Prover must also have a witness. (4) Prover and Verifier expand the given instance and witness. (5) The implementation of a (cryptographic) ZKP technology gets as input the circuit, and the instance and the witness in the form ingested by the circuit. (6) ZKP protocol is executed. (7) Verifier obtains the verification result.}\label{fig:deploymentmodel}
\end{figure}

\paragraph{Semantics} The statement to be proven is commonly seen as a binary relation, but its meaning is really the left projection of that relation. The semantics of \ZKSC precisely defines the meaning of its programs
.

The language constructions and the type system of \ZKSC capture the essential details common to all ZK proofs. A ZK proof aims to give both confidentiality and integrity guarantees. \ZKSC precisely fixes the possible movements of data between different domains, enforcing the ``no read up'' and ``no write down'' properties without exceptions. The viability of such strong restrictions shows that the confidentiality guarantees are similar for all ZK proofs, and do not really depend on the particular relation, to which the ZK proofs technique has been applied. On the other hand, the necessary checks for making sure that the relation is satisfied are very much a part of the description of that relation; their inclusion is the responsibility of the programmer encoding this relation in \ZKSC. Hence, the language cannot offer much of \emph{formal} support in verifying that the desired integrity properties have been specified by the programmer.

We start this paper in Sect.~\ref{sec:zkscexample} with an example program in \ZKSC, showing off its features. We continue with the description of the syntax of \ZKSC in Sect.~\ref{sec:syntax} and its type system in Sect.~\ref{sec:typesystem}. The execution of a  well-typed program can be split between different domains in the manner that we desire, with necessary data available at each domain. An arithmetic circuit, suitable as an input to a ZKP technique, can also be statically extracted from a well-typed program. In Sect.~\ref{sec:semantics} we give a formal semantics of \ZKSC, stating the language that is accepted by a \ZKSC program. Following up on it, we describe the compilation into an arithmetic circuit in Sect.~\ref{sec:compilation}. We continue with the evaluation of expressivity and efficiency of \ZKSC, first showing in Sect.~\ref{sec:extensions} how some data structures and methods useful for ZKP can be straightforwardly encoded, and then discussing in Sect.~\ref{sec:evaluation} the circuits output by \ZKSC compiler for various example problems. We finish the paper in Sect.~\ref{sec:relwork} by comparing \ZKSC against other existing languages and means of specifying statements proved in ZK. The appendices provide full proofs of the lemmas and theorems.


\section{\ZKSC on an Example}\label{sec:zkscexample}
Let us start the description of the language with an example. Suppose that Prover and Verifier both know a large integer~$z$. Prover wants to convince Verifier that he also knows a factor~$x$ of~$z$ such that $1<x<z$. In ZKP terms, $z$~is the \emph{instance} and $x$~is the \emph{witness}. The relation between $z$ and~$x$ can be specified in \ZKSC as shown in Fig.~\ref{language:factor}. This program specifies an arithmetic circuit that Prover and Verifier have to execute on top of a ZKP technology of their choice, as well as the local computations they both should perform. 

\begin{figure}[tbp]
\begin{lstlisting}
fn bitextract_pre[@D](x : uint[N] $pre @D, fbw : uint $pre @public) -> list[bool[N] $pre @D] {
	let rec xx = for i in 0 .. fbw {                         // generates a list of length fbw
		if (i == 0) { x } else { xx[i - 1] / 2 }
	};
	for i in 0 .. fbw { // generates a list of booleans, which is returned
		let b = xx[i] % 2;
		b == 1
	}
}
fn check_bitextract[@D](x : uint[N] $post @D, xb : list[bool[N] $post @D]) {
	let mut s = xb[length(xb) - 1] as uint[N]; // variable s is mutable
	for i in 0 .. length(xb) - 1 { s = 2 * s + xb[length(xb) - i - 2] as uint[N]; };
	assert_zero(x - s);
}
fn bitextract[@D](x : uint[N] $post @D, fbw : uint $pre @public) -> list[bool[N] $post @D] {
	let xb_pre = bitextract_pre(x as $pre, fbw);
	let xb = for i in 0 .. length(xb_pre) { wire { xb_pre[i] } };
	if (@prover <= @D) { check_bitextract(x, xb); };
	xb
}
fn less_than[@D1, @D2, @D](x : uint[N] $post @D1, y : uint[N] $post @D2, fbw : uint $pre @public) -> bool[N] $post @D
	where @D1 <= @D, @D2 <= @D {
	let xb = bitextract(x as @D, fbw);
	let yb = bitextract(y as @D, fbw);
	// lexicographic comparison of lists of bits xb and yb omitted
}
fn main() {
	let fbw : uint $pre @public = get_public("fixed bit width");
	let z : uint[N] $post @verifier = wire { get_instance("z") };
	let x : uint[N] $post @prover = wire { get_witness("x") };
	let y = wire { z as $pre as @prover / x as $pre };
	assert_zero(x * y - (z as @prover));
	assert(less_than(x, z, fbw));
	assert(less_than(y, z, fbw));
}
\end{lstlisting}
\caption{\label{language:factor}
A \ZKSC program for verifying that Prover knows a proper factor of a given positive integer}
\end{figure}

Execution of the program starts from the function \lstinline{main}. It first loads a public constant \lstinline+fbw+, determining the size of the inputs handled by the circuit. As next, the two inputs $z$ and $x$ are loaded, with the former being visible to both Prover and Verifier, while the latter is seen by Prover only. Both inputs of the program are made inputs to the arithmetic circuit that the compilation of the program produces, using the \lstinline+wire+ construct. The \emph{stage} \lstinline+$pre+ denotes that the value is only available for the local computations by Prover and Verifier, while \lstinline+$post+ denotes its availability in the circuit. The \emph{domain} \lstinline+@prover+ denotes that only Prover knows this value, while \lstinline+@verifier+ denotes that Verifier also has knowledge of it, and \lstinline+@public+ means that this value is known at compile time. Note that the domain only informs about the availability of the value in \emph{local} computations.

The ZKP technology will interpret the operations of the arithmetic circuit over some finite field, and its structure may be important in specifying the relation. We have found that hiding the size of the field from the programmer is not a sensible choice, as too much depends on it. However, the code can be polymorphic over that size. Hence the \emph{data type} \lstinline+uint[N]+ denotes unsigned integers modulo an integer \lstinline+N+ (a compile-time parameter). Unbounded integers, available only in \lstinline+$pre+ stage, have the data type \lstinline+uint+.

In the next line, Prover \emph{expands the witness}. In order to show that $x$ divides~$z$, Prover has to come up with a value~$y$ satisfying $x\cdot y=z$. The division operation is only available for local computations. Hence, in the program, the values \lstinline+z+ and \lstinline+x+ are turned back into values available for local computations using the \lstinline+as $pre+ operation, and divided; the result of the division is turned to another input of the arithmetic circuit. The domain of \lstinline+z as $pre+ is still \lstinline+@verifier+, hence it is cast up to \lstinline+@prover+, because the division operator expects its arguments to have equal types. The type of~\lstinline+y+ is automatically inferred as \lstinline+uint[N] $post @prover+. The types of \lstinline+fbw+, \lstinline+z+, and \lstinline+x+ could have been inferred automatically, too.

The next line in the function \lstinline+main+ specifies the check that the expression \lstinline+x*y-z+ is evaluated to zero, i.e., the product of \lstinline+x+ and \lstinline+y+ indeed equals \lstinline+z+. Multiplication and subtraction operations are available in \lstinline+$post+ stage (both expect arguments to have equal types), likewise is the check that a number is zero. Even though the previous line has stated that Prover should compute $y$ as $z/x$, Verifier cannot trust that it was computed like this, hence this check is necessary. The computations done in the circuit are trusted by Verifier.

The last two lines in \lstinline+main+ check that $x<z$ and $y<z$. The comparison is made by obtaining the bitwise representations (of width \lstinline+fbw+) of both arguments, which can be straightforwardly compared. We see that the function \lstinline+less_than+ is polymorphic in its argument and result domains. The size \lstinline+N+ also matters for booleans, because they are represented as integers when translated to the arithmetic circuit.

The bit representations are computed by the function \lstinline+bitextract+. Well, they are actually computed in the stage \lstinline+$pre+ by the function \lstinline+bitextract_pre+, making use of operations not available in the arithmetic circuit. The result is then added as inputs to the circuit by \lstinline+bitextract+ that also checks its correctness if needed. The correctness check is necessary only if the argument and the result are in the domain \lstinline+@prover+. In the case of domains of lower privacy, Verifier can check correctness of the result directly. Representing some integer as a sequence of bits is a typical instance/witness expansion.

As this example shows, \ZKSC (intentionally) enables interleaving circuit computation with local computation and specifying both of them within one language. Although one cannot be sure that computations outside the circuit are performed exactly like in the code, specifying the intended behavior using the same notation is good for readability and reduces the amount of code (as \ZKSC allows stage polymorphism). As mingling of circuit computation and local computation can be arbitrarily complex, attempts to keep the corresponding pieces of code separate lead to big difficulties.



\newcommand{\mebbep}[1]{\mebbe{(#1)}} 
\newcommand{\mebbe}[1]{{#1}^{?}}
\newcommand{\mnterm}[1]{\mebbe{\nterm{#1}}}
\newcommand{\mterm}[1]{\mebbe{\term{#1}}}
\newcommand{\alt}{\mid}
\newcommand{\term}[1]{\texttt{#1}}
\newcommand{\nterm}[1]{\textit{#1\/}}
\newcommand{\pp}{\mathrel{::=}}
\newcommand{\lizt}[1]{{#1}^{\star}}
\newcommand{\csl}[1]{{#1}^{\mbox{\tiny csl}}}
\newcommand{\inpars}[1]{\mbox{\lstinline{(}}#1\mbox{\lstinline{)}}}
\newcommand{\inbracks}[1]{\mbox{\lstinline{[}}#1\mbox{\lstinline{]}}}
\newcommand{\inbraces}[1]{\mbox{\lstinline!\{!}\ #1\ \mbox{\lstinline!\}!}}
\newcommand{\prestg}{\mbox{\lstinline+$pre+}}
\newcommand{\poststg}{\mbox{\lstinline+$post+}}
\newcommand{\publicdom}{\mbox{\lstinline+@public+}}
\newcommand{\verifierdom}{\mbox{\lstinline+@verifier+}}
\newcommand{\proverdom}{\mbox{\lstinline+@prover+}}
\newcommand{\unitty}{\mbox{\lstinline+()+}}
\newcommand{\listty}[1]{\mbox{\lstinline+list+}\inbracks{#1}}
\newcommand{\qualty}[3]{#1\ #2\ #3}
\newcommand{\qual}[2]{#1\ #2}
\newcommand{\ofty}[2]{#1\ \mbox{\lstinline+:+}\ #2}
\newcommand{\uintty}{\mbox{\lstinline+uint+}}
\newcommand{\boolty}{\mbox{\lstinline+bool+}}
\newcommand{\uintmodty}[1]{\uintty\inbracks{#1}}
\newcommand{\boolmodty}[1]{\boolty\inbracks{#1}}
\newcommand{\letexpr}[2]{\mbox{\lstinline+let+}\ #1\ \mbox{\lstinline+=+}\ #2}
\newcommand{\ifexpr}[3]{\mbox{\lstinline+if+}\ #1\ \inbraces{#2}\ \mbox{\lstinline+else+}\ \inbraces{#3}}
\newcommand{\forexpr}[4]{\mbox{\lstinline+for+}\ #1\ \mbox{\lstinline+in+}\ #2\ \mbox{\lstinline+..+}\ #3\ \inbraces{#4}}
\newcommand{\wireexpr}[1]{\mbox{\lstinline+wire+}\ \inbraces{#1}}
\newcommand{\assignexpr}[2]{#1\ \mbox{\lstinline+=+}\ #2}
\newcommand{\loadexpr}[2]{#1\inbracks{#2}}
\newcommand{\opexpr}[3]{#1\ #2\ #3}
\newcommand{\castexpr}[2]{#1\ \mbox{\lstinline+as+}\ #2}

\section{The Syntax of \ZKSC}\label{sec:syntax}

Figure~\ref{syntax:ebnf} describes the syntax of a subset of \ZKSC that contains its most important features via the extended Backus-Naur forms (EBNF). We use the following conventions:
\begin{itemize}
\item The l.h.s. and r.h.s. of productions are separated by the symbol $\pp$;
\item Non-terminals and terminals are written in italic and typewriter font, respectively;
\item Zero or more repetitions of a term~$t$ is denoted by $\lizt{t\/}$;
\item An optional occurrence of a term~$t$ is denoted by $\mebbe{t}$.
\end{itemize}

\begin{figure}[tbp]
\figsize\newcommand{\newitem}[1]{\multicolumn{3}{@{}l}{\bullet\hspace{1em}\mbox{#1}}}
\[\setlength{\arraycolsep}{4pt}\begin{array}{@{}ll@{}}
\begin{array}[t]{lcl}
\newitem{Top-level structure of \ZKSC programs:}\\
&&\\
\nterm{prog}
&\pp&\lizt{\nterm{fundef}}\\ 
\nterm{fundef}
&\pp&\nterm{sig}\ \inbraces{\nterm{ex}}\\
\nterm{sig}
&\pp&\nterm{name}\ \inpars{\lizt{\nterm{param}}}\ \nterm{rettype}\\
\nterm{param}
&\pp&\ofty{\nterm{name}}{\nterm{type}}\\ 
\nterm{rettype}
&\pp&\mbox{\lstinline{->}}\ \nterm{type}\\
&&\\
\newitem{Types:}\\
&&\\
\nterm{type}
&\pp&\nterm{qualtype}\ \mnterm{rettype}\\
\nterm{qualtype}
&\pp&\nterm{datatype}\ \mnterm{stage}\ \mnterm{domain}\\
\nterm{datatype}
&\pp&\nterm{valtype}
\alt\unitty
\alt\listty{\nterm{type}}\\
\nterm{valtype}
&\pp&\uintty
\alt\uintmodty{\nterm{mod}}
\alt\boolty
\alt\boolmodty{\nterm{mod}}\\
\nterm{mod}
&\pp&\nterm{uintlit}\\
\nterm{stage}
&\pp&\prestg
\alt\poststg\\
\nterm{domain}
&\pp&\publicdom
\alt\verifierdom
\alt\proverdom
\end{array}&\begin{array}[t]{lcl@{}}
\newitem{Expressions and statements:}\\
&&\\
\nterm{seqex}
&\pp&\nterm{ex}\ \mbox{\lstinline{;}}\ \nterm{ex}\\
\nterm{vdex}
&\pp&\nterm{vardef}\ \mbox{\lstinline{;}}\ \nterm{ex}\\
\nterm{vardef}
&\pp&\letexpr{\mebbe{\mbox{\lstinline{mut}}}\ \nterm{name}\ \mebbep{\mbox{\lstinline{:}}\ \nterm{type}}}{\nterm{ex}}\\
\nterm{ifex}
&\pp&\mbox{\lstinline{if}}\ \nterm{ex}\ \inbraces{\nterm{ex}}\ \mbox{\lstinline{else}}\ \inbraces{\nterm{ex}}\\ 
\nterm{forex}
&\pp&\forexpr{\nterm{name}}{\nterm{ex}}{\nterm{ex}}{\nterm{ex}}\\
\nterm{wireex}
&\pp&\wireexpr{\nterm{ex}}\\ 
\nterm{assignex}
&\pp&\assignexpr{\nterm{lvalex}}{\nterm{ex}}\\
\nterm{lvalex}
&\pp&\nterm{name}
\alt\nterm{loadex}\\
\nterm{loadex}
&\pp&\loadexpr{\nterm{lvalex}}{\nterm{ex}}\\ 
\nterm{opex}
&\pp&\opexpr{\nterm{ex}}{\nterm{oper}}{\nterm{ex}}\\ 
\nterm{callex}
&\pp&\nterm{name}\ \inpars{\lizt{\nterm{ex}}}\\ 
\nterm{castex}
&\pp&\castexpr{\nterm{ex}}{\nterm{casttype}}\\ 
\nterm{casttype}
&\pp&\nterm{qualtype}
\alt\nterm{stage}
\alt\nterm{domain}
\end{array}\\
\multicolumn{2}{c}{
\begin{array}{lcl@{}}
&&\\
\nterm{ex}
&\pp&\nterm{seqex}
\alt\nterm{vdex}
\alt\nterm{ifex}
\alt\nterm{forex}
\alt\nterm{wireex}
\alt\nterm{assignex}
\alt\nterm{loadex}
\alt\nterm{opex} 
\alt\nterm{callex}
\alt\nterm{castex}
\alt\nterm{name}
\alt\nterm{uintlit}
\alt\nterm{boollit}
\end{array}
}
\end{array}
\]
\caption{\label{syntax:ebnf}
The syntax of a subset of \ZKSC}
\end{figure}

A \ZKSC \emph{program} consists of \emph{function definitions}, including the main function. Every function definition contains its \emph{signature} and \emph{body}, the latter of which is an expression. Although \ZKSC supports type inference, due to which the types of local variables can be omitted, declaring the types of parameters in function definitions is still required.

The type system of \ZKSC supports both parametric and ad-hoc polymorphism. The example in Fig.~\ref{language:factor} showed the use of \emph{type parameters} and \emph{type predicates} in the signatures of functions. An example of ad-hoc polymorphism occurs in the function \lstinline{bitextract}, where the control flow depends on the value of the type parameter \lstinline{@D}. Polymorphism is important in the reuse of \ZKSC code (i.e. libraries), but since it is orthogonal to the ZKP aspects of the language, we leave it out of consideration in Sects.~\ref{sec:syntax}--\ref{sec:compilation}. The polymorphic language can be converted to an internal monomorphic form using standard monomorphization techniques (code duplication, instantiation) since we do not support polymorphic recursion.

The outermost structure of a \emph{type} represents it as a (curried) function type with zero or more argument types, each of which is a \emph{qualified type}. All qualified types in \ZKSC are triples consisting of a \emph{data type}, a \emph{stage} (prefixed with \lstinline{$}) and a \emph{domain} (prefixed with \lstinline{@}); the latter two are called \emph{qualifiers}. Stage and domain are allowed to be omitted; an omitted stage is inferred, whereas an omitted domain is read as the domain \lstinline{@public}.

Currently, the primitive types in \ZKSC are \emph{Booleans}, \emph{(unsigned) integers}, and the unit type~$\unitty$ consisting of a single value. A list is a data structure of linear shape where all elements have equal type; the element type is given to the list type as parameter. The qualifiers of a list type do not necessarily coincide with those of the element type. For example, a list of domain $\publicdom$ can contain elements of domain $\proverdom$, meaning that the shape of the list is known to the compiler while the elements are known to Prover only. 


The \emph{expressions} and \emph{statements} (which are not distinguished in \ZKSC; e.g. the if-then-else construct can be used both as a statement with side effects in the branches, or as a pure expression similar to the ternary conditional operator in C-style languages) are largely self-explanatory. Each expression returns a value (which may be of type~$\unitty$ in particular). Sequential execution ignores the return value of the first expression. A \emph{variable definition} introduces and initializes an immutable or a mutable (if the modifier \lstinline{mut} is present) variable; its type may be omitted, in which case it will be inferred.
%
%
A \emph{for loop} of the form $\forexpr{x}{e_1}{e_2}{e_3}$ introduces a new variable~$x$, executes the loop body $e_3$ for each value of~$x$ in the half-open segment between the values of $e_1$ and~$e_2$, and returns all values of the loop body as a list.
A \emph{wire expression} transforms local values to circuit inputs. 



We have omitted definitions for names and literals as they are intuitive. Most of the given definitions of composed expressions are intuitive, too. The nonterminal $\nterm{opex}$ expands to expressions constructed via binary operator application. Standard operator precedence is assumed which allows some pairs of parentheses to be omitted (it is not made explicit in the grammar). A \emph{cast expression} enables the programmer to convert between different types. One can provide either a complete type expression or just the domain or the stage instead.

\def\vd{\mbox{\ $\vdash$\ }}
\newcommand*{\defeq}{\stackrel{\text{def}}{=}}
\newcommand*{\mapstom}{\stackrel{*}{\mapsto}}
\newcommand*{\eff}{!}
\newcommand{\type}{\uopl{datatype}}
\newcommand{\assert}{\uopl{assert}}
\newcommand{\mut}{\uopl{mut}}
\newcommand{\var}{\uopl{var}}
\newcommand{\allpre}{\uopl{allpre}}
\newcommand{\gen}[1]{\left\langle #1\right\rangle}
\newcommand{\getexpr}[2]{\mbox{\lstinline{get}}_{#1}\inpars{#2}}
\newcommand{\addexpr}[2]{\opexpr{#1}{\mbox{\lstinline{+}}}{#2}}
\newcommand{\assertexpr}[1]{\mbox{\lstinline{assert}}\inpars{#1}}
\newcommand{\stmtcomp}[2]{\opexpr{#1}{\mbox{\lstinline{;} }}{#2}}
\newcommand{\singleton}{\mbox{\textoneoldstyle}}
\newcommand{\lookup}[2]{#1\left(#2\right)}
\newcommand{\vars}{\uopl{vars}}
\newcommand{\vals}{\uopl{vals}}

\section{Static Semantics}\label{sec:typesystem}

The security guarantees of \ZKSC are established by its type system. Types are checked (and inferred in certain cases) during compile time, hence the type system is part of the static semantics of \ZKSC.
In this section, we describe a static semantics that traces also effects (e.g., assertions and mutable variable updates) that expression evaluation can cause. 
The type system makes sure that the program can be translated into the circuit, and into Prover's and Verifier's local computations.

In Sec.~\ref{sec:syntax}, we explained our non-treatment of type parameters. For analogous reasons, we skip function definitions and user-defined function calls in our treatment; a few most important or representative built-in functions are considered.

Assertions in the type rules are of the form $\Gamma \vd e : t \eff D$. Here, $\Gamma$~is a type environment, $e$~is an expression, $t$~is a qualified type, and $D$~is an upward closed set of domains, assuming the ordering $\publicdom\subtype\verifierdom\subtype\proverdom$ of growing privacy. A type assertion states that under the constraints imposed by~$\Gamma$, the expression~$e$ has type~$t$ and running it can cause effects in domains belonging to the set~$D$. For example, an assignment to a mutable variable whose domain is $\proverdom$ causes an effect in domain $\proverdom$. In discussions, we will sometimes omit effects from type assertions if the effects are not important.

A type environment is a finite association list consisting of the following kinds of components:
\begin{itemize}
\item Variable typings written in the form $x:q$, where $x$~is a variable and $q$ is a qualified type;
\item Mutability statements in the form $\mut x:b$ where $x$ is a variable and $b$~is its mutability status ($0$ or $1$).
\end{itemize}
We write $(z:w),\Gamma$ to denote a new type environment containing the association $(z:w)$ followed by all associations in~$\Gamma$.
We also write $\lookup{\Gamma}{z}$ for lookup of~$z$ in the type environment~$\Gamma$, i.e., if 
$\Gamma=((z_1:w_1),\ldots,(z_n:w_n))$ then $\lookup{\Gamma}{z}=w_i$ where $i$ is the least index such that $z_i=z$. Note that $z$ is a variable possibly equipped with $\mut$, and $w$ is either a qualified type or a mutability status, respectively. If the associations of variables to qualified types of $\Gamma$ are $(x_1:q_1),\ldots,(x_n:q_n)$, in this order, then we denote $\vars\Gamma=(x_1,\ldots,x_n)$.

\begin{figure*}\figsize\renewcommand{\arraystretch}{2.2}
\begin{tabular}{@{}c@{}}
\begin{prooftree}
    \infer0{\Gamma \vd \epsilon : \qualty{\unitty}{\prestg}{\publicdom} \eff \varnothing}
\end{prooftree}\quad
\begin{prooftree}
\hypo{n\in\mathbb{N}}
\infer1{\Gamma \vd \overline{n} : \qualty{\uintmodty{\mbox{\lstinline{N}}}}{s}{d} \eff \gen{s}}
\end{prooftree}\quad
\begin{prooftree}
\hypo{b\in\mathbb{B}}
\infer1{\Gamma \vd \overline{b} : \qualty{\boolmodty{\mbox{\lstinline{N}}}}{s}{d} \eff \gen{s}}
\end{prooftree}\quad
\begin{prooftree}
    \hypo{\lookup{\Gamma}{x}=(\qualty{t}{s}{d})}
    \infer1{\Gamma \vd x : \qualty{t}{s}{d} \eff \varnothing}
\end{prooftree}\\
\begin{prooftree}
  \hypo{\Gamma \vd e_1 : \qualty{\uintmodty{\mbox{\lstinline{N}}}}{s}{d} \eff D_1}
  \hypo{\Gamma \vd e_2 : \qualty{\uintmodty{\mbox{\lstinline{N}}}}{s}{d} \eff D_2}
  \infer2{\Gamma \vd \addexpr{e_1}{e_2} : \qualty{\uintmodty{\mbox{\lstinline+N+}}}{s}{d} \eff \gen{s}\cup D_1\cup D_2}
\end{prooftree}\quad
\begin{prooftree}
    \hypo{\Gamma \vd e : \qualty{\boolmodty{\mbox{\lstinline{N}}}}{\poststg}{d} \eff D}
    \infer1{\Gamma \vd \assertexpr{e} : \qualty{\unitty}{\prestg}{\publicdom} \eff \gen{\publicdom}}
\end{prooftree}\\
\begin{prooftree}
    \hypo{\allpre_{d'}(\qualty{t}{s}{d})}
    \infer1{\Gamma \vd (\ofty{\getexpr{d'}{k}\!}{\!\qualty{t}{\!s\!}{d}})\!:\!\qualty{t}{\!s\!}{d}\eff\!\varnothing}
\end{prooftree}\quad
\begin{prooftree}
    \hypo{\Gamma \vd e_1 : \qualty{\boolmodty{\mbox{\lstinline{N}}}}{\prestg}{d'} \eff D_1}
    \hypo{\Gamma \vd e_i : \qualty{t}{s}{d} \eff D_i\ (i=2,3)}
    \hypo{\gen{d'}\supseteq\gen{s}\cup\gen{d} \cup D_2 \cup D_3}
    \infer3{\Gamma \vd \ifexpr{e_1}{e_2}{e_3} : \qualty{t}{s}{d} \eff D_1 \cup D_2 \cup D_3}
\end{prooftree}\\
\begin{prooftree}
    \hypo{\Gamma \vd e_i : \qualty{\uintty}{\prestg}{d'} \eff D_i\ (i=1,2)}
    \hypo{(x : \qualty{\uintty}{\prestg}{d'}), \Gamma \vd e_3 : \qualty{t}{s}{d} \eff D_3}
    \hypo{\gen{d'} \supseteq \gen{s}\cup\gen{d}\cup D_3}
    \infer3{\Gamma \vd \forexpr{x}{e_1}{e_2}{e_3} : \qualty{\listty{\qualty{t}{s}{d}}}{\prestg}{d'} \eff D_1 \cup D_2 \cup D_3}
\end{prooftree}\\
\begin{prooftree}
    \hypo{\Gamma \vd e : \qualty{t}{\prestg}{d} \eff D}
    \hypo{t\in\set{\uintmodty{\mbox{\lstinline{N}}},\boolmodty{\mbox{\lstinline{N}}}}}
    \infer2{\Gamma \vd \wireexpr{e} : \qualty{t}{\poststg}{d}\eff \gen{\publicdom}}
\end{prooftree}\qquad
\begin{prooftree}
    \hypo{\Gamma \vd e : \qualty{t}{s}{d} \eff D}
    \hypo{s \subtype s'}
    \hypo{d \subtype d'}
    \hypo{\gen{d'}\supseteq\gen{t}}
    \infer4{\Gamma \vd \castexpr{e}{\qualty{t}{s'}{d'}} : \qualty{t}{s'}{d'} \eff D}
\end{prooftree}\\
\begin{prooftree}
    \hypo{\Gamma \vd e : \qualty{t}{s}{d} \eff D}
    \hypo{\Gamma \vd l : \qualty{t}{s}{d} \eff D'}
    \hypo{\lookup{\Gamma}{\mut(\var l)}=1}
    \infer3{\Gamma \vd \assignexpr{l}{e} : \qualty{\unitty}{\prestg}{\publicdom} \eff \gen{s}\cup\gen{d} \cup D \cup D'}
\end{prooftree}\quad
\begin{prooftree}
    \hypo{\Gamma \vd l : \qualty{\listty{\qualty{t}{s}{d}}}{s'}{d'} \eff D'}
    \hypo{\Gamma \vd e : \qualty{\uintty}{s'}{d'} \eff D}
    \infer2{\Gamma \vd \loadexpr{l}{e} : \qualty{t}{s}{d} \eff D \cup D'}
\end{prooftree}\\
\begin{prooftree}
    \hypo{\Gamma \vd e_1 : \qualty{t_1}{s_1}{d_1} \eff D_1}
    \hypo{(\mut x : b), (x : \qualty{t_1}{s_1}{d_1}), \Gamma \vd e_2 : \qualty{t_2}{s_2}{d_2} \eff D_2}
    \infer2{\Gamma \vd \stmtcomp{\letexpr{\mbox{\lstinline{mut}}^b\ x}{e_1}}{e_2} : \qualty{t_2}{s_2}{d_2} \eff \gen{d_1} \cup D_1 \cup D_2}
\end{prooftree}\quad
\begin{prooftree}
    \hypo{\Gamma \vd e_1 : \qualty{t_1}{s_1}{d_1} \eff D_1}
    \hypo{\Gamma \vd e_2 : \qualty{t_2}{s_2}{d_2} \eff D_2}
    \infer2{\Gamma \vd \stmtcomp{e_1}{e_2} 
    : \qualty{t_2}{s_2}{d_2} \eff D_1 \cup D_2}
\end{prooftree}
\end{tabular}
\caption{\label{typesystem:exprstmt}
Typing rules of expressions and statements of a subset of \ZKSC without parametric polymorphism}
\end{figure*}

%
The static semantics is presented in Fig.~\ref{typesystem:exprstmt}. The data type derivation parts of the rules are standard. Hence we mainly comment on stages, domains and effects. There exist four upward closed domain sets linearly ordered by inclusion: $\varnothing\subset\gen{\proverdom}\subset\gen{\verifierdom}\subset\gen{\publicdom}$, where $\gen{d}$ denotes the set consisting of domain~$d$ and all larger domains. Thus the union of upward closed sets always equals the largest set in the union. We also order stages as $\poststg\subtype\prestg$, reflecting that data computed in the circuit are also computed locally in the corresponding domain in order to be ready to provide expanded instances/witnesses to the circuit. \ZKSC requires conversion to supertype to be made explicit using the \lstinline{as} keyword.

The first rule in Fig.~\ref{typesystem:exprstmt} handles missing expression~$\epsilon$; it is needed for the case where the last expression of a sequential execution is absent. Integer and Boolean literals are denoted by overlined constants. We show only rules for $\uintmodty{\mbox{\lstinline{N}}}$ and $\boolmodty{\mbox{\lstinline{N}}}$; rules for $\uintty$ and $\boolty$ in stage $\prestg$ are similar. Literals can be typed with any stage and domain (i.e., no type cast is required). The effect of literals depends on their actual stage: In stage $\prestg$ they do not have any effect, while in stage $\poststg$, they have public effect since there they contribute to constructing of the circuit. In general, any operation in the circuit is considered a public effect. To specify the possible effects concisely in the rules, we extend the $\gen{\cdot}$ notation to stages by $\gen{\prestg}=\varnothing$, $\gen{\poststg}=\gen{\publicdom}$, and also to data types by $\gen{t}=\varnothing$ if $t$~{}is a primitive type and $\gen{t}=\gen{t'}\cup\gen{s'}\cup\gen{d'}$ if $t=\listty{\qualty{t'}{s'}{d'}}$. (This way, the result of the $\gen{\cdot}$ operation is always an upward closed set of domains.)

Types of variables are read directly from the type environment without an effect. The rule for addition allows this operation to be performed in any domain and stage, but the domain of the arguments and the result must be the same and similarly for stages. The rule for \lstinline+assert+ establishes the result data type to be the unit type. The unit type, as well as list types, is always in the stage $\prestg$ since the circuit does not deal with values of these types. 

In any domain, the slice of the program available to local computation of that domain can read input data. The \ZKSC functions for that are \lstinline{get_public}, \lstinline{get_instance} and \lstinline{get_witness} which read the public constants, the instance, and the witness, respectively (all non-expanded). Here we denote these functions uniformly by $\mbox{\lstinline{get}}_d$ where $d=\publicdom$, $d=\verifierdom$ and $d=\proverdom$, respectively, so the rule for $\mbox{\lstinline+get+}_d$ captures all three cases. An argument of the function $\mbox{\lstinline{get}}_d$ is a key of a dictionary and the function returns the corresponding value in the dictionary. String literals are used as keys.

The type of an expression of the form $\getexpr{d'}{k}$ cannot be derived from its constituents as the dictionary need not be available for the type checker. Therefore, the type rule for such expressions requires them to be explicitly typed. The only restriction we impose on the result type is $\allpre_{d'}(\qualty{t}{s}{d})$ which is to denote that all types occurring in $\qualty{t}{s}{d}$ must be qualified with $\qual{\prestg}{d'}$. More precisely, $\allpre_{d'}(\qualty{t}{s}{d})$ is true iff $s=\prestg$, $d=d'$ and in the case $t=\listty{q'}$ also $\allpre_{d'}(q')$. 

The rules for conditional expressions and loops restrict the stage of guards to $\prestg$, meaning that branchings happen only in local computations. For conditional expressions, the restriction $\gen{d'}\supseteq\gen{s}\cup\gen{d}\cup D_2\cup D_3$ serves the following purposes:
\begin{itemize}
\item $\gen{d'}\supseteq\gen{d}$ (i.e., $d'\subtype d$) makes the output of the conditional expression to have at least as high privacy level as the guard (the \emph{no-read-up} property);
\item $\gen{d'}\supseteq D_2\cup D_3$ disallows information flows through side effects of the branches (the \emph{no-write-down} property);
\item $\gen{d'}\supseteq\gen{s}$ ensures that, if the expression is computed by the circuit, then the branch that must be taken is known at the compile-time (also an instance of \emph{no-write-down}, as all computations by the circuit have visible side-effects).
\end{itemize}
The restriction $\gen{d'}\supseteq\gen{s}\cup\gen{d}\cup D_3$ in the loop rule provides similar guarantees. For example, if the body of a loop performs assertions then the loop bounds must be in the domain $\publicdom$. 

%
A type cast can extend the type of a given expression but not reduce it. 
Moreover, the condition $\gen{d'}\supseteq\gen{t}$ of the cast rule guarantees a type invariant that prevents list elements from revealing information about the list structure to domains of lower privacy (via lookups). This invariant is formally established in Definition~\ref{typesystem:def} and Theorem~\ref{typesystem:thm}.

The l.h.s. of an assignment can be complex, consisting of a variable followed by an index vector. Let $\var l$ for any L-value expression~$l$ denote the variable whose mutability permits assignment to this L-value, i.e., $\var l=l$ if $l$~is a variable, and $\var l=\var l'$ if $l$ is $\mbox{$l'$\inbracks{$e$}}$. The assignment rule states that the return value of an assignment expression is of the unit data type and the effect of an assignment belongs to the domain $\publicdom$ if the operands are in stage $\poststg$, otherwise the effect belongs to the same domain~$d$ as the operands. The type structure and effects of the contents of the l.h.s. is the topic of Lemma~\ref{typesystem:lhslemma}.

In the list element access rule, the index must be in the same domain as the list structure since accessing an element via its index may reveal information about the length of the list. Reading does not introduce new effects.

In the rule for let statements, $\mbox{\lstinline{mut}}^1$ stands for $\mbox{\lstinline{mut}}$ keyword and $\mbox{\lstinline{mut}}^0$ means empty string, so the rule applies to both immutable and mutable variable definitions. Mutability information is reflected in the type environment when type checking the expression~$e_2$. We consider variable definition to be effectful, whence~$\gen{d_1}$ is added into the set of effects. The last rule similarly handles statement sequences but applies to the case where the first statement does not define new variables.

%
The following definition~\ref{typesystem:def} and theorem~\ref{typesystem:thm} are essential. Proof of the theorem goes by induction on the structure of~$e$. The details are given in Appendix~\ref{typesystem-proofs}.

\begin{definition}\label{typesystem:def}
Let $q=\qualty{t}{s}{d}$ be a qualified type. We call $q$ \emph{well-structured} if either $t$~{}is a primitive type, whereby $s=\prestg$ in the case of $t=\unitty$; or $s=\prestg$ and  $t=\listty{\qualty{t'}{s'}{d'}}$ such that $\gen{d}\supseteq\gen{s'}\cup\gen{d'}$ and $\qualty{t'}{s'}{d'}$ is well-structured. Call a type environment $\Gamma$ \emph{well-structured} if all qualified types occurring in it are well-structured.
\end{definition}

\begin{theorem}\label{typesystem:thm}
If\/ $\Gamma$ is well-structured and $\Gamma\vd e : q\eff D$, then $q$ is well-structured.
\end{theorem}

Proofs of many theorems of the following sections rely on Lemma~\ref{typesystem:lhslemma}, which is itself proved by induction on the length of index vector:

\begin{lemma}\label{typesystem:lhslemma}
Let $e=\loadexpr{\loadexpr{\loadexpr{x}{y_1}}{y_2}\ldots}{y_n}$ where $x$ is a variable.
Let $\Gamma\vd e:q\eff D$, where $q=(\qualty{t}{s}{d})$ and $\Gamma$ is well-structured.
Then there exist domains $d_1,\ldots,d_n$ and upward closed domain sets $D_1,\ldots,D_n$ such that $\Gamma\vd y_i:\qualty{\uintty}{\prestg}{d_i}\eff D_i$ for each $i=1,\ldots,n$ and
\[
\Gamma\vd x:\qualty{\listty{\ldots\listty{\qualty{\listty{\qualty{\listty{q}}{\prestg}{d_n}}}{\prestg}{d_{n-1}}}\ldots}}{\prestg}{d_1}\eff\varnothing\mbox{,}
\]
whereby $d_1\subtype\ldots\subtype d_{n-1}\subtype d_n\subtype d$ and $D=D_1\cup\ldots\cup D_n$.
\end{lemma}

The details of the proof are given in Appendix~\ref{typesystem-proofs}.

\newcommand{\sem}[1]{{\left\llbracket{#1}\right\rrbracket}}
\newcommand{\semd}[1]{\sem{#1}_d}
\newcommand{\m}{T}
\newcommand{\semm}[1]{\sem{#1}^\m}
\newcommand{\tru}{\mbox{\bfseries tt}}
\newcommand{\fls}{\mbox{\bfseries ff}}
\newcommand{\return}{\uopl{pure}}
\newcommand{\throw}{\uopl{throw}}
\newcommand{\modify}{\uopl{modify}}
\newcommand{\update}{\uopl{upd}}
\newcommand{\guard}{\uopl{guard}}
\newcommand{\forM}{\uopl{for}}
\newcommand{\monadic}[1]{\hat{\ifthenelse{\equal{#1}{i}}{\imath}{#1}}}
\newcommand{\lam}[2]{\leftthreetimes #1.\,#2}
\newcommand{\mcomp}[1]{\uopl{do}\{#1\}}
\newcommand{\mcompv}[1]{\begin{array}[t]{@{}l@{}l@{}}\uopl{do}\{&\begin{array}[t]{@{}l@{}}#1\end{array}\\\}&\end{array}}
\newcommand{\hstop}{;\;}
\newcommand{\ap}{\,}
\newcommand{\Env}{\mathbf{Env}}
\newcommand{\In}{\mathbf{In}}
\newcommand{\Out}{\mathbf{Out}}
\let\omicron=o
\newcommand{\allpure}{\uopl{allpure}}
\newcommand{\head}{\uopl{head}}
\newcommand{\tail}{\uopl{tail}}
\newcommand{\new}{\uopl{new}}
\newcommand{\coin}[2]{\mathbin{\mbox{$\sim\hspace{-0.9em}\substack{#1\\#2}\hspace{0.15em}$}}}
\newcommand{\coincirc}[1]{\coin{#1}{{}}}
\newcommand{\coinpure}[1]{\coin{\phantom{q}}{#1}}
\newcommand{\branching}[1]{\set{\renewcommand{\arraystretch}{1}\begin{array}{@{}l@{\;}l@{}}#1\end{array}}}

\section{Dynamic Semantics}\label{sec:semantics}

We present dynamic semantics of our language in a denotational style. In fact, the overall setting assumes four different dynamic semantics, loosely corresponding to the views of three domains and the circuit. We call the three semantics corresponding to the domains \emph{local} since they describe what is computed by different parties locally. For example, Prover's view contains all computations that are performed as $\qual{\prestg}{\proverdom}$, Verifier's view contains all computations performed as $\qual{\prestg}{\verifierdom}$, etc. Values of all expressions and statements in Prover's domain are unknown from Verifier's point of view; we denote the unknown value by~$\top$. Likewise, the local semantics for the public domain evaluates all expressions and statements in the higher domains to~$\top$. Computations in the public domain are performed by the compiler.

As $\prestg$ and $\proverdom$ are the topmost elements of the stage and domain hierarchy, Prover's view encompasses the whole program. Basically, this view describes \emph{what should actually happen}, nevertheless ignoring the special way computation is performed in the circuit due to limited supply of operations.

The circuit semantics describes computations performed as $\qual{\poststg}{\mbox{\lstinline{@D}}}$ for any \lstinline{@D}, and also everything in $\qual{\prestg}{\publicdom}$. Although not computed by the circuit, values in $\qual{\prestg}{\publicdom}$ are inevitably needed in performing branching computations as conditions of if expressions and loop bounds belong to stage $\prestg$. In reality, the compiler unrolls conditionals and loops for the circuit as the latter has no means for branching.

We define all three local semantics via a common set of equations. The differences arise from domain inclusion conditions that can be either true or false depending on the party and can introduce~$\top$. The type system ensures that computing the program parts of the lower domains is not impeded by not knowing the values of the higher domains.

The main notation and types are summarized in Fig.~\ref{semantics:types}. The set of values that our semantics can produce consists of non-negative integers, booleans $\tru$ and~$\fls$, the only value of the unit type, and finite sequences of (possibly unknown) values. We need two sets of values, $U$ and $V$, both defined recursively as the least fixpoint satisfying the corresponding equation. Values in the set~$U$ (which we sometimes call the ``core'' values) are built without making use of~$\top$, while the definition of~$V$ involves also $\top$. Core values are used only for representing the input of local computation; we can be sure that each party can fully read its input whence $\top$ will never occur there. By $\mathds{1}+A$, we denote the disjoint sum of a singleton set and set~$A$; in the case of~$M\ap V$, we assume $\mathds{1}=\set{\top}$. We refer to the elements of the main summand (i.e., not $\mathds{1}$) as \emph{pure}.

The local semantics of an expression takes a value environment and a triple of input dictionaries (one for each domain) as arguments, and normally produces a triple containing the value of the expression, an updated environment and a pair of finite sequences of values to be delivered to the circuit (one sequence for each of $\proverdom$ and $\verifierdom$; the circuit does not take public input). In exceptional cases, the semantics can fail, which is shown by the addend~$\mathds{1}$ in the equation for $C_d\ap A$ and means a runtime error. For simplicity, we ignore runtime errors other than assertion failures in the semantics; in practice, all other runtime errors are considered semantically equivalent to a failed assertion. 


\begin{figure}\figsize
\renewcommand{\arraycolsep}{2pt}
\begin{tabular}{@{}l@{}}
$\begin{array}{lcl}
\mathbb{N}&=&\set{0,1,2,\ldots}\\
\mathbb{B}&=&\set{\tru,\fls}\\
()&=&\set{\singleton}
\end{array}$\qquad
$\begin{array}{lcl}
X&&\mbox{the set of variables}\\
K&&\mbox{the set of input keys}\\
U&=&\mathbb{N}\cup\mathbb{B}\cup()\cup U^*\quad\mbox{the set of core values}\\
V&=&\mathbb{N}\cup\mathbb{B}\cup()\cup(M\ap V)^*\quad\mbox{the set of values}
\end{array}$\\[4.5ex]
$\begin{array}{lcl}
M\ap A&=&\mathds{1}+A\\
C_d\ap A&=&\Env\to\In^3\to\mathds{1}+A\times\Env\times\Out^2\\
C\ap A&=&\Env\to\In^3\to\Out^2\to\mathds{1}+A\times\Env\times\Out^2
\end{array}$\\[4.5ex]
$\begin{array}{lcl@{\quad}l}
\Env&=&(X\times M\ap V)^+&\mbox{value environments}\\
\In&=&K\to U&\mbox{inputs of local computation}\\
\Out&=&(M\ap(\mathbb{N}\cup\mathbb{B}))^*&\mbox{streams of values delivered to the circuit}
\end{array}$\\[4.5ex]
$\begin{array}{lcl@{\quad}l}
\semd{e}&:&C_d\ap(M\ap V)&\mbox{the local semantics of expression~$e$ in domain~$d$}\\
\sem{e}&:&C\ap(M\ap V)&\mbox{the circuit semantics of expression~$e$}
\end{array}$
\end{tabular}
\caption{\label{semantics:types}
Types of semantic objects}
\end{figure}

A value environment is a finite association list. It is operated as a stack. We use the following notation for lookup and update of variable~$x$ in any environment~$\gamma=((x_1,a_1),\ldots,(x_n,a_n))$, where $i$ is the least index such that $x_i=x$:
\[
\lookup{\gamma}{x}=a_i\mbox{,}\quad
[x\mapsto a]\gamma=((x_1,a_1),\ldots,(x_{i-1},a_{i-1}),(x_i,a),(x_{i+1},a_{i+1}),\ldots,(x_n,a_n))\mbox{.}
\]
Note that an update is performed in the stack element where the variable~$x$ is defined. If a new definition of~$x$ is desired, we instead write $(x,a),\gamma$ which means pushing a new association to the stack. To omit the topmost element of the stack~$\gamma$, we write $\tail\gamma$. If $\gamma=((x_1,a_1),\ldots,(x_n,a_n))$ then we write $\vars\gamma=(x_1,\ldots,x_n)$.

The definition of local semantics is given in Fig.~\ref{semantics:exprstmt}. The notation of syntactic objects coincides with that in the type rules (e.g., $x$ stands for a variable etc.); in addition, $\gamma$, $\phi$ and $\omicron$ denote value environments, inputs and outputs, respectively. To avoid the need to study exceptional cases separately, we use the monad comprehension syntax of the functional programming language Haskell. This notation was first advocated by Wadler~\cite{DBLP:journals/mscs/Wadler92} for succinct description of computations that may involve side effects. We only use the notation for the \emph{maybe monad} $A\mapsto\mathds{1}+A$. For example, the sum of values $\monadic{v}_1,\monadic{v}_2\in M\ap V$ (the static type system ensures that their types are correct, but either value can be unknown) is written as
\[
\mcomp{i_1\gets\monadic{v}_1\hstop i_2\gets\monadic{v}_2\hstop\return(i_1+i_2)}\mbox{.}
\]
Here, the first two clauses define $i_1$ and $i_2$ as pure representatives of $\monadic{v}_1$ and $\monadic{v}_2$, respectively, and the last clause specifies the sum $i_1+i_2$ as the final outcome. The latter is wrapped into a monadic value (i.e., an element of $M\ap V$) by function $\return$. Any of the clauses evaluating to~$\top$ turns the final result $\top$ immediately. In general, monad comprehension can contain any finite number of clauses, all of which except the last one may bind new pure values. Evaluation is strict and progresses from left to right. Note that here and below, we denote monadic values by letters with hat for clarity. 

Since $C_d$ also involves exceptional cases, we use monad comprehension for~$C_d$, too. So we have a two-layer monadic specification of semantics ($M$ is inside and $C_d$ outside). In the outer layer, we use function $\guard$ that on a false condition raises an exception (and jumps out of comprehension) and has no effect otherwise. Due to the two-layer representation, the unknown value~$\top$ causes no exception in the outer layer.

\begin{figure*}\figsize\renewcommand{\arraycolsep}{2pt}\renewcommand{\arraystretch}{1.2}
\[\begin{array}{lcl}
\semd{\epsilon}\gamma\phi&=&\return(\return\singleton,\gamma,\epsilon)\\
\semd{\overline{n}}\gamma\phi&=&\branching{\return(\return n,\gamma,\epsilon)&\mbox{if $\overline{n}$ is in domain~$d$ or lower}\\\return(\top,\gamma,\epsilon)&\mbox{otherwise}}\\
\semd{\overline{b}}\gamma\phi&=&\branching{\return(\return b,\gamma,\epsilon)&\mbox{if $\overline{b}$ is in domain~$d$ or lower}\\\return(\top,\gamma,\epsilon)&\mbox{otherwise}}\\
\semd{x}\gamma\phi&=&\return(\lookup{\gamma}{x},\gamma,\epsilon)\\
\semd{\addexpr{e_1}{e_2}}\gamma_0\phi&=&\mcompv{(\monadic{v}_1,\gamma_1,\omicron_1)\gets\semd{e_1}\gamma_0\phi\hstop\\(\monadic{v}_2,\gamma_2,\omicron_2)\gets\semd{e_2}\gamma_1\phi\hstop\\\return(\mcomp{v_1\gets\monadic{v}_1\hstop v_2\gets\monadic{v}_2\hstop\return(v_1+v_2)},\gamma_2,\omicron_1\omicron_2)}\\
\semd{\assertexpr{e}}\gamma_0\phi&=&\mcomp{(\monadic{v},\gamma_1,\omicron_1)\gets\semd{e}\gamma_0\phi\hstop\guard(\monadic{v}\ne\return\fls)\hstop\return(\return\singleton,\gamma_1,\omicron_1)}\\
\semd{\getexpr{d'}{k}}\gamma\phi&=&\branching{\return(\allpure(\phi_{d'}(k)),\gamma,\epsilon)&\mbox{if $d'\subtype d$}\\\return(\top,\gamma,\epsilon)&\mbox{otherwise}}\\
\semd{\ifexpr{e_1}{e_2}{e_3}}\gamma_0\phi&=&\mcompv{(\monadic{v}_1,\gamma_1,\omicron_1)\gets\semd{e_1}\gamma_0\phi\hstop\\\branching{\mcomp{(\monadic{v}_2,\gamma_2,\omicron_2)\gets\semd{e_2}\gamma_1\phi\hstop\return(\monadic{v}_2,\gamma_2,\omicron_1\omicron_2)}&\mbox{if $\monadic{v}_1=\return\tru$}\\\mcomp{(\monadic{v}_3,\gamma_3,\omicron_3)\gets\semd{e_3}\gamma_1\phi\hstop\return(\monadic{v}_3,\gamma_3,\omicron_1\omicron_3)}&\mbox{if $\monadic{v}_1=\return\fls$}\\\return(\top,\gamma_1,\omicron_1)&\mbox{otherwise}}}\\
\semd{\forexpr{x}{e_1}{e_2}{e_3}}\gamma_0\phi&=&\begin{array}[t]{@{}l@{}}\mcompv{(\monadic{v}_1,\gamma_1,\omicron_1)\gets\semd{e_1}\gamma_0\phi\hstop\\(\monadic{v}_2,\gamma_2,\omicron_2)\gets\semd{e_2}\gamma_1\phi\hstop\\\branching{\uopl{main}(\max(0,i_2-i_1))&\mbox{where $\monadic{v}_j=\return i_j$, $j=1,2$}\\\return(\top,\gamma_2,\omicron_1\omicron_2)&\mbox{if $\monadic{v}_1=\top$ or $\monadic{v}_2=\top$}}}\end{array}\\
\multicolumn{2}{r}{\mbox{where}}&\uopl{main}(n)=\mcompv{(\monadic{a}_1,\gamma'_1,\omicron'_1)\gets\semd{e}((x,\return i_1),\gamma_2)\phi\hstop\\\forall k=2,\ldots,n:(\monadic{a}_k,\gamma'_k,\omicron'_k)\gets\semd{e_3}([x\mapsto\return(i_1\!\!+\!\!k\!\!-\!\!1)]\gamma'_{k-1})\phi\hstop\\\return(\return(\monadic{a}_1,\ldots,\monadic{a}_n),\tail\gamma'_n,\omicron_1\omicron_2\omicron'_1\ldots\omicron'_n)}\\
\semd{\wireexpr{e}}\gamma_0\phi&=&\mcomp{(\monadic{v},\gamma_1,\omicron_1)\gets\semd{e}\gamma_0\phi\hstop\return(\monadic{v},\gamma_1,\lam{d'}{\branching{(\omicron_1)_{d'}\monadic{v}&\mbox{if $e$ is in domain~$d'$}\\(\omicron_1)_{d'}&\mbox{otherwise}}})}\\
\semd{\castexpr{e}{d'}}\gamma_0\phi&=&\branching{\semd{e}\gamma_0\phi&\mbox{if $d'\subtype d$}\\\mcomp{(\monadic{v}_1,\gamma_1,\omicron_1)\gets\semd{e}\gamma_0\phi\hstop\return(\top,\gamma_1,\omicron_1)}&\mbox{otherwise}}\\
\semd{\assignexpr{l}{e}}\gamma_0\phi&=&\begin{array}[t]{@{}l@{}}\mcompv{\monadic{a}\gets\return(\lookup{\gamma_0}{x})\hstop\\\forall k=1,\ldots,n:(\monadic{i}_k,\gamma_k,\omicron_k)\gets\semd{y_k}\gamma_{k-1}\phi\hstop\\(\monadic{v},\gamma',\omicron')\gets\semd{e}\gamma_n\phi\hstop\\\return(\return\singleton,[x\mapsto\update(\monadic{a},\monadic{i}_1\ldots\monadic{i}_n,\monadic{v})]\gamma',\omicron_1\ldots\omicron_n\omicron')}\end{array}\\
\multicolumn{2}{r}{\mbox{where}}&\loadexpr{\loadexpr{\loadexpr{x}{y_1}}{y_2}\ldots}{y_n}=l\\
\semd{\loadexpr{l}{e}}\gamma_0\phi&=&\mcompv{(\monadic{a},\gamma_1,\omicron_1)\gets\semd{l}\gamma_0\phi\hstop\\(\monadic{i},\gamma_2,\omicron_2)\gets\semd{e}\gamma_1\phi\hstop\\\return(\mcomp{a\gets\monadic{a}\hstop i\gets\monadic{i}\hstop a_i},\gamma_2,\omicron_1\omicron_2)}\\
\semd{\stmtcomp{\letexpr{x}{e_1}}{e_2}}\gamma_0\phi&=&\mcompv{(\monadic{v}_1,\gamma_1,\omicron_1)\gets\semd{e_1}\gamma_0\phi\hstop\\(\monadic{v}_2,\gamma_2,\omicron_2)\gets\semd{e_2}((x,\monadic{v}_1),\gamma_1)\phi\hstop\\\return(\monadic{v}_2,\tail\gamma_2,\omicron_1\omicron_2)}\\
\semd{\stmtcomp{e_1}{e_2}}\gamma_0\phi&=&\mcomp{(\monadic{v}_1,\gamma_1,\omicron_1)\gets\semd{e_1}\gamma_0\phi\hstop(\monadic{v}_2,\gamma_2,\omicron_2)\gets\semd{e_2}\gamma_1\phi\hstop\return(\monadic{v}_2,\gamma_2,\omicron_1\omicron_2)}
\end{array}
\]
\caption{\label{semantics:exprstmt}
Local dynamic semantics of expressions and statements}
\end{figure*}

The semantics of~$+$ is included as an example of a built-in operator. Arithmetic is implicitly performed modulo some positive integer if that is required by the type ($\uintmodty{\mbox{\lstinline{N}}}$).

The only case that uses core values is that of $\mbox{\lstinline{get}}$. To transform a core value to a value in $M\ap V$, we use the function $\allpure:U\to M\ap V$ defined as
\[
\allpure v=\branching{\return v&\mbox{if $v$ belongs to a primitive type}\\\return(\allpure v_1,\ldots,\allpure v_n)&\mbox{if $v=(v_1,\ldots,v_n)$}}\mbox{.}
\]

Concerning type casts, we show only the variant with domain cast as only domain matters here.

The assignment case uses an auxiliary function $\update:M\ap V\times(M\ap V)^*\times M\ap V\to M\ap V$ that takes a value that can be a list with $0$ or more dimensions (i.e., a primitive value or a list of primitives or a matrix etc.), an index vector, and a value, and returns a new list where the cell indicated by the index vector has been updated with the given value. More formally,
\[
\update(\monadic{a},\monadic{i}_1\ldots\monadic{i}_n,\monadic{v})=\branching{\monadic{v}&\mbox{if $n=0$}\\\mcomp{a\gets\monadic{a}\hstop i_1\gets\monadic{i}_1\hstop\return([i_1\mapsto\update(a_{i_1},\monadic{i}_2\ldots\monadic{i}_n,\monadic{v})]a)}&\mbox{if $n>0$}}\mbox{.}
\]

%
The circuit semantics is defined mostly analogously; the definition is given in Fig.~\ref{semantics:circuit}. Again, we present only one case of type cast; the other cases are defined in similar lines. The most important difference from local semantics is concerning the \lstinline{wire} construct that reads non-public values from the output streams of local computations. For this reason, the circuit semantics takes a pair of streams as a supplementary argument. Execution of each wire expression in a non-public domain removes the first value from the stream corresponding to the domain of that expression; the updated pair of streams is included in the result. 

\begin{figure*}\figsize\renewcommand{\arraycolsep}{2pt}\renewcommand{\arraystretch}{1.2}
\[
\begin{array}{lcl}
\sem{\epsilon}\gamma\phi\omicron&=&\return(\return\singleton,\gamma,\omicron)\\
\sem{\overline{n}}\gamma\phi\omicron&=&\branching{\return(\return n,\gamma,\omicron)&\mbox{if $\overline{n}$ is in \lstinline+$post+ or \lstinline+@public+}\\\return(\top,\gamma,\omicron)&\mbox{otherwise}}\\
\sem{\overline{b}}\gamma\phi\omicron&=&\branching{\return(\return b,\gamma,\omicron)&\mbox{if $\overline{b}$ is in \lstinline+$post+ or \lstinline+@public+}\\\return(\top,\gamma,\omicron)&\mbox{otherwise}}\\
\sem{x}\gamma\phi\omicron&=&\return(\lookup{\gamma}{x},\gamma,\omicron)\\
\sem{\addexpr{e_1}{e_2}}\gamma_0\phi\omicron_0&=&\mcompv{(\monadic{v}_1,\gamma_1,\omicron_1)\gets\sem{e_1}\gamma_0\phi\omicron_0\hstop\\(\monadic{v}_2,\gamma_2,\omicron_2)\gets\sem{e_2}\gamma_1\phi\omicron_1\hstop\\\return(\mcomp{v_1\gets\monadic{v}_1\hstop v_2\gets\monadic{v}_2\hstop\return(v_1+v_2)},\gamma_2,\omicron_2)}\\
\sem{\assertexpr{e}}\gamma_0\phi\omicron_0&=&\mcomp{(\monadic{v},\gamma_1,\omicron_1)\gets\sem{e}\gamma_0\phi\omicron_0\hstop\guard(\monadic{v}\ne\return\fls)\hstop\return(\return\singleton,\gamma_1,\omicron_1)}\\
\sem{\getexpr{d}{k}}\gamma\phi\omicron&=&\branching{\return(\allpure(\phi_{d}(k)),\gamma,\omicron)&\mbox{if $d=\publicdom$}\\\return(\top,\gamma,\omicron)&\mbox{otherwise}}\\
\sem{\ifexpr{e_1}{e_2}{e_3}}\gamma_0\phi\omicron_0&=&\mcomp{(\monadic{v}_1,\gamma_1,\omicron_1)\gets\sem{e_1}\gamma_0\phi\omicron_0\hstop\branching{\sem{e_2}\gamma_1\phi\omicron_1&\mbox{if $\monadic{v}_1=\return\tru$}\\\sem{e_3}\gamma_1\phi\omicron_1&\mbox{if $\monadic{v}_1=\return\fls$}\\\return(\top,\gamma_1,\omicron_1)&\mbox{otherwise}}}\\
\sem{\forexpr{x}{e_1}{e_2}{e_3}}\gamma_0\phi\omicron_0&=&\begin{array}[t]{@{}l@{}}\mcompv{(\monadic{v}_1,\gamma_1,\omicron_1)\gets\sem{e_1}\gamma_0\phi\omicron_0\hstop\\(\monadic{v}_2,\gamma_2,\omicron_2)\gets\sem{e_2}\gamma_1\phi\omicron_1\hstop\\\branching{\uopl{main}(\max(0,i_2-i_1))&\mbox{where $\monadic{v}_j=\return i_j$, $j=1,2$}\\\return(\top,\gamma_2,\omicron_2)&\mbox{if $\monadic{v}_1=\top$ or $\monadic{v}_2=\top$}}}\end{array}\\
\multicolumn{2}{r}{\mbox{where}}&\uopl{main}(n)=\mcompv{(\monadic{a}_1,\gamma'_1,\omicron'_1)\gets\sem{e_3}((x,\return i_1),\gamma_2)\phi\omicron_2\hstop\\\forall k=2,...,n:(\monadic{a}_k,\!\gamma'_k,\!\omicron'_k)\gets\sem{e_3}([x\mapsto\return(i_1\!\!+\!\!k\!\!-\!\!1)]\gamma'_{k-1})\phi\omicron'_{k-1}\hstop\\\return(\return(\monadic{a}_1,\ldots,\monadic{a}_n),\tail\gamma'_n,\omicron'_n)}\\
\sem{\wireexpr{e}}\gamma\phi\omicron&=&\mcompv{(\monadic{v},\gamma',\omicron')\gets\sem{e}\gamma\phi\omicron\hstop\\\return(\branching{\monadic{v}&\mbox{if $d=\publicdom$}\\\head\omicron'_d&\mbox{otherwise}},\gamma',\lam{d'}{\branching{\tail\omicron'_{d'}&\mbox{if $d'=d$}\\\omicron'_{d'}&\mbox{otherwise}}})}\\
\multicolumn{2}{r}{\mbox{where}}&\mbox{$d$ is the domain of~$e$}\\
\sem{\castexpr{e}{\qualty{t}{s}{d}}}\gamma\phi\omicron&=&\branching{\sem{e}\gamma\phi\omicron&\mbox{if $s=\poststg$ or $d=\publicdom$}\\\mcomp{(\monadic{v},\gamma',\omicron')\gets\sem{e}\gamma\phi\omicron\hstop\return(\top,\gamma',\omicron')}&\mbox{otherwise}}\\
\sem{\assignexpr{l}{e}}\gamma_0\phi\omicron_0&=&\begin{array}[t]{@{}l@{}}\mcompv{\monadic{a}\gets\return(\lookup{\gamma_0}{x})\hstop\\\forall k=1,\ldots,n:(\monadic{i}_k,\gamma_k,\omicron_k)\gets\sem{y_k}\gamma_{k-1}\phi\omicron_{k-1}\hstop\\(\monadic{v},\gamma',\omicron')\gets\sem{e}\gamma_n\phi\omicron_n\hstop\\\return(\return\singleton,[x\mapsto\update(\monadic{a},\monadic{i}_1\ldots\monadic{i}_n,\monadic{v})]\gamma',\omicron')}\end{array}\\
\multicolumn{2}{r}{\mbox{where}}&\loadexpr{\loadexpr{\loadexpr{x}{y_1}}{y_2}\ldots}{y_n}=l\\
\sem{\loadexpr{l}{e}}\gamma_0\phi\omicron_0&=&\mcompv{(\monadic{a},\gamma_1,\omicron_1)\gets\sem{l}\gamma_0\phi\omicron_0\hstop\\(\monadic{i},\gamma_2,\omicron_2)\gets\sem{e}\gamma_1\phi\omicron_1\hstop\\\return(\mcomp{a\gets\monadic{a}\hstop i\gets\monadic{i}\hstop a_i},\gamma_2,\omicron_2)}\\
\sem{\stmtcomp{\letexpr{x}{e_1}}{e_2}}\gamma_0\phi\omicron_0&=&\begin{array}[t]{@{}l@{}}\mcompv{(\monadic{v}_1,\gamma_1,\omicron_1)\gets\sem{e_1}\gamma_0\phi\omicron_0\hstop\\(\monadic{v}_2,\gamma_2,\omicron_2)\gets\sem{e_2}((x,\monadic{v}_1),\gamma_1)\phi\omicron_1\hstop\\\return(\monadic{v}_2,\tail\gamma_2,\omicron_2)}\end{array}\\
\sem{\stmtcomp{e_1}{e_2}}\gamma_0\phi\omicron_0&=&\mcomp{(\monadic{v}_1,\gamma_1,\omicron_1)\gets\sem{e_1}\gamma_0\phi\omicron_0\hstop\sem{e_2}\gamma_1\phi\omicron_1}
\end{array}
\]
\caption{\label{semantics:circuit}
Circuit semantics}
\end{figure*}

We can prove Theorems \ref{semantics:compthm}--\ref{semantics:correctness} below. Theorem~\ref{semantics:localthm} implies (via repeated application) that every party can compute all data that belong to its domain or lower domains despite not knowing values of the higher domains. Theorem~\ref{semantics:effthm} states that evaluating an expression can change only those values of the value environment that live in domains where the expression is effectful according to the type system. Theorem~\ref{semantics:outputthm} states that an expression can output values to the circuit only if the expression is effectful in $\publicdom$. Theorem~\ref{semantics:effcircthm} is similar to Theorem~\ref{semantics:effthm} but is concerning values in the environment that are visible to the circuit. Theorem~\ref{semantics:safety} states that inputs of higher domains do not influence computation results in the lower domains.
Theorem~\ref{semantics:soundness} states that executions of the same code in different domains agree on values visible in the lower domain.  Theorem~\ref{semantics:correctness} establishes that if a program succeeds in Prover's semantics then it succeeds in the circuit semantics, provided that it is given the same input and Prover's and Verifier's correct output. Proofs of the theorems use induction on the structure of the expression; the details are given in Appendix~\ref{semantics-proofs}.

Before the theorems can be precisely formulated, a few notions must be introduced which the formulations rely on. The notions basically specify, for a fixed domain's or the circuit's point of view, what are good relationships between monadic values and types, and between two monadic values. 

\begin{definition}\label{semantics:exposeddef}
Let a predicate~$P$ on qualified types be fixed. For any well-structured qualified type $q=\qualty{t}{s}{d}$ and $\monadic{v}\in M\ap V$, we say that $\monadic{v}$ is \emph{$q$-exposed in~$P$} if one of the following alternatives holds:
\begin{enumerate}
\item $P(q)$ is true and $t$~{}is a primitive type and $\monadic{v}=\return v$ where $v\in t$ (e.g., if $t=\boolmodty{\mbox{\lstinline{N}}}$ then $v\in\set{\tru,\fls}$);
\item $P(q)$ is true and $t=\listty{q'}$ and $\monadic{v}=\return(\monadic{v}_1,\ldots,\monadic{v}_n)$ where $n\in\NN$ and all $\monadic{v}_1,\ldots,\monadic{v}_n\in M\ap V$ are $q'$-exposed in~$P$;
\item $P(q)$ is false.
\end{enumerate}
\end{definition}

\begin{definition}\label{semantics:exactdef}
Let a predicate~$P$ on qualified types be fixed. For any well-structured qualified type $q=\qualty{t}{s}{d}$ and $\monadic{v}\in M\ap V$, we say that $\monadic{v}$ is \emph{$q$-exact in~$P$} if one of the following alternatives holds:
\begin{enumerate}
\item $P(q)$ is true and $t$~{}is a primitive type and $\monadic{v}=\return v$ where $v\in t$ (e.g., if $t=\boolmodty{\mbox{\lstinline{N}}}$ then $v\in\set{\tru,\fls}$);
\item $P(q)$ is true and $t=\listty{q'}$ and $\monadic{v}=\return(\monadic{v}_1,\ldots,\monadic{v}_n)$ where $n\in\NN$ and all $\monadic{v}_1,\ldots,\monadic{v}_n\in M\ap V$ are $q'$-exact in~$P$;
\item $P(q)$ is false and $\monadic{v}=\top$.
\end{enumerate}
\end{definition}

\begin{definition}\label{semantics:coincidentdef}
Let a predicate~$P$ on qualified types be fixed. For any well-structured qualified type $q=\qualty{t}{s}{d}$ and $\monadic{v},\monadic{v}'\in M\ap V$, we say that $\monadic{v}$ and $\monadic{v}'$ are \emph{$q$-coincident in~$P$} and write $\monadic{v}\coin{q}{P}\monadic{v}'$ iff one of the following alternatives holds:
\begin{enumerate}
\item $P(q)$ is true and $t$~{}is a primitive type and $\monadic{v}=\monadic{v}'=\return v$ where $v\in t$;
\item $P(q)$ is true and $t=\listty{q'}$ and $\monadic{v}=\return(\monadic{v}_1,\ldots,\monadic{v}_n)$, $\monadic{v}'=\return(\monadic{v}'_1,\ldots,\monadic{v}'_n)$ where $n\in\NN$ and $\monadic{v}_i$, $\monadic{v}'_i$ are $q'$-coincident in~$P$ for every $i=1,\ldots,n$;
\item $P(q)$ is false.
\end{enumerate}
\end{definition}

\begin{definition}
Let $\Gamma$ be a well-structured type environment and $P$ be a predicate defined on qualified types. 
\begin{enumerate}
\item We say that $\gamma\in\mathbf{Env}$ is \emph{$\Gamma$-exposed in~$P$} iff $\vars\Gamma=\vars\gamma$ and, for every association $(x_i:q_i)$ occurring in~$\Gamma$, the value in the corresponding association $(x_i,\monadic{v}_i)$ in $\gamma$ is $q_i$-exposed in~$P$. 
\item We say that $\gamma\in\mathbf{Env}$ is \emph{$\Gamma$-exact in~$P$} iff $\vars\Gamma=\vars\gamma$ and, for every association $(x_i:q_i)$ occurring in $\Gamma$, the value in the corresponding association $(x_i,\monadic{v}_i)$ in $\gamma$ is $q_i$-exact in~$P$. 
\item We say that $\gamma,\gamma'\in\Env$ are \emph{$\Gamma$-coincident in~$P$} and write $\gamma\coin{\Gamma}{P}\gamma'$ iff $\vars\Gamma=\vars\gamma=\vars\gamma'$ and, for every association $(x_i:q_i)$ occurring in $\Gamma$, the values in the corresponding associations $(x_i,\monadic{v}_i)$ and $(x_i,\monadic{v}'_i)$ in $\gamma$ and $\gamma'$, respectively, are $q_i$-coincident in~$P$.
\end{enumerate}
\end{definition}

\begin{definition}\label{semantics:interpretdef}
For any fixed domain~$d'$, we shall say ``-exposed in~$d'$'', ``-exact in~$d'$'' and ``\mbox{-}co\-incident in~$d'$'' instead of ``-exposed in~$P$'', ``-exact in~$P$'' and ``-coincident in~$P$'' where $P(\qualty{t}{s}{d})\rightleftharpoons(d\subtype d')$. We shall say ``-exact in circuit'' and ``-coincident in circuit'' instead of ``-exact in~$P$'' and ``\mbox{-}co\-incident in~$P$'' where $P(\qualty{t}{s}{d})\rightleftharpoons(s=\poststg\vee d=\publicdom)$. We write $\coin{q}{d}$ and $\coin{\Gamma}{d}$ for coincidence in~$d$, and $\coincirc{q}$ and $\coincirc{\Gamma}$ for coincidence in circuit.
\end{definition}

\begin{definition}
Call a predicate~$P$ defined on qualified types \emph{data insensitive} if $P(\qualty{t_1}{s}{d})=P(\qualty{t_2}{s}{d})$ for all data types $t_1,t_2$, stage~$s$ and domain~$d$.
\end{definition}

Note that all predicates used in Definition~\ref{semantics:interpretdef} are data insensitive.

\begin{definition}
Let $d$ be a fixed domain. 
\begin{enumerate}
\item Let $\omicron\in\Out^2$. We say that $\omicron$ is \emph{exact in~$d$} iff, for any $d'\in\set{\proverdom,\verifierdom}$, each component of $\omicron_{d'}$ is of the form $\return v$ for $v\in\NN\cup\BB$ if $d'\subtype d$ and $\top$ if $d'$ is a strict superdomain of~$d$.
\item Let $\omicron,\omicron'\in\Out^2$. We say that $\omicron$ and $\omicron'$ are \emph{coincident in~$d$} and write $\omicron\coinpure{d}\omicron'$ iff, for any $d'\in\set{\proverdom,\verifierdom}$, the lengths of $\omicron_{d'}$ and $\omicron'_{d'}$ are equal and if $d'\subtype d$ then $\omicron_{d'}=\omicron'_{d'}$.
\end{enumerate}
\end{definition}

\begin{theorem}\label{semantics:compthm}
Let $\Gamma \vd e : q\eff D$ with well-structured~$\Gamma$ and $\gamma\in\mathbf{Env}$ be $\Gamma$-exact in~$d$ for some domain~$d$. Assume that for all subexpressions of~$e$ of the form $\ofty{\getexpr{d'}{k}}{q'}$ where $d'\subtype d$, the value $\allpure(\phi_{d'}(k))$ is $q'$-exact in~$d$. Assume that $\semd{e}\gamma\phi=\return(\monadic{v},\gamma',\omicron)$. Then:
\begin{thmlist}
\item\label{semantics:localthm}
$\monadic{v}$~{}is $q$-exact, $\gamma'$~{}is $\Gamma$-exact and $\omicron$~{}is exact in~$d$;
\item\label{semantics:effthm}
$\gamma\coin{\Gamma}{d''}\gamma'$ for any domain~$d''$ such that $d''\subtype d$ and $d''\notin D$;
\item\label{semantics:outputthm}
If $\publicdom\notin D$ then $\omicron=\epsilon$;
\item\label{semantics:effcircthm}
If $d=\proverdom$ and $\publicdom\notin D$ then $\gamma\coincirc{\Gamma}\gamma'$.
\end{thmlist}
\end{theorem}

%
%
%
\begin{theorem}\label{semantics:safety}
If $\semd{e}\gamma\phi=\return(\monadic{v},\gamma',\omicron)$ and $\phi'_{d'}=\phi_{d'}$ for every $d'\subtype d$ then $\semd{e}\gamma\phi'=\semd{e}\gamma\phi$.
\end{theorem}

\begin{theorem}\label{semantics:soundness}
Let $\Gamma\vd e:q\eff D$ with well-structured $\Gamma$. Let $d,d'$ be domains such that $d'\subtype d$. Let $\gamma_d,\gamma_{d'}\in\mathbf{Env}$ be $\Gamma$-exact in~$d$ and $d'$, respectively, and let $\gamma_d\coin{\Gamma}{d'}\gamma_{d'}$. Assume that for all subexpressions of~$e$ of the form $\ofty{\getexpr{d''}{k}}{q'}$, the value $\allpure(\phi_{d''}(k))$ is $q'$-exact in~$d''$. Assume that there exist $\monadic{v}_d,\gamma'_d,\omicron$ such that $\semd{e}\gamma_d\phi=\return(\monadic{v}_d,\gamma'_d,\omicron)$. Then there exist $\monadic{v}_{d'},\gamma'_{d'},\omicron'$ such that $\sem{e}_{d'}\gamma_{d'}\phi=\return(\monadic{v}_{d'},\gamma'_{d'},\omicron')$, whereby $\monadic{v}_d\coin{q}{d'}\monadic{v}_{d'}$, $\gamma_d\coin{\Gamma}{d'}\gamma_{d'}$ and $\omicron\coinpure{d'}\omicron'$.
\end{theorem}

\begin{theorem}\label{semantics:correctness}
Let $\Gamma \vd e : q\eff D$ with well-structured~$\Gamma$. Let $\gamma_d$, $\gamma\in\mathbf{Env}$ be $\Gamma$-exact in $\proverdom$  and in circuit, respectively, such that $\gamma_d\coincirc{\Gamma}\gamma$. Assume that, for all subexpressions of~$e$ of the form $\ofty{\getexpr{d'}{k}}{q'}$, the value $\allpure(\phi_{d'}(k))$ is $q'$-exact in $\proverdom$. Assume that, for $d=\proverdom$, there exist $\monadic{v}_d,\gamma'_d,\omicron$ such that $\semd{e}\gamma_d\phi=\return(\monadic{v}_d,\gamma'_d,\omicron)$. If $\rho$~{}is any pair of stream continuations (one for each of $\proverdom$ and $\verifierdom$) then $\sem{e}\gamma\phi(\omicron\rho)=\return(\monadic{v},\gamma',\rho)$, where $\monadic{v}$ is $q$-exact and $\gamma'$ is $\Gamma$-exact in circuit. Thereby, $\monadic{v}_d\coincirc{q}\monadic{v}$ and $\gamma'_d\coincirc{\Gamma}\gamma'$. (Here, $\omicron\rho$ denotes the pointwise concatenation of $\omicron$ and $\rho$.)
\end{theorem}

%

\newcommand{\semc}[1]{\sem{#1}_\mathrm{C}}
\newcommand{\Vc}{V_\mathrm{C}}
\newcommand{\Cc}{C_\mathrm{C}}
\newcommand{\Envc}{\mathbf{Env}_\mathrm{C}}
\newcommand{\Inc}{\mathbf{In}_\mathrm{C}}
\newcommand{\Outc}{\mathbf{Out}_\mathrm{C}}
\newcommand{\wire}{\uopl{wireCirc}}
\newcommand{\mknode}[1]{\uopl{node}{\![#1]}}
\newcommand{\mkoper}{\uopl{op}}
\newcommand{\mkoutput}{\uopl{out}}
\newcommand{\mkinput}{\uopl{in}}
\newcommand{\mkconst}{\uopl{con}}
\newcommand{\allpuretop}{\uopl{allpuretop}}
\newcommand{\drop}{\uopl{drop}}
\newcommand{\map}{\uopl{map}}
\newcommand{\updatec}{\uopl{upd}_{\mathrm{C}}}

\section{Compilation}\label{sec:compilation}

\ZKSC programs are compiled into arithmetic circuits corresponding to the circuit semantics defined in Sect.~\ref{sec:semantics}. An \emph{arithmetic circuit $\mathcal{C}$ over a ring $R$} is a directed acyclic graph, the nodes of which are partitioned into \emph{input}, \emph{constant}, and \emph{operation} nodes, such that each operation node has exactly two incoming arcs and other nodes have none. Additionally, $\mathcal{C}$ assigns an element of~$R$ to each constant node, an operation --- either addition or multiplication --- to each operation node, and a domain --- either \lstinline+@prover+ or \lstinline+@verifier+ --- to each input node in it, and specifies a subset of nodes as output nodes. Also, $\mathcal{C}$~defines an enumeration of its input nodes of each domain. 

Let $\mathcal{V}$ be the set of all nodes in~$\mathcal{C}$ and $\mathcal{I}\subseteq\mathcal{V}$ the set of all input nodes. An assignment $\alpha\in R^{\mathcal{I}}$ of values to the input nodes extends naturally to an assignment $\alpha^\star\in R^{\mathcal{V}}$ to all nodes (the values for the constant nodes are given in the definition of~$\mathcal{C}$ and the value for each operation node is found by applying the operation in it to the values of its predecessors). We say that $\mathcal{C}$ \emph{accepts input $\alpha\in R^{\mathcal{I}}$}, if $\alpha^\star$ assigns~$0$ to all output nodes. If $R$ is a finite field of characteristic \lstinline+N+, then such circuits can be evaluated by various ZKP techniques (perhaps with additional restrictions on \lstinline+N+).

\begin{figure}[b]\figsize
\renewcommand{\arraycolsep}{2pt}
\begin{tabular}{@{}l@{}}
$\begin{array}{lcl}
\Vc&=&\mathbb{N}\cup\mathbb{B}\cup()\cup(M\ap\Vc\times M\ap T)^*\quad\mbox{the set of composite values}
\end{array}$\\[1.5ex]
$\begin{array}{lcl}
\Cc\ap A&=&\Envc\to\In^3\to\NN^2\to\mathds{1}+A\times\Envc\times\Outc\times\NN^2\\
\end{array}$\\[1.5ex]
$\begin{array}{lcl@{\quad}l}
\Envc&=&(X\times(M\ap\Vc\times M\ap T))^+&\mbox{composite value environments}\\
\Outc&=&T^*&\mbox{streams of subcircuits}
\end{array}$\\[4.5ex]
$\begin{array}{lcl@{\quad}l}
\semc{e}&:&\Cc\ap(M\ap\Vc\times M\ap T)&\mbox{the result of compilation of expression~$e$}\\
\end{array}$
\end{tabular}
\caption{\label{compilation:types}
Types of semantic objects}
\end{figure}

Denote the set of all circuits by~$T$. Compilation to an arithmetic circuit proceeds in the lines of the dynamic semantics, using a new monad~$\Cc$. To get all types right, we replace the sets $V$, $\Env$, $\Out$ defined in Fig.~\ref{semantics:types} by $\Vc$, $\Envc$, $\Outc$ defined in Fig.~\ref{compilation:types}. The main difference is that $M\ap V$ is replaced with $M\ap\Vc\times M\ap T$ at most places, meaning that value-circuit pairs occur here as results of computation. We call these pairs \emph{composite values}. The compiler still has to carry values along with circuits for making stage casts from $\poststg$ to $\prestg$ if necessary. The value and the circuit component of a composite value can independently of each other be missing. For instance, if $e$ is in $\qual{\prestg}{\publicdom}$ then the value component is known but there is no circuit but if $e$ is in $\qual{\poststg}{\proverdom}$ then the circuit exists but the value is unknown. In $\qual{\poststg}{\publicdom}$, both components are given. The set $\Outc$ contains output streams of subtrees of the circuit under construction rooted at its output nodes (this is to define the output nodes). The monad~$\Cc$ has an extra inner state not occurring in dynamic semantics, a pair of natural numbers, for counting how many input nodes of domains $\proverdom$ and $\verifierdom$ have been created.

\begin{figure*}\figsize
\renewcommand{\arraycolsep}{2pt}\renewcommand{\arraystretch}{1.2}
\[
\begin{array}{lcl}
\semc{\epsilon}\gamma\phi\nu&=&\return((\return\singleton,\top),\gamma,\epsilon,\nu)\\
\semc{\overline{n}}\gamma\phi\nu&=&\return((\branching{\return n&\mbox{if $d=\publicdom$}\\\top&\mbox{otherwise}},\branching{\return(\mknode{\mkconst(n)})&\mbox{if $s=\poststg$}\\\top&\mbox{otherwise}}),\gamma,\epsilon,\nu)\\
\semc{\overline{b}}\gamma\phi\nu&=&\return((\branching{\return b&\mbox{if $d=\publicdom$}\\\top&\mbox{otherwise}},\branching{\return(\mknode{\mkconst(|b|)})&\mbox{if $s=\poststg$}\\\top&\mbox{otherwise}}),\gamma,\epsilon,\nu)\\
\semc{x}\gamma\phi\nu&=&\return(\lookup{\gamma}{x},\gamma,\epsilon,\nu)\\
\semc{\addexpr{e_1}{e_2}}\gamma_0\phi\nu_0&=&\mcompv{((\monadic{v}_1,\monadic{c}_1),\gamma_1,\omicron_1,\nu_1)\gets\semc{e_1}\gamma_0\phi\nu_0\hstop\\((\monadic{v}_2,\monadic{c}_2),\gamma_2,\omicron_2,\nu_2)\gets\semc{e_2}\gamma_1\phi\nu_1\hstop\\\return((\monadic{s},\monadic{t}),\gamma_2,\omicron_1\omicron_2,\nu_2)}\\
\multicolumn{2}{r}{\mbox{where}}&\monadic{s}=\mcomp{i_1\gets\monadic{v}_1\hstop i_2\gets\monadic{v}_2\hstop\return(i_1+i_2)},\\&&\monadic{t}=\mcomp{c_1\gets\monadic{c}_1\hstop c_2\gets\monadic{c}_2\hstop\return(\mknode{\mkoper{(+)}}(c_1,c_2))}\\
\semc{\assertexpr{e}}\gamma_0\phi\nu_0&=&\mcomp{((\monadic{v}_1,\!\monadic{c}_1),\!\gamma_1,\!\omicron_1,\!\nu_1)\gets\semc{e}\gamma_0\phi\nu_0\hstop\guard(\monadic{c}_1\ne\top)\hstop\return((\return\singleton,\!\top),\!\gamma_1,\!\omicron_1c_1,\!\nu_1)}\\
\multicolumn{2}{r}{\mbox{where}}&\monadic{c}_1=\return(c_1)\\
\semc{\getexpr{d}{k}}\gamma\phi\nu&=&\branching{\return(\allpuretop(\phi_d(k)),\gamma,\epsilon,\nu)&\mbox{if $d=\publicdom$}\\\return((\top,\top),\gamma,\epsilon,\nu)&\mbox{otherwise}}\\
\semc{\castexpr{e}{\qualty{t}{s}{d}}}\gamma_0\phi\nu_0&=&\mcompv{((\monadic{v}_1,\monadic{c}_1),\gamma_1,\omicron_1,\nu_1)\gets\semc{e}\gamma_0\phi\nu_0\hstop\\\return((\branching{\monadic{v}_1&\mbox{if $d=\publicdom$}\\\top&\mbox{otherwise}},\branching{\monadic{c}_1&\mbox{if $s=\poststg$}\\\top&\mbox{otherwise}}),\gamma_1,\omicron_1,\nu_1)}\\
\semc{\assignexpr{l}{e}}\gamma_0\phi\nu_0&=&\mcompv{(\monadic{a},\monadic{c})\gets\return(\lookup{\gamma_0}{x})\hstop\\\forall k=1,\ldots,n:((\monadic{i}_k,\monadic{c}_k),\gamma_k,\omicron_k,\nu_k)\gets\semc{y_k}\gamma_{k-1}\phi\nu_{k-1}\hstop\\((\monadic{v},\monadic{c}'),\gamma',\omicron',\nu')\gets\semc{e}\gamma_n\phi\nu_n\hstop\\\return((\return\singleton,\top),[x\mapsto\updatec((\monadic{a},\monadic{c}),\monadic{i}_1\ldots\monadic{i}_n,(\monadic{v},\monadic{c}'))]\gamma',\omicron_1\ldots\omicron_n\omicron',\nu')}\\
\multicolumn{2}{r}{\mbox{where}}&\loadexpr{\loadexpr{\loadexpr{x}{y_1}}{y_2}\ldots}{y_n}=l\\
\semc{\loadexpr{l}{e}}\gamma_0\phi\nu_0&=&\mcompv{((\monadic{a},\monadic{c}_1),\gamma_1,\omicron_1,\nu_1)\gets\semc{l}\gamma_0\phi\nu_0\hstop\\((\monadic{i},\monadic{c}_2),\gamma_2,\omicron_2,\nu_2)\gets\semc{e}\gamma_1\phi\nu_1\hstop\\\monadic{r}\gets\return(\mcomp{a\gets\monadic{a}\hstop i\gets\monadic{i}\hstop\return a_i})\hstop\\\branching{\return((\monadic{v},\monadic{c}),\gamma_2,\omicron_1\omicron_2,\nu_2)&\mbox{if $\monadic{r}=\return(\monadic{v},\monadic{c})$}\\\return((\top,\top),\gamma_2,\omicron_1\omicron_2,\nu_2)&\mbox{otherwise}}}
\end{array}
\]
\caption{\label{fig:compilationsmall}
Compilation to an arithmetic circuit: ``Small'' expressions}
\end{figure*}

\begin{figure*}\figsize
\renewcommand{\arraycolsep}{2pt}\renewcommand{\arraystretch}{1.2}
\[
\begin{array}{lcl@{}}
\semc{\ifexpr{e_1}{e_2}{e_3}}\gamma_0\phi\nu_0&=&\mcompv{((\monadic{v}_1,\monadic{c}_1),\gamma_1,\omicron_1,\nu_1)\gets\semc{e_1}\gamma_0\phi\nu_0\hstop\\\branching{\mcomp{(\!(\monadic{v}_2,\!\monadic{c}_2),\!\gamma_2,\!\omicron_2,\!\nu_2)\gets\semc{e_2}\gamma_1\phi\nu_1\hstop\return(\!(\monadic{v}_2,\!\monadic{c}_2),\!\gamma_2,\!\omicron_1\omicron_2,\!\nu_2)}&\mbox{if $\monadic{v}_1\!=\!\return\tru$}\\\mcomp{(\!(\monadic{v}_3,\!\monadic{c}_3),\!\gamma_3,\!\omicron_3,\!\nu_3)\gets\semc{e_3}\gamma_1\phi\nu_1\hstop\return(\!(\monadic{v}_3,\!\monadic{c}_3),\!\gamma_3,\!\omicron_1\omicron_3,\!\nu_3)}&\mbox{if $\monadic{v}_1\!=\!\return\fls$}\\\return(\!(\top,\!\top),\!\gamma_1,\!\omicron_1,\!\nu_1)&\mbox{otherwise}}}\\
\semc{\forexpr{x}{e_1}{e_2}{e_3}}\gamma_0\phi\nu_0&=&\mcompv{((\monadic{v}_1,\monadic{c}_1),\gamma_1,\omicron_1,\nu_1)\gets\semc{e_1}\gamma_0\phi\nu_0\hstop\\((\monadic{v}_2,\monadic{c}_2),\gamma_2,\omicron_2,\nu_2)\gets\semc{e_2}\gamma_1\phi\nu_1\hstop\\\branching{\uopl{main}(\max(0,i_2-i_1))&\mbox{where $\monadic{v}_j=\return i_j$}\\\return((\top,\top),\gamma_2,\omicron_1\omicron_2,\nu_2)&\mbox{if $\monadic{v}_1=\top$ or $\monadic{v}_2=\top$}}}\\
\multicolumn{2}{r@{}}{\mbox{where}}&\uopl{main}(n)=\mcompv{((\monadic{a}_1,\monadic{c}'_1),\gamma'_1,\omicron'_1,\nu'_1)\gets\semc{e_3}((x,(\return i_1,\top)),\gamma_2)\phi\nu_2\hstop\\\forall k\!\!=\!\!2,\!...,\!n\!:\!(\!(\!\monadic{a}_{k},\!\monadic{c}'_{\!k}\!),\!\gamma'_{\!k},\!\omicron'_{\!k},\!\!\nu'_{\!k})\!\!\gets\!\!\semc{\!e_3\!}\!(\!(x,\!(\return(\!i_1\!+\!k\!-\!1)\!,\!\top)\!),\gamma'_{k\!-\!1})\phi\nu'_{k\!-\!1}\!\hstop\\\return((\return((\monadic{a}_1,\monadic{c}'_1),\!...,\!(\monadic{a}_n,\monadic{c}'_n)),\top),\tail\gamma'_n,\omicron_1\omicron_2\omicron'_1...\omicron'_n,\nu'_n)}\\
\semc{\wireexpr{e}}\gamma\phi\nu&=&\mcompv{((\monadic{v}',\monadic{c}'),\gamma',\omicron',\nu')\gets\semc{e}\gamma\phi\nu\hstop\\\return((\monadic{v}',\branching{\return(\mknode{\mkconst(n)})&\mbox{if $\monadic{v}'=\return n$, $n\in\NN$}\\\return(\mknode{\mkconst(|b|)})&\mbox{if $\monadic{v}'=\return b$, $b\in\BB$}\\\return(\mknode{\mkinput_d(\nu'_d)})&\mbox{if $\monadic{v}'=\top$}}),\gamma',\omicron',\nu'')}\\
\multicolumn{2}{r}{\mbox{where}}&\mbox{$d$ is the domain of~$e$ and $\nu''=\lam{d'}{\branching{\nu'_{d'}+1&\mbox{if $d'=d$}\\\nu'_{d'}&\mbox{otherwise}}}$}\\
\semc{\stmtcomp{\letexpr{x}{e_1}}{e_2}}\gamma_0\phi\nu_0&=&\mcompv{((\monadic{v}_1,\monadic{c}_1),\gamma_1,\omicron_1,\nu_1)\gets\semc{e_1}\gamma_0\phi\nu_0\hstop\\((\monadic{v}_2,\monadic{c}_2),\gamma_2,\omicron_2,\nu_2)\gets\semc{e_2}((x,(\monadic{v}_1,\monadic{c}_1)),\gamma_1)\phi\nu_1\hstop\\\return((\monadic{v}_2,\monadic{c}_2),\tail\gamma_2,\omicron_1\omicron_2,\nu_2)}\\
\semc{\stmtcomp{e_1}{e_2}}\gamma_0\phi\nu_0&=&\mcompv{((\monadic{v}_1,\monadic{c}_1),\gamma_1,\omicron_1,\nu_1)\gets\semd{e_1}\gamma_0\phi\nu_0\hstop\\((\monadic{v}_2,\monadic{c}_2),\gamma_2,\omicron_2,\nu_2)\gets\semd{e_2}\gamma_1\phi\nu_1\hstop\\\return((\monadic{v}_2,\monadic{c}_2),\gamma_2,\omicron_1\omicron_2)}
\end{array}
\]
\caption{\label{fig:compilationlarge}
Compilation to an arithmetic circuit: ``Large'' expressions and statements}
\end{figure*}

Our compilation semantics $\semc{\cdot}$ is presented in Figures~\ref{fig:compilationsmall} and~\ref{fig:compilationlarge}. We denote by $\mknode{\mathit{t}}(c_1,\ldots,c_n)$ the circuit with a node of the given class~$t$ as root and the circuits referenced by $c_1,\ldots,c_n$ as subtrees. Here, $n$ is $0$ or $2$ depending on~$t$, and $t$ can be $\mkconst(v)$ (a node of constant value~$v$), $\mkoper(\oplus)$ (a node of operation~$\oplus$), or $\mkinput_d(k)$ (the input node number~$k$ of domain~$d$). So the compiler builds up the arithmetic circuit from small pieces as subtrees. As explained above, we consider a node to be an output node iff it is a root of some output subcircuit. Output occurs in the \lstinline{assert} case only. 

Any Boolean~$b$ is replaced with its representative integer~$|b|$ where $|\tru|=0$, $|\fls|=1$ because the circuit only does (modular) arithmetic. Representing~$\tru$ by~$0$ arises from the fact that a \ZKSC program succeeds if all arguments of occurrences of \lstinline{assert} evaluate to~$\tru$ whereas the circuit checks if all its outputs are zero.

Input nodes are created if a \lstinline{wire} construct is applied to a value in domain \lstinline{@prover} or \lstinline{@verifier}. This way, Verifier's and Prover's private values passed to the circuit via the \lstinline{wire} construct form an assignment that the circuit either accepts or not.

Note that the $\updatec$ function used in the assignment clause is here defined slightly differently from $\update$:
\[\renewcommand{\arraycolsep}{2pt}
\begin{array}{@{}lcl@{}}
\updatec((\monadic{v},\monadic{c}),\monadic{i}_1\ldots\monadic{i}_n,(\monadic{v}',\monadic{c}'))&=&\branching{(\monadic{v}',\monadic{c}')&\mbox{if $n=0$}\\\left(\monadic{v}'',\top\right)&\mbox{if $n>0$}}\\
\multicolumn{2}{r}{\mbox{where}}&\monadic{v}''=\mcomp{a\gets\monadic{a}\hstop i_1\gets\monadic{i}_1\hstop\return([i_1\!\mapsto\!\updatec(a_{i_1},\!\monadic{i}_2\ldots\monadic{i}_n,\!(\monadic{v}',\!\monadic{c}'))]a)}\mbox{.}
\end{array}
\]

In the clause for \lstinline+get+, we use a new variant of the $\allpure$ operation, $\allpuretop$, that lifts a value in $U$ to $M\ap\Vc\times M\ap T$ in a straightforward manner by inserting $\return$ and $\top$ at each level:
\[
\allpuretop(v)=\branching{(\return v,\top)&\mbox{if $v$ is of a primitive type}\\(\return(\allpuretop(v_1),\ldots,\allpuretop(v_n)),\top)&\mbox{if $v=(v_1,\ldots,v_n)$}}
\]

In the following, we refer to the following auxiliary operation $|\cdot|:T\to M\ap\NN$ that evaluates a given circuit:
\[\figsize
|c|=\branching{\return n&\mbox{if $c=\mknode{\mkconst(n)}$}\\\mcomp{i_1\gets|c_1|\hstop i_2\gets|c_2|\hstop\return(i_1\oplus i_2)}&\mbox{if $c=\mknode{\mkoper(\oplus)}(c_1,c_2)$}\\\top&\mbox{if $c=\mknode{\mkinput_d(k)}$}}
\]
The results must be monadic since the input nodes do not refer to particular values and must be mapped to~$\top$. We also need an analogous operation that takes the circuit input into account. For any circuit~$c$ and $\pi\in(\NN^*)^2$ (a sequence of input values for both Prover and Verifier), we define
\[\figsize
c(\pi)=\branching{n&\mbox{if $c=\mknode{\mkconst(n)}$}\\c_1(\pi)\oplus c_2(\pi)&\mbox{if $c=\mknode{(\mkoper(\oplus))}(c_1,c_2)$}\\\pi_d(i)&\mbox{if $c=\mknode{\mkinput_d(i)}$}}
\]
This operation always results in a pure value. If the output of compilation of an expression of a statement is $\omicron$ then the resulting circuit accepts an input~$\pi$ iff $\omicron_i(\pi)=0$ for all~$i$ such that $\omicron_i$ exists (where $\omicron_i$ denotes the $i$th component of the sequence~$\omicron$).

We can prove Theorems~\ref{thm:can_compile} and~\ref{thm:correct_compilation} below. Theorem~\ref{thm:can_compile} establishes conditions under which well-typed programs can be compiled into a circuit. Theorem~\ref{thm:correct_compilation} states that compilation into a circuit preserves semantics. For establishing the claims in a mathematically precise form, we need several new notions. Definition~\ref{compilation:exact} introduces one more exactness property that is analogous to those in Sect.~\ref{sec:semantics}. Definition~\ref{compilation:eval} introduces an operation that calculates a monadic value that represents the given composite value. It uses the first component of the given composite value if it is pure, and finds the value of the circuit (in the second component of the composite value) on the given input if the circuit exists. This operation is needed to establish correspondence between composite values the compiler is computing with and the monadic values in the circuit semantics. Definition~\ref{compilation:equiv} introduces an equivalence relation on $M\ap V$ that identifies integer and Boolean values that are encoded the same in the circuit. Definition~\ref{compilation:env} lifts the introduced operations pointwise to value environments.

\begin{definition}\label{compilation:exact}
For any well-structured qualified type $q=\qualty{t}{s}{d}$ and $(\monadic{v},\monadic{c})\in M\ap\Vc\times M\ap T$, we say that $(\monadic{v},\monadic{c})$ is \emph{$q$-exact} if all the following implications hold:
\begin{enumerate}
\item If $s=\poststg$ then there is a circuit~$c$ such that $\monadic{c}=\return c$
, whereby $\monadic{v}=\return n$ with $n\in\NN$ implies $|c|=\return n$, and $\monadic{v}=\return b$ with $b\in\BB$ implies $|c|=\return|b|$, and in both latter cases $|c|=\return(c(\pi))$ for any input sequence~$\pi$.
\item If $d=\publicdom$ and $t$~{}is a primitive type then $\monadic{v}=\return v$ where $v\in t$;
\item If $d=\publicdom$ and $t=\listty{q'}$ then $\monadic{v}=\return((\monadic{v}_1,\monadic{c}_1),\ldots,(\monadic{v}_n,\monadic{c}_n))$ where $n\in\NN$ and all $(\monadic{v}_1,\monadic{c}_1),\ldots,(\monadic{v}_n,\monadic{c}_n)\in M\ap\Vc\times M\ap T$ are $q'$-exact;
\item If $s=\prestg$ then $\monadic{c}=\top$;
\item If $d\ne\publicdom$ then $\monadic{v}=\top$.
\end{enumerate}
\end{definition}

\begin{definition}\label{compilation:eval}
Define $(\_,\_)\bullet\_:(M\ap\Vc\times M\ap T)\times(\NN^*)^2\to M\ap V$ as follows: For every pair $(\monadic{v},\monadic{c})\in M\ap\Vc\times M\ap T$ and circuit input $\pi\in(\NN^*)^2$,
\[
(\monadic{v},\monadic{c})\bullet\pi=\branching{\return n&\mbox{if $\monadic{v}=\return n$, $n\in\NN$}\\\return|b|&\mbox{if $\monadic{v}=\return b$, $b\in\BB$}\\\return\singleton&\mbox{if $\monadic{v}=\return\singleton$}\\\return((\monadic{v}_1,\monadic{c}_1)\bullet\pi,\ldots,(\monadic{v}_n,\monadic{c}_n)\bullet\pi)&\mbox{if $\monadic{v}=\return((\monadic{v}_1,\monadic{c}_1),\ldots,(\monadic{v}_n,\monadic{c}_n))$}\\\return(c(\pi))&\mbox{if $\monadic{v}=\top$ but $\monadic{c}=\return c$}\\\top&\mbox{if $\monadic{v}=\top$ and $\monadic{c}=\top$}}
\]
\end{definition}

\begin{definition}\label{compilation:equiv}
Define a binary predicate $\sim$ on $M\ap V$ as follows: For every pair $(\monadic{v},\monadic{v}')\in M\ap V\times M\ap V$, we write $\monadic{v}\sim\monadic{v}'$ iff one of the following alternatives holds:
\begin{enumerate}
\item $\monadic{v}=\return v$ and $\monadic{v}'=\return v'$ where $v=v'$ or $|v|=v'$ or $v=|v'|$ or $|v| = |v'|$ (where $|\cdot|$ is applied to Boolean values but not to integers);
\item $\monadic{v}=\return\singleton=\monadic{v}'$;
\item $\monadic{v}=\return(\monadic{v}_1,\ldots,\monadic{v}_n)$ and $\monadic{v}'=\return(\monadic{v}'_1,\ldots,\monadic{v}'_n)$ where $\monadic{v}_1\sim\monadic{v}'_1, \ldots, \monadic{v}_n\sim\monadic{v}'_n$;
\item $\monadic{v}=\top=\monadic{v}'$.
\end{enumerate}
\end{definition}

%
\begin{definition}\label{compilation:env}
\begin{enumerate}
\item Let $\Gamma$ be a well-structured type environment. We say that $\gamma\in\Envc$ is \emph{$\Gamma$-exact} iff $\vars\gamma=\vars\Gamma$ and, for every association $(x_i:q_i)$ occurring in $\Gamma$, the value in the corresponding association $(x_i,(\monadic{v}_i,\monadic{c}_i))$ is $q_i$-exact.
\item For $\gamma\in\Envc$ and $\pi\in(\NN^*)^2$, we define $\gamma\bullet\pi\in\Env$ by lifting the bullet operation pointwise to association lists.
\item For $\gamma,\gamma'\in\Env$, we write $\gamma\sim\gamma'$ iff $\vars\gamma=\vars\gamma'$ and $\lookup{\gamma}{x}\sim\lookup{\gamma'}{x}$ for each defined variable~$x$.
\end{enumerate}
\end{definition}

We also extend $\sim$ pointwise to sequences (of equal length) of monadic values and, furthermore, to tuples of such sequences. For any natural number~$n$ and list~$l$, we write $\drop n\ap l$ to denote the part of list~$l$ remaining if the first $n$ elements are removed. We write $\map f\ap l$ for the result of applying~$f$ to all elements of~$l$. Using these auxiliary operations, we extend the bullet notation to denote $\nu\bullet\pi=\lam{d}{\map\return\left(\drop\nu_d\pi_d\right)}$ where $\nu\in\NN^2$, $\pi\in(\NN^*)^2$ and $d$ ranges over $\set{\proverdom,\verifierdom}$. 

\begin{theorem}\label{thm:can_compile}
Let $\Gamma \vd e : q\eff D$ with well-structured~$\Gamma$. Let $\gamma\in\Envc$ be $\Gamma$-exact. Let $\phi\in\In^3$ such that, for all subexpressions of~$e$ of the form $\ofty{\getexpr{d}{k}}{q'}$ where $d=\publicdom$, the value $\allpure(\phi_{d}(k))$ is $q'$-exact in~$d$. Let $\nu\in\NN^2$. Unless an array lookup fails due to an index being out of bounds, we have $\semc{e}\gamma\phi\nu=\return((\monadic{v},\monadic{c}),\gamma',\omicron,\nu')$ for some $\monadic{v}$, $\monadic{c}$, $\gamma'$, $\omicron$ and $\nu'$, whereby $(\monadic{v},\monadic{c})$ is $q$-exact and $\gamma'$ is $\Gamma$-exact.
\end{theorem}

In the following theorem:
\begin{itemize}
\item Letters without tilde denote the original states and those with tilde denote result states;
\item Letters without prime denote compiler states while those with prime denote states in circuit semantics.
\end{itemize}

\begin{theorem}\label{thm:correct_compilation}
Let $\Gamma \vd e:\qualty{t}{s}{d}\eff D$ with well-structured~$\Gamma$. Let $\gamma\in\Envc$ be $\Gamma$-exact, $\phi\in\In^3$ be a triple of type correct input dictionaries, and $\nu\in\NN^2$. Let $\semc{e}\gamma\phi\nu=\return((\monadic{v},\monadic{c}),\tilde{\gamma},\tilde{\omicron},\tilde{\nu})$. Let $\pi\in(\NN^*)^2$, $\gamma'\in\Env$, $\omicron'\in\Out^2$ be such that $\gamma'$ is $\Gamma$-exact in circuit, $\gamma\bullet\pi\sim\gamma'$ and $\nu\bullet\pi\sim\omicron'$. Then:
\begin{enumerate}
\item If $\sem{e}\gamma'\phi\omicron'=\return(\monadic{v}',\tilde{\gamma}',\tilde{\omicron}')$ then $\tilde{\omicron}$ accepts $\pi$; moreover, $\tilde{\gamma}\bullet\pi\sim\tilde{\gamma}'$ and $\tilde{\nu}\bullet\pi\sim\tilde{\omicron}'$, and $d=\publicdom$ or $s=\poststg$ implies $(\monadic{v},\monadic{c})\bullet\pi\sim\monadic{v}'$;
\item If $\sem{e}\gamma'\phi\omicron'$ fails (i.e., it is not of the form $\return(\monadic{v}',\tilde{\gamma}',\tilde{\omicron}')$),  while all circuits that arise during compilation need only inputs in~$\pi$, then $\tilde{\omicron}$~{}does not accept $\pi$.
\end{enumerate}
\end{theorem}

The theorems are proved by structural induction on~$e$. The details are given in Appendix~\ref{compilation-proofs}.

\section{Useful constructions}\label{sec:extensions}\label{ssec:dictionaries}

Our choice of the details of the type system is validated by the ease of implementing certain constructions that often occur in the statements to be proved in ZK. We will describe them next, after two simple extensions to \ZKSC.

First, whenever $s_1,\ldots,s_n$ are identifiers and $t_1,\ldots,t_n$ are qualified types, then $\{s_1:t_1,\ldots,s_n:t_n\}$ is a qualified type. Records are defined by specifying the values of all fields, and the values of the fields can be read. The record itself is implicitly qualified \lstinline+$pre @public+.
\ZKSC allows to declare type synonyms for record types using the keyword \lstinline{struct}.

Second, parameters of functions can also be passed by reference. Changes to a by-reference parameter inside the called function are visible in the calling function.
A parameter preceded by the keyword \lstinline+ref+ in a function declaration is passed by reference.

\paragraph{Bit extraction}
Arithmetic circuits naturally support addition and multiplication of values. To compare two values, they have to be split into bits as in Fig.~\ref{language:factor}. We see that each invocation of \lstinline+less_than+ causes the bits of both of its arguments being computed. Hence the value of variable \lstinline+z+ is split into bits twice, resulting in different lists of bits created by the \lstinline+wire+ expression in the function \lstinline+bitextract+, but with equal values.

It is conceivable that an optimizing compiler is able to detect that these lists of bits have to be equal. However, in a DSL like \ZKSC, we prefer to be able to explicitly indicate the availability of such language-specific optimizations. If the value of a variable \lstinline+v+ of type \lstinline+uint[N] $post @D+ could be split into bits, then we could define it with the following type instead (note polymorphism over \lstinline+N+, as well as \lstinline+@D+):
\begin{lstlisting}
struct uint_with_bits[N : Nat, @D] {
    value : uint[N] $post @D,
    hasBits : bool $pre @public,
    bits : list[bool[N] $post @D]
}
\end{lstlisting}
When creating or updating the value, we define the field \lstinline+value+ and set \lstinline+hasBits+ to \lstinline+false+. The function \lstinline+bitextract+ checks the field \lstinline+hasBits+ of its argument (which is passed by reference) and returns the content of the field \lstinline+bits+ if it is true. The function performs the actual splitting into bits only if the field \lstinline+hasBits+ is false; in this case it updates \lstinline+bits+ and sets \lstinline+hasBits+.


\paragraph{Dictionaries}
Dictionaries generalize arrays, allowing the keys (indices) to come from any set, not just from a segment of integers, and supporting the operations \lstinline+load(key)+ and \lstinline+store(key,value)+. In \ZKSC, lists (which play the role of arrays) have rather restrictive typing rules associated with them, making sure that computations with them can be converted into circuit operations. Having the keys in \lstinline+$post+ and in an arbitrary domain requires the use of \emph{Oblivious RAM} (ORAM)~\cite{DBLP:journals/jacm/GoldreichO96}, which has had a number of solutions proposed in the context of ZK proofs~\cite{DBLP:conf/ccs/Neff01,buffet,DBLP:conf/asiacrypt/BootleCGJM18}. A general method for ORAM in ZKP context~\cite{buffet} performs no correctness checks while the \lstinline+load+ and \lstinline+store+ operations are executed. Rather, all checks for a particular dictionary will be performed when it is no longer used; these checks involve sorting the list of operations by the values of keys, and checking that the equality of certain keys implies the equality of accompanying values. We refer to~\cite{buffet} for details.

The type and the supported operations of a dictionary are given in Fig.~\ref{fig:dictionaries}. During the execution, the performed operations have to be recorded. Both the \lstinline+load+ and \lstinline+store+ operations log the key, the value, and the performed operation (where \lstinline+false+ means loading, and \lstinline+true+ means storing). When logging, the key and the value (which may be private) are kept in the stage \lstinline+$post+, thereby fixing the values in the circuit. The performed operation, however, is public, because \lstinline+load+ and \lstinline+store+ can only be invoked in a public context. The functions \lstinline+find+ and \lstinline+update+ work with values in the stage \lstinline+$pre+ only, looking up a binding or updating them; their specification is mundane. The function \lstinline+finalize+ is run on the field \lstinline+log+ of a dictionary before that dictionary goes out of scope, it implements the checks in~\cite{buffet}. \ZKSC contains syntactic sugar for dictionary creation, and loading and storing of values. It also automatically adds the calls to \lstinline+finalize+ at the end of the block containing the creation of the dictionary.
\begin{figure}
\begin{lstlisting}
struct keyValuePair[N : Nat,$S,@D] {
    key : uint[N] $S @D,
    val : uint[N] $S @D
}
struct operation[N : Nat,@D] {
    kv : keyValuePair[N,$post,@D],
    op : bool $pre @public
}
struct map[N : Nat,@D] {
    bindings : list[keyValuePair[N,$pre,@D]],
    log : list[operation[N,@D]]
}
fn create[N : Nat,@D]() -> map[N,@D] {
    { bindings = [], log = [] }
}
fn load[N : Nat,@D](ref d : map[N,@D], k : uint[N] $post @D) -> uint[N] $post @D {
    let v : uint[N] $post @D = wire { find(d.bindings, k as $pre) };
    d.log = cons({ kv = { key = k, val = v }, op = false }, d.log);
    v
}
fn store[N : Nat,@D](ref d : map[N,@D], k : uint[N] $post @D, v : uint[N] $post @D) {
	d.bindings = update(d.bindings, k as $pre, v as $pre);
    d.log = cons({ kv = { key = k, val = v }, op = true }, d.log);
}
fn finalize[N : Nat,@D](ops : list[operation[N,@D]]) {
    // sorting and assertions omitted
}
\end{lstlisting}
\caption{Implementing dictionaries}\label{fig:dictionaries}
\end{figure}

Depending on the number of loads and stores, as well as other characteristics of the operations of the dictionary, different ways of performing finalization may be most efficient~\cite{xJSnark}. If these characteristics can be derived or predicted from public data, the best way can be chosen in \lstinline+finalize+.


\section{Evaluation}\label{sec:evaluation}

We have implemented \ZKSC compiler in about 20 kLoC of Haskell code, which includes the parser, type-checker, \lstinline+@public+ precomputation engine, and translator to circuits. With it, we illustrate the expressiveness and efficiency of the \ZKSC language on a number of small examples. We have implemented them, and translated them to arithmetic circuits. The \ZKSC source of our examples is given in the supplementary material. Table~\ref{tbl:circsize} shows the size of generated circuits, including the number of inputs (both for the instance, and the witness), and different arithmetic operations. We distinguish between linear and non-linear (i.e. multiplication) operations, as only the latter are costly for a cryptographic ZKP technique. If our generated circuits compute modulo $2^{61}-1$, then they can be ingested by EMP toolkit~\cite{emp-toolkit}, in which case we also report the running time, network traffic, and memory use of Prover and Verifier. The Verifier process, running on a laptop, talks to a Prover process running on a server. They are connected over a link with 5 ms latency and bandwidth of around 450 Mbit/s.

\begin{table}
\caption{Size of circuits}\label{tbl:circsize}
\figsize
\begin{tabular}{l@{}|r@{}r|r@{\;}r|r@{\;}r||r|rr|r}
Ex.\,\null& \multicolumn{2}{c|}{size} & inst. & wit. & \multicolumn{2}{c||}{operations} & \multicolumn{1}{c|}{exec.} & \multicolumn{2}{c|}{RAM (MB)} & n/w \\
& & & size & size & \multicolumn{2}{c||}{\hfill lin. \hfill n.-lin.} & t. (s) & \multicolumn{1}{c}{P} & \multicolumn{1}{c|}{V} & MB \\ \hline
F. & \multicolumn{2}{c|}{N/A} & 123 & 124 & 2.1k & 845 & 1.47 & 327 & 329 & 7.3 \\ \hline
& \multicolumn{1}{c}{A1} & \multicolumn{1}{c|}{A2} & \multicolumn{2}{c|}{} & \multicolumn{2}{c||}{} \\ \cline{2-3}
& 10 & 10 & 0 & 1.4k & 16k & 3.9k & 1.35 & 328 & 330 & 8.1 \\
& 10 & 50 & 0 & 4.0k & 46k & 11k & 1.39 & 332 & 330 & 8.8  \\
& 50 & 50 & 0 & 6.5k & 75k & 18k & 1.39 & 338 & 332 & 9.5  \\
& 500 & 1k & 0 & 96k & 1.1M & 278k & 3.05 & 468 & 466 & 32  \\
\rotatebox{90}{\rlap{~Millionaires}}&
 10k & 50k & 0 & 3.8M & 45M & 11M & $*$ & $*$ & $*$ & $*$ \\ \hline
& \multicolumn{1}{c}{A1} & \multicolumn{1}{c|}{A2} & \multicolumn{2}{c|}{} & \multicolumn{2}{c||}{} \\ \cline{2-3}
& 10 & 10 & 10 & 1.4k & 30k & 5.8k & 1.54 & 328 & 330 & 8.5  \\
& 10 & 50 & 50 & 4.4k & 97k & 20k & 1.47 & 334 & 333 & 10  \\
& 50 & 50 & 50 & 7.4k & 168k & 36k & 1.63 & 341 & 338 & 11  \\
& 500 & 1k & 1k & 117k & 2.9M & 703k & 3.47 & 570 & 525 & 76  \\
\rotatebox{90}{\rlap{~~Subset check}}&
 10k & 50k & 50k & 5.0M & 137M & 37M & $*$ & $*$ & $*$ & $*$ \\ \hline
& \multicolumn{2}{c|}{Curve}  & \multicolumn{2}{c|}{} & \multicolumn{2}{c||}{} \\ \cline{2-3}
& \multicolumn{2}{c|}{P-256} & 2 & 1.4k & 22k & 4.9k \\
& \multicolumn{2}{c|}{P-384} & 2 & 1.9k & 33k & 7.2k \\
\rotatebox{90}{\rlap{~ECDSA}}& \multicolumn{2}{c|}{P-521} &
 2 & 2.5k & 44k & 9.7k \\ \hline
& \multicolumn{1}{c}{$n$} & \multicolumn{1}{c|}{$m$}  & \multicolumn{2}{c|}{} & \multicolumn{2}{c||}{}\\ \cline{2-3}
& 10 & 10 & 0 & 5.1k & 79k & 17k & 1.26 & 331 & 335 & 8.9  \\
& 10 & 50 & 0 & 13k & 175k & 38k & 1.35 & 346 & 335 & 11  \\
& 50 & 50 & 0 & 40k & 662k & 139k & 1.69 & 392 & 370 & 22  \\
& 500 & 1k & 0 & 769k & 13M & 2.7M & $*$ & $*$ & $*$ & $*$  \\
\rotatebox{90}{\rlap{$\qquad$SSSP}}&
 10k & 50k & 0 & 35M & 590M & 123M & $*$ & $*$ & $*$ & $*$  \\ \hline
\end{tabular}
\begin{minipage}{\columnwidth}
    \vspace{1em}
    $*$ did not compile with EMP tools in reasonable time and resources
\end{minipage}
\end{table}

\paragraph{Knowing the factors} Our first example is the program in Fig.~\ref{language:factor}, with performance reported in row ``F'' of Table~\ref{tbl:circsize}.


\paragraph{Historical millionaires} In this example~\cite{Viaduct}, two millionaires are comparing their net worths over time, and want to figure out whose minimum net worth was larger. In~\cite{Viaduct}, this was an example for secure two-party computation. In our setting, we let both arrays of historical net worths (expressed as integers) to be a part of the witness, and check that the smallest element of the first array is larger than the smallest element of the second array. As such, the instance of the relation appears to be empty, and the relation itself trivially satisfied. The setting is still interesting if we consider it to be part of some larger application where Prover has committed to both arrays, the commitments are part of the instance, and the verification that the witness matches the commitments happens elsewhere.



\paragraph{Subset check} In this example, which is part of a larger application developed in \ZKSC, the input consists of two arrays, one held by Verifier and the other one by Prover, and Prover wants to convince Verifier that all elements of the first array are also elements of the second. We show this by sorting the concatenation of two arrays and checking that each element which originated from the first array is followed by an element which is equal to it.


\paragraph{ECDSA verification} In this example (again part of a larger application), Verifier has a public key for the ECDSA digital signature scheme~\cite{kerry2013fips}. Prover wants to convince Verifier that he has a message digest and a signature that verifies with respect to the public key. We express this by implementing the ECDSA verification procedure in \ZKSC. The main part of this verification is the computation of two scalar multiplications in the elliptic curve group, where the scalars originate from the signature and the digest. In one of them, the multiplied point is the public key. We use the standard double-and-add method for computing the scalar multiple. In the second multiplication, the multiplied point is the generator of the group. The second point is public, and we use a windowed method to compute its scalar multiples. The implementation of point doublings and additions follows their standard definitions, which require the inversion of certain elements of the field, which we have implemented with the standard compute-and-verify technique.

In ECDSA verification, computations are performed in two different fields. The computing of scalar multiplications takes place in the field $\ZZ_P$, over which the elliptic curve has been defined. But some multiplications and comparisons also take place in the field $\ZZ_Q$, where $Q$ is the cardinality of the elliptic curve group. \ZKSC thus supports the use of several different moduli (e.g. \lstinline+uint[P]+ and \lstinline+uint[Q]+) in the same program. In this case, the compilation procedure creates several circuits. In order to relate the computations by different moduli, there is an operation for asserting the equality of two values in \lstinline+$post+ stage, these values may belong to different types \lstinline+uint[N]+. These assertions are translated into assertions about the equality of values in the wires of different circuits. It is up to the cryptographic ZKP technique to correctly interpret the produced circuits and relationships between them. Techniques like MPC-in-the-head~\cite{DBLP:conf/stoc/IshaiKOS07,ZKBoo} can likely handle them, using the ring conversions developed for secure multiparty computation protocols~\cite{DBLP:journals/ijisec/BogdanovNTW12}.

In Table~\ref{tbl:circsize} we report the sizes of the circuit for different standardized curves~\cite{kerry2013fips}, adding up the numbers of operations for both moduli. As EMP toolkit does not support fields of this size, we do not report any running times.


\paragraph{Single-Source Shortest Paths (SSSP)} In this example, the input, known by Prover, is a directed graph with weighted edges. The graph has an \emph{initial vertex}, from which all other vertices are reachable. The graph is represented in sparse manner, i.e. the representation consists of the number of vertices $n$, the number of edges $m$ (also known to Verifier), and, for each edge, its starting vertex, its ending vertex, and its length (all known only to Prover). Prover wants to convince Verifier that the distances from a the initial vertex to all other vertices have certain values.

Here Prover finds the shortest distances by himself, as well as the shortest-path tree represented by giving for each vertex (except the initial) the last edge on the shortest path from the origin to this vertex. He makes both of these sets of data available to the computation, after which the circuit checks that they match with the edges. The check consists of three parts. First, for each vertex we check that its reported distance from the origin is equal to the length of the last edge on the shortest path (found through the shortest-path tree), plus the distance of the source vertex of that last edge. Second, for each edge we check that it is \emph{relaxed
}, i.e. the distance to its starting vertex, plus its length, is not less than the distance to its ending vertex. Third, we check that the shortest-path tree is indeed a tree rooted in the origin vertex. We do this by pointer-jumping~\cite[Sec.~30.1]{cormen2009introduction} the shortest-path tree for $\lceil \log n\rceil$ times, and check that in the result, all pointers point to the origin vertex.

Due to the sparse representation of both the graph itself and the shortest-path tree, the SSSP example makes extensive use of dictionaries (Sec.~\ref{ssec:dictionaries}). It is a nice example of interleaving computation and verification. The shortest distances, and the shortest-path tree are verified in the circuit, which involves the computation of permutations sorting the keys used in the dictionaries, which are verified by comparing the neighboring elements in the sorted list of keys, which requires the computation of bit extractions of these elements, the correctness of which is again verified.


\section{Related Work}\label{sec:relwork}

The ZKP use-case proliferation has brought with it a number of tools for either generating the circuits for some cryptographic ZK technique, or for directly describing the computation that runs under a ZKP technique. If the toolset contains a domain-specific language, then its features may often be given the following description in terms \ZKSC. The included DSL has mainly imperative features, it is strongly typed, and the possible qualifiers of the types are \lstinline+$post @prover+ and \lstinline+$pre @public+. The latter type is used in computations defining the structure of the circuit (e.g. for loop counters), while the computations with the former are translated into the circuit or invoked using the ZKP technique. The instance and witness to the computation both receive the qualifiers \lstinline+$post @prover+. The language also allows local computations. Typically, it allows casts from \lstinline+$pre @public+ to \lstinline+$post @prover+, and perhaps to use \lstinline+$post @prover+ values also in the computations at \lstinline+$pre @public+, the results of which must be cast back to \lstinline+$post @prover+.

Such languages include \textsc{circom}~\cite{circom} and ZoKrates~\cite{zokrates}. The latter also includes branchings over \lstinline+$post @prover+ conditions, translating them into executions of both branches, followed by an oblivious choice. We have purposefully excluded such construction, believing it will confuse the developer. The languages to generate circuits (or rank-1 constraint systems) also include \textsc{Sn\r{a}rkl}~\cite{DBLP:conf/padl/StewartML18}, a DSL embedded in Haskell. The necessary local computations are introduced during compilation, in the style of \textsc{Pinocchio}~\cite{DBLP:journals/cacm/ParnoHG016} and Setty et al.~\cite{DBLP:conf/uss/SettyVPBBW12}.

\textsc{Pinocchio} is an early example of a system for \emph{verifiable computation}, aiming to make the correctness of the execution of C programs verifiable in a manner that is cheaper than re-running the program. This approach, refined in~\cite{Geppetto}, as well as in~\cite{ben2014succinct,DBLP:conf/sp/WahbyHGSW16} for hardware-like descriptions, uses a common high-level language to define the circuit, with any \lstinline+$pre+-stage computation inserted through compiler optimizations.

We consider the language of the xJSnark system~\cite{xJSnark} to be the closest to \ZKSC. It follows the description given above. For local computations, it offers blocks of code which take values with both \lstinline+$post @prover+ and \lstinline+$pre @public+ qualifiers, and return the results back to \lstinline+$post @prover+. While it is more expressive than~\cite{circom,zokrates}, we consider our type system to be superior to its method for mixing local computations with those on the circuit. Indeed, our type system offers distinction between \lstinline+$pre @prover+ and \lstinline+$pre @verifier+, and the integrity properties inherent in them. We made use of this distinction in Fig.~\ref{language:factor}. Additionally, in \ZKSC, values in the \lstinline+$pre+-stage can be long-lived, the usefulness of which we showed in Sec.~\ref{ssec:dictionaries}. Finally, we can be polymorphic over stages, allowing same or similar computations performed either locally or in the circuit to be expressed only once.



Hastings et al.~\cite{DBLP:conf/sp/HastingsHNZ19} review 11 different MPC suites from the point of view of the language support the offer for the specification of secure MPC protocols. While the proposed languages often distinguish between private and public values, only Wysteria~\cite{Wysteria,Wys-star} offers constructions to specify the computations done by one or several parties, either locally or using a secure multiparty computation protocol. Their handling of parties is very expressive and general, but the notion of malicious parties is lacking. The recent Viaduct suite~\cite{Viaduct} expands on Wysteria's type system with integrity types, and incorporates also ZK proofs and other secure computation techniques besides MPC in its back-end. However, Viaduct's type system does not distinguish between \lstinline+@verifier+ and \lstinline+@public+, and thus does not support the compilation into a circuit. The distinctions between other qualifiers of \ZKSC are present in Viaduct. Additionally, the handling of data structures by Viaduct's type system is simplified, with e.g. no distinction between read- and write-access to arrays.
\begin{acks}                            
This research has been funded by the \grantsponsor{Spon1}{Defense Advanced Research Projects Agency}{https://www.darpa.mil}
(DARPA) under contract \grantnum{Spon1}{HR0011-20-C-0083}. The views, opinions, and/or findings
expressed are those of the author(s) and should not be interpreted as representing
the official views or policies of the Department of Defense or the U.S. Government.
This research has also been supported by \grantsponsor{Spon2}{European Regional Development Fund}{https://www.europarl.europa.eu/factsheets/et/sheet/95/euroopa-regionaalarengu-fond-erf-}
through the \grantnum{Spon2}{Estonian Centre of Excellence in ICT Research (EXITE)}.
\end{acks}

\bibliography{provenance}

\appendix
\section{Proof of Theorem~\ref{typesystem:thm} and Lemma~\ref{typesystem:lhslemma}}\label{typesystem-proofs}

Before proving Theorem~\ref{typesystem:thm}, we have to establish that $\allpre_{d'}$ holds on well-structured types only.

\begin{lemma}\label{typesystem:lemma}
Let $q=\qualty{t}{s}{d}$ be a qualified type and $d'$ be a domain. If $\allpre_{d'}q$ then $q$ is well-structured.
\end{lemma}

\begin{proof}
By induction on the structure of~$t$:
\begin{itemize}
\item If $t$ is a primitive type then $q=\qualty{t}{\prestg}{d'}$ which is well-structured by Definition~\ref{typesystem:def}.
\item Let $t=\listty{q'}$. Then $q=\qualty{\listty{q'}}{\prestg}{d'}$, i.e., $s=\prestg$ and $d=d'$, and $\allpre_{d'}q'$. By the induction hypothesis, $q'$ is well-structured. We know by the definition of $\allpre_{d'}$ that $q'=\qualty{t'}{\prestg}{d'}$ for some datatype $t'$. Hence $\gen{d}=\gen{d'}=\varnothing\cup\gen{d'}=\gen{\prestg}\cup\gen{d'}$ which completes the proof.
\end{itemize}
\end{proof}

\begin{theorem}[Theorem~\ref{typesystem:thm}]
If\/ $\Gamma\vd e : q\eff D$ with a well-structured~$\Gamma$ then $q$ is well-structured.
\end{theorem}

\begin{proof}
Let $q=\qualty{t}{s}{d}$. We proceed by induction of the structure of~$e$:
\begin{itemize}
\item If $e=\epsilon$, $e=\overline{n}$, $e=\overline{b}$, $e=\mbox{\lstinline+assert+}(e')$, $e=\mbox{\lstinline+wire+ }\inbraces{e'}$ or $e=(\mbox{$l$ \lstinline+=+ $e'$})$ then $t$ is primitive, whereby $t=()$ implies $s=\prestg$. Hence $q$ is well-structured.
\item If $e=x$ then the claim follows from the assumption that $\Gamma$ is well-structured.
\item If $e=\ofty{\getexpr{d}{k}}{q}$ then $\allpre_dq$ and the claim follows from Lemma~\ref{typesystem:lemma}.
\item Suppose $e=\mbox{$e_1$ \lstinline{+} $e_2$}$. Then $t=\mbox{\lstinline{uint[N]}}$, and $\Gamma\vd e_1:t\ s\ d\eff D_1$ and $\Gamma\vd e_2:t\ s\ d\eff D_2$. Hence by the induction hypothesis, $t\ s\ d$ is well-structured, and the desired claim follows.
\item Suppose $e=\mbox{\lstinline+if+ $e_1$ \inbraces{$e_2$} \lstinline+else+ \inbraces{$e_3$}}$. Then $\Gamma\vd e_2:t\ s\ d\eff D_2$ and $\Gamma\vd e_3:t\ s\ d\eff D_3$. Hence by the induction hypothesis, $t\ s\ d$ is well-structured, and the desired claim follows.
\item Suppose $e=\mbox{\lstinline+for+ $x$ \lstinline+in+ $e_1$ \lstinline+..+ $e_2$ }\inbraces{e_3}$. Then $t=\mbox{\lstinline+list+ }\inbracks{t'\ s'\ d'}$, $s=\mbox{\lstinline+$pre+}$ and $\Gamma\vd e_i:\mbox{\lstinline{uint[N]} \lstinline+$pre+}\ d\eff D_i$ ($i=1,2$), $(x:\mbox{\lstinline+uint[N]+ \lstinline+$pre+}\ d),\Gamma\vd e_3:t'\ s'\ d'\eff D_3$, whereby $\gen{d}\supseteq\gen{s'}\cup\gen{d'}\cup D_3$. By the induction hypothesis for $e_1$ and $e_2$, the qualified type $\mbox{\lstinline{uint[N]} \lstinline+$pre+}\ d$ is well-structured. Therefore, the extended type environment $(x:\mbox{\lstinline+uint[N]+ \lstinline+$pre+}\ d),\Gamma$ is well-structured. Hence the induction hypothesis for $e_3$ implies that $t'\ s'\ d'$ is well-structured. This completes checking the assumptions of $\mbox{\lstinline+list+ }\inbracks{t'\ s'\ d'}\ s\ d$ being well-structured by Definition~\ref{typesystem:def}, hence the desired claim follows.
\item Suppose $e=e'\mbox{ \lstinline{as} }t\ s\ d$. Then $\Gamma\vd e':t\ s'\ d'$, $s'\subtype s$, $d'\subtype d$, and $\gen{d}\supseteq\gen{t}$. By the induction hypothesis, $t\ s'\ d'$ is well-structured. If $t$ is a primitive type then $t\ s\ d$ is well-structured, too, as $s'=\mbox{\lstinline+$pre+}$ implies $s=\mbox{\lstinline+$pre+}$. Now let $t=\mbox{\lstinline+list+ }\inbracks{t''\ s''\ d''}$; by $t\ s'\ d'$ being well-structured, we have $s'=\mbox{\lstinline+$pre+}$, $\gen{d'}\supseteq\gen{s''}\cup\gen{d''}$ and 
$t''\ s''\ d''$ well-structured. Therefore $s=\mbox{\lstinline+$pre+}$. Moreover, using $\gen{d}\supseteq\gen{t}$ and the definition of $\gen{t}$, we obtain $\gen{d}\supseteq\gen{t''}\cup\gen{s''}\cup\gen{d''}\supseteq\gen{s''}\cup\gen{d''}$. Hence $t\ s\ d$ is well-structured by Definition~\ref{typesystem:def}.
\item Suppose $e=l\inbracks{e'}$. Then $\Gamma\vd l:t'\ s'\ d'\eff D'$ where $t'=\mbox{\lstinline+list+}\inbracks{t\ s\ d}$. By the induction hypothesis, $t'\ s'\ d'$ is well-structured. Definition~\ref{typesystem:def} now implies $t\ s\ d$ also being well-structured.
\item Suppose $e=\mbox{\lstinline+let+ }x:t'\ s'\ d'\mbox{ \lstinline+=+ }e_1\mbox{\lstinline{;}}\ e_2$. Then $\Gamma\vd e_1:t'\ s'\ d'\eff D_1$ and $(x:t'\ s'\ d'),\Gamma\vd e_2:t\ s\ d\eff D_2$. By the induction hypothesis for $e_1$, qualified type $t'\ s'\ d'$ is well-structured. Therefore, the extended type environment $(x:t'\ s'\ d'),\Gamma$ is well-structured. Hence the induction hypothesis for $e_2$ implies that $t\ s\ d$ is well-structured.
\item If $e=e_1\mbox{\lstinline{;}}\ e_2$ then $\Gamma\vd e_2:t\ s\ d\eff D_2$, whence the induction hypothesis implies $t\ s\ d$ being well-structured.
\end{itemize}
\end{proof}

\begin{lemma}[Lemma~\ref{typesystem:lhslemma}]
Let $\Gamma\vd e:q\eff D$, where $q=(\qualty{t}{s}{d})$ and $\Gamma$ is well-structured. Let $e=\loadexpr{\loadexpr{\loadexpr{x}{y_1}}{y_2}\ldots}{y_n}$ where $x$ is a variable. Then there exist domains $d_1,\ldots,d_n$ and upward closed domain sets $D_1,\ldots,D_n$ such that $\Gamma\vd y_i:\qualty{\uintty}{\prestg}{d_i}\eff D_i$ for each $i=1,\ldots,n$ and
\[
\Gamma\vd x:\qualty{\listty{\ldots\listty{\qualty{\listty{\qualty{\listty{q}}{\prestg}{d_n}}}{\prestg}{d_{n-1}}}\ldots}}{\prestg}{d_1}\eff\varnothing\mbox{,}
\]
whereby $d_1\subtype\ldots\subtype d_{n-1}\subtype d_n\subtype d$ and $D=D_1\cup\ldots\cup D_n$.
\end{lemma}

\begin{proof}
We prove the claim by induction on~$n$. If $n=0$ then $x=e$ and the claim holds trivially. 
Now suppose that $n>0$ and the claim holds for $n-1$. As $\Gamma\vd\loadexpr{\loadexpr{\loadexpr{x}{y_1}}{y_2}\ldots}{y_n}:q\eff D$, we must have a domain $d_n$ and upward closed domain sets~$D',D_n$ such that
\[
\begin{array}{l}
\Gamma\vd\loadexpr{\loadexpr{\loadexpr{x}{y_1}}{y_2}\ldots}{y_{n-1}}:\qualty{\listty{q}}{\prestg}{d_n}\eff D'\mbox{,}\\
\Gamma\vd y_n:\qualty{\uintty}{\prestg}{d_n}\eff D_n\mbox{,}\\
D=D'\cup D_n\mbox{.}
\end{array}
\]
By the induction hypothesis, $\Gamma\vd y_i:\qualty{\uintty}{\prestg}{d_i}\eff D_i$ for each $i=1,\ldots,n-1$, and
\[
\Gamma\vd x:\qualty{\listty{\ldots\listty{\qualty{\listty{\qualty{\listty{q}}{\prestg}{d_n}}}{\prestg}{d_{n-1}}}\ldots}}{\prestg}{d_1}\eff\varnothing\mbox{,}
\]
whereby $d_1\subtype\ldots\subtype d_{n-1}\subtype d_n$ and $D'=D_1\cup\ldots\cup D_{n-1}$. As $D=D'\cup D_n$, the latter implies $D=D_1\cup\ldots\cup D_n$. By $\Gamma$ being well-structured and Theorem~\ref{typesystem:thm}, the type of~$x$ is well-structured, whence also $d_n\subtype d$. 
\end{proof}

\section{Proofs of Theorems \ref{semantics:compthm}--\ref{semantics:correctness}}\label{semantics-proofs}

We start with proving some auxiliary lemmas.

%
%
\begin{lemma}\label{semantics:reintlemma}
\begin{enumerate}
\item Let $q=(\qualty{t}{s}{d})$ be a qualified type and $P$ be a predicate defined on qualified types. If $\monadic{v}\in M\ap V$ is $q$-exact in~$P$ then $\monadic{v}$ is $q$-exposed in~$P$.
\item Let $\Gamma$ be a well-structured type environment and $P$ be a predicate defined on qualified types. If $\gamma\in\Env$ is $\Gamma$-exact in~$P$ then $\gamma$ is $\Gamma$-exposed in~$P$.
\end{enumerate}
\end{lemma}

\begin{proof}
\begin{enumerate}
\item If $P(q)$ is false then the claim holds trivially by Definition~\ref{semantics:exposeddef}. Assume $P(q)$ being true. We proceed by induction on the structure of the data type~$t$. If $t$ is a primitive type then the claim again follows directly from Definitions~\ref{semantics:exposeddef} and \ref{semantics:exactdef}. If $t=\listty{q'}$ then, by Definition~\ref{semantics:exactdef}, $\monadic{v}=\return(\monadic{v}_1,\ldots,\monadic{v}_n)$ where all $\monadic{v}_i$ are $q'$-exact in~$P$. By the induction hypothesis, $\monadic{v}_i$ are $q'$-exposed in~$P$. Hence by Definition~\ref{semantics:exposeddef}, $\monadic{v}$ is $q$-exposed in~$P$.
\item Let $\vars\Gamma=(x_1,\ldots,x_n)$ and the corresponding types be $q_1,\ldots,q_n$. By $\gamma$ being $\Gamma$-exact in~$P$, $\gamma=((x_1,\monadic{v}_1),\ldots,(x_n,\monadic{v}_n))$ for some $\monadic{v}_1,\ldots,\monadic{v}_n$ such that, for every $i=1,\ldots,n$, the corresponding value $\monadic{v}_i$ is $q_i$-exact in~$P$. By part~(1), $\monadic{v}_i$ is $q_i$-exposed in~$P$. Hence $\gamma$ is \mbox{$\Gamma$-exposed} in~$P$.
\end{enumerate}
\end{proof}

\begin{lemma}\label{semantics:approxlemma}
\begin{enumerate}
\item Let $q=(\qualty{t}{s}{d})$ be a qualified type and $P,Q$ be predicates on qualified types such that $Q(q')\Rightarrow P(q')$ for all qualified types~$q'$. If $\monadic{v}\in M\ap V$ is $q$-exposed in~$P$ then $\monadic{v}$ is $q$-exposed in~$Q$.
\item Let $\Gamma$ be a well-structured type environment and $P,Q$ be predicates on qualified types such that $Q(q')\Rightarrow P(q')$ for all qualified types~$q'$. If $\gamma\in\Env$ is $\Gamma$-exposed in~$P$ then $\gamma$ is $\Gamma$-exposed in~$Q$.
\end{enumerate}
\end{lemma}

\begin{proof}
\begin{enumerate}
\item If $Q(q)$ does not hold then the desired claim holds by clause~3 of Definition~\ref{semantics:exposeddef}. Assume now $Q(q)$. Then also $P(q)$. We proceed by induction on the structure of the data type~$t$. If $t$ is a primitive type then clause~1 of Definition~\ref{semantics:exposeddef} implies $\monadic{v}=\return v$, $v\in t$, which in turn gives $\monadic{v}$ being $q$-exposed in~$Q$. If $t=\listty{q'}$ then clause~2 of Definition~\ref{semantics:exposeddef} implies $\monadic{v}=\return(\monadic{v}_1,\ldots,\monadic{v}_n)$ where $\monadic{v}_i$ being $q'$-exposed in~$P$ for all $i=1,\ldots,n$. By the induction hypothesis, all $\monadic{v}_i$ are $q'$-exposed in~$Q$. So, by clause~2 of Definition~\ref{semantics:exposeddef}, $\monadic{v}$ is $q$-exposed in~$Q$.
\item Let $\vars\Gamma=(x_1,\ldots,x_n)$ and the corresponding types be $q_1,\ldots,q_n$. Let $\gamma$ be $\Gamma$-exposed in~$P$. Then $\gamma=((x_1,\monadic{v}_1),\ldots,(x_n,\monadic{v}_n))$ where, for every $i=1,\ldots,n$, $\monadic{v}_i$ is $q_i$-exposed in~$P$. By part~(1), $\monadic{v}_i$ is $q_i$-exposed in~$Q$. Hence $\gamma$ is $\Gamma$-exposed in~$Q$.
\end{enumerate}
\end{proof}

\begin{lemma}\label{semantics:approxcirclemma}
Let $q$ be a qualified type such that $\allpre_{d_0}q$ where $d_0=\publicdom$. If $\monadic{v}$ is $q$-exact in $\proverdom$ then $\monadic{v}$ is $q$-exact in circuit.
\end{lemma}

\begin{proof}
Let $q=(\qualty{t}{s}{d})$. We proceed by induction on the structure of~$t$. If $t$ is a primitive type then, by $\monadic{v}$ being $q$-exact in $\proverdom$, $\monadic{v}=\return v$ where $v\in t$. But $\allpre_{d_0}q$ implies $d=d_0=\publicdom$. Hence $\return v$ is $q$-exact in circuit. Now assume $t=\listty{q'}$. By $\monadic{v}$ being $q$-exact in $\proverdom$, $\monadic{v}=\return(\monadic{v}_1,\ldots,\monadic{v}_n)$ where $n\in\NN$ and all $\monadic{v}_i$ are $q'$-exact in $\proverdom$. Note that $\allpre_{d_0}q$ implies $\allpre_{d_0}q'$. Hence by the induction hypothesis, all $\monadic{v}_i$ are $q'$-exact in circuit. Consequently, $\monadic{v}$ is $q$-exact in circuit.
\end{proof}

%
%
\begin{lemma}\label{semantics:castcoinlemma}
Let $q_0=\qualty{t}{s}{d_0}$ and $q_1=\qualty{t}{s}{d_1}$ where $d_1\subtype d_0$. Let $d$ be any domain. If $\monadic{v}\coin{q_1}{d}\monadic{v}'$ then $\monadic{v}\coin{q_0}{d}\monadic{v}'$.
\end{lemma}

\begin{proof}
Suppose $d_0\subtype d$. Then also $d_1\subtype d$. If $t$ is a primitive type then $\monadic{v}\coin{q_1}{d}\monadic{v}'$ implies $\monadic{v}=\return v=\monadic{v}'$ where $v\in t$. Hence $\monadic{v}\coin{q_0}{d}\monadic{v}'$. If $t=\listty{q'}$ then $\monadic{v}\coin{q_1}{d}\monadic{v}'$ implies $\monadic{v}=\return(\monadic{v}_1,\ldots,\monadic{v}_n)$, $\monadic{v}'=\return(\monadic{v}'_1,\ldots,\monadic{v}'_n)$ where $\monadic{v}_i\coin{q'}{d}\monadic{v}'_i$ for every $i=1,\ldots,n$. Consequently, $\monadic{v}\coin{q_0}{d}\monadic{v}'$.

If $d_0$ is a strict superdomain of~$d$ then the claim holds vacuously.
\end{proof}

\begin{lemma}\label{semantics:coinlemma}
\begin{enumerate}
\item Let $q=(\qualty{t}{s}{d})$ be a qualified type and $P$~be a predicate defined on qualified types. The relation $\coin{q}{P}$ is an equivalence on monadic values that are $q$-exposed in~$P$.
\item Let $\Gamma$ be a well-structured type environment and $P$~be a predicate defined on qualified types. The relation $\coin{\Gamma}{P}$ is an equivalence on value environments that are $\Gamma$-exposed in~$P$.
\end{enumerate}
\end{lemma}

\begin{proof}
In both parts, we have to establish reflexivity, transitivity and symmetry. 

\begin{enumerate}
\item If $P(q)$ does not hold then all $\monadic{v}\in M\ap V$ are $q$-coincident in~$P$, whence all required properties hold. Now assume $P(q)$. We proceed by induction on the structure of the data type~$t$. 

Let $t$ be a primitive type. For reflexivity, take $\monadic{v}$ which is $q$-exposed in~$P$. Then $\monadic{v}=\return v$ where $v\in t$. As $\monadic{v}=\monadic{v}$, this implies $\monadic{v}\coin{q}{P}\monadic{v}$. For transitivity, assume $\monadic{v}\coin{q}{P}\monadic{v}'$ and $\monadic{v}'\coin{q}{P}\monadic{v}''$. Then $\monadic{v}=\monadic{v}'=\monadic{v}''=\return v$ where $v\in t$. Consequently, $\monadic{v}\coin{q}{P}\monadic{v}''$. For symmetry, assume $\monadic{v}\coin{q}{P}\monadic{v}'$. Then $\monadic{v}=\monadic{v}'=\return v$ where $v\in t$. Consequently, $\monadic{v}'\coin{q}{P}\monadic{v}$.

Now let $t=\listty{q'}$ and $q'=(\qualty{t'}{s'}{d'})$. For reflexivity, take $\monadic{v}$ which is $q$-exposed in~$P$. Then $\monadic{v}=\return(\monadic{v}_1,\ldots,\monadic{v}_n)$ where $\monadic{v}_i$ is $q'$-exposed in~$P$ for every $i=1,\ldots,n$. By the induction hypothesis, reflexivity applies to each~$i$ and gives $\monadic{v}_i\coin{q'}{P}\monadic{v}_i$. Hence $\monadic{v}\coin{q}{P}\monadic{v}$. For transitivity, assume $\monadic{v}\coin{q}{P}\monadic{v}'$ and $\monadic{v}'\coin{q}{P}\monadic{v}''$. Then $\monadic{v}=\return(\monadic{v}_1,\ldots,\monadic{v}_n)$, $\monadic{v}'=\return(\monadic{v}'_1,\ldots,\monadic{v}'_n)$, $\monadic{v}''=\return(\monadic{v}''_1,\ldots,\monadic{v}''_n)$, where $\monadic{v}_i\coin{q'}{P}\monadic{v}'_i$ and $\monadic{v}'_i\coin{q'}{P}\monadic{v}''_i$ for every $i=1,\ldots,n$. By the induction hypothesis, transitivity applies to each~$i$ and gives $\monadic{v}_i\coin{q'}{P}\monadic{v}''_i$. Hence $\monadic{v}\coin{q}{P}\monadic{v}''$. For symmetry, assume $\monadic{v}\coin{q}{P}\monadic{v}'$. Then $\monadic{v}=\return(\monadic{v}_1,\ldots,\monadic{v}_n)$, $\monadic{v}'=\return(\monadic{v}'_1,\ldots,\monadic{v}'_n)$ where $\monadic{v}_i\coin{q'}{P}\monadic{v}'_i$ for every $i=1,\ldots,n$. By the induction hypothesis, symmetry applies to each~$i$ and gives $\monadic{v}'_i\coin{q'}{P}\monadic{v}_i$. Hence $\monadic{v}'\coin{q'}{P}\monadic{v}$.

\item Let $\vars\Gamma=(x_1,\ldots,x_n)$ and the corresponding types be $q_1,\ldots,q_n$. 

For reflexivity, take $\gamma$ being $\Gamma$-exposed in~$P$. Then $\gamma=((x_1,\monadic{v}_1),\ldots,(x_n,\monadic{v}_n))$ where $\monadic{v}_i$ is $q_i$-exposed for every $i=1,\ldots,n$. By part~(1), $\monadic{v}_i\coin{q_i}{P}\monadic{v}_i$ for every $i=1,\ldots,n$. The desired claim follows.

For transitivity, assume $\gamma\coin{\Gamma}{P}\gamma'$ and $\gamma'\coin{\Gamma}{P}\gamma''$. Then we must have $\gamma=((x_1,\monadic{v}_1),\ldots,(x_n,\monadic{v}_n))$, $\gamma'=((x_1,\monadic{v}'_1),\ldots,(x_n,\monadic{v}'_n))$ and $\gamma''=((x_1,\monadic{v}''_1),\ldots,(x_n,\monadic{v}''_n))$ where $\monadic{v}_i\coin{q_i}{P}\monadic{v}'_i$ and $\monadic{v}'_i\coin{q_i}{P}\monadic{v}''_i$ for every $i=1,\ldots,n$. By part~(1), $\monadic{v}_i\coin{q_i}{P}\monadic{v}''_i$ for every $i=1,\ldots,n$. The desired claim follows.

For symmetry, assume $\gamma\coin{\Gamma}{P}\gamma'$. Then $\gamma=((x_1,\monadic{v}_1),\ldots,(x_n,\monadic{v}_n))$ and $\gamma'=((x_1,\monadic{v}'_1),\ldots,(x_n,\monadic{v}'_n))$ where $\monadic{v}_i\coin{q_i}{P}\monadic{v}'_i$ for every $i=1,\ldots,n$. By part~(1), $\monadic{v}'_i\coin{q_i}{P}\monadic{v}_i$ for every $i=1,\ldots,n$. The desired claim follows.
\end{enumerate}
\end{proof}

\begin{lemma}
Let $\Gamma$ be a well-structured type environment. Let $x$ be a variable occurring in $\vars\Gamma$ and let $q=\lookup{\Gamma}{x}$. 
\begin{thmlist}
\item\label{semantics:lookupexposedlemma}
Let $P$ be a predicate defined on qualified types. If $\gamma\in\Env$ is $\Gamma$-exposed in~$P$ then $\lookup{\gamma}{x}$ is $q$-exposed in~$P$.
\item\label{semantics:lookupexactlemma}
Let $P$ be a predicate defined on qualified types. If $\gamma\in\Env$ is $\Gamma$-exact in~$P$ then $\lookup{\gamma}{x}$ is $q$-exact in~$P$.
\item\label{semantics:lookupcoinlemma}
Let $P$, $Q$ be predicates defined on qualified types such that $Q(q')\Rightarrow P(q')$ for every qualified type~$q'$. Let $\gamma,\gamma'\in\mathbf{Env}$ be $\Gamma$-exact in~$P$ and $Q$, respectively. If $\gamma\coin{\Gamma}{Q}\gamma'$ then $\lookup{\gamma}{x}\coin{q}{Q}\lookup{\gamma'}{x}$.
\end{thmlist}
\end{lemma}

\begin{proof}
Let $\vars\Gamma=(x_1,\ldots,x_n)$ and the corresponding types be $q_1,\ldots,q_n$.
\begin{enumerate}
\item Let $\gamma$ be $\Gamma$-exposed in~$P$. Then $\gamma=((x_1,\monadic{v}_1),\ldots,(x_n,\monadic{v}_n))$ where $\monadic{v}_i$ is $q_i$-exposed in~$P$ for every $i=1,\ldots,n$. Let $k$ be the least index such that $x_k=x$. Then $q=\lookup{\Gamma}{x}=q_k$ and $\lookup{\gamma}{x}=\monadic{v}_k$. Hence $\lookup{\gamma}{x}$ is $q$-exposed in~$P$.
\item Similar to the proof of Lemma~\ref{semantics:lookupexposedlemma}.
\item Let $\gamma,\gamma'\in\mathbf{Env}$ be $\Gamma$-exact in~$P$ and $Q$, respectively, such that $\gamma\coin{\Gamma}{Q}\gamma'$. Then we have $\gamma=((x_1,\monadic{v}_1),\ldots,(x_n,\monadic{v}_n))$ and $\gamma'=((x_1,\monadic{v}'_1),\ldots,(x_n,\monadic{v}'_n))$ where, for every $i=1,\ldots,n$, $\monadic{v}_i$ and $\monadic{v}'_i$ are $q_i$-exact in $P$ and $Q$, respectively, and $\monadic{v}_i\coin{q_i}{Q}\monadic{v}'_i$. Let $k$ be the least index such that $x_k=x$. Then $q=\lookup{\Gamma}{x}=q_k$ and $\lookup{\gamma}{x}=\monadic{v}_k$, $\lookup{\gamma'}{x}=\monadic{v}'_k$. Hence $\lookup{\gamma}{x}\coin{q}{Q}\lookup{\gamma'}{x}$.
\end{enumerate}
\end{proof}

\begin{lemma}\label{semantics:updexactlemma}
Let $P$ be any data insensitive predicate defined on qualified types. Let $d_1,\ldots,d_n$ be domains and $q$ be a qualified type. Moreover, let $q_k=(\qualty{\uintty}{\prestg}{d_k})$ for every $k=1,\ldots,n$ and  $q'=(\qualty{\listty{\ldots\listty{\qualty{\listty{\qualty{\listty{q}}{\prestg}{d_n}}}{\prestg}{d_{n-1}}}\ldots}}{\prestg}{d_1})$. Assume $\monadic{a}$ being $q'$-exact in~$P$, let $\monadic{i}_k$ be $q_k$-exact in~$P$ for every $k=1,\ldots,n$, and let $\monadic{v}$ be $q$-exact in~$P$. Then $\update(\monadic{a},\monadic{i}_1\ldots\monadic{i}_n,\monadic{v})$ (assuming that it is well defined) is $q'$-exact in~$P$.
\end{lemma}

\begin{proof}
We proceed by induction on~$n$. If $n=0$ then $q'=q$ and $\update(\monadic{a},\monadic{i}_1\ldots\monadic{i}_n,\monadic{v})=\monadic{v}$ which is $q$-exact in~$P$ by assumption. Let now $n>0$ and assume that the claim holds for $n-1$. Denoting $q'_1=(\qualty{\listty{\ldots\listty{\qualty{\listty{\qualty{\listty{q}}{\prestg}{d_n}}}{\prestg}{d_{n-1}}}\ldots}}{\prestg}{d_2})$, consider two cases:
\begin{itemize}
\item If $P(q')$ is true then, by data insensitivity, also $P(q_1)$ holds. By exactness, $\monadic{a}=\return(\monadic{a}_1,\ldots,\monadic{a}_{l_1})$ where all $\monadic{a}_k$ ($k=1,\ldots,l_1$) are $q'_1$-exact in~$P$, and also $\monadic{i}_1=\return i_1$ with $i_1\in\NN$. As $\update(\monadic{a},\monadic{i}_1\ldots\monadic{i}_n,\monadic{v})$ is well-defined, $i_1\leq l_1$. Now $\update(\monadic{a}_{i_1},\monadic{i}_2\ldots\monadic{i}_n,\monadic{v})$ is $q'_1$-exact in~$P$ by the induction hypothesis. Hence also $\return(\monadic{a}_1,\ldots,\monadic{a}_{i_1-1},\update(\monadic{a}_{i_1},\monadic{i}_2\ldots\monadic{i}_n,\monadic{v}),\monadic{a}_{i_1+1},\ldots,\monadic{a}_{l_1})$ must be $q'$-exact in~$P$. The desired claim follows.
\item If $P(q')$ does not hold then, by exactness, $\monadic{a}=\top$. Hence $\update(\monadic{a},\monadic{i}_1\ldots\monadic{i}_n,\monadic{v})=\top$ which is $q'$-exact in~$P$.
\end{itemize}
\end{proof}

\begin{lemma}\label{semantics:updcoin0lemma}
Let $P,Q$ be data insensitive predicates defined on qualified types such that $Q(q'')\Rightarrow P(q'')$ for every qualified type $q''$. Let $d_1,\ldots,d_n$ be domains and $q$ be a qualified type such that $Q(q)$ does not hold. Denote $q'=(\qualty{\listty{\ldots\listty{\qualty{\listty{\qualty{\listty{q}}{\prestg}{d_n}}}{\prestg}{d_{n-1}}}\ldots}}{\prestg}{d_1})$ and $q_k=(\qualty{\uintty}{\prestg}{d_k})$ for every $k=1,\ldots,n$. Let $\monadic{a}$ be $q'$-exposed in~$P$, let $\monadic{i}_k$ be $q_k$-exposed in~$P$ for every $k=1,\ldots,n$ and let $\monadic{v}$ be $q$-exposed in~$P$. Then $\monadic{a}\coin{q'}{Q}\update(\monadic{a},\monadic{i}_1\ldots\monadic{i}_n,\monadic{v})$ (provided that the right-hand side is well-defined).
\end{lemma}

\begin{proof}
We proceed by induction on~$n$. If $n=0$ then $\update(\monadic{a},\monadic{i}_1\ldots\monadic{i}_n,\monadic{v})=\monadic{v}$ and $q'=q$. By Definition~\ref{semantics:coincidentdef}, $\monadic{a}\coin{q}{Q}\monadic{v}$ since $Q(q)$ is false. Now suppose that $n>0$ and the claim holds for $n-1$. Denoting $q'_1=(\qualty{\listty{\ldots\listty{\qualty{\listty{\qualty{\listty{q}}{\prestg}{d_n}}}{\prestg}{d_{n-1}}}\ldots}}{\prestg}{d_2})$, consider two cases:
\begin{itemize}
\item If $P(q')$ holds then, by data insensitivity, also $P(q_1)$ is true. We have $\monadic{a}=\return(\monadic{a}_1,\ldots,\monadic{a}_{l_1})$ and $\monadic{i}_1=\return i_1$ where all $\monadic{a}_k$ are $q'_1$-exposed and $l_1,i_1\in\NN$ (and, by $\update(\monadic{a},\monadic{i_1}\ldots\monadic{i}_n,\monadic{v})$ being well-defined, $i_1\leq l_1$). Hence
\[
\update(\monadic{a},\monadic{i}_1\ldots\monadic{i}_n,\monadic{v})=\return(\monadic{a}_1,\ldots,\monadic{a}_{i_1-1},\update(\monadic{a}_{i_1},\monadic{i}_2\ldots\monadic{i}_n,\monadic{v}),\monadic{a}_{i_1+1},\ldots,\monadic{a}_{l_1})\mbox{.}
\]
For every $k=1,\ldots,l_1$ except~$k=i_1$, the $k$th component of $\monadic{a}$ and that of $\update(\monadic{a},\monadic{i}_1\ldots\monadic{i}_n,\monadic{v})$ are equal. By the induction hypothesis, $\monadic{a}_{i_1}\coin{q'_1}{Q}\update(\monadic{a}_{i_1},\monadic{i}_2\ldots\monadic{i}_n,\monadic{v})$. Consequently, we obtain $\monadic{a}\coin{q'}{Q}\update(\monadic{a},\monadic{i}_1\ldots\monadic{i}_n,\monadic{v})$.
\item If $P(q')$ does not hold then $Q(q')$ does not hold either. Therefore $\monadic{a}\coin{q'}{Q}\update(\monadic{a},\monadic{i}_1\ldots\monadic{i}_n,\monadic{v})$ by clause~3 of Definition~\ref{semantics:coincidentdef}.
\end{itemize}
\end{proof}

\begin{lemma}\label{semantics:updexactcoinlemma}
Let $P,Q$ be data insensitive predicates defined on qualified types such that $Q(q'')\Rightarrow P(q'')$ for every qualified type $q''$. Let $d_1,\ldots,d_n$ be domains and $q$ be a qualified type. Denote $q_k=(\qualty{\uintty}{\prestg}{d_k})$ for every $k=1,\ldots,n$ and  $q'=(\qualty{\listty{\ldots\listty{\qualty{\listty{\qualty{\listty{q}}{\prestg}{d_n}}}{\prestg}{d_{n-1}}}\ldots}}{\prestg}{d_1})$. Let $\monadic{a}_P,\monadic{a}_Q$ be $q'$-exact in $P$ and $Q$, respectively; for all $k=1,\ldots,n$, let $\monadic{i}_P^k,\monadic{i}_Q^k$ be $q_k$-exact in $P$ and $Q$, respectively; let $\monadic{v}_P,\monadic{v}_Q$ be $q$-exact in $P$ and $Q$, respectively. Moreover, assume that $\monadic{a}_P\coin{q'}{Q}\monadic{a}_Q$, for every $k=1,\ldots,n$ we have $\monadic{i}_P^k\coin{q_k}{Q}\monadic{i}_Q^k$, and $\monadic{v}_P\coin{q}{Q}\monadic{v}_Q$. If $\update(\monadic{a}_P,\monadic{i}_P^1\ldots\monadic{i}_P^n,\monadic{v}_P)$ is well-defined then $\update(\monadic{a}_Q,\monadic{i}_Q^1\ldots\monadic{i}_Q^n,\monadic{v}_Q)$ is well-defined, too, whereby $\update(\monadic{a}_P,\monadic{i}_P^1\ldots\monadic{i}_P^n,\monadic{v}_P)\coin{q'}{Q}\update(\monadic{a}_Q,\monadic{i}_Q^1\ldots\monadic{i}_Q^n,\monadic{v}_Q)$.
\end{lemma}

\begin{proof}
We proceed by induction on~$n$. If $n=0$ then $q'=q$, whereby $\update(\monadic{a}_P,\monadic{i}_P^1\ldots\monadic{i}_P^n,\monadic{v}_P)=\monadic{v}_P$ and $\update(\monadic{a}_Q,\monadic{i}_Q^1\ldots\monadic{i}_Q^n,\monadic{v}_Q)=\monadic{v}_Q$. By assumption, $\monadic{v}_P\coin{q}{Q}\monadic{v}_Q$. Let now $n>0$ and assume that the claim holds for $n-1$. Denoting $q'_1=(\qualty{\listty{\ldots\listty{\qualty{\listty{\qualty{\listty{q}}{\prestg}{d_n}}}{\prestg}{d_{n-1}}}\ldots}}{\prestg}{d_2})$, consider two cases:
\begin{itemize}
\item If $Q(q')$ holds then also $P(q')$ holds. By data insensitivity, $Q(q_1)$ and $P(q_1)$ are true, too. Hence $\monadic{a}_P\coin{q'}{Q}\monadic{a}_Q$ implies $\monadic{a}_P=\return(\monadic{a}_P^1,\ldots,\monadic{a}_P^{l_1})$, $\monadic{a}_Q=\return(\monadic{a}_Q^1,\ldots,\monadic{a}_Q^{l_1})$ where $\monadic{a}_P^k\coin{q'_1}{Q}\monadic{a}_Q^k$ for every $k=1,\ldots,n$ and $\monadic{i}_P^1\coin{q_1}{Q}\monadic{i}_Q^1$ implies $\monadic{i}_P^1=\monadic{i}_Q^1=\return i_1$ with $i_1\in\NN$. By the induction hypothesis, $\update(\monadic{a}_Q^{i_1},\monadic{i}_Q^2\ldots\monadic{i}_Q^n,\monadic{v}_Q)$ is well-defined and $\update(\monadic{a}_P^{i_1},\monadic{i}_P^2\ldots\monadic{i}_P^n,\monadic{v}_P)\coin{q'_1}{Q}\update(\monadic{a}_Q^{i_1},\monadic{i}_Q^2\ldots\monadic{i}_Q^n,\monadic{v}_Q)$. As $\update(\monadic{a}_P,\monadic{i}_P^1\ldots\monadic{i}_P^n,\monadic{v}_P)$ is well-defined, $i_1\leq l_1$. All this implies that $\update(\monadic{a}_Q,\monadic{i}_Q^1\ldots\monadic{i}_Q^n,\monadic{v}_Q)$ is also well-defined and, since
\[
\begin{array}{l}
\update(\monadic{a}_P,\monadic{i}_P^1\ldots\monadic{i}_P^n,\monadic{v}_P)=\return(\monadic{a}_P^1,\ldots,\monadic{a}_P^{i_1-1},\update(\monadic{a}_P^{i_1},\monadic{i}_P^2\ldots\monadic{i}_P^n,\monadic{v}_P),\monadic{a}_P^{i_1+1},\ldots,\monadic{a}_P^{l_1})\mbox{,}\\
\update(\monadic{a}_Q,\monadic{i}_Q^1\ldots\monadic{i}_Q^n,\monadic{v}_Q)=\return(\monadic{a}_Q^1,\ldots,\monadic{a}_Q^{i_1-1},\update(\monadic{a}_Q^{i_1},\monadic{i}_Q^2\ldots\monadic{i}_Q^n,\monadic{v}_Q),\monadic{a}_Q^{i_1+1},\ldots,\monadic{a}_Q^{l_1})\mbox{,}
\end{array}
\]
we obtain $\update(\monadic{a}_P,\monadic{i}_P^1\ldots\monadic{i}_P^n,\monadic{v}_P)\coin{q'}{Q}\update(\monadic{a}_Q,\monadic{i}_Q^1\ldots\monadic{i}_Q^n,\monadic{v}_Q)$ as desired.
\item Assume that $Q(q')$ is false. By assumption, $\monadic{a}_Q$ is $q'$-exact in~$Q$, whence $\monadic{a}_Q=\top$. Hence also $\update(\monadic{a}_Q,\monadic{i}_Q^1\ldots\monadic{i}_Q^n,\monadic{v}_Q)=\top$ which is well-defined. The claim $\update(\monadic{a}_P,\monadic{i}_P^1\ldots\monadic{i}_P^n,\monadic{v}_P)\coin{q'}{Q}\top$ holds vacuously.
\end{itemize}
\end{proof}

\begin{lemma}\label{semantics:outputlemma}
Let $d$ be any domain and let $\omicron,\omicron'\in\Out^2$ be both exact in~$d$. Then the pointwise concatenation $\omicron\omicron'$ is exact in~$d$.
\end{lemma}

\begin{proof}
Trivial.
\end{proof}

\begin{lemma}\label{semantics:outputcoinlemma}
Let $d,d'$ be domains such that $d'\subtype d$ and let $\omicron_1,\omicron_2,\omicron'_1,\omicron'_2\in\Out^2$ such that $\omicron_1\coinpure{d'}\omicron'_1$ and $\omicron_2\coinpure{d'}\omicron'_2$. Then $\omicron_1\omicron_2\coinpure{d'}\omicron'_1\omicron'_2$.
\end{lemma}

\begin{proof}
Trivial.
\end{proof}

\begin{theorem}[Theorem~\ref{semantics:localthm}]
Let $\Gamma \vd e : q\eff D$ with well-structured~$\Gamma$ and $\gamma\in\mathbf{Env}$ be $\Gamma$-exact in~$d$ for some domain~$d$. Assume that for all subexpressions of~$e$ of the form $\ofty{\getexpr{d'}{k}}{q'}$ where $d'\subtype d$, the value $\allpure(\phi_{d'}(k))$ is $q'$-exact in~$d$. If $\semd{e}\gamma\phi=\return(\monadic{v},\gamma',\omicron)$ then $\monadic{v}$~{}is $q$-exact, $\gamma'$~{}is $\Gamma$-exact and $\omicron$~{}is exact in~$d$.
\end{theorem}

\begin{proof}
Let $q=t_0\ s_0\ d_0$. We proceed by induction on the structure of~$e$:
\begin{itemize}
\item Let $e=\epsilon$. Then $t_0=\unitty$, $d_0=\publicdom$ and we have $\monadic{v}=\return\singleton$, $\gamma'=\gamma$, $\omicron=\epsilon$. As $d_0\subtype d$ and $t_0$ is primitive, establishing that $\monadic{v}$ is $q$-exact in~$d$ reduces to clause~1 of Definition~\ref{semantics:exactdef}. It holds since $\singleton\in\unitty$. The environment~$\gamma'$ is $\Gamma$-exact in~$d$ by assumption and the output~$\epsilon$ is exact in~$d$ trivially.

\item Let $e=\overline{n}$ where $n\in\NN$. Then $t_0=\uintmodty{\mbox{\lstinline+N+}}$, $\gamma'=\gamma$, $\omicron=\epsilon$. We have to study two cases:
\begin{itemize}
\item If $d_0\subtype d$ then $\monadic{v}=\return n$. As $t_0$ is primitive, establishing that $\monadic{v}$ is $q$-exact in~$d$ reduces to clause~1 of Definition~\ref{semantics:exactdef}. It holds since $n$ is an integer.
\item If $d_0$ is a strict superdomain of~$d$ then $\monadic{v}=\top$ and establishing that $\monadic{v}$ is $q$-exact in~$d$ reduces to clause~3 of Definition~\ref{semantics:exactdef}. The former equality is exactly what clause~3 requires.
\end{itemize}
The environment~$\gamma'$ is $\Gamma$-exact in~$d$ by assumption and the output~$\epsilon$ is exact in~$d$ trivially.

\item Let $e=\overline{b}$ where $b\in\BB$. This case is analogous to the previous one.

\item Let $e=x$. Then $\lookup{\Gamma}{x}=q$ and $\monadic{v}=\lookup{\gamma}{x}$, $\gamma'=\gamma$, $\omicron=\epsilon$. Hence $\monadic{v}$ is $q$-exact in~$d$ by the assumption that $\gamma$ is $\Gamma$-exact in~$d$ and Lemma~\ref{semantics:lookupexactlemma}. Also $\gamma'$ being $\Gamma$-exact in~$d$ directly follows from assumption and the output~$\epsilon$ is exact in~$d$ trivially.

\item Let $e=\addexpr{e_1}{e_2}$. Then $t_0=\uintmodty{\mbox{\lstinline+N+}}$ and
\[
\begin{array}{l}
\Gamma\vd e_1:q\eff D_1\mbox{,}\\
\Gamma\vd e_2:q\eff D_2\mbox{,}
\end{array}
\]
and also
\[
\begin{array}{l}
\semd{e_1}\gamma\phi=\return(\monadic{v}_1,\gamma_1,\omicron_1)\mbox{,}\\
\semd{e_2}\gamma_1\phi=\return(\monadic{v}_2,\gamma_2,\omicron_2)\mbox{,}\\
\monadic{v}=\mcomp{v_1\gets\monadic{v}_1\hstop v_2\gets\monadic{v}_2\hstop\return(v_1+v_2)},\quad\gamma'=\gamma_2\mbox{,}\quad\omicron=\omicron_1\omicron_2\mbox{.}
\end{array}
\]
By the induction hypothesis about~$e_1$, $\monadic{v}_1$ is $q$-exact, $\gamma_1$ is $\Gamma$-exact and $\omicron_1$ is exact in~$d$. Now by the induction hypothesis about~$e_2$, $\monadic{v}_2$ is $q$-exact, $\gamma_2$ is $\Gamma$-exact and $\omicron_2$ is exact in~$d$. Finally, we have to consider two cases:
\begin{itemize}
\item If $d_0\subtype d$ then, by $\monadic{v}_1$ and $\monadic{v}_2$ being $q$-exact in~$d$, we have $\monadic{v}_1=\return v_1$ and $\monadic{v}_2=\return v_2$ for integers $v_1,v_2$. Hence $\monadic{v}=\return(v_1+v_2)$ which shows that $\monadic{v}$ is $q$-exact in~$d$ by clause~1 of Definition~\ref{semantics:exactdef}.
\item If $d_0$ is a strict superdomain of~$d$ then, by $\monadic{v}_1$ and $\monadic{v}_2$ being $q$-exact in~$d$, we have $\monadic{v}_1=\monadic{v}_2=\top$ and hence also $\monadic{v}=\top$. By clause~3 of Definition~\ref{semantics:exactdef}, $\monadic{v}$ is $q$-exact in~$d$.
\end{itemize}
The desired exactness claim about~$\gamma'$ holds by the above and $\omicron$ is exact in~$d$ by Lemma~\ref{semantics:outputlemma}.

\item Let $e=\assertexpr{e_1}$. Then $t_0=\unitty$, $d_0=\publicdom$ and
\[
\Gamma\vd e_1:\qualty{\boolmodty{\mbox{\lstinline+N+}}}{\poststg}{d_1}\eff D
\]
and also
\[
\begin{array}{l}
\semd{e_1}\gamma\phi=\return(\monadic{v}_1,\gamma_1,\omicron_1)\mbox{,}\\
\monadic{v}_1\ne\return\fls\mbox{,}\quad\monadic{v}=\return\singleton\mbox{,}\quad\gamma'=\gamma_1\mbox{,}\quad\omicron=\omicron_1\mbox{.}
\end{array}
\]
As $d_0\subtype d$ and $t_0$ is primitive, establishing that $\monadic{v}$ is $q$-exact in~$d$ reduces to clause~1 of Definition~\ref{semantics:exactdef} which holds as $\singleton\in\unitty$. By the induction hypothesis about~$e_1$, $\gamma_1$ is $\Gamma$-exact and $\omicron_1$ is exact in~$d$ which establish the desired exactness claims about~$\gamma'$ and $\omicron$.

\item Let $e=\ofty{\getexpr{d''}{k}}{q}$. Then $\allpre_{d''}(q)$, implying $d_0=d''$. Moreover, $\gamma'=\gamma$, whence $\gamma'$ is $\Gamma$-exact in~$d$ by assumption. The output $\epsilon$ is exact in~$d$ trivially. For the remaining desired exactness claim, consider two cases:
\begin{itemize}
\item If $d''\subtype d$ then $\monadic{v}=\allpure(\phi_{d''}(k))$ which is $q$-exact in~$d$ by assumption of the theorem.

\item If $d''$ is a strict superdomain of~$d$ then $\monadic{v}=\top$ and establishing that $\monadic{v}$ is $q$-exact in~$d$ reduces to clause~3 of Definition~\ref{semantics:exactdef}. The former equality is exactly what clause~3 requires.
\end{itemize}

\item Let $e=\ifexpr{e_1}{e_2}{e_3}$. Then
\[
\begin{array}{l}
\Gamma\vd e_1:\qualty{\boolmodty{\mbox{\lstinline+N+}}}{\prestg}{d_1}\eff D_1\mbox{,}\\
\Gamma\vd e_2:q\eff D_2\mbox{,}\\
\Gamma\vd e_3:q\eff D_3\mbox{,}\\
\gen{d_1}\supseteq\gen{s_0}\cup \gen{d_0}\cup D_2\cup D_3\mbox{,}
\end{array}
\]
and
\[
\semd{e_1}\gamma\phi=\return(\monadic{v}_1,\gamma_1,\omicron_1)\mbox{.}
\]
By the induction hypothesis about $e_1$, $\monadic{v}_1$ is $(\qualty{\boolmodty{\mbox{\lstinline+N+}}}{\prestg}{d_1})$-exact, $\gamma_1$ is $\Gamma$-exact and $\omicron_1$ is exact in~$d$. We have to consider three cases:
\begin{itemize}
\item If $\monadic{v}_1=\return\tru$ then $\semd{e_2}\gamma_1\phi=\return(\monadic{v}_2,\gamma_2,\omicron_2)$, $\monadic{v}=\monadic{v}_2$, $\gamma'=\gamma_2$, $\omicron=\omicron_1\omicron_2$. By the induction hypothesis about $e_2$, $\monadic{v}_2$ is $q$-exact, $\gamma_2$ is $\Gamma$-exact and $\omicron_2$ is exact in~$d$. Hence the desired claim follows.
\item If $\monadic{v}_1=\return\fls$ then $\semd{e_3}\gamma_1\phi=\return(\monadic{v}_3,\gamma_3,\omicron_3)$, $\monadic{v}=\monadic{v}_3$, $\gamma'=\gamma_3$, $\omicron=\omicron_1\omicron_3$. By the induction hypothesis about $e_3$, $\monadic{v}_3$ is $q$-exact, $\gamma_3$ is $\Gamma$-exact and $\omicron_3$ is exact in~$d$. Hence the desired claim follows.
\item If $\monadic{v}_1\ne\return b$ for $b\not\in\BB$ then $\monadic{v}=\top$, $\gamma'=\gamma_1$, $\omicron=\omicron_1$. As $\monadic{v}_1$ is $(\qualty{\boolmodty{\mbox{\lstinline+N+}}}{\prestg}{d_1})$-exact in~$d$, the only possibility is $\monadic{v}_1=\top$, whereby $d_1$ has to be a strict superdomain of~$d$. But $\gen{d_1}\supseteq\gen{d_0}$ implies $d_1\subtype d_0$, meaning that also $d_0$ must be a strict superdomain of~$d$. Hence by clause~3 of Definition~\ref{semantics:exactdef}, $\top$ is $q$-exact in~$d$. The desired claim follows.
\end{itemize}

\item Let $e=\forexpr{x}{e_1}{e_2}{e_3}$. Then
\[
\begin{array}{l}
\Gamma\vd e_1:\qualty{\uintty}{\prestg}{d_0}\eff D_1\mbox{,}\\
\Gamma\vd e_2:\qualty{\uintty}{\prestg}{d_0}\eff D_2\mbox{,}\\
(x:\qualty{\uintty}{\prestg}{d_0}),\Gamma\vd e_3:\qualty{t_1}{s_1}{d_1}\eff D_3\mbox{,}\\
\gen{d_0}\supseteq\gen{s_1}\cup\gen{d_1}\cup D_3\mbox{,}\\
t_0=\listty{\qualty{t_1}{s_1}{d_1}}\mbox{,}\quad s_0=\prestg\mbox{,}
\end{array}
\]
and
\[
\begin{array}{l}
\semd{e_1}\gamma\phi=\return(\monadic{v}_1,\gamma_1,\omicron_1)\mbox{,}\\
\semd{e_2}\gamma_1\phi=\return(\monadic{v}_2,\gamma_2,\omicron_2)\mbox{.}
\end{array}
\]
By the induction hypothesis about~$e_1$, $\monadic{v}_1$ is $(\qualty{\uintty}{\prestg}{d_0})$-exact, $\gamma_1$ is $\Gamma$-exact and $\omicron_1$ is exact in~$d$. Now by the induction hypothesis about~$e_2$, $\monadic{v}_2$ is $(\qualty{\uintty}{\prestg}{d_0})$-exact, $\gamma_2$ is $\Gamma$-exact and $\omicron_2$ is exact in~$d$. We have to consider two cases:
\begin{itemize}
\item If $\monadic{v}_1=\return i_1$, $\monadic{v}_2=\return i_2$ then, by exactness, $i_1,i_2\in\NN$. Denoting $n=\max(0,i_2-i_1)$,
\[
\begin{array}{l}
\semd{e_3}((x,\return i_1),\gamma_2)\phi=\return(\monadic{v}_3,\gamma_3,\omicron_3)\mbox{,}\\
\semd{e_3}([x\mapsto\return(i_1+1)]\gamma_3)\phi=\return(\monadic{v}_4,\gamma_4,\omicron_4)\mbox{,}\\
\dotfill\\
\semd{e_3}([x\mapsto\return(i_1+n-1)]\gamma_{n+1})\phi=\return(\monadic{v}_{n+2},\gamma_{n+2},\omicron_{n+2})\mbox{,}\\
\monadic{v}=\return(\monadic{v}_3,\ldots,\monadic{v}_{n+2})\mbox{,}\quad\gamma'=\tail\gamma_{n+2}\mbox{,}\quad\omicron=\omicron_1\ldots\omicron_{n+2}\mbox{.}
\end{array}
\]
As $\gamma_2$ is $\Gamma$-exact and $\return i_1$ is $(\qualty{\uintty}{\prestg}{d_0})$-exact in~$d$, the updated environment $(x,\return i_1),\gamma_2$ is $((x:\qualty{\uintty}{\prestg}{d_0}),\Gamma)$-exact in~$d$. Hence the induction hypothesis about $e_3$ applies and gives $\monadic{v}_3$ being $(\qualty{t_1}{s_1}{d_1})$-exact, $\gamma_3$ being $((x:\qualty{\uintty}{\prestg}{d_0}),\Gamma)$-exact and $\omicron_3$ being exact in~$d$. Replacing $i_1$ by $i_1+1$ does not violate exactness, so $[x\mapsto\return(i_1+1)]\gamma_3$ is also $((x:\qualty{\uintty}{\prestg}{d_0}),\Gamma)$-exact in~$d$. Hence the induction hypothesis about $e_3$ applies and gives $\monadic{v}_4$ being $(\qualty{t_1}{s_1}{d_1})$-exact, $\gamma_4$ being $((x:\qualty{\uintty}{\prestg}{d_0}),\Gamma)$-exact and $\omicron_4$ being exact in~$d$. Analogously we obtain $\monadic{v}_k$ being $(\qualty{t_1}{s_1}{d_1})$-exact and $\omicron_k$ being exact for all $k=3,4,\ldots,n+2$, as well as $\gamma_{n+2}$ being $((x:\qualty{\uintty}{\prestg}{d_0}),\Gamma)$-exact in~$d$. Obviously the latter implies $\tail\gamma_{n+2}$ being $\Gamma$-exact in~$d$. By clause~2 of Definition~\ref{semantics:exactdef}, $\return(\monadic{v}_3,\ldots,\monadic{v}_{n+2})$ is $(\qualty{\listty{\qualty{t_1}{s_1}{d_1}}}{\prestg}{d_0})$-exact in~$d$. The desired claim follows.
\item If $\monadic{v}_1=\top$ or $\monadic{v}_2=\top$ then $\monadic{v}=\top$, $\gamma'=\gamma_2$, $\omicron=\omicron_1\omicron_2$ and, by exactness, $d_0$ is a strict supertype of~$d$. Hence $\top$ is also $q$-exact in~$d$ by clause~3 of Definition~\ref{semantics:exactdef}. The desired claim follows.
\end{itemize}

\item Let $e=\wireexpr{e_1}$. Then
\[
\begin{array}{l}
\Gamma\vd e_1:\qualty{t_0}{\prestg}{d_0}\eff D_1\mbox{,}\\
s_0=\poststg\mbox{,}\quad\mbox{$t_0$ is $\uintmodty{\mbox{\lstinline+N+}}$ or $\boolmodty{\mbox{\lstinline+N+}}$}\mbox{,}
\end{array}
\]
and
\[
\begin{array}{l}
\semd{e_1}\gamma\phi=\return(\monadic{v}_1,\gamma_1,\omicron_1)\mbox{,}\\
\monadic{v}=\monadic{v}_1\mbox{,}\quad\gamma'=\gamma_1\mbox{,}\quad\omicron=\lam{d'}{\branching{(\omicron_1)_{d'}\monadic{v}_1&\mbox{if $d'=d_0$}\\(\omicron_1)_{d'}&\mbox{otherwise}}}\mbox{.}
\end{array}
\]
By the induction hypothesis about~$e_1$, $\monadic{v}_1$ is $(\qualty{t_0}{\prestg}{d_0})$-exact, $\gamma_1$ is $\Gamma$-exact and $\omicron_1$ is exact in~$d$. Then $\monadic{v}_1$ is also $(\qualty{t_0}{\poststg}{d_0})$-exact in~$d$. As $t_0$ is $\uintmodty{\mbox{\lstinline+N+}}$ or $\boolmodty{\mbox{\lstinline+N+}}$, $\monadic{v}_1=\return v_1$ for some $v_1\in\NN\cup\BB$ if $d_0\subtype d$ and $\monadic{v}_1=\top$ if $d_0$ is a strict superdomain of~$d$. As $\monadic{v}_1$ is concatenated to $(\omicron_1)_{d_0}$, the component corresponding to~$d_0$ of $\omicron$ contains only pure values if $d_0\subtype d$ and only tops if $d_0$ is a strict superdomain of~$d$. Thus $\omicron$ is exact in~$d$. The desired claim follows.

\item Let $e=\castexpr{e_1}{d_0}$. Let the domain of $e_1$ be $d_1$, i.e.,
\[
\begin{array}{l}
\Gamma\vd e_1:\qualty{t_0}{s_0}{d_1}\eff D\mbox{,}\\
d_1\subtype d_0\mbox{,}\quad\gen{d_0}\supseteq\gen{t_0}\mbox{.}
\end{array}
\]
We have
\[
\begin{array}{l}
\semd{e_1}\gamma\phi=\return(\monadic{v}_1,\gamma_1,\omicron_1)\mbox{,}\\
\monadic{v}=\branching{\monadic{v}_1&\mbox{if $d_0\subtype d$}\\\top&\mbox{otherwise}}\mbox{,}\quad\gamma'=\gamma_1\mbox{,}\quad\omicron=\omicron_1\mbox{.}
\end{array}
\]
By the induction hypothesis, $\monadic{v}_1$ is $(\qualty{t_0}{s_0}{d_1})$-exact and $\gamma_1$ is $\Gamma$-exact in~$d$. Consider two cases:
\begin{itemize}
\item If $d_0\subtype d$ then $\monadic{v}=\monadic{v}_1$. As $d_1\subtype d_0$, we also have $d_1\subtype d$. Hence $\monadic{v}_1$ being $(\qualty{t_0}{s_0}{d_1})$-exact in~$d$ implies $\monadic{v}_1$ being $(\qualty{t_0}{s_0}{d_0})$-exact in~$d$. The desired claim follows.
\item If $d_0$ is a strict supertype of~$d$ then $\monadic{v}=\top$. The desired claim follows as $\top$ is $(\qualty{t_0}{s_0}{d_0})$-exact in~$d$.
\end{itemize}

\item Let $e=(\assignexpr{e_1}{e_2})$. Then
\[
\begin{array}{l}
\Gamma\vd e_1:\qualty{t_1}{s_1}{d_1}\mbox{,}\\
\Gamma\vd e_2:\qualty{t_1}{s_1}{d_1}\mbox{,}\\
t_0=\unitty\mbox{,}\quad s_0=\prestg\mbox{,}\quad d_0=\publicdom\mbox{.}
\end{array}
\]
Let $e_1=\loadexpr{\loadexpr{\loadexpr{x}{y_1}}{y_2}\ldots}{y_n}$ and denote $\monadic{a}=\lookup{\gamma}{x}$; then
\[
\begin{array}{l}
\semd{y_1}\gamma\phi=\return(\monadic{i}_1,\gamma_1,\omicron_1)\mbox{,}\\
\semd{y_2}\gamma_1\phi=\return(\monadic{i}_2,\gamma_2,\omicron_2)\mbox{,}\\
\dotfill\\
\semd{y_n}\gamma_{n-1}\phi=\return(\monadic{i}_n,\gamma_n,\omicron_n)\mbox{,}\\
\semd{e_2}\gamma_n\phi=\return(\monadic{v}_1,\gamma_{n+1},\omicron_{n+1})\mbox{,}\\
\monadic{v}=\return\singleton\mbox{,}\quad\gamma'=[x\mapsto\update(\monadic{a},\monadic{i}_1\ldots\monadic{i}_n,\monadic{v}_1)]\gamma_{n+1}\mbox{,}\quad\omicron=\omicron_1\ldots\omicron_{n+1}\mbox{.}
\end{array}
\]
By Lemma~\ref{typesystem:lhslemma}, $\Gamma\vd y_i:\qualty{\uintty}{\prestg}{d'_i}$ for each $i=1,\ldots,n$ and
\[
\Gamma\vd x:\qualty{\listty{\ldots\listty{\qualty{\listty{\qualty{\listty{\qualty{t'}{s'}{d'_{n+1}}}}{\prestg}{d'_n}}}{\prestg}{d'_{n-1}}}\ldots}}{\prestg}{d'_1}\mbox{.}
\]
Hence $\monadic{a}$ is $(\qualty{\listty{\ldots\listty{\qualty{\listty{\qualty{\listty{\qualty{t'}{s'}{d'_{n+1}}}}{\prestg}{d'_n}}}{\prestg}{d'_{n-1}}}\ldots}}{\prestg}{d'_1})$-exact in~$d$ by exactness of~$\gamma$. By the induction hypothesis about~$y_1$, $\monadic{i}_1$ is $(\qualty{\uintty}{\prestg}{d'_1})$-exact, $\gamma_1$ is $\Gamma$-exact and $\omicron_1$~{}is exact in~$d$. Now by the induction hypothesis about~$y_2$, $\monadic{i}_2$ is $(\qualty{\uintty}{\prestg}{d'_2})$-exact, $\gamma_2$ is $\Gamma$-exact and $\omicron_2$ is exact in~$d$. Similarly, we obtain that, for each $i=1,2,\ldots,n$, $y_i$ is $(\qualty{\uintty}{\prestg}{d'_i})$-exact, $\gamma_i$ is $\Gamma$-exact and $\omicron_k$ is exact in~$d$. The induction hypothesis about~$e_2$ implies that $\monadic{v}_1$ is $(\qualty{t_1}{s_1}{d_1})$-exact, $\gamma_{n+1}$ is $\Gamma$-exact and $\omicron_{n+1}$ is exact in~$d$. Now $\update(\monadic{a},\monadic{i}_1\ldots\monadic{i}_n,\monadic{v}_1)$ is $(\qualty{\listty{\ldots\listty{\qualty{\listty{\qualty{\listty{\qualty{t'}{s'}{d'_{n+1}}}}{\prestg}{d'_n}}}{\prestg}{d'_{n-1}}}\ldots}}{\prestg}{d'_1})$-exact in~$d$ by Lemma~\ref{semantics:updexactlemma}. Hence $[x\mapsto\update(\monadic{a},\monadic{i}_1\ldots\monadic{i}_n,\monadic{v}_1)]\gamma_{n+1}$ is $\Gamma$-exact in~$d$. As $\return\singleton$ is $(\qualty{\unitty}{\prestg}{\publicdom})$-exact in~$d$, we are done.

\item Let $e=\loadexpr{e_1}{e_2}$. Then
\[
\begin{array}{l}
\Gamma\vd e_1:\qualty{\listty{q}}{s_1}{d_1}\mbox{,}\\
\Gamma\vd e_2:\qualty{\uintty}{s_1}{d_1}\mbox{,}
\end{array}
\]
and
\[
\begin{array}{l}
\semd{e_1}\gamma\phi=\return(\monadic{a},\gamma_1,\omicron_1)\mbox{,}\\
\semd{e_2}\gamma_1\phi=\return(\monadic{i},\gamma_2,\omicron_2)\mbox{,}\\
\monadic{v}=\mcomp{a\gets\monadic{a}\hstop i\gets\monadic{i}\hstop a_i}\mbox{,}\quad\gamma'=\gamma_2\mbox{,}\quad\omicron=\omicron_1\omicron_2\mbox{.}
\end{array}
\]
By the induction hypothesis about $e_1$, $\monadic{a}$ is $(\qualty{\listty{q}}{s_1}{d_1})$-exact, $\gamma_1$ is $\Gamma$-exact and $\omicron_1$ is exact in~$d$. Now by the induction hypothesis about $e_2$, $\monadic{i}$ is $(\qualty{\uintty}{s_1}{d_1})$-exact, $\gamma_2$ is $\Gamma$-exact and $\omicron_2$~{}is exact in~$d$. Consider two cases:
\begin{itemize}
\item If $d_1\subtype d$ then, by exactness, $\monadic{a}=\return(\monadic{a}_1,\ldots,\monadic{a}_n)$ for some $n\in\NN$ and $\monadic{i}=\return i$ for some $i\in\NN$, whereby $\monadic{a}_1,\ldots,\monadic{a}_n$ are $q$-exact in~$d$. Moreover, to obtain a pure value, we must have $i\leq n$. Hence $\monadic{v}=\monadic{a}_i$ exists and is $q$-exact in~$d$.
\item If $d_1$ is a strict supertype of~$d$ then, by exactness, $\monadic{a}=\top$ and $\monadic{i}=\top$, whence $\monadic{v}=\top$. As $\Gamma$ is well-structured, $\qualty{\listty{q}}{s_1}{d_1}$ is well-structured by Theorem~\ref{typesystem:thm}. Hence $d_1\subtype d_0$, implying that $d_0$ is a strict supertype of~$d$. By clause~3 of Definition~\ref{semantics:exactdef}, $\top$ is $q$-exact in~$d$.
\end{itemize}

\item Let $e=\stmtcomp{\letexpr{x}{e_1}}{e_2}$. Then
\[
\begin{array}{l}
\Gamma\vd e_1:q_1\eff D_1\mbox{,}\\
(x:q_1),\Gamma\vd e_2:q\eff D_2\mbox{,}
\end{array}
\]
and
\[
\begin{array}{l}
\semd{e_1}\gamma\phi=\return(\monadic{v}_1,\gamma_1,\omicron_1)\mbox{,}\\
\semd{e_2}((x,\monadic{v}_1),\gamma_1)\phi=\return(\monadic{v}_2,\gamma_2,\omicron_2)\mbox{,}\\
\monadic{v}=\monadic{v}_2\mbox{,}\quad\gamma'=\tail\gamma_2\mbox{,}\quad\omicron=\omicron_1\omicron_2\mbox{.}
\end{array}
\]
By the induction hypothesis about~$e_1$, $\monadic{v}_1$ is $q_1$-exact, $\gamma_1$ is $\Gamma$-exact and $\omicron_1$ is exact in~$d$. Hence $(x,\monadic{v}_1),\gamma_1$ is $((x:q_1),\Gamma)$-exact in~$d$. Now the induction hypothesis about~$e_2$ applies and implies $\monadic{v}_2$ being $q$-exact, $\gamma_2$ being $((x:q_1),\Gamma)$-exact and $\omicron_2$ being exact in~$d$. Then $\tail\gamma_2$ is $\Gamma$-exact in~$d$. The desired claim follows.

\item Let $e=\stmtcomp{e_1}{e_2}$. Then
\[
\begin{array}{l}
\Gamma\vd e_1:q_1\eff D_1\mbox{,}\\
\Gamma\vd e_2:q\eff D_2\mbox{,}
\end{array}
\]
and
\[
\begin{array}{l}
\semd{e_1}\gamma\phi=\return(\monadic{v}_1,\gamma_1,\omicron_1)\mbox{,}\\
\semd{e_2}\gamma_1\phi=\return(\monadic{v}_2,\gamma_2,\omicron_2)\mbox{,}\\
\monadic{v}=\monadic{v}_2\mbox{,}\quad\gamma'=\gamma_2\mbox{,}\quad\omicron=\omicron_1\omicron_2\mbox{.}
\end{array}
\]
By the induction hypothesis about~$e_1$, $\monadic{v}_1$ is $q_1$-exact, $\gamma_1$ is $\Gamma$-exact and $\omicron_1$ is exact in~$d$. Now by the induction hypothesis about~$e_2$, $\monadic{v}_2$ is $q$-exact, $\gamma_2$ is $\Gamma$-exact and $\omicron_2$ is exact in~$d$. The desired claim follows.

\end{itemize}
\end{proof}

\begin{theorem}[Theorem~\ref{semantics:effthm}]
Let $\Gamma \vd e : q\eff D$ with well-structured~$\Gamma$ and $\gamma\in\mathbf{Env}$ be $\Gamma$-exact in~$d$ for some domain~$d$. Assume that for all subexpressions of~$e$ of the form $\ofty{\getexpr{d'}{k}}{q'}$ where $d'\subtype d$, the value $\allpure(\phi_{d'}(k))$ is $q'$-exact in~$d$. If $\semd{e}\gamma\phi=\return(\monadic{v},\gamma',\omicron)$ then, for any domain~$d''$ such that $d''\subtype d$ and $d''\notin D$, we have $\gamma\coin{\Gamma}{d''}\gamma'$.
\end{theorem}

\begin{proof}
By Theorem~\ref{semantics:localthm}, $\gamma'$ is $\Gamma$-exact in~$d$. By Lemma~\ref{semantics:reintlemma}, $\gamma$ and $\gamma'$ are $\Gamma$-exposed in~$d$. Take $d''$ such that $d''\subtype d$ and $d''\notin D$. By Lemma~\ref{semantics:approxlemma}, $\gamma$ and $\gamma'$ are also $\Gamma$-exposed in~$d''$. Let $q=\qualty{t_0}{s_0}{d_0}$. We proceed by induction on the structure of~$e$. 

If $e=\epsilon$ or $e=\overline{n}$ where $n\in\NN$ or $e=\overline{b}$ where $b\in\BB$ or $e=x$ or $e=\getexpr{d'}{k}$ then $\gamma=\gamma'$. Hence by Lemma~\ref{semantics:coinlemma}, $\gamma\coin{\Gamma}{d''}\gamma'$. If $e=\assertexpr{e_1}$ or $e=\wireexpr{e_1}$ then the claim holds vacuously since $D=\gen{\publicdom}$. We study the remaining cases.

\begin{itemize}
\item Let $e=\addexpr{e_1}{e_2}$. Then
\[
\begin{array}{l}
\Gamma\vd e_1:q\eff D_1\mbox{,}\\
\Gamma\vd e_2:q\eff D_2\mbox{,}\\
D=\gen{s_0}\cup D_1\cup D_2\mbox{,}
\end{array}
\]
whence we must have $d''\notin D_1$ and $d''\notin D_2$. We also have
\[
\begin{array}{l}
\semd{e_1}\gamma\phi=\return(\monadic{v}_1,\gamma_1,\omicron_1)\mbox{,}\\
\semd{e_2}\gamma_1\phi=\return(\monadic{v}_2,\gamma_2,\omicron_2)\mbox{,}\\
\gamma'=\gamma_2\mbox{.}
\end{array}
\]
By the induction hypothesis, $\gamma\coin{\Gamma}{d''}\gamma_1$ and $\gamma_1\coin{\Gamma}{d''}\gamma_2$. By Lemma~\ref{semantics:coinlemma}, $\gamma\coin{\Gamma}{d''}\gamma_2$. The desired claim follows.

\item Let $e=\ifexpr{e_1}{e_2}{e_3}$. Then
\[
\begin{array}{l}
\Gamma\vd e_1:\qualty{\boolmodty{\mbox{\lstinline+N+}}}{\prestg}{d_1}\eff D_1\mbox{,}\\
\Gamma\vd e_2:q\eff D_2\mbox{,}\\
\Gamma\vd e_3:q\eff D_3\mbox{,}\\
\gen{d_1}\supseteq\gen{s_0}\cup\gen{d_0}\cup D_2\cup D_3\mbox{,}\quad D=D_1\cup D_2\cup D_3\mbox{,}
\end{array}
\]
whence we must have $d''\notin D_1$, $d''\notin D_2$ and $d''\notin D_3$. We also have
\[
\begin{array}{l}
\semd{e_1}\gamma\phi=\return(\monadic{v}_1,\gamma_1,\omicron_1)\mbox{.}
\end{array}
\]
Consider three cases:
\begin{itemize}
\item If $\monadic{v}_1=\return\tru$ then
\[
\begin{array}{l}
\semd{e_2}\gamma_1\phi=\return(\monadic{v}_2,\gamma_2,\omicron_2)\mbox{,}\\
\gamma'=\gamma_2\mbox{.}
\end{array}
\]
By the induction hypothesis, $\gamma\coin{\Gamma}{d''}\gamma_1$ and $\gamma_1\coin{\Gamma}{d''}\gamma_2$. By Lemma~\ref{semantics:coinlemma}, $\gamma\coin{\Gamma}{d''}\gamma_2$. The desired claim follows.
\item If $\monadic{v}_1=\return\fls$ then
\[
\begin{array}{l}
\semd{e_3}\gamma_1\phi=\return(\monadic{v}_3,\gamma_3,\omicron_3)\mbox{,}\\
\gamma'=\gamma_3\mbox{.}
\end{array}
\]
By the induction hypothesis, $\gamma\coin{\Gamma}{d''}\gamma_1$ and $\gamma_1\coin{\Gamma}{d''}\gamma_3$. By Lemma~\ref{semantics:coinlemma}, $\gamma\coin{\Gamma}{d''}\gamma_3$. The desired claim follows.
\item If $\monadic{v}_1\ne\return b$ where $b\in\BB$ then $\gamma'=\gamma_1$. By the induction hypothesis, $\gamma\coin{\Gamma}{d''}\gamma_1$. The desired claim follows.
\end{itemize}

\item Let $e=\forexpr{x}{e_1}{e_2}{e_3}$. Then
\[
\begin{array}{l}
\Gamma\vd e_1:\qualty{\uintty}{\prestg}{d_0}\eff D_1\mbox{,}\\
\Gamma\vd e_2:\qualty{\uintty}{\prestg}{d_0}\eff D_2\mbox{,}\\
(x:\qualty{\uintty}{\prestg}{d_0}),\Gamma\vd e_3:\qualty{t_1}{s_1}{d_1}\eff D_3\mbox{,}\\
\gen{d_0}\supseteq\gen{s_1}\cup\gen{d_1}\cup D_3\mbox{,}\quad D=D_1\cup D_2\cup D_3\mbox{,}
\end{array}
\]
whence we have $d''\notin D_1$, $d''\notin D_2$ and $d''\notin D_3$. We also have
\[
\begin{array}{l}
\semd{e_1}\gamma\phi=\return(\monadic{v}_1,\gamma_1,\omicron_1)\mbox{,}\\
\semd{e_2}\gamma_1\phi=\return(\monadic{v}_2,\gamma_2,\omicron_2)\mbox{.}
\end{array}
\]
Consider two cases:
\begin{itemize}
\item If $\monadic{v}_1=\return i_1$ and $\monadic{v}_2=\return i_2$ for $i_1,i_2\in\NN$ then, denoting $n=\max(0,i_2-i_1)$, we have
\[
\begin{array}{l}
\semd{e_3}((x,\return i_1),\gamma_2)\phi=\return(\monadic{v}_3,\gamma_3,\omicron_3)\mbox{,}\\
\semd{e_3}([x\mapsto\return(i_1+1)]\gamma_3)\phi=\return(\monadic{v}_4,\gamma_4,\omicron_4)\mbox{,}\\
\dotfill\\
\semd{e_3}([x\mapsto\return(i_1+n-1)]\gamma_{n+1})\phi=\return(\monadic{v}_{n+2},\gamma_{n+2},\omicron_{n+2})\mbox{,}\\
\gamma'=\tail\gamma_{n+2}\mbox{.}
\end{array}
\]
By the induction hypothesis, $\gamma\coin{\Gamma}{d''}\gamma_1$ and $\gamma_1\coin{\Gamma}{d''}\gamma_2$, and also $(x,\return i_1),\gamma_2\coin{\Gamma'}{d''}\gamma_3$ and $[x\mapsto\return(i_1+k-3)]\gamma_{k-1}\coin{\Gamma'}{d''}\gamma_k$ for each $k=4,\ldots,n+2$ where $\Gamma'=((x:\qualty{\uintty}{\prestg}{d_0}),\Gamma)$. Then also $\gamma_2\coin{\Gamma}{d''}\tail\gamma_3$ and $\tail\gamma_{k-1}\coin{\Gamma}{d''}\tail\gamma_k$ for every $k=4,\ldots,n+2$. By Lemma~\ref{semantics:coinlemma}, $\gamma\coin{\Gamma}{d''}\tail\gamma_{n+2}$. The desired claim follows.
\item If $\monadic{v}_1\ne\return i_1$ or $\monadic{v}_2\ne\return i_2$ with $i_1,i_2\in\NN$ then $\gamma'=\gamma_2$. By the induction hypothesis, $\gamma\coin{\Gamma}{d''}\gamma_1$ and $\gamma_1\coin{\Gamma}{d''}\gamma_2$. By Lemma~\ref{semantics:coinlemma}, $\gamma\coin{\Gamma}{d''}\gamma_2$. The desired claim follows.
\end{itemize}

\item Let $e=\castexpr{e_1}{d_0}$. Then
\[
\Gamma\vd e_1:\qualty{t_1}{s_1}{d_1}\eff D\mbox{,}
\]
and
\[
\begin{array}{l}
\semd{e_1}\gamma\phi=\return(\monadic{v}_1,\gamma_1,\omicron_1)\mbox{,}\\
\gamma'=\gamma_1\mbox{.}
\end{array}
\]
By the induction hypothesis, $\gamma\coin{\Gamma}{d''}\gamma_1$. The desired claim follows.

\item Let $e=(\assignexpr{e_1}{e_2})$. Let $e_1=\loadexpr{\loadexpr{\loadexpr{x}{y_1}}{y_2}\ldots}{y_n}$. By Lemma~\ref{typesystem:lhslemma},
\[
\begin{array}{l}
\Gamma\vd y_i:\qualty{\uintty}{\prestg}{d_i}\eff D_i\mbox{,}\\
\Gamma\vd x:\qualty{\listty{\ldots\listty{\qualty{\listty{\qualty{\listty{\qualty{t}{s}{d_{n+1}}}}{\prestg}{d_n}}}{\prestg}{d_{n-1}}}\ldots}}{\prestg}{d_1}\eff D_0
\end{array}
\]
and
\[
\begin{array}{l}
\Gamma\vd e_1:\qualty{t}{s}{d_{n+1}}\eff D_1\cup\ldots\cup D_n\mbox{,}\\
\Gamma\vd e_2:\qualty{t}{s}{d_{n+1}}\eff D_{n+1}\mbox{,}\\
D=\gen{s}\cup\gen{d_{n+1}}\cup D_1\cup\ldots\cup D_{n+1}\mbox{,}
\end{array}
\]
for some $D_1,\ldots,D_{n+1}$. Hence we must have $d''\notin D_k$ for every $k=1,\ldots,n+1$ and $d''\notin\gen{d_{n+1}}$ meaning that $d_{n+1}$ is a strict superdomain of~$d''$. We also have
\[
\begin{array}{l}
\semd{y_1}\gamma\phi=\return(\monadic{i}_1,\gamma_1,\omicron_1)\mbox{,}\\
\semd{y_2}\gamma_1\phi=\return(\monadic{i}_2,\gamma_2,\omicron_2)\mbox{,}\\
\dotfill\\
\semd{y_n}\gamma_{n-1}\phi=\return(\monadic{i}_n,\gamma_n,\omicron_n)\mbox{,}\\
\semd{e_2}\gamma_n\phi=\return(\monadic{v},\gamma_{n+1},\omicron_{n+1})\mbox{,}\\
\gamma'=[x\mapsto\update(\lookup{\gamma}{x},\monadic{i}_1\ldots\monadic{i}_n,\monadic{v})]\gamma_{n+1}\mbox{.}
\end{array}
\]
By the induction hypothesis, $\gamma_{k-1}\coin{\Gamma}{d''}\gamma_k$ for each $k=1,\ldots,n+1$ (denoting $\gamma_0=\gamma$). Let $\vars\Gamma=(z_1,\ldots,z_l)$ and let the corresponding types be $q_1,\ldots,q_l$. Let $k$ be the least index such that $z_k=x$. For every $i\ne k$, we have $\lookup{\gamma_{n+1}}{z_i}=\lookup{\gamma'}{z_i}$. Hence, by Lemma~\ref{semantics:coinlemma}, $\lookup{\gamma}{z_i}\coin{q_i}{d''}\lookup{\gamma'}{z_i}$. Moreover, Lemma~\ref{semantics:updcoin0lemma} implies $\lookup{\gamma}{z_k}\coin{q_k}{d''}\update(\lookup{\gamma}{z_k},\monadic{i}_1\ldots\monadic{i}_n,\monadic{v})$. Consequently, $\gamma\coin{\Gamma}{d''}\gamma'$. The desired claim follows.

\item Let $e=\loadexpr{e_1}{e_2}$. Then
\[
\begin{array}{l}
\Gamma\vd e_1:\qualty{\listty{q}}{\prestg}{d_1}\eff D_1\mbox{,}\\
\Gamma\vd e_2:\qualty{\uintty}{\prestg}{d_1}\eff D_2\mbox{,}\\
D=D_1\cup D_2\mbox{,}
\end{array}
\]
whence we have $d''\notin D_1$ and $d''\notin D_2$. We also have
\[
\begin{array}{l}
\semd{e_1}\gamma\phi=\return(\monadic{a},\gamma_1,\omicron_1)\mbox{,}\\
\semd{e_2}\gamma_1\phi=\return(\monadic{i},\gamma_2,\omicron_2)\mbox{,}\\
\gamma'=\gamma_2\mbox{.}
\end{array}
\]
By the induction hypothesis, $\gamma\coin{\Gamma}{d''}\gamma_1$ and $\gamma_1\coin{\Gamma}{d''}\gamma_2$. The desired claim follows by Lemma~\ref{semantics:coinlemma}.

\item Let $e=(\stmtcomp{\letexpr{x}{e_1}}{e_2})$. Then
\[
\begin{array}{l}
\Gamma\vd e_1:\qualty{t_1}{s_1}{d_1}\eff D_1\mbox{,}\\
(x:\qualty{t_1}{s_1}{d_1}),\Gamma\vd e_2:q\mbox{,}\\
D=\gen{d_1}\cup D_1\cup D_2\mbox{,}
\end{array}
\]
whence $d''\notin D_1$ and $d''\notin D_2$. We also have
\[
\begin{array}{l}
\semd{e_1}\gamma\phi=\return(\monadic{v}_1,\gamma_1,\omicron_1)\mbox{,}\\
\semd{e_2}((x,\monadic{v}_1),\gamma_1)\phi=\return(\monadic{v}_2,\gamma_2,\omicron_2)\mbox{,}\\
\gamma'=\tail\gamma_2\mbox{.}
\end{array}
\]
The induction hypothesis implies $\gamma\coin{\Gamma}{d''}\gamma_1$ and $(x,\monadic{v}_1),\gamma_1\coin{\Gamma'}{d''}\gamma_2$ where $\Gamma'=((x:\qualty{t_1}{s_1}{d_1}),\Gamma)$. The latter implies $\gamma_1\coin{\Gamma}{d''}\tail\gamma_2$. The desired claim follows by Lemma~\ref{semantics:coinlemma}.

\item Let $e=\stmtcomp{e_1}{e_2}$. Then
\[
\begin{array}{l}
\Gamma\vd e_1:\qualty{t_1}{s_1}{d_1}\eff D_1\mbox{,}\\
\Gamma\vd e_2:q\mbox{,}\\
D=D_1\cup D_2\mbox{,}
\end{array}
\]
whence $d''\notin D_1$ and $d''\notin D_2$. We also have
\[
\begin{array}{l}
\semd{e_1}\gamma\phi=\return(\monadic{v}_1,\gamma_1,\omicron_1)\mbox{,}\\
\semd{e_2}\gamma_1\phi=\return(\monadic{v}_2,\gamma_2,\omicron_2)\mbox{,}\\
\gamma'=\gamma_2\mbox{.}
\end{array}
\]
By the induction hypothesis, $\gamma\coin{\Gamma}{d''}\gamma_1$ and $\gamma_1\coin{\Gamma}{d''}\gamma_2$. The desired claim follows by Lemma~\ref{semantics:coinlemma}.

\end{itemize}

\end{proof}

\begin{theorem}[Theorem~\ref{semantics:outputthm}]
Let $\Gamma \vd e : q\eff D$ with well-structured~$\Gamma$ and $\gamma\in\mathbf{Env}$ be $\Gamma$-exact in~$d$ for some domain~$d$. Assume that for all subexpressions of~$e$ of the form $\ofty{\getexpr{d'}{k}}{q'}$ where $d'\subtype d$, the value $\allpure(\phi_{d'}(k))$ is $q'$-exact in~$d$. If $\semd{e}\gamma\phi=\return(\monadic{v},\gamma',\omicron)$ and $\publicdom\notin D$ then $\omicron=\epsilon$.
\end{theorem}

\begin{proof}
Let $q=\qualty{t_0}{s_0}{d_0}$. We proceed by induction on the structure of~$e$. 

If $e=\epsilon$ or $e=\overline{n}$ where $n\in\NN$ or $e=\overline{b}$ where $b\in\BB$ or $e=x$ or $e=\getexpr{d'}{k}$ then $\omicron=\epsilon$. If $e=\assertexpr{e_1}$ or $e=\wireexpr{e_1}$ then the claim holds vacuously since $D=\gen{\publicdom}$. We study the remaining cases.

\begin{itemize}
\item Let $e=\addexpr{e_1}{e_2}$. Then
\[
\begin{array}{l}
\Gamma\vd e_1:q\eff D_1\mbox{,}\\
\Gamma\vd e_2:q\eff D_2\mbox{,}\\
D=\gen{s_0}\cup D_1\cup D_2\mbox{,}
\end{array}
\]
whence we must have $\publicdom\notin D_1$ and $\publicdom\notin D_2$. We also have
\[
\begin{array}{l}
\semd{e_1}\gamma\phi=\return(\monadic{v}_1,\gamma_1,\omicron_1)\mbox{,}\\
\semd{e_2}\gamma_1\phi=\return(\monadic{v}_2,\gamma_2,\omicron_2)\mbox{,}\\
\omicron=\omicron_1\omicron_2\mbox{.}
\end{array}
\]
By Theorem~\ref{semantics:localthm}, $\gamma_1$ and $\gamma_2$ are $q$-exact in~$d$. By the induction hypothesis, $\omicron_1=\epsilon$ and $\omicron_2=\epsilon$. Hence $\omicron=\epsilon$.

\item Let $e=\ifexpr{e_1}{e_2}{e_3}$. Then
\[
\begin{array}{l}
\Gamma\vd e_1:\qualty{\boolmodty{\mbox{\lstinline+N+}}}{\prestg}{d_1}\eff D_1\mbox{,}\\
\Gamma\vd e_2:q\eff D_2\mbox{,}\\
\Gamma\vd e_3:q\eff D_3\mbox{,}\\
\gen{d_1}\supseteq\gen{s_0}\cup\gen{d_0}\cup D_2\cup D_3\mbox{,}\quad D=D_1\cup D_2\cup D_3\mbox{,}
\end{array}
\]
whence we must have $\publicdom\notin D_1$, $\publicdom\notin D_2$ and $\publicdom\notin D_3$. We also have
\[
\begin{array}{l}
\semd{e_1}\gamma\phi=\return(\monadic{v}_1,\gamma_1,\omicron_1)\mbox{.}
\end{array}
\]
Consider three cases:
\begin{itemize}
\item If $\monadic{v}_1=\return\tru$ then
\[
\begin{array}{l}
\semd{e_2}\gamma_1\phi=\return(\monadic{v}_2,\gamma_2,\omicron_2)\mbox{,}\\
\omicron=\omicron_1\omicron_2\mbox{.}
\end{array}
\]
By Theorem~\ref{semantics:localthm}, $\gamma_1$ and $\gamma_2$ are $\Gamma$-exact in~$d$. By the induction hypothesis, $\omicron_1=\epsilon$ and $\omicron_2=\epsilon$. Hence $\omicron=\epsilon$.
\item If $\monadic{v}_1=\return\fls$ then
\[
\begin{array}{l}
\semd{e_3}\gamma_1\phi=\return(\monadic{v}_3,\gamma_3,\omicron_3)\mbox{,}\\
\omicron=\omicron_1\omicron_3\mbox{.}
\end{array}
\]
By Theorem~\ref{semantics:localthm}, $\gamma_1$ and $\gamma_3$ are $\Gamma$-exact in~$d$. By the induction hypothesis, $\omicron_1=\epsilon$ and $\omicron_3=\epsilon$. Hence $\omicron=\epsilon$.
\item If $\monadic{v}_1\ne\return b$ where $b\in\BB$ then $\omicron=\omicron_1$. By the induction hypothesis, $\omicron_1=\epsilon$. The desired claim follows.
\end{itemize}

\item Let $e=\forexpr{x}{e_1}{e_2}{e_3}$. Then
\[
\begin{array}{l}
\Gamma\vd e_1:\qualty{\uintty}{\prestg}{d_0}\eff D_1\mbox{,}\\
\Gamma\vd e_2:\qualty{\uintty}{\prestg}{d_0}\eff D_2\mbox{,}\\
(x:\qualty{\uintty}{\prestg}{d_0}),\Gamma\vd e_3:\qualty{t_1}{s_1}{d_1}\eff D_3\mbox{,}\\
\gen{d_0}\supseteq\gen{s_1}\cup\gen{d_1}\cup D_3\mbox{,}\quad D=D_1\cup D_2\cup D_3\mbox{,}
\end{array}
\]
whence we have $\publicdom\notin D_1$, $\publicdom\notin D_2$ and $\publicdom\notin D_3$. We also have
\[
\begin{array}{l}
\semd{e_1}\gamma\phi=\return(\monadic{v}_1,\gamma_1,\omicron_1)\mbox{,}\\
\semd{e_2}\gamma_1\phi=\return(\monadic{v}_2,\gamma_2,\omicron_2)\mbox{.}
\end{array}
\]
By Theorem~\ref{semantics:localthm}, $\monadic{v}_1$ and $\monadic{v}_2$ are $(\qualty{\uintty}{\prestg}{d_0})$-exact in~$d$ and $\gamma_1$ and $\gamma_2$ are $\Gamma$-exact in~$d$. Consider two cases:
\begin{itemize}
\item If $\monadic{v}_1=\return i_1$ and $\monadic{v}_2=\return i_2$ for $i_1,i_2\in\NN$ then, denoting $n=\max(0,i_2-i_1)$, we have
\[
\begin{array}{l}
\semd{e_3}((x,\return i_1),\gamma_2)\phi=\return(\monadic{v}_3,\gamma_3,\omicron_3)\mbox{,}\\
\semd{e_3}([x\mapsto\return(i_1+1)]\gamma_3)\phi=\return(\monadic{v}_4,\gamma_4,\omicron_4)\mbox{,}\\
\dotfill\\
\semd{e_3}([x\mapsto\return(i_1+n-1)]\gamma_{n+1})\phi=\return(\monadic{v}_{n+2},\gamma_{n+2},\omicron_{n+2})\mbox{,}\\
\omicron=\omicron_1\ldots\omicron_{n+2}\mbox{.}
\end{array}
\]
By the above, $((x,\return i_1),\gamma_2)$ is $\Gamma'$-exact in~$d$ where $\Gamma'=((x:\qualty{\uintty}{\prestg}{d_0}),\Gamma)$. By Theorem~\ref{semantics:localthm}, $\gamma_3,\ldots,\gamma_{n+2}$ are also $\Gamma'$-exact in~$d$ since increasing $i_1$ does not change this property. Hence by the induction hypothesis, $\omicron_1=\ldots=\omicron_{n+2}=\epsilon$, implying $\omicron=\epsilon$.
\item If $\monadic{v}_1\ne\return i_1$ or $\monadic{v}_2\ne\return i_2$ with $i_1,i_2\in\NN$ then $\omicron=\omicron_1\omicron_2$. By the induction hypothesis, $\omicron_1=\epsilon$ and $\omicron_2=\epsilon$. Hence $\omicron=\epsilon$.
\end{itemize}

\item Let $e=\castexpr{e_1}{d_0}$. Then
\[
\Gamma\vd e_1:\qualty{t_1}{s_1}{d_1}\eff D\mbox{,}
\]
and
\[
\begin{array}{l}
\semd{e_1}\gamma\phi=\return(\monadic{v}_1,\gamma_1,\omicron_1)\mbox{,}\\
\omicron=\omicron_1\mbox{.}
\end{array}
\]
By the induction hypothesis, $\omicron_1=\epsilon$. The desired claim follows.

\item Let $e=(\assignexpr{e_1}{e_2})$. Let $e_1=\loadexpr{\loadexpr{\loadexpr{x}{y_1}}{y_2}\ldots}{y_n}$. By Lemma~\ref{typesystem:lhslemma},
\[
\begin{array}{l}
\Gamma\vd y_i:\qualty{\uintty}{\prestg}{d_i}\eff D_i\mbox{,}\\
\Gamma\vd x:\qualty{\listty{\ldots\listty{\qualty{\listty{\qualty{\listty{\qualty{t}{s}{d_{n+1}}}}{\prestg}{d_n}}}{\prestg}{d_{n-1}}}\ldots}}{\prestg}{d_1}\eff D_0
\end{array}
\]
and
\[
\begin{array}{l}
\Gamma\vd e_1:\qualty{t}{s}{d_{n+1}}\eff D_1\cup\ldots\cup D_n\mbox{,}\\
\Gamma\vd e_2:\qualty{t}{s}{d_{n+1}}\eff D_{n+1}\mbox{,}\\
D=\gen{s}\cup\gen{d_{n+1}}\cup D_1\cup\ldots\cup D_{n+1}\mbox{,}
\end{array}
\]
for some $D_1,\ldots,D_{n+1}$. Hence we must have $\publicdom\notin D_k$ for every $k=1,\ldots,n+1$. We also have
\[
\begin{array}{l}
\semd{y_1}\gamma\phi=\return(\monadic{i}_1,\gamma_1,\omicron_1)\mbox{,}\\
\semd{y_2}\gamma_1\phi=\return(\monadic{i}_2,\gamma_2,\omicron_2)\mbox{,}\\
\dotfill\\
\semd{y_n}\gamma_{n-1}\phi=\return(\monadic{i}_n,\gamma_n,\omicron_n)\mbox{,}\\
\semd{e_2}\gamma_n\phi=\return(\monadic{v},\gamma_{n+1},\omicron_{n+1})\mbox{,}\\
\omicron=\omicron_1\ldots\omicron_{n+1}\mbox{.}
\end{array}
\]
By Theorem~\ref{semantics:localthm}, $\gamma_1,\ldots,\gamma_n$ are $\Gamma$-exact in~$d$. By the induction hypothesis, $\omicron_1=\ldots=\omicron_{n+1}=\epsilon$, implying $\omicron=\epsilon$.

\item Let $e=\loadexpr{e_1}{e_2}$. Then
\[
\begin{array}{l}
\Gamma\vd e_1:\qualty{\listty{q}}{\prestg}{d_1}\eff D_1\mbox{,}\\
\Gamma\vd e_2:\qualty{\uintty}{\prestg}{d_1}\eff D_2\mbox{,}\\
D=D_1\cup D_2\mbox{,}
\end{array}
\]
whence we have $\publicdom\notin D_1$ and $\publicdom\notin D_2$. We also have
\[
\begin{array}{l}
\semd{e_1}\gamma\phi=\return(\monadic{a},\gamma_1,\omicron_1)\mbox{,}\\
\semd{e_2}\gamma_1\phi=\return(\monadic{i},\gamma_2,\omicron_2)\mbox{,}\\
\omicron=\omicron_1\omicron_2\mbox{.}
\end{array}
\]
By Theorem~\ref{semantics:localthm}, $\gamma_1$ and $\gamma_2$ are $\Gamma$-exact in~$d$. By the induction hypothesis, $\omicron_1=\epsilon$ and $\omicron_2=\epsilon$. Hence $\omicron=\epsilon$.

\item Let $e=(\stmtcomp{\letexpr{x}{e_1}}{e_2})$. Then
\[
\begin{array}{l}
\Gamma\vd e_1:\qualty{t_1}{s_1}{d_1}\eff D_1\mbox{,}\\
(x:\qualty{t_1}{s_1}{d_1}),\Gamma\vd e_2:q\mbox{,}\\
D=\gen{d_1}\cup D_1\cup D_2\mbox{,}
\end{array}
\]
whence $\publicdom\notin D_1$ and $\publicdom\notin D_2$. We also have
\[
\begin{array}{l}
\semd{e_1}\gamma\phi=\return(\monadic{v}_1,\gamma_1,\omicron_1)\mbox{,}\\
\semd{e_2}((x,\monadic{v}_1),\gamma_1)\phi=\return(\monadic{v}_2,\gamma_2,\omicron_2)\mbox{,}\\
\omicron=\omicron_1\omicron_2\mbox{.}
\end{array}
\]
By Theorem~\ref{semantics:localthm}, $\monadic{v}_1$ is $(\qualty{t_1}{s_1}{d_1})$-exact and $\gamma_1$ is $\Gamma$-exact in~$d$. Hence $\left((x,\monadic{v}_1),\gamma_1\right)$ is $\Gamma'$-exact in~$d$ where $\Gamma'=((x:\qualty{t_1}{s_1}{d_1}),\Gamma)$. The induction hypothesis implies $\omicron_1=\epsilon$ and $\omicron_2=\epsilon$. Hence $\omicron=\epsilon$.

\item Let $e=\stmtcomp{e_1}{e_2}$. Then
\[
\begin{array}{l}
\Gamma\vd e_1:\qualty{t_1}{s_1}{d_1}\eff D_1\mbox{,}\\
\Gamma\vd e_2:q\mbox{,}\\
D=D_1\cup D_2\mbox{,}
\end{array}
\]
whence $\publicdom\notin D_1$ and $\publicdom\notin D_2$. We also have
\[
\begin{array}{l}
\semd{e_1}\gamma\phi=\return(\monadic{v}_1,\gamma_1,\omicron_1)\mbox{,}\\
\semd{e_2}\gamma_1\phi=\return(\monadic{v}_2,\gamma_2,\omicron_2)\mbox{,}\\
\omicron=\omicron_1\omicron_2\mbox{.}
\end{array}
\]
By Theorem~\ref{semantics:localthm}, $\gamma_1$ and $\gamma_2$ are $\Gamma$-exact in~$d$. By the induction hypothesis, $\omicron_1=\epsilon$ and $\omicron_2=\epsilon$. Hence $\omicron=\epsilon$.

\end{itemize}

\end{proof}

\begin{theorem}[Theorem~\ref{semantics:effcircthm}]
Let $\Gamma \vd e : q\eff D$ with well-structured~$\Gamma$ and $\gamma\in\mathbf{Env}$ be $\Gamma$-exact in $d$ for $d=\proverdom$. Assume that for all subexpressions of~$e$ of the form $\ofty{\getexpr{d'}{k}}{q'}$, the value $\allpure(\phi_{d'}(k))$ is $q'$-exact in~$d$. If $\semd{e}\gamma\phi=\return(\monadic{v},\gamma',\omicron)$ and $\publicdom\notin D$ then $\gamma\coincirc{\Gamma}\gamma'$.
\end{theorem}

\begin{proof}
By Theorem~\ref{semantics:localthm}, $\gamma'$ is $\Gamma$-exact in $\proverdom$. By Lemma~\ref{semantics:reintlemma}, $\gamma$ and $\gamma'$ are $\Gamma$-exposed in $\proverdom$. By Lemma~\ref{semantics:approxlemma}, $\gamma$ and $\gamma'$ are also $\Gamma$-exposed in circuit. Let $q=\qualty{t_0}{s_0}{d_0}$. We proceed by induction on the structure of~$e$. 

If $e=\epsilon$ or $e=\overline{n}$ where $n\in\NN$ or $e=\overline{b}$ where $b\in\BB$ or $e=x$ or $e=\getexpr{d'}{k}$ then $\gamma=\gamma'$. By Lemma~\ref{semantics:coinlemma}, $\gamma\coincirc{\Gamma}\gamma'$. If $e=\assertexpr{e_1}$ or $e=\wireexpr{e_1}$ then the claim holds vacuously since $D=\gen{\publicdom}$. We study the remaining cases.

\begin{itemize}
\item Let $e=\addexpr{e_1}{e_2}$. Then
\[
\begin{array}{l}
\Gamma\vd e_1:q\eff D_1\mbox{,}\\
\Gamma\vd e_2:q\eff D_2\mbox{,}\\
D=\gen{s_0}\cup D_1\cup D_2\mbox{,}
\end{array}
\]
whence we must have $\publicdom\notin D_1$ and $\publicdom\notin D_2$. We also have
\[
\begin{array}{l}
\semd{e_1}\gamma\phi=\return(\monadic{v}_1,\gamma_1,\omicron_1)\mbox{,}\\
\semd{e_2}\gamma_1\phi=\return(\monadic{v}_2,\gamma_2,\omicron_2)\mbox{,}\\
\gamma'=\gamma_2\mbox{.}
\end{array}
\]
By the induction hypothesis, $\gamma\coincirc{\Gamma}\gamma_1$ and $\gamma_1\coincirc{\Gamma}\gamma_2$. By Lemma~\ref{semantics:coinlemma}, $\gamma\coincirc{\Gamma}\gamma_2$. Hence $\gamma\coincirc{\Gamma}\gamma'$.

\item Let $e=\ifexpr{e_1}{e_2}{e_3}$. Then
\[
\begin{array}{l}
\Gamma\vd e_1:\qualty{\boolmodty{\mbox{\lstinline+N+}}}{\prestg}{d_1}\eff D_1\mbox{,}\\
\Gamma\vd e_2:q\eff D_2\mbox{,}\\
\Gamma\vd e_3:q\eff D_3\mbox{,}\\
\gen{d_1}\supseteq\gen{s_0}\cup\gen{d_0}\cup D_2\cup D_3\mbox{,}\quad D=D_1\cup D_2\cup D_3\mbox{,}
\end{array}
\]
whence we must have $\publicdom\notin D_1$, $\publicdom\notin D_2$ and $\publicdom\notin D_3$. We also have
\[
\begin{array}{l}
\semd{e_1}\gamma\phi=\return(\monadic{v}_1,\gamma_1,\omicron_1)\mbox{.}
\end{array}
\]
By Theorem~\ref{semantics:localthm}, $\monadic{v}_1$ is $(\qualty{\boolmodty{\mbox{\lstinline+N+}}}{\prestg}{d_1})$-exact in $\proverdom$, whence we have to consider two cases:
\begin{itemize}
\item If $\monadic{v}_1=\return\tru$ then
\[
\begin{array}{l}
\semd{e_2}\gamma_1\phi=\return(\monadic{v}_2,\gamma_2,\omicron_2)\mbox{,}\\
\gamma'=\gamma_2\mbox{.}
\end{array}
\]
By the induction hypothesis, $\gamma\coincirc{\Gamma}\gamma_1$ and $\gamma_1\coincirc{\Gamma}\gamma_2$. By Lemma~\ref{semantics:coinlemma}, $\gamma\coincirc{\Gamma}\gamma_2$. Hence $\gamma\coincirc{\Gamma}\gamma'$.
\item If $\monadic{v}_1=\return\fls$ then
\[
\begin{array}{l}
\semd{e_3}\gamma_1\phi=\return(\monadic{v}_3,\gamma_3,\omicron_3)\mbox{,}\\
\gamma'=\gamma_3\mbox{.}
\end{array}
\]
By the induction hypothesis, $\gamma\coincirc{\Gamma}\gamma_1$ and $\gamma_1\coincirc{\Gamma}\gamma_3$. By Lemma~\ref{semantics:coinlemma}, $\gamma\coincirc{\Gamma}\gamma_3$. Hence $\gamma\coincirc{\Gamma}\gamma'$.
\end{itemize}

\item Let $e=\forexpr{x}{e_1}{e_2}{e_3}$. Then
\[
\begin{array}{l}
\Gamma\vd e_1:\qualty{\uintty}{\prestg}{d_0}\eff D_1\mbox{,}\\
\Gamma\vd e_2:\qualty{\uintty}{\prestg}{d_0}\eff D_2\mbox{,}\\
(x:\qualty{\uintty}{\prestg}{d_0}),\Gamma\vd e_3:\qualty{t_1}{s_1}{d_1}\eff D_3\mbox{,}\\
\gen{d_0}\supseteq\gen{s_1}\cup\gen{d_1}\cup D_3\mbox{,}\quad D=D_1\cup D_2\cup D_3\mbox{,}
\end{array}
\]
whence we have $\publicdom\notin D_1$, $\publicdom\notin D_2$ and $\publicdom\notin D_3$. We also have
\[
\begin{array}{l}
\semd{e_1}\gamma\phi=\return(\monadic{v}_1,\gamma_1,\omicron_1)\mbox{,}\\
\semd{e_2}\gamma_1\phi=\return(\monadic{v}_2,\gamma_2,\omicron_2)\mbox{.}
\end{array}
\]
By Theorem~\ref{semantics:localthm}, $\monadic{v}_1$ and $\monadic{v}_2$ are $(\qualty{\uintty}{\prestg}{d_0})$-exact in $\proverdom$, whence $\monadic{v}_1=\return i_1$ and $\monadic{v}_2=\return i_2$ for $i_1,i_2\in\NN$. Denoting $n=\max(0,i_2-i_1)$, we have
\[
\begin{array}{l}
\semd{e_3}((x,\return i_1),\gamma_2)\phi=\return(\monadic{v}_3,\gamma_3,\omicron_3)\mbox{,}\\
\semd{e_3}([x\mapsto\return(i_1+1)]\gamma_3)\phi=\return(\monadic{v}_4,\gamma_4,\omicron_4)\mbox{,}\\
\dotfill\\
\semd{e_3}([x\mapsto\return(i_1+n-1)]\gamma_{n+1})\phi=\return(\monadic{v}_{n+2},\gamma_{n+2},\omicron_{n+2})\mbox{,}\\
\gamma'=\tail\gamma_{n+2}\mbox{.}
\end{array}
\]
By the induction hypothesis, $\gamma\coincirc{\Gamma}\gamma_1$ and $\gamma_1\coincirc{\Gamma}\gamma_2$, and also $\left((x,\return i_1),\gamma_2\right)\coincirc{\Gamma'}\gamma_3$ and $[x\mapsto\return(i_1+k-3)]\gamma_{k-1}\coincirc{\Gamma'}\gamma_k$ for each $k=4,\ldots,n+2$ where $\Gamma'=((x:\qualty{\uintty}{\prestg}{d_0}),\Gamma)$. Then also $\gamma_2\coincirc{\Gamma}\tail\gamma_3$ and $\tail\gamma_{k-1}\coincirc{\Gamma}\tail\gamma_k$ for every $k=4,\ldots,n+2$. By Lemma~\ref{semantics:coinlemma}, $\gamma\coincirc{\Gamma}\tail\gamma_{n+2}$. Hence $\gamma\coincirc{\Gamma}\gamma'$.

\item Let $e=\castexpr{e_1}{d_0}$. Then
\[
\Gamma\vd e_1:\qualty{t_1}{s_1}{d_1}\eff D\mbox{,}
\]
and
\[
\begin{array}{l}
\semd{e_1}\gamma\phi=\return(\monadic{v}_1,\gamma_1,\omicron_1)\mbox{,}\\
\gamma'=\gamma_1\mbox{.}
\end{array}
\]
By the induction hypothesis, $\gamma\coincirc{\Gamma}\gamma_1$. Hence $\gamma\coincirc{\Gamma}\gamma'$.

\item Let $e=(\assignexpr{e_1}{e_2})$. Let $e_1=\loadexpr{\loadexpr{\loadexpr{x}{y_1}}{y_2}\ldots}{y_n}$. By Lemma~\ref{typesystem:lhslemma},
\[
\begin{array}{l}
\Gamma\vd y_i:\qualty{\uintty}{\prestg}{d_i}\eff D_i\mbox{,}\\
\Gamma\vd x:\qualty{\listty{\ldots\listty{\qualty{\listty{\qualty{\listty{\qualty{t}{s}{d_{n+1}}}}{\prestg}{d_n}}}{\prestg}{d_{n-1}}}\ldots}}{\prestg}{d_1}\eff D_0
\end{array}
\]
and
\[
\begin{array}{l}
\Gamma\vd e_1:\qualty{t}{s}{d_{n+1}}\eff D_1\cup\ldots\cup D_n\mbox{,}\\
\Gamma\vd e_2:\qualty{t}{s}{d_{n+1}}\eff D_{n+1}\mbox{,}\\
D=\gen{s}\cup\gen{d_{n+1}}\cup D_1\cup\ldots\cup D_{n+1}\mbox{,}
\end{array}
\]
for some $D_1,\ldots,D_{n+1}$. Hence we must have $\publicdom\notin D_k$ for every $k=1,\ldots,n+1$ and $\publicdom\notin\gen{d_{n+1}}$ meaning that $d_{n+1}$ is a strict superdomain of~$\publicdom$. We also have
\[
\begin{array}{l}
\semd{y_1}\gamma\phi=\return(\monadic{i}_1,\gamma_1,\omicron_1)\mbox{,}\\
\semd{y_2}\gamma_1\phi=\return(\monadic{i}_2,\gamma_2,\omicron_2)\mbox{,}\\
\dotfill\\
\semd{y_n}\gamma_{n-1}\phi=\return(\monadic{i}_n,\gamma_n,\omicron_n)\mbox{,}\\
\semd{e_2}\gamma_n\phi=\return(\monadic{v},\gamma_{n+1},\omicron_{n+1})\mbox{,}\\
\gamma'=[x\mapsto\update(\lookup{\gamma}{x},\monadic{i}_1\ldots\monadic{i}_n,\monadic{v})]\gamma_{n+1}\mbox{.}
\end{array}
\]
By the induction hypothesis, $\gamma_{k-1}\coincirc{\Gamma}\gamma_k$ for each $k=1,\ldots,n+1$ (denoting $\gamma_0=\gamma$). Let $\vars\Gamma=(z_1,\ldots,z_l)$ and let the corresponding types be $q_1,\ldots,q_l$. Let $k$ be the least index such that $z_k=x$. For every $i\ne k$, we have $\lookup{\gamma_{n+1}}{z_i}=\lookup{\gamma'}{z_i}$. Hence, by Lemma~\ref{semantics:coinlemma}, $\lookup{\gamma}{z_i}\coincirc{q_i}\lookup{\gamma'}{z_i}$. Moreover, Lemma~\ref{semantics:updcoin0lemma} implies $\lookup{\gamma}{z_k}\coincirc{q_k}\update(\lookup{\gamma}{z_k},\monadic{i}_1\ldots\monadic{i}_n,\monadic{v})$. Consequently, $\gamma\coincirc{\Gamma}\gamma'$. The desired claim follows.

\item Let $e=\loadexpr{e_1}{e_2}$. Then
\[
\begin{array}{l}
\Gamma\vd e_1:\qualty{\listty{q}}{\prestg}{d_1}\eff D_1\mbox{,}\\
\Gamma\vd e_2:\qualty{\uintty}{\prestg}{d_1}\eff D_2\mbox{,}\\
D=D_1\cup D_2\mbox{,}
\end{array}
\]
whence we have $\publicdom\notin D_1$ and $\publicdom\notin D_2$. We also have
\[
\begin{array}{l}
\semd{e_1}\gamma\phi=\return(\monadic{a},\gamma_1,\omicron_1)\mbox{,}\\
\semd{e_2}\gamma_1\phi=\return(\monadic{i},\gamma_2,\omicron_2)\mbox{,}\\
\gamma'=\gamma_2\mbox{.}
\end{array}
\]
By the induction hypothesis, $\gamma\coincirc{\Gamma}\gamma_1$ and $\gamma_1\coincirc{\Gamma}\gamma_2$. By Lemma~\ref{semantics:coinlemma}, $\gamma\coincirc{\Gamma}\gamma_2$. Hence $\gamma\coincirc{\Gamma}\gamma'$.

\item Let $e=(\stmtcomp{\letexpr{x}{e_1}}{e_2})$. Then
\[
\begin{array}{l}
\Gamma\vd e_1:\qualty{t_1}{s_1}{d_1}\eff D_1\mbox{,}\\
(x:\qualty{t_1}{s_1}{d_1}),\Gamma\vd e_2:q\mbox{,}\\
D=\gen{d_1}\cup D_1\cup D_2\mbox{,}
\end{array}
\]
whence $\publicdom\notin D_1$ and $\publicdom\notin D_2$. We also have
\[
\begin{array}{l}
\semd{e_1}\gamma\phi=\return(\monadic{v}_1,\gamma_1,\omicron_1)\mbox{,}\\
\semd{e_2}((x,\monadic{v}_1),\gamma_1)\phi=\return(\monadic{v}_2,\gamma_2,\omicron_2)\mbox{,}\\
\gamma'=\tail\gamma_2\mbox{.}
\end{array}
\]
The induction hypothesis implies $\gamma\coincirc{\Gamma}\gamma_1$ and $\left((x,\monadic{v}_1),\gamma_1\right)\coincirc{\Gamma'}\gamma_2$ where $\Gamma'=((x:\qualty{t_1}{s_1}{d_1}),\Gamma)$. The latter implies $\gamma_1\coincirc{\Gamma}\tail\gamma_2$. The desired claim follows by Lemma~\ref{semantics:coinlemma}.

\item Let $e=\stmtcomp{e_1}{e_2}$. Then
\[
\begin{array}{l}
\Gamma\vd e_1:\qualty{t_1}{s_1}{d_1}\eff D_1\mbox{,}\\
\Gamma\vd e_2:q\mbox{,}\\
D=D_1\cup D_2\mbox{,}
\end{array}
\]
whence $\publicdom\notin D_1$ and $\publicdom\notin D_2$. We also have
\[
\begin{array}{l}
\semd{e_1}\gamma\phi=\return(\monadic{v}_1,\gamma_1,\omicron_1)\mbox{,}\\
\semd{e_2}\gamma_1\phi=\return(\monadic{v}_2,\gamma_2,\omicron_2)\mbox{,}\\
\gamma'=\gamma_2\mbox{.}
\end{array}
\]
By the induction hypothesis, $\gamma\coincirc{\Gamma}\gamma_1$ and $\gamma_1\coincirc{\Gamma}\gamma_2$. The desired claim follows by Lemma~\ref{semantics:coinlemma}.

\end{itemize}

\end{proof}

\begin{theorem}[Theorem~\ref{semantics:safety}]
If $\semd{e}\gamma\phi=\return(\monadic{v},\gamma',\omicron)$ and $\phi'_{d'}=\phi_{d'}$ for every $d'\subtype d$ then $\semd{e}\gamma\phi'=\semd{e}\gamma\phi$.
\end{theorem}

\begin{proof}
Let $\phi':\In^3$ be such that $\phi_{d'}=\phi'_{d'}$ for every $d'\subtype d$. We proceed by induction on the structure of~$e$. 

\begin{itemize}
\item If $e=\epsilon$, $e=\overline{n}$ where $n\in\NN$, $e=\overline{b}$ where $b\in\BB$, or $e=x$, then $\semd{e}\gamma\phi'=\semd{e}\gamma\phi$ because the semantics in these cases does not depend on the input dictionaries.

\item Let $e=\addexpr{e_1}{e_2}$. Then
\[
\begin{array}{l}
\semd{e_1}\gamma\phi=\return(\monadic{v}_1,\gamma_1,\omicron_1)\mbox{,}\\
\semd{e_2}\gamma_1\phi=\return(\monadic{v}_2,\gamma_2,\omicron_2)\mbox{,}\\
\monadic{v}=\mcomp{v_1\gets\monadic{v}_1\hstop v_2\gets\monadic{v}_2\hstop\return(v_1+v_2)}\mbox{,}\quad\gamma'=\gamma_2\mbox{,}\quad\omicron=\omicron_1\omicron_2\mbox{.}
\end{array}
\]
By the induction hypothesis about $e_1$ and $e_2$,
\[
\begin{array}{l}
\semd{e_1}\gamma\phi'=\return(\monadic{v}_1,\gamma_1,\omicron_1)\mbox{,}\\
\semd{e_2}\gamma_1\phi'=\return(\monadic{v}_2,\gamma_2,\omicron_2)\mbox{.}
\end{array}
\]
Hence $\semd{\addexpr{e_1}{e_2}}\gamma\phi'=\return(\mcomp{v_1\gets\monadic{v}_1\hstop v_2\gets\monadic{v}_2\hstop\return(v_1+v_2)},\gamma_2,\omicron_1\omicron_2)=\semd{\addexpr{e_1}{e_2}}\gamma\phi$ as needed.

\item Let $e=\assertexpr{e_1}$. Then
\[
\begin{array}{l}
\semd{e_1}\gamma\phi=\return(\monadic{v}_1,\gamma_1,\omicron_1)\mbox{,}\\
\monadic{v}_1\ne\return\fls\mbox{,}\quad\monadic{v}=\return\singleton\mbox{,}\quad\gamma'=\gamma_1\mbox{,}\quad\omicron=\omicron_1\mbox{.}
\end{array}
\]
By the induction hypothesis about~$e_1$,
\[
\begin{array}{l}
\semd{e_1}\gamma\phi'=\return(\monadic{v}_1,\gamma_1,\omicron_1)\mbox{.}
\end{array}
\]
Hence $\semd{\assertexpr{e_1}}\gamma\phi'=\return(\return\singleton,\gamma_1,\omicron_1)=\semd{\assertexpr{e_1}}\gamma\phi$.

\item Let $e=\ofty{\getexpr{d''}{k}}{q}$. Then $\semd{e}\gamma\phi'=\semd{e}\gamma\phi$ since the semantics depends only on inputs of subdomains of~$d$.

\item Let $e=\ifexpr{e_1}{e_2}{e_3}$. Then
\[
\begin{array}{l}
\semd{e_1}\gamma\phi=\return(\monadic{v}_1,\gamma_1,\omicron_1)
\end{array}
\]
and, by the induction hypothesis about~$e_1$,
\[
\begin{array}{l}
\semd{e_1}\gamma\phi'=\return(\monadic{v}_1,\gamma_1,\omicron_1)\mbox{.}
\end{array}
\]
Consider three cases:
\begin{itemize}
\item If $\monadic{v}_1=\return\tru$ then
\[
\begin{array}{l}
\semd{e_2}\gamma_1\phi=\return(\monadic{v}_2,\gamma_2,\omicron_2)\mbox{,}\\
\monadic{v}=\monadic{v}_2\mbox{,}\quad\gamma'=\gamma_2\mbox{,}\quad\omicron=\omicron_1\omicron_2\mbox{.}
\end{array}
\]
By the induction hypothesis about~$e_2$,
\[
\semd{e_2}\gamma_1\phi'=\return(\monadic{v}_2,\gamma_2,\omicron_2)\mbox{.}
\]
Hence $\semd{\ifexpr{e_1}{e_2}{e_3}}\gamma\phi'=\return(\monadic{v}_2,\gamma_2,\omicron_1\omicron_2)=\semd{\ifexpr{e_1}{e_2}{e_3}}\gamma\phi$.
\item If $\monadic{v}_1=\return\fls$ then
\[
\begin{array}{l}
\semd{e_3}\gamma_1\phi=\return(\monadic{v}_3,\gamma_3,\omicron_3)\mbox{,}\\
\monadic{v}=\monadic{v}_3\mbox{,}\quad\gamma'=\gamma_3\mbox{,}\quad\omicron=\omicron_1\omicron_3\mbox{.}
\end{array}
\]
By the induction hypothesis about~$e_3$,
\[
\semd{e_3}\gamma_1\phi'=\return(\monadic{v}_3,\gamma_3,\omicron_3)\mbox{.}
\]
Hence $\semd{\ifexpr{e_1}{e_2}{e_3}}\gamma\phi'=\return(\monadic{v}_3,\gamma_3,\omicron_1\omicron_3)=\semd{\ifexpr{e_1}{e_2}{e_3}}\gamma\phi$.
\item If $\monadic{v}_1\ne\return b$ where $b\in\BB$ then $\monadic{v}=\top$, $\gamma'=\gamma_1$ and $\omicron=\omicron_1$. Since we similarly have $\semd{\ifexpr{e_1}{e_2}{e_3}}\gamma\phi'=\return(\top,\gamma_1,\omicron_1)$, the desired claim follows.
\end{itemize}

\item Let $e=\forexpr{x}{e_1}{e_2}{e_3}$. Then
\[
\begin{array}{l}
\semd{e_1}\gamma\phi=\return(\monadic{v}_1,\gamma_1,\omicron_1)\mbox{,}\\
\semd{e_2}\gamma_1\phi=\return(\monadic{v}_2,\gamma_2,\omicron_2)\mbox{,}
\end{array}
\]
and, by the induction hypothesis about $e_1$ and $e_2$,
\[
\begin{array}{l}
\semd{e_1}\gamma\phi'=\return(\monadic{v}_1,\gamma_1,\omicron_1)\mbox{,}\\
\semd{e_2}\gamma\phi'=\return(\monadic{v}_2,\gamma_2,\omicron_2)\mbox{.}
\end{array}
\]
Consider two cases:
\begin{itemize}
\item If $\monadic{v}_1=\return i_1$ and $\monadic{v}_2=\return i_2$ for $i_1,i_2\in\NN$ then, denoting $n=\max(0,i_2-i_1)$, we have
\[
\begin{array}{l}
\semd{e_3}((x,\return i_1),\gamma_2)\phi=\return(\monadic{v}_3,\gamma_3,\omicron_3)\mbox{,}\\
\semd{e_3}([x\mapsto\return(i_1+1)]\gamma_3)\phi=\return(\monadic{v}_4,\gamma_4,\omicron_4)\mbox{,}\\
\dotfill\\
\semd{e_3}([x\mapsto\return(i_1+n-1)]\gamma_{n+1})\phi=\return(\monadic{v}_{n+2},\gamma_{n+2},\omicron_{n+2})\mbox{,}\\
\monadic{v}=\return(\monadic{v}_3,\ldots,\monadic{v}_{n+2})\mbox{,}\quad\gamma'=\tail\gamma_{n+2}\mbox{,}\quad\omicron=\omicron_1\ldots\omicron_{n+2}\mbox{.}
\end{array}
\]
By the induction hypothesis about $e_3$,
\[
\begin{array}{l}
\semd{e_3}((x,\return i_1),\gamma_2)\phi'=\return(\monadic{v}_3,\gamma_3,\omicron_3)\mbox{,}\\
\semd{e_3}([x\mapsto\return(i_1+1)]\gamma_3)\phi'=\return(\monadic{v}_4,\gamma_4,\omicron_4)\mbox{,}\\
\dotfill\\
\semd{e_3}([x\mapsto\return(i_1+n-1)]\gamma_{n+1})\phi'=\return(\monadic{v}_{n+2},\gamma_{n+2},\omicron_{n+2})\mbox{.}
\end{array}
\]
Hence $\semd{\forexpr{x}{e_1}{e_2}{e_3}}\gamma\phi'=\return(\return(\monadic{v}_3,\ldots,\monadic{v}_{n+2}),\tail\gamma_{n+2},\omicron_1\ldots\omicron_{n+2})$ and the desired claim follows.
\item If $\monadic{v}_1\ne\return i_1$ or $\monadic{v}_2\ne\return i_2$ with $i_1,i_2\in\NN$ then $\monadic{v}=\top$, $\gamma'=\gamma_2$ and $\omicron=\omicron_1\omicron_2$. As we similarly get $\semd{\forexpr{x}{e_1}{e_2}{e_3}}\gamma\phi'=\return(\top,\gamma_2,\omicron_1\omicron_2)$, the desired claim follows.
\end{itemize}

\item Let $e=\wireexpr{e_1}$. Then
\[
\begin{array}{l}
\semd{e_1}\gamma\phi=\return(\monadic{v}_1,\gamma_1,\omicron_1)\mbox{,}\\
\monadic{v}=\monadic{v}_1\mbox{,}\quad\gamma'=\gamma_1\mbox{,}\quad\omicron=\lam{d'}{\branching{(\omicron_1)_{d'}\monadic{v}_1&\mbox{if $e_1$ is in domain~$d'$}\\(\omicron_1)_{d'}&\mbox{otherwise}}}\mbox{.}
\end{array}
\]
By the induction hypothesis about~$e_1$,
\[
\semd{e_1}\gamma\phi'=\return(\monadic{v}_1,\gamma_1,\omicron_1)\mbox{.}
\]
Hence
\[
\begin{array}{lcl}
\semd{\wireexpr{e_1}}\gamma\phi'
&=&\return(\monadic{v}_1,\gamma_1,\lam{d'}{\branching{(\omicron_1)_{d'}\monadic{v}_1&\mbox{if $e_1$ is in domain~$d'$}\\(\omicron_1)_{d'}&\mbox{otherwise}}})\\
&=&\semd{\wireexpr{e_1}}\gamma\phi
\end{array}
\]
and the desired claim follows.

\item Let $e=\castexpr{e_1}{d_0}$. Let the domain of~$e_1$ be $d_1$. Then
\[
\begin{array}{l}
\semd{e_1}\gamma\phi=\return(\monadic{v}_1,\gamma_1,\omicron_1)\mbox{,}\\
\monadic{v}=\branching{\monadic{v}_1&\mbox{if $d_0\subtype d$}\\\top&\mbox{otherwise}}\mbox{,}\quad\gamma'=\gamma_1\mbox{,}\quad\omicron=\omicron_1\mbox{.}
\end{array}
\]
By the induction hypothesis,
\[
\semd{e_1}\gamma\phi'=\return(\monadic{v}_1,\gamma_1,\omicron_1)\mbox{.}
\]
Hence $\semd{\castexpr{e_1}{d_0}}\gamma\phi'=\return(\branching{\monadic{v}_1&\mbox{if $d_0\subtype d$}\\\top&\mbox{otherwise}},\gamma_1,\omicron_1)=\semd{\castexpr{e_1}{d_0}}\gamma\phi$ and the desired claim follows.

\item Let $e=(\assignexpr{e_1}{e_2})$. Let $e_1=\loadexpr{\loadexpr{\loadexpr{x}{y_1}}{y_2}\ldots}{y_n}$. Then
\[
\begin{array}{l}
\semd{y_1}\gamma\phi=\return(\monadic{i}_1,\gamma_1,\omicron_1)\mbox{,}\\
\semd{y_2}\gamma_1\phi=\return(\monadic{i}_2,\gamma_2,\omicron_2)\mbox{,}\\
\dotfill\\
\semd{y_n}\gamma_{n-1}\phi=\return(\monadic{i}_n,\gamma_n,\omicron_n)\mbox{,}\\
\semd{e_2}\gamma_n\phi=\return(\monadic{r},\gamma_{n+1},\omicron_{n+1})\mbox{,}\\
\monadic{v}=\return\singleton\mbox{,}\quad\gamma'=[x\mapsto\update(\lookup{\gamma}{x},\monadic{i}_1\ldots\monadic{i}_n,\monadic{r})]\gamma_{n+1}\mbox{,}\quad\omicron=\omicron_1\ldots\omicron_{n+1}\mbox{.}
\end{array}
\]
By the induction hypothesis,
\[
\begin{array}{l}
\semd{y_1}\gamma\phi'=\return(\monadic{i}_1,\gamma_1,\omicron_1)\mbox{,}\\
\semd{y_2}\gamma_1\phi'=\return(\monadic{i}_2,\gamma_2,\omicron_2)\mbox{,}\\
\dotfill\\
\semd{y_n}\gamma_{n-1}\phi'=\return(\monadic{i}_n,\gamma_n,\omicron_n)\mbox{,}\\
\semd{e_2}\gamma_n\phi'=\return(\monadic{r},\gamma_{n+1},\omicron_{n+1})\mbox{.}
\end{array}
\]
Hence
\[
\semd{\assignexpr{e_1}{e_2}}\gamma\phi'=\return(\return\singleton,[x\mapsto\update(\lookup{\gamma}{x},\monadic{i}_1\ldots\monadic{i}_n,\monadic{v})]\gamma_{n+1},\omicron_1\ldots\omicron_{n+1})=\semd{\assignexpr{e_1}{e_2}}\gamma\phi
\]
and the desired claim follows.

\item Let $e=\loadexpr{e_1}{e_2}$. Then
\[
\begin{array}{l}
\semd{e_1}\gamma\phi=\return(\monadic{a},\gamma_1,\omicron_1)\mbox{,}\\
\semd{e_2}\gamma_1\phi=\return(\monadic{i},\gamma_2,\omicron_2)\mbox{,}\\
\monadic{v}=\mcomp{a\gets\monadic{a}\hstop i\gets\monadic{i}\hstop a_i}\mbox{,}\quad\gamma'=\gamma_2\mbox{,}\quad\omicron=\omicron_1\omicron_2\mbox{.}
\end{array}
\]
By the induction hypothesis,
\[
\begin{array}{l}
\semd{e_1}\gamma\phi'=\return(\monadic{a},\gamma_1,\omicron_1)\mbox{,}\\
\semd{e_2}\gamma_1\phi'=\return(\monadic{i},\gamma_2,\omicron_2)\mbox{.}\\
\end{array}
\]
Hence $\semd{\loadexpr{e_1}{e_2}}\gamma\phi'=\return(\mcomp{a\gets\monadic{a}\hstop i\gets\monadic{i}\hstop a_i},\gamma_2,\omicron_1\omicron_2)=\semd{\loadexpr{e_1}{e_2}}$ and the desired claim follows.

\item Let $e=\stmtcomp{\letexpr{x}{e_1}}{e_2}$. Then
\[
\begin{array}{l}
\semd{e_1}\gamma\phi=\return(\monadic{v}_1,\gamma_1,\omicron_1)\mbox{,}\\
\semd{e_2}((x,\monadic{v}_1),\gamma_1)\phi=\return(\monadic{v}_2,\gamma_2,\omicron_2)\mbox{,}\\
\monadic{v}=\monadic{v}_2\mbox{,}\quad\gamma'=\tail\gamma_2\mbox{,}\quad\omicron=\omicron_1\omicron_2\mbox{.}
\end{array}
\]
By the induction hypothesis,
\[
\begin{array}{l}
\semd{e_1}\gamma\phi'=\return(\monadic{v}_1,\gamma_1,\omicron_1)\mbox{,}\\
\semd{e_2}((x,\monadic{v}_1),\gamma_1)\phi'=\return(\monadic{v}_2,\gamma_2,\omicron_2)\mbox{.}
\end{array}
\]
Hence $\semd{\stmtcomp{\letexpr{x}{e_1}}{e_2}}\gamma\phi'=\return(\monadic{v}_2,\tail\gamma_2,\omicron_1\omicron_2)=\semd{\stmtcomp{\letexpr{x}{e_1}}{e_2}}\gamma\phi$ and the desired claim follows.

\item Let $e=\stmtcomp{e_1}{e_2}$. Then
\[
\begin{array}{l}
\semd{e_1}\gamma\phi=\return(\monadic{v}_1,\gamma_1,\omicron_1)\mbox{,}\\
\semd{e_2}\gamma_1\phi=\return(\monadic{v}_2,\gamma_2,\omicron_2)\mbox{,}\\
\monadic{v}=\monadic{v}_2\mbox{,}\quad\gamma'=\gamma_2\mbox{,}\quad\omicron=\omicron_1\omicron_2\mbox{.}
\end{array}
\]
By the induction hypothesis,
\[
\begin{array}{l}
\semd{e_1}\gamma\phi'=\return(\monadic{v}_1,\gamma_1,\omicron_1)\mbox{,}\\
\semd{e_2}\gamma_1\phi'=\return(\monadic{v}_2,\gamma_2,\omicron_2)\mbox{.}
\end{array}
\]
Hence $\semd{\stmtcomp{e_1}{e_2}}\gamma\phi'=\return(\monadic{v}_2,\gamma_2,\omicron_1\omicron_2)=\semd{\stmtcomp{e_1}{e_2}}\gamma\phi$ and the desired claim follows.

\end{itemize}

\end{proof}

\begin{theorem}[Theorem~\ref{semantics:soundness}]
Let $\Gamma\vd e:q\eff D$ with well-structured $\Gamma$. Let $d,d'$ be domains such that $d'\subtype d$. Let $\gamma_d,\gamma_{d'}\in\mathbf{Env}$ be $\Gamma$-exact in~$d$ and $d'$, respectively, and let $\gamma_d\coin{\Gamma}{d'}\gamma_{d'}$. Assume that for all subexpressions of~$e$ of the form $\ofty{\getexpr{d''}{k}}{q'}$, the value $\allpure(\phi_{d''}(k))$ is $q'$-exact in~$d''$. Assume that there exist $\monadic{v}_d,\gamma'_d,\omicron$ such that $\semd{e}\gamma_d\phi=\return(\monadic{v}_d,\gamma'_d,\omicron)$. Then there exist $\monadic{v}_{d'},\gamma'_{d'},\omicron'$ such that $\sem{e}_{d'}\gamma_{d'}\phi=\return(\monadic{v}_{d'},\gamma'_{d'},\omicron')$, whereby $\monadic{v}_d\coin{q}{d'}\monadic{v}_{d'}$, $\gamma_d\coin{\Gamma}{d'}\gamma_{d'}$ and $\omicron\coinpure{d'}\omicron'$.
\end{theorem}

\begin{proof}
Let $q=\qualty{t_0}{s_0}{d_0}$. We proceed by induction on the structure of~$e$:

\begin{itemize}
\item Let $e=\epsilon$. Then $\monadic{v}_d=\return\singleton$, $\gamma'_d=\gamma_d$, $\omicron=\epsilon$ and $\sem{\epsilon}_{d'}\gamma_{d'}\phi=\return(\return\singleton,\gamma_{d'},\epsilon)$. The desired claim follows since $\return\singleton\coin{q}{d'}\return\singleton$ and, by assumption, $\gamma_d\coin{\Gamma}{d'}\gamma_{d'}$.

\item Let $e=\overline{n}$ where $n\in\NN$. Then $\monadic{v}_d=\branching{\return n&\mbox{if $d_0\subtype d$}\\\top&\mbox{otherwise}}$, $\gamma'_d=\gamma_d$, $\omicron=\epsilon$. On the other hand, $\sem{\overline{n}}_{d'}\gamma_{d'}\phi=\return(\branching{\return n&\mbox{if $d_0\subtype d'$}\\\top&\mbox{otherwise}},\gamma_{d'},\epsilon)$. Consider two cases:
\begin{itemize}
\item If $d_0\subtype d'$ then also $d_0\subtype d$, whence $\monadic{v}_d=\return n=\monadic{v}_{d'}$. Consequently, $\monadic{v}_d\coin{q}{d'}\monadic{v}'$.
\item If $d_0$ is a strict superdomain of~$d'$ then $\monadic{v}_d\coin{q}{d'}\monadic{v}'$ vacuously.
\end{itemize}
Moreover, $\gamma_{d}\coin{\Gamma}{d'}\gamma_{d'}$ by assumption. The desired claim follows.

\item The case $e=\overline{b}$ with $b\in\BB$ is similar to the previous case.

\item If $e=x$ then $\monadic{v}_d=\lookup{\gamma_d}{x}$, $\gamma'_d=\gamma_d$, $\omicron=\epsilon$. On the other hand, $\sem{x}_{d'}\gamma_{d'}\phi=\return(\lookup{\gamma_{d'}}{x},\gamma_{d'},\epsilon)$. The desired claim follows since, by assumption, $\gamma_d\coin{\Gamma}{d'}\gamma_{d'}$ which also implies $\lookup{\gamma_d}{x}\coin{q}{d'}\lookup{\gamma_{d'}}{x}$ by Lemma~\ref{semantics:lookupcoinlemma}.

\item Let $e=\addexpr{e_1}{e_2}$. Then
\[
\begin{array}{l}
\semd{e_1}\gamma_d\phi=\return(\monadic{v}_d^1,\gamma_d^1,\omicron_1)\mbox{,}\\
\semd{e_2}\gamma_d^1\phi=\return(\monadic{v}_d^2,\gamma_d^2,\omicron_2)\mbox{,}\\
\monadic{v}_d=\mcomp{v_d^1\gets\monadic{v}_d^1\hstop v_d^2\gets\monadic{v}_d^2\hstop\return(v_d^1+v_d^2)}\mbox{,}\quad\gamma'_d=\gamma_d^2\mbox{,}\quad\omicron=\omicron_1\omicron_2\mbox{.}
\end{array}
\]
By the induction hypothesis about $e_1$, $\sem{e_1}_{d'}\gamma_{d'}\phi=\return(\monadic{v}_{d'}^1,\gamma_{d'}^1,\omicron'_1)$ where $\monadic{v}_d^1\coin{q}{d'}\monadic{v}_{d'}^1$, $\gamma_d^1\coin{\Gamma}{d'}\gamma_{d'}^1$ and $\omicron_1\coinpure{d'}\omicron'_1$. Hence by the induction hypothesis about $e_2$, $\sem{e_2}_{d'}\gamma_{d'}^1\phi=\return(\monadic{v}_{d'}^2,\gamma_{d'}^2,\omicron'_2)$ where $\monadic{v}_d^2\coin{q}{d'}\monadic{v}_{d'}^2$, $\gamma_d^2\coin{\Gamma}{d'}\gamma_{d'}^2$ and $\omicron_2\coinpure{d'}\omicron'_2$. Thus
\[
\begin{array}{l}
\sem{\addexpr{e_1}{e_2}}_{d'}\gamma_{d'}\phi=\return(\monadic{v}_{d'},\gamma'_{d'},\omicron')\mbox{,}\\
\monadic{v}_{d'}=\mcomp{v_{d'}^1\gets\monadic{v}_{d'}^1\hstop v_{d'}^2\gets\monadic{v}_{d'}^2\hstop\return(v_{d'}^1+v_{d'}^2)}\mbox{,}\quad\gamma'_{d'}=\gamma_{d'}^2\mbox{,}\quad\omicron'=\omicron'_1\omicron'_2\mbox{.}
\end{array}
\]
Consider two cases:
\begin{itemize}
\item If $d_0\subtype d'$ then also $d_0\subtype d$ and we have $\monadic{v}_d^1=\return v_d^1$, $\monadic{v}_2=\return v_d^2$, $\monadic{v}_{d'}^1=\return v_{d'}^1$, $\monadic{v}_{d'}^2=\return v_{d'}^2$ by exactness. This implies $\monadic{v}_d=\return(v_d^1+v_d^2)$ and $\monadic{v}_{d'}=\return(v_{d'}^1+v_{d'}^2)$. By $\monadic{v}_d^1\coin{q}{d'}\monadic{v}_{d'}^1$ and $\monadic{v}_d^2\coin{q}{d'}\monadic{v}_{d'}^2$, we have $v_d^1=v_{d'}^1$ and $v_d^2=v_{d'}^2$, whence $\monadic{v}_d=\monadic{v}_{d'}$. Consequently, $\monadic{v}_d\coin{q}{d'}\monadic{v}_{d'}$.
\item If $d_0$ is a strict supertype of~$d'$ then $\monadic{v}_d\coin{q}{d'}\monadic{v}_{d'}$ vacuously.
\end{itemize}
Along with $\gamma_d^2\coin{\Gamma}{d'}\gamma_{d'}^2$ and $\omicron_1\omicron_2\coinpure{d'}\omicron'_1\omicron'_2$, this implies the desired claim.

\item Let $e=\assertexpr{e_1}$. Then
\[
\begin{array}{l}
\Gamma\vd e_1:q'\mbox{,}\\
q'=\qualty{\boolmodty{\mbox{\lstinline+N+}}}{\poststg}{d_1}
\end{array}
\]
and
\[
\begin{array}{l}
\semd{e_1}\gamma_d\phi=\return(\monadic{v}_d^1,\gamma_d^1,\omicron_1)\mbox{,}\\
\monadic{v}_d^1\ne\return\fls\mbox{,}\quad\monadic{v}_d=\return\singleton\mbox{,}\quad\gamma'_d=\gamma_d^1\mbox{,}\quad\omicron=\omicron_1\mbox{.}
\end{array}
\]
By the induction hypothesis about~$e_1$, $\sem{e_1}_{d'}\gamma_{d'}\phi=\return(\monadic{v}_{d'}^1,\gamma_{d'}^1,\omicron'_1)$ where $\monadic{v}_d^1\coin{q'}{d'}\monadic{v}_{d'}^1$, $\gamma_d^1\coin{\Gamma}{d'}\gamma_{d'}^1$ and $\omicron_1\coinpure{d'}\omicron'_1$. If $d_1\subtype d'$ then $\monadic{v}_d^1\coin{q'}{d'}\monadic{v}_{d'}^1$ implies $\monadic{v}_d^1=\return\tru=\monadic{v}_{d'}^1$, establishing $\monadic{v}_{d'}^1\ne\return\fls$. If $d_1$ is a strict superdomain of $d'$ then, by exactness, $\monadic{v}_{d'}^1=\top\ne\return\fls$. Therefore,
\[
\begin{array}{l}
\sem{\assertexpr{e_1}}_{d'}\gamma_{d'}\phi=\return(\monadic{v}_{d'},\gamma'_{d'},\omicron')\mbox{,}\\
\monadic{v}_{d'}=\return\singleton\mbox{,}\quad\gamma'_{d'}=\gamma_{d'}^1\mbox{,}\quad\omicron'=\omicron'_1\mbox{.}
\end{array}
\]
The desired claim follows since $\return\singleton\coin{q}{d'}\return\singleton$, $\gamma_d^1\coin{\Gamma}{d'}\gamma_{d'}^1$ and $\omicron_1\coinpure{d'}\omicron'_1$.

\item Let $e=\ofty{\getexpr{d''}{e_1}}{q}$. Then $\monadic{v}_d=\branching{\allpure(\phi_{d''}(k))&\mbox{if $d''\subtype d$}\\\top&\mbox{otherwise}}$, $\gamma'_d=\gamma_d$ and $\omicron=\epsilon$. On the other hand, $\sem{\getexpr{d''}{e_1}}_{d'}\gamma_{d'}\phi=\return(\monadic{v}_{d'},\gamma'_{d'},\omicron')$ where $\monadic{v}_{d'}=\branching{\allpure(\phi_{d''}(k))&\mbox{if $d''\subtype d'$}\\\top&\mbox{otherwise}}$, $\gamma'_{d'}=\gamma_{d'}$, $\omicron'=\epsilon$. Consider three cases:
\begin{itemize}
\item If $d''\subtype d'$ then $d''\subtype d$, whence $\monadic{v}_d=\allpure(\phi_{d''}(k))=\monadic{v}_{d'}$ and therefore $\monadic{v}_d\coin{q}{d'}\monadic{v}_{d'}$;
\item If $d''$ is a strict superdomain of~$d'$ then $\monadic{v}_d\coin{q}{d'}\monadic{v}_{d''}$ vacuously.
\end{itemize}
Moreover, $\gamma_d\coin{\Gamma}{d'}\gamma_{d'}$ by assumption. The desired claim follows.

\item Let $e=\ifexpr{e_1}{e_2}{e_3}$. Then
\[
\begin{array}{l}
\Gamma\vd e_1:q'\eff D_1\mbox{,}\\
\Gamma\vd e_2:q\eff D_2\mbox{,}\\
\Gamma\vd e_3:q\eff D_3\mbox{,}\\
q'=\qualty{\boolmodty{\mbox{\lstinline+N+}}}{\prestg}{d_1}\mbox{,}\quad
\gen{d_1}\supseteq\gen{s_0}\cup\gen{d_0}\cup D_2\cup D_3\mbox{,}
\end{array}
\]
and
\[
\semd{e_1}\gamma_d\phi=\return(\monadic{v}_d^1,\gamma_d^1,\omicron_1)\mbox{.}
\]
By the induction hypothesis, $\sem{e_1}_{d'}\gamma_{d'}\phi=\return(\monadic{v}_{d'}^1,\gamma_{d'}^1,\omicron'_1)$ where $\monadic{v}_d^1\coin{q'}{d'}\monadic{v}_{d'}^1$, $\gamma_d^1\coin{\Gamma}{d'}\gamma_{d'}^1$ and $\omicron_1\coinpure{d'}\omicron'_1$. Consider three cases:
\begin{itemize}
\item If $\monadic{v}_d^1=\return\tru$ then $\semd{e_2}\gamma_d^1\phi=\return(\monadic{v}_d^2,\gamma_d^2,\omicron_2)$ and $\monadic{v}_d=\monadic{v}_d^2$, $\gamma'_d=\gamma_d^2$, $\omicron=\omicron_1\omicron_2$. By the induction hypothesis, $\sem{e_2}_{d'}\gamma_{d'}^1\phi=\return(\monadic{v}_{d'}^2,\gamma_{d'}^2,\omicron'_2)$ where $\monadic{v}_d^2\coin{q}{d'}\monadic{v}_{d'}^2$, $\gamma_d^2\coin{\Gamma}{d'}\gamma_{d'}^2$, $\omicron_2\coinpure{d'}\omicron'_2$. 

If $d_1\subtype d'$ then $\monadic{v}_{d'}^1=\return\tru$, whence $\sem{\ifexpr{e_1}{e_2}{e_3}}_{d'}\gamma_{d'}\phi=\return(\monadic{v}_{d'},\gamma'_{d'},\omicron')$ where $\monadic{v}_{d'}=\monadic{v}_{d'}^2$, $\gamma'_{d'}=\gamma_{d'}^2$, $\omicron'=\omicron'_1\omicron'_2$. The desired claim follows.

If $d_1$ is a strict supertype of~$d'$ then, by exactness, $\monadic{v}_{d'}^1=\top$. Therefore we obtain $\sem{\ifexpr{e_1}{e_2}{e_3}}_{d'}\gamma_{d'}\phi=\return(\monadic{v}_{d'},\gamma'_{d'},\omicron')$ where $\monadic{v}_{d'}=\top$, $\gamma'_{d'}=\gamma_{d'}^1$, $\omicron'=\omicron'_1$. As $d_1\subtype d_0$, also $d_0$ is a strict supertype of~$d'$, vacuously implying $\monadic{v}_d\coin{q}{d'}\monadic{v}_{d'}$. Finally, note that $d'\in D_2$ would imply $d_1\subtype d'$ and introduce contradiction, thus $d'\notin D_2$. By Theorem~\ref{semantics:effthm}, $\gamma_d^1\coin{\Gamma}{d'}\gamma_d^2$. As also $\gamma_d^1\coin{\Gamma}{d'}\gamma_{d'}^1$, Lemma~\ref{semantics:coinlemma} establishes $\gamma_d^2\coin{\Gamma}{d'}\gamma_{d'}^1$. By $D_2$ being upward closed, we also have $\publicdom\notin D_2$, whence Theorem~\ref{semantics:outputthm} gives $\omicron_2=\epsilon$. Thus $\omicron=\omicron_1$. The desired result follows.

\item The case $\monadic{v}_d^1=\return\fls$ is similar to the previous case.

\item If $\monadic{v}_d^1=\top$ then $\monadic{v}_d=\top$, $\gamma'_d=\gamma_d^1$, $\omicron=\omicron_1$. By exactness, $d_1$ is a strict supertype of~$d$. Hence $d_1$ is a strict supertype of~$d'$, implying $\monadic{v}_{d'}^1=\top$. We get $\sem{\ifexpr{e_1}{e_2}{e_3}}_{d'}\gamma_{d'}\phi=\return(\monadic{v}_{d'},\gamma'_{d'},\omicron')$ where $\monadic{v}_{d'}=\top$, $\gamma'_{d'}=\gamma_{d'}^1$, $\omicron'=\omicron'_1$. The desired claim follows.

\end{itemize}

\item Let $e=\forexpr{x}{e_1}{e_2}{e_3}$. Then
\[
\begin{array}{l}
\Gamma\vd e_1:q'\eff D_1\mbox{,}\\
\Gamma\vd e_2:q'\eff D_2\mbox{,}\\
(x:q'),\Gamma\vd e_3:q''\eff D_3\mbox{,}\\
q'=\qualty{\uintty}{\prestg}{d_0}\mbox{,}\quad q''=\qualty{t_1}{s_1}{d_1}\mbox{,}\quad t_0=\listty{q''}\mbox{,}\quad s_0=\prestg\mbox{,}\\
\gen{d_0}\supseteq\gen{s_1}\cup\gen{d_1}\cup D_3\mbox{,}
\end{array}
\]
and
\[
\begin{array}{l}
\semd{e_1}\gamma_d\phi=\return(\monadic{v}_d^1,\gamma_d^1,\omicron_1)\mbox{,}\\
\semd{e_2}\gamma_d^1\phi=\return(\monadic{v}_d^2,\gamma_d^2,\omicron_2)\mbox{.}
\end{array}
\]
By the induction hypothesis, $\sem{e_1}_{d'}\gamma_{d'}\phi=\return(\monadic{v}_{d'}^1,\gamma_{d'}^1,\omicron'_1)$ where $\monadic{v}_d^1\coin{q'}{d'}\monadic{v}_{d'}^1$, $\gamma_d^1\coin{\Gamma}{d'}\gamma_{d'}^1$, $\omicron_1\coinpure{d'}\omicron'_1$. Now the induction hypothesis gives $\sem{e_2}_{d'}\gamma_{d'}^1\phi=\return(\monadic{v}_{d'}^2,\gamma_{d'}^2,\omicron'_2)$ where $\monadic{v}_d^2\coin{q'}{d'}\monadic{v}_{d'}^2$, $\gamma_d^2\coin{\Gamma}{d'}\gamma_{d'}^2$, $\omicron_2\coinpure{d'}\omicron'_2$. Consider two cases:

\begin{itemize}
\item If $\monadic{v}_d^1=\return i_1$, $\monadic{v}_d^2=\return i_2$ then, denoting $n=\max(0,i_2-i_1)$, we have
\[
\begin{array}{l}
\semd{e_3}((x,\return i_1),\gamma_d^2)\phi=\return(\monadic{v}_d^3,\gamma_d^3,\omicron_3)\mbox{,}\\
\semd{e_3}([x\mapsto\return(i_1+1)]\gamma_d^3)\phi=\return(\monadic{v}_d^4,\gamma_d^4,\omicron_4)\mbox{,}\\
\dotfill\\
\semd{e_3}([x\mapsto\return(i_1+n-1)]\gamma_d^{n+1})\phi=\return(\monadic{v}_d^{n+2},\gamma_d^{n+2},\omicron_{n+2})\mbox{,}\\
\monadic{v}_d=\return(\monadic{v}_d^3,\ldots,\monadic{v}_d^{n+2})\mbox{,}\quad\gamma'_d=\tail\gamma_d^{n+2}\mbox{,}\quad\omicron=\omicron_1\ldots\omicron_{n+2}\mbox{.}
\end{array}
\]
As $\return i_1\coin{q'}{d'}\return i_1$ and $\gamma_d^2\coin{\Gamma}{d'}\gamma_{d'}^2$, we have $(x,\return i_1),\gamma_d^2\coin{\Gamma'}{d'}\left((x,\return i_1),\gamma_{d'}^2\right)$ where $\Gamma'=((x:q'),\Gamma)$. By the induction hypothesis, $\sem{e_3}_{d'}((x,\return i_1),\gamma_{d'}^2)\phi=\return(\monadic{v}_{d'}^3,\gamma_{d'}^3,\omicron'_3)$ where $\monadic{v}_d^3\coin{q''}{d'}\monadic{v}_{d'}^3$, $\gamma_d^3\coin{\Gamma'}{d'}\gamma_{d'}^3$ and $\omicron_3\coinpure{d'}\omicron'_3$. Replacing $i_1$ with $i_1+k$ does not violate the necessary properties, so we similarly get $\sem{e_3}_{d'}([x\mapsto\return(i_1+k-3)]\gamma_{d'}^{k-1})\phi=\return(\monadic{v}_{d'}^k,\gamma_{d'}^k,\omicron'_k)$ where $\monadic{v}_d^k\coin{q''}{d'}\monadic{v}_{d'}^k$, $\gamma_d^k\coin{\Gamma'}{d'}\gamma_{d'}^k$ and $\omicron_k\coinpure{d'}\omicron'_k$, for every $k=4,\ldots,n+2$.

If $d_0\subtype d'$ then the above implies $\monadic{v}_{d'}^1=\return i_1$ and $\monadic{v}_{d'}^2=\return i_2$, whence we obtain $\sem{\forexpr{x}{e_1}{e_2}{e_3}}_{d'}\gamma_{d'}\phi=\return(\monadic{v}_{d'},\gamma'_{d'},\omicron')$ where $\monadic{v}_{d'}=\return(\monadic{v}_{d'}^3,\ldots,\monadic{v}_{d'}^{n+2})$, $\gamma'_{d'}=\gamma_{d'}^{n+2}$ and $\omicron'=\omicron'_1\ldots\omicron'_{n+2}$. Thus $\monadic{v}_d\coin{q_0}{d'}\monadic{v}_{d'}$, $\gamma'_d\coin{\Gamma}{d'}\gamma'_{d'}$, $\omicron\coinpure{d'}\omicron'$ and the desired claim follows.

If $d_0$ is a strict superdomain of~$d'$ then $\monadic{v}_{d'}^1=\monadic{v}_{d'}^2=\top$ and $\sem{\forexpr{x}{e_1}{e_2}{e_3}}_{d'}\gamma_{d'}\phi=\return(\monadic{v}_{d'},\gamma'_{d'},\omicron')$ where $\monadic{v}_{d'}=\top$, $\gamma'_{d'}=\gamma_{d'}^2$, $\omicron'=\omicron'_1\omicron'_2$. Then $\monadic{v}_d^{n+2}\coin{q_0}{d'}\top$ vacuously. Assuming $d'\in D_3$ would give $d_0\subtype d'$ by $\gen{d_0}\supseteq D_3$, hence $d'\notin D_3$. So Theorem~\ref{semantics:effthm} implies
\[
\begin{array}{l}
\left((x,\return i_1),\gamma_d^2\right)\coin{\Gamma'}{d'}\gamma_d^3\mbox{,}\\{}
[x\mapsto\return(i_1+1)]\gamma_d^3\coin{\Gamma'}{d'}\gamma_d^4\mbox{,}\\
\dotfill\\{}
[x\mapsto\return(i_1+n-1)]\gamma_d^{n+1}\coin{\Gamma'}{d'}\gamma_d^{n+2}\mbox{.}
\end{array}
\]
We also have $\gamma_d^k\coin{\Gamma'}{d'}[x\mapsto\return(i_1+k-2)]\gamma_d^k$ for every $k=3,\ldots,n+1$ since monadic values of the additional variable~$x$ are coincident vacuously. Lemma~\ref{semantics:coinlemma} now implies $\left((x,\return i_1),\gamma_d^2\right)\coin{\Gamma'}{d'}\gamma_d^{n+2}$. Hence also $\gamma_d^2\coin{\Gamma}{d'}\tail\gamma_d^{n+2}$. As we also have $\gamma_d^2\coin{\Gamma}{d'}\gamma_{d'}^2$, Lemma~\ref{semantics:coinlemma} gives $\gamma_d^{n+2}\coin{\Gamma}{d'}\gamma_{d'}^2$. Furthermore, as $d'\notin D_3$ implies $\publicdom\notin D_3$ by $D_3$ being upward closed, Theorem~\ref{semantics:outputthm} gives $\omicron_3=\ldots=\omicron_{n+2}=\epsilon$. Hence $\omicron=\omicron_1\omicron_2$ and the desired result follows.

\item If $\monadic{v}_d^1=\top$ or $\monadic{v}_d^2=\top$ then
\[
\monadic{v}_d=\top\mbox{,}\quad\gamma'_d=\gamma_d^2\mbox{,}\quad\omicron=\omicron_1\omicron_2\mbox{.}
\]
By exactness, $d_0$ is a strict superdomain of~$d$, implying that $d_0$ is also a strict superdomain of~$d'$. Hence, by exactness, $\monadic{v}_{d'}^1=\top$ or $\monadic{v}_{d'}^2=\top$, implying $\sem{\forexpr{x}{e_1}{e_2}{e_3}}_{d'}\gamma_{d'}\phi=\return(\monadic{v}_{d'},\gamma'_{d'},\omicron')$ where $\monadic{v}_{d'}=\top$, $\gamma'_{d'}=\gamma_{d'}^2$, $\omicron'=\omicron'_1\omicron'_2$. The desired claim follows.

\end{itemize}

\item Let $e=\wireexpr{e_1}$. Then
\[
\begin{array}{l}
\Gamma\vd e_1:q'\mbox{,}\\
s_0=\poststg\mbox{,}\quad q'=\qualty{t_0}{\prestg}{d_0}\mbox{,}
\end{array}
\]
where $t_0$ is $\uintmodty{\mbox{\lstinline+N+}}$ or $\boolmodty{\mbox{\lstinline+N+}}$, and
\[
\begin{array}{l}
\semd{e_1}\gamma_d\phi=\return(\monadic{v}_d^1,\gamma_d^1,\omicron_1)\mbox{,}\\
\monadic{v}_d=\monadic{v}_d^1\mbox{,}\quad\gamma'_d=\gamma_d^1\mbox{,}\quad\omicron=\lam{d''}{\branching{(\omicron_1)_{d''}\monadic{v}_d^1&\mbox{if $d''=d_0$}\\(\omicron_1)_{d''}&\mbox{otherwise}}}\mbox{.}
\end{array}
\]
By the induction hypothesis about~$e_1$, $\sem{e_1}_{d'}\gamma_{d'}\phi=\return(\monadic{v}_{d'}^1,\gamma_{d'}^1,\omicron'_1)$ where $\monadic{v}_d^1\coin{q'}{d'}\monadic{v}_{d'}^1$, $\gamma_d^1\coin{\Gamma}{d'}\gamma_{d'}^1$ and $\omicron_1\coinpure{d'}\omicron'_1$. Hence
\[
\begin{array}{l}
\sem{\wireexpr{e_1}}_{d'}\gamma_{d'}\phi=\return(\monadic{v}_{d'},\gamma'_{d'},\omicron')\mbox{,}\\
\monadic{v}_{d'}=\monadic{v}_{d'}^1\mbox{,}\quad\gamma'_{d'}=\gamma_{d'}^1\mbox{,}\quad\omicron'=\lam{d''}{\branching{(\omicron'_1)_{d''}\monadic{v}_{d'}^1&\mbox{if $d''=d_0$}\\(\omicron'_1)_{d''}&\mbox{otherwise}}}\mbox{.}
\end{array}
\]
The claim $\monadic{v}_d\coin{q}{d'}\monadic{v}_{d'}$ follows since relations $\coin{q}{d'}$ and $\coin{q'}{d'}$ are equal. The claim $\gamma'_d\coin{\Gamma}{d'}\gamma'_{d'}$ is implied by the above. Finally, $\monadic{v}_d^1\coin{q}{d'}\monadic{v}_{d'}^1$ means that if $d_0\subtype d'$ then $\monadic{v}_d^1=\return v_d^1=\monadic{v}_{d'}^1$ where $v_d^1\in\NN\cup\BB$. Hence if $d_0\subtype d'$ then $\omicron_{d_0}=\omicron'_{d_0}$. If $d_0$~{}is a strict superdomain of~$d'$ then the lengths of $\omicron_{d_0}$ and $\omicron'_{d_0}$ are equal since both are one more than the common length of $(\omicron_1)_{d_0}$ and $(\omicron'_1)_{d_0}$.

\item Let $e=\castexpr{e_1}{d_0}$. Then
\[
\begin{array}{l}
\Gamma\vd e_1:q'\mbox{,}\\
q'=\qualty{t_0}{s_0}{d_1}\mbox{,}\quad d_1\subtype d_0\mbox{,}\quad\gen{d_0}\supseteq\gen{t_0}
\end{array}
\]
and
\[
\begin{array}{l}
\semd{e_1}\gamma_d\phi=\return(\monadic{v}_d^1,\gamma_d^1,\omicron_1)\mbox{,}\\
\monadic{v}_d=\branching{\monadic{v}_d^1&\mbox{if $d_0\subtype d$}\\\top&\mbox{otherwise}}\mbox{,}\quad\gamma'_d=\gamma_d^1\mbox{,}\quad\omicron=\omicron_1\mbox{.}
\end{array}
\]
By the induction hypothesis about~$e_1$, $\sem{e_1}_{d'}\gamma_{d'}\phi=\return(\monadic{v}_{d'}^1,\gamma_{d'}^1,\omicron'_1)$ where $\monadic{v}_d^1\coin{q'}{d'}\monadic{v}_{d'}^1$, $\gamma_d^1\coin{\Gamma}{d'}\gamma_{d'}^1$, $\omicron_1\coinpure{d'}\omicron'_1$. By Lemma~\ref{semantics:castcoinlemma}, the first of these implies $\monadic{v}_d^1\coin{q}{d'}\monadic{v}_{d'}^1$. We obtain $\sem{\castexpr{e_1}{d_0}}_{d'}\gamma_{d'}\phi=\return(\monadic{v}_{d'},\gamma'_{d'},\omicron')$ where $\monadic{v}_{d'}=\branching{\monadic{v}_{d'}^1&\mbox{if $d_0\subtype d'$}\\\top&\mbox{otherwise}}$, $\gamma'_{d'}=\gamma_{d'}^1$, $\omicron'=\omicron'_1$. Consider two cases:
\begin{itemize}
\item If $d_0\subtype d'$ then $d_0\subtype d$, implying $\monadic{v}_d=\monadic{v}_d^1$ and $\monadic{v}_{d'}=\monadic{v}_{d'}^1$. The desired claim follows.
\item If $d_0$ is a strict superdomain of~$d'$ then $\monadic{v}_d\coin{q}{d'}\monadic{v}_{d'}$ holds vacuously. The desired claim follows.
\end{itemize}

\item Let $e=(\assignexpr{e_1}{e_2})$. Then $q=\qualty{\unitty}{\prestg}{\publicdom}$. Let $e_1=\loadexpr{\loadexpr{\loadexpr{x}{y_1}}{y_2}\ldots}{y_n}$. By Lemma~\ref{typesystem:lhslemma},
\[
\begin{array}{l}
\Gamma\vd y_k:q_k\mbox{,}\\
\Gamma\vd x:q'\mbox{,}\\
q_k=\qualty{\uintty}{\prestg}{d_k}\mbox{,}\quad q'=\qualty{\listty{\ldots\listty{\qualty{\listty{\qualty{\listty{q_{n+1}}}{\prestg}{d_n}}}{\prestg}{d_{n-1}}}\ldots}}{\prestg}{d_1}
\end{array}
\]
and
\[
\begin{array}{l}
\Gamma\vd e_1:q_{n+1}\mbox{,}\\
\Gamma\vd e_2:q_{n+1}\mbox{,}\\
q_{n+1}=\qualty{t}{s}{d_{n+1}}\mbox{.}
\end{array}
\]
We also have
\[
\begin{array}{l}
\semd{y_1}\gamma_d\phi=\return(\monadic{i}_d^1,\gamma_d^1,\omicron_1)\mbox{,}\\
\semd{y_2}\gamma_d^1\phi=\return(\monadic{i}_d^2,\gamma_d^2,\omicron_2)\mbox{,}\\
\dotfill\\
\semd{y_n}\gamma_d^{n-1}\phi=\return(\monadic{i}_d^n,\gamma_d^n,\omicron_n)\mbox{,}\\
\semd{e_2}\gamma_d^n\phi=\return(\monadic{r}_d,\gamma_d^{n+1},\omicron_{n+1})\mbox{,}\\
\monadic{v}_d=\return\singleton\mbox{,}\quad\gamma'_d=[x\mapsto\update(\lookup{\gamma_d}{x},\monadic{i}_d^1\ldots\monadic{i}_d^n,\monadic{r}_d)]\gamma_d^{n+1}\mbox{,}\quad\omicron=\omicron_1\ldots\omicron_{n+1}\mbox{.}
\end{array}
\]
Applying repeatedly the induction hypothesis, we get $\sem{y_k}_{d'}\gamma_{d'}^{k-1}\phi=\return(\monadic{i}_{d'}^k,\gamma_{d'}^k,\omicron'_k)$ where $\monadic{i}_d^k\coin{q_k}{d'}\monadic{i}_{d'}^k$, $\gamma_d^k\coin{\Gamma}{d'}\gamma_{d'}^k$ and $\omicron_k\coinpure{d'}\omicron'_k$ for each $k=1,\ldots,n$ (denoting $\gamma_{d'}^0=\gamma_{d'}$). In addition, $\sem{e_2}_{d'}\gamma_{d'}^n\phi=\return(\monadic{r}_{d'},\gamma_{d'}^{n+1},\omicron'_{k+1})$ where $\monadic{r}_d\coin{q_{n+1}}{d'}\monadic{r}_{d'}$, $\gamma_d^{n+1}\coin{\Gamma}{d'}\gamma_{d'}^{n+1}$ and $\omicron_{n+1}\coinpure{d'}\omicron'_{n+1}$. By Lemma~\ref{semantics:updexactcoinlemma}, $\update(\lookup{\gamma_{d'}}{x},\monadic{i}_{d'}^1\ldots\monadic{i}_{d'}^n,\monadic{r}_{d'})$ is well-defined and
\[
\update(\lookup{\gamma_d}{x},\monadic{i}_d^1\ldots\monadic{i}_d^n,\monadic{r}_d)\coin{q'}{d'}\update(\lookup{\gamma_{d'}}{x},\monadic{i}_{d'}^1\ldots\monadic{i}_{d'}^n,\monadic{r}_{d'})\mbox{.}
\]
We conclude $[x\mapsto\update(\lookup{\gamma_d}{x},\monadic{i}_d^1\ldots\monadic{i}_d^n,\monadic{r}_d)]\gamma_d^{n+1}\coin{\Gamma}{d'}[x\mapsto\update(\lookup{\gamma_{d'}}{x},\monadic{i}_{d'}^1\ldots\monadic{i}_{d'}^n,\monadic{r}_{d'})]\gamma_{d'}^{n+1}$ and the desired claim follows.

\item Let $e=\loadexpr{e_1}{e_2}$. Then
\[
\begin{array}{l}
\Gamma\vd e_1:q_1\mbox{,}\\
\Gamma\vd e_2:q_2\mbox{,}\\
q_1=\qualty{\listty{q}}{\prestg}{d_1}\mbox{,}\quad q_2=\qualty{\uintty}{\prestg}{d_1}
\end{array}
\]
and
\[
\begin{array}{l}
\semd{e_1}\gamma_d\phi=\return(\monadic{a}_d,\gamma_d^1,\omicron_1)\mbox{,}\\
\semd{e_2}\gamma_d^1\phi=\return(\monadic{i}_d,\gamma_d^2,\omicron_2)\mbox{,}\\
\monadic{v}_d=\mcomp{a_d\gets\monadic{a}_d\hstop i_d\gets\monadic{i}_d\hstop (a_d)_{i_d}}\mbox{,}\quad\gamma'_d=\gamma_d^2\mbox{,}\quad\omicron=\omicron_1\omicron_2\mbox{.}
\end{array}
\]
By the induction hypothesis about $e_1$, $\sem{e_1}_{d'}\gamma_{d'}\phi=\return(\monadic{a}_{d'},\gamma_{d'}^1,\omicron'_1)$ where $\monadic{a}_d\coin{q_1}{d'}\monadic{a}_{d'}$, $\gamma_d^1\coin{\Gamma}{d'}\gamma_{d'}^1$ and $\omicron_1\coinpure{d'}\omicron'_1$. Hence by the induction hypothesis about $e_2$, $\sem{e_2}_{d'}\gamma_{d'}^1\phi=\return(\monadic{i}_{d'},\gamma_{d'}^2,\omicron'_2)$ where $\monadic{i}_d\coin{q_2}{d'}\monadic{i}_{d'}$, $\gamma_d^2\coin{\Gamma}{d'}\gamma_{d'}^2$ and $\omicron_2\coinpure{d'}\omicron'_2$. Thus
\[
\begin{array}{l}
\sem{\loadexpr{e_1}{e_2}}_{d'}\gamma_{d'}\phi=\return(\monadic{v}_{d'},\gamma'_{d'},\omicron')\mbox{,}\\
\monadic{v}_{d'}=\mcomp{a_{d'}\gets\monadic{a}_{d'}\hstop i_{d'}\gets\monadic{i}_{d'}\hstop(a_{d'})_{i_{d'}}}\mbox{,}\quad\gamma'_{d'}=\gamma_{d'}^2\mbox{,}\quad\omicron'=\omicron'_1\omicron'_2\mbox{.}
\end{array}
\]
Consider two cases:
\begin{itemize}
\item If $d_0\subtype d'$ then, by $q$ being well-structured, also $d_1\subtype d_0\subtype d'\subtype d$. Hence we have $\monadic{a}_d=\return(\monadic{a}_d^1,\ldots,\monadic{a}_d^n)$, $\monadic{i}_d=\return i_d$, $\monadic{a}_{d'}=\return(\monadic{a}_{d'}^1,\ldots,\monadic{a}_{d'}^n)$, $\monadic{i}_{d'}=\return i_{d'}$ where $\monadic{a}_d^k\coin{q}{d'}\monadic{a}_{d'}^k$ for every $k=1,\ldots,n$ and $n\geq i_d=i_{d'}\in\NN$. This implies $\monadic{v}_d=\monadic{a}_d^{i_d}$ and $\monadic{v}_{d'}=\monadic{a}_{d'}^{i_{d'}}$, whence $\monadic{v}_d\coin{q}{d'}\monadic{v}_{d'}$.
\item If $d_0$ is a strict supertype of~$d'$ then $\monadic{v}_d\coin{q}{d'}\monadic{v}_{d'}$ holds vacuously.
\end{itemize}
Along with $\gamma_d^2\coin{\Gamma}{d'}\gamma_{d'}^2$ and $\omicron_1\omicron_2\coinpure{d'}\omicron'_1\omicron'_2$, this implies the desired claim.

\item Let $e=\stmtcomp{\letexpr{x}{e_1}}{e_2}$. Then
\[
\begin{array}{l}
\Gamma\vd e_1:q_1\mbox{,}\\
(x:q_1),\Gamma\vd e_2:q
\end{array}
\]
and
\[
\begin{array}{l}
\semd{e_1}\gamma_d\phi=\return(\monadic{v}_d^1,\gamma_d^1,\omicron_1)\mbox{,}\\
\semd{e_2}((x,\monadic{v}_d^1),\gamma_d^1)\phi=\return(\monadic{v}_d^2,\gamma_d^2,\omicron_2)\mbox{,}\\
\monadic{v}_d=\monadic{v}_d^2\mbox{,}\quad\gamma'_d=\tail\gamma_d^2\mbox{,}\quad\omicron=\omicron_1\omicron_2\mbox{.}
\end{array}
\]
By the induction hypothesis about $e_1$, $\sem{e_1}_{d'}\gamma_{d'}\phi=\return(\monadic{v}_{d'}^1,\gamma_{d'}^1,\omicron'_1)$ where $\monadic{v}_d^1\coin{q_1}{d'}\monadic{v}_{d'}^1$, $\gamma_d^1\coin{\Gamma}{d'}\gamma_{d'}^1$ and $\omicron_1\coinpure{d'}\omicron'_1$. Denoting $\Gamma'=((x:q_1),\Gamma)$, we therefore obtain $\left((x,\monadic{v}_d^1),\gamma_d^1\right)\coin{\Gamma'}{d'}\left((x,\monadic{v}_{d'}^1),\gamma_{d'}^1\right)$. By the induction hypothesis about $e_2$, $\sem{e_2}_{d'}((x,\monadic{v}_{d'}^1),\gamma_{d'}^1)\phi=\return(\monadic{v}_{d'}^2,\gamma_{d'}^2,\omicron'_2)$ where $\monadic{v}_d^2\coin{q}{d'}\monadic{v}_{d'}^2$, $\gamma_d^2\coin{\Gamma'}{d'}\gamma_{d'}^2$ and $\omicron_2\coinpure{d'}\omicron'_2$. Thus also $\tail\gamma_d^2\coin{\Gamma}{d'}\tail\gamma_{d'}^2$. As
\[
\begin{array}{l}
\sem{\stmtcomp{\letexpr{x}{e_1}}{e_2}}_{d'}\gamma_{d'}\phi=\return(\monadic{v}_{d'},\gamma'_{d'},\omicron')\mbox{,}\\
\monadic{v}_{d'}=\monadic{v}_{d'}^2\mbox{,}\quad\gamma'_{d'}=\tail\gamma_{d'}^2\mbox{,}\quad\omicron'=\omicron'_1\omicron'_2\mbox{,}
\end{array}
\]
the desired claim follows.

\item Let $e=\stmtcomp{e_1}{e_2}$. Then
\[
\begin{array}{l}
\Gamma\vd e_1:q_1\mbox{,}\\
\Gamma\vd e_2:q
\end{array}
\]
and
\[
\begin{array}{l}
\semd{e_1}\gamma_d\phi=\return(\monadic{v}_d^1,\gamma_d^1,\omicron_1)\mbox{,}\\
\semd{e_2}\gamma_d^1\phi=\return(\monadic{v}_d^2,\gamma_d^2,\omicron_2)\mbox{,}\\
\monadic{v}_d=\monadic{v}_d^2\mbox{,}\quad\gamma'_d=\gamma_d^2\mbox{,}\quad\omicron=\omicron_1\omicron_2\mbox{.}
\end{array}
\]
By the induction hypothesis about $e_1$, $\sem{e_1}_{d'}\gamma_{d'}\phi=\return(\monadic{v}_{d'}^1,\gamma_{d'}^1,\omicron'_1)$ where $\monadic{v}_d^1\coin{q_1}{d'}\monadic{v}_{d'}^1$, $\gamma_d^1\coin{\Gamma}{d'}\gamma_{d'}^1$ and $\omicron_1\coinpure{d'}\omicron'_1$. By the induction hypothesis about $e_2$, $\sem{e_2}_{d'}\gamma_{d'}^1\phi=\return(\monadic{v}_{d'}^2,\gamma_{d'}^2,\omicron'_2)$ where $\monadic{v}_d^2\coin{q}{d'}\monadic{v}_{d'}^2$, $\gamma_d^2\coin{\Gamma}{d'}\gamma_{d'}^2$ and $\omicron_2\coinpure{d'}\omicron'_2$. As
\[
\begin{array}{l}
\sem{\stmtcomp{e_1}{e_2}}_{d'}\gamma_{d'}\phi=\return(\monadic{v}_{d'},\gamma'_{d'},\omicron')\mbox{,}\\
\monadic{v}_{d'}=\monadic{v}_{d'}^2\mbox{,}\quad\gamma'_{d'}=\gamma_{d'}^2\mbox{,}\quad\omicron'=\omicron'_1\omicron'_2\mbox{,}
\end{array}
\]
the desired claim follows.

\end{itemize}

\end{proof}

\begin{theorem}[Theorem~\ref{semantics:correctness}]
Let $\Gamma \vd e : q\eff D$ with well-structured~$\Gamma$. Let $\gamma_d$, $\gamma\in\mathbf{Env}$ be $\Gamma$-exact in $\proverdom$  and in circuit, respectively, such that $\gamma_d\coincirc{\Gamma}\gamma$. Assume that, for all subexpressions of~$e$ of the form $\ofty{\getexpr{d'}{k}}{q'}$, the value $\allpure(\phi_{d'}(k))$ is $q'$-exact in $\proverdom$. Assume that, for $d=\proverdom$, there exist $\monadic{v}_d,\gamma'_d,\omicron$ such that $\semd{e}\gamma_d\phi=\return(\monadic{v}_d,\gamma'_d,\omicron)$. If $\rho$~{}is any pair of stream continuations (one for each of $\proverdom$ and $\verifierdom$) then $\sem{e}\gamma\phi(\omicron\rho)=\return(\monadic{v},\gamma',\rho)$, where $\monadic{v}$ is $q$-exact and $\gamma'$ is $\Gamma$-exact in circuit. Thereby, $\monadic{v}_d\coincirc{q}\monadic{v}$ and $\gamma'_d\coincirc{\Gamma}\gamma'$. (Here, $\omicron\rho$ denotes the pointwise concatenation of $\omicron$ and $\rho$.)
\end{theorem}

\begin{proof}
Let $q=\qualty{t_0}{s_0}{d_0}$. We proceed by induction on the structure of~$e$:
\begin{itemize}
\item Let $e=\epsilon$. Then $t_0=\unitty$, $s_0=\prestg$, $d_0=\publicdom$ and we have $\monadic{v}_d=\return\singleton$, $\gamma'_d=\gamma_d$, $\omicron=\epsilon$. We also have $\sem{e}\gamma\phi\rho=\return(\monadic{v},\gamma',\rho)$ where $\monadic{v}=\return\singleton$ and $\gamma'=\gamma$. As $d_0=\publicdom$ and $t_0$ is primitive, establishing that $\return\singleton$ is $q$-exact in circuit reduces to clause~1 of Definition~\ref{semantics:exactdef}. It holds since $\singleton\in\unitty$. Similarly, we obtain $\return\singleton\coincirc{q}\return\singleton$. By assumption, $\gamma'$ is $\Gamma$-exact in circuit and $\gamma'_d\coincirc{\Gamma}\gamma'$.

\item Let $e=\overline{n}$ where $n\in\NN$. Then $t_0=\uintmodty{\mbox{\lstinline+N+}}$ and $\monadic{v}_d=\return n$, $\gamma'_d=\gamma_d$, $\omicron=\epsilon$. We also have $\sem{e}\gamma\phi\rho=\return(\monadic{v},\gamma',\rho)$ where $\monadic{v}=\branching{\return n&\mbox{if $s_0=\poststg$ or $d_0=\publicdom$}\\\top&\mbox{otherwise}}$ and $\gamma'=\gamma$. By assumption, $\gamma'$ is $\Gamma$-exact in circuit and $\gamma'_d\coincirc{\Gamma}\gamma'$. Furthermore, we have to study two cases:
\begin{itemize}
\item If $s_0=\poststg$ or $d_0=\publicdom$ then $\monadic{v}=\return n$. As $t_0$ is primitive, establishing that $\monadic{v}$ is $q$-exact in circuit reduces to clause~1 of Definition~\ref{semantics:exactdef}. It holds since $n$ is an integer. Similarly, we obtain $\return n\coincirc{q}\return n$.
\item If $s_0=\prestg$ and $d_0\ne\publicdom$ then $\monadic{v}=\top$ and establishing that $\monadic{v}$ is $q$-exact in circuit reduces to clause~3 of Definition~\ref{semantics:exactdef}. The former equality is exactly what clause~3 requires. The claim $\monadic{v}_d\coincirc{q}\monadic{v}$ holds vacuously.
\end{itemize}

\item Let $e=\overline{b}$ where $b\in\BB$. This case is analogous to the previous one.

\item Let $e=x$. Then $\lookup{\Gamma}{x}=q$ and $\monadic{v}_d=\lookup{\gamma_d}{x}$, $\gamma'_d=\gamma_d$, $\omicron=\epsilon$. We also have $\sem{e}\gamma\phi\rho=\return(\monadic{v},\gamma',\rho)$ where $\monadic{v}=\lookup{\gamma}{x}$, $\gamma'=\gamma$. Hence $\monadic{v}$ is $q$-exact in circuit by the assumption that $\gamma$ is $\Gamma$-exact in circuit and Lemma~\ref{semantics:lookupexactlemma}. Also $\gamma'$ being $\Gamma$-exact in circuit directly follows from assumption. Finally, $\monadic{v}_d\coincirc{q}\monadic{v}$ by $\gamma_d\coincirc{\Gamma}\gamma$, and $\gamma'_d\coincirc{\Gamma}\gamma'$ by assumption.

\item Let $e=\addexpr{e_1}{e_2}$. Then $t_0=\uintmodty{\mbox{\lstinline+N+}}$ and
\[
\begin{array}{l}
\Gamma\vd e_1:q\eff D_1\mbox{,}\\
\Gamma\vd e_2:q\eff D_2\mbox{,}
\end{array}
\]
and also
\[
\begin{array}{l}
\semd{e_1}\gamma_d\phi=\return(\monadic{v}_d^1,\gamma_d^1,\omicron_1)\mbox{,}\\
\semd{e_2}\gamma_d^1\phi=\return(\monadic{v}_d^2,\gamma_d^2,\omicron_2)\mbox{,}\\
\monadic{v}_d=\mcomp{v_d^1\gets\monadic{v}_d^1\hstop v_d^2\gets\monadic{v}_d^2\hstop\return(v_d^1+v_d^2)},\quad\gamma'_d=\gamma_d^2\mbox{,}\quad\omicron=\omicron_1\omicron_2\mbox{.}
\end{array}
\]
By $\monadic{v}_d^1,\monadic{v}_d^2$ being $q$-exact in $\proverdom$, we must have $\monadic{v}_d^1=\return v_d^1$ and $\monadic{v}_d^2=\return v_d^2$, whence $\monadic{v}_d=\return(v_d^1+v_d^2)$. By the induction hypothesis about~$e_1$, $\sem{e_1}\gamma\phi(\omicron_1\omicron_2\rho)=\return(\monadic{v}^1,\gamma^1,\omicron_2\rho)$ where $\monadic{v}^1$ is $q$-exact and $\gamma^1$ is $\Gamma$-exact in circuit, $\monadic{v}_d^1\coincirc{q}\monadic{v}^1$ and $\gamma_d^1\coincirc{\Gamma}\gamma^1$. Now by the induction hypothesis about~$e_2$, $\sem{e_2}\gamma^1\phi(\omicron_1\rho)=\return(\monadic{v}^2,\gamma^2,\rho)$ where $\monadic{v}^2$ is $q$-exact and $\gamma^2$ is $\Gamma$-exact in circuit, $\monadic{v}_d^2\coincirc{q}\monadic{v}^2$ and $\gamma_d^2\coincirc{\Gamma}\gamma^2$. Hence $\sem{e}\gamma\phi(\omicron\rho)=\return(\monadic{v},\gamma',\rho)$ where
\[
\monadic{v}=\mcomp{v^1\gets\monadic{v}^1\hstop v^2\gets\monadic{v}^2\hstop\return(v^1+v^2)},\quad\gamma'=\gamma^2\mbox{.}
\]
Then $\gamma'$ being $\Gamma$-exact in circuit and $\gamma'_d\coincirc{\Gamma}\gamma'$ are implied by the above. Finally, we have to consider two cases:
\begin{itemize}
\item If $s_0=\poststg$ or $d_0=\publicdom$ then, by $\monadic{v}^1$ and $\monadic{v}^2$ being $q$-exact in circuit, we have $\monadic{v}^1=\return v^1$ and $\monadic{v}^2=\return v^2$ for integers $v^1,v^2$. Hence $\monadic{v}=\return(v^1+v^2)$ which shows that $\monadic{v}$ is $q$-exact in circuit by clause~1 of Definition~\ref{semantics:exactdef}. Moreover, $\monadic{v}_d^1\coincirc{q}\monadic{v}^1$ and $\monadic{v}_d^2\coincirc{q}\monadic{v}^2$ imply $v_d^1=v^1$ and $v_d^2=v^2$. Hence $\monadic{v}_d\coincirc{q}\monadic{v}$.
\item If $s_0=\prestg$ and $d_0\ne\publicdom$ then, by $\monadic{v}^1$ and $\monadic{v}^2$ being $q$-exact in circuit, we have $\monadic{v}^1=\monadic{v}^2=\top$ and hence also $\monadic{v}=\top$. By clause~3 of Definition~\ref{semantics:exactdef}, $\monadic{v}$ is $q$-exact in circuit. The claim $\monadic{v}_d\coincirc{q}\monadic{v}$ holds vacuously.
\end{itemize}

\item Let $e=\assertexpr{e_1}$. Then $t_0=\unitty$, $s_0=\prestg$, $d_0=\publicdom$ and
\[
\Gamma\vd e_1:\qualty{\boolmodty{\mbox{\lstinline+N+}}}{\poststg}{d_1}\eff D
\]
and also
\[
\begin{array}{l}
\semd{e_1}\gamma_d\phi=\return(\monadic{v}_d^1,\gamma_d^1,\omicron_1)\mbox{,}\\
\monadic{v}_d^1\ne\return\fls\mbox{,}\quad\monadic{v}_d=\return\singleton\mbox{,}\quad\gamma'_d=\gamma_d^1\mbox{,}\quad\omicron=\omicron_1\mbox{.}
\end{array}
\]
Denote $q_1=(\qualty{\boolmodty{\mbox{\lstinline+N+}}}{\poststg}{d_1})$. By $\monadic{v}_d^1$ being $q_1$-exact in $\proverdom$ and $\monadic{v}_d^1\ne\return\fls$, the only possibility is $\monadic{v}_d^1=\return\tru$. By the induction hypothesis about $e_1$, $\sem{e_1}\gamma\phi(\omicron_1\rho)=\return(\monadic{v}^1,\gamma^1,\rho)$ where $\monadic{v}^1$ is $q_1$-exact and $\gamma^1$ is $\Gamma$-exact in circuit, $\monadic{v}_d^1\coincirc{q_1}\monadic{v}^1$ and $\gamma_d^1\coincirc{\Gamma}\gamma^1$. Hence $\monadic{v}^1\ne\return\fls$ since $e_1$ is in $\poststg$ which implies that $\return\tru$ and $\return\fls$ cannot be $q_1$-coincident in circuit. Consequently, $\sem{e}\gamma\phi(\omicron\rho)=\return(\monadic{v},\gamma',\rho)$ where $\monadic{v}=\return\singleton$, $\gamma'=\gamma^1$, whence we have $\monadic{v}$ being $q$-exact in circuit and $\monadic{v}_d\coincirc{q}\monadic{v}$. By the above, we also have $\gamma'$ being $\Gamma$-exact in circuit and $\gamma'_d\coincirc{\Gamma}\gamma'$.

\item Let $e=\ofty{\getexpr{d'}{k}}{q}$. Then $\allpre_{d'}(q)$, implying $s_0=\prestg$ and $d_0=d'$. Moreover, $\monadic{v}_d=\allpure(\phi_{d'}(k))$, $\gamma'_d=\gamma_d$ and $\omicron=\epsilon$. We also have $\sem{e}\gamma\phi\rho=\return(\monadic{v},\gamma',\rho)$ where $\monadic{v}=\branching{\allpure(\phi_{d'}(k))&\mbox{if $d'=\publicdom$}\\\top&\mbox{otherwise}}$ and $\gamma'=\gamma$. Hence, by assumptions, $\gamma'$ is $\Gamma$-exact in circuit and $\gamma'_d\coincirc{\Gamma}\gamma'$. For the other desired claims, consider two cases:

\begin{itemize}
\item If $d'=\publicdom$ then $\monadic{v}_d=\allpure(\phi_{d'}(k))=\monadic{v}$, implying $\monadic{v}_d\coincirc{q}\monadic{v}$. Moreover, $\monadic{v}$ is $q$-exact in circuit by Lemma~\ref{semantics:approxcirclemma}.
\item If $d'\ne\publicdom$ then $\monadic{v}=\top$ which is $q$-exact since $s_0=\prestg$. Finally, $\monadic{v}_d\coincirc{q}\monadic{v}$ holds vacuously.
\end{itemize}

\item Let $e=\ifexpr{e_1}{e_2}{e_3}$. Then
\[
\begin{array}{l}
\Gamma\vd e_1:\qualty{\boolmodty{\mbox{\lstinline+N+}}}{\prestg}{d_1}\eff D_1\mbox{,}\\
\Gamma\vd e_2:q\eff D_2\mbox{,}\\
\Gamma\vd e_3:q\eff D_3\mbox{,}\\
\gen{d_1}\supseteq\gen{s_0}\cup \gen{d_0}\cup D_2\cup D_3\mbox{,}
\end{array}
\]
and
\[
\semd{e_1}\gamma_d\phi=\return(\monadic{v}_d^1,\gamma_d^1,\omicron_1)\mbox{.}
\]
Denote $q_1=(\qualty{\boolmodty{\mbox{\lstinline+N+}}}{\prestg}{d_1})$. Note that $\monadic{v}_d^1$ is $q_1$-exact in $\proverdom$, whence $\monadic{v}_d^1=\return v_d^1$ where $v_d^1\in\left\{\tru,\fls\right\}$. Thus we have to consider two cases:
\begin{itemize}
\item If $\monadic{v}_d^1=\return\tru$ then $\semd{e_2}\gamma_d^1\phi\omicron_1=\return(\monadic{v}_d^2,\gamma_d^2,\omicron_2)$, $\monadic{v}_d=\monadic{v}_d^2$, $\gamma'_d=\gamma_d^2$, $\omicron=\omicron_1\omicron_2$. By the induction hypothesis about $e_1$, $\sem{e_1}\gamma\phi(\omicron_1\omicron_2\rho)=\return(\monadic{v}^1,\gamma^1,\omicron_2\rho)$ where $\monadic{v}^1$ is $q_1$-exact and $\gamma^1$ is $\Gamma$-exact in circuit, $\monadic{v}_d^1\coincirc{q_1}\monadic{v}^1$ and $\gamma_d^1\coincirc{\Gamma}\gamma^1$. If $d_1=\publicdom$ then $\monadic{v}_d^1\coincirc{q_1}\monadic{v}^1$ implies $\monadic{v}^1=\return\tru$. By the induction hypothesis about $e_2$, $\sem{e_2}\gamma^1\phi(\omicron_2\rho)=\return(\monadic{v}^2,\gamma^2,\rho)$ where $\monadic{v}^2$ is $q$-exact and $\gamma^2$ is $\Gamma$-exact in circuit, $\monadic{v}_d^2\coincirc{q}\monadic{v}^2$ and $\gamma_d^2\coincirc{\Gamma}\gamma^2$. Hence $\sem{e}\gamma\phi(\omicron\rho)=\return(\monadic{v},\gamma',\rho)$ where $\monadic{v}=\monadic{v}^2$ and $\gamma'=\gamma^2$. The desired claim follows in this case. If $d_1\ne\publicdom$ then $\monadic{v}^1$ being $q_1$-exact in circuit implies $\monadic{v}^1=\top$. As $\gen{d_1}\supseteq D_2$ implies $\publicdom\notin D_2$, Theorems~\ref{semantics:outputthm} and~\ref{semantics:effcircthm} imply $\omicron_2=\epsilon$ and $\gamma_d^1\coincirc{q}\gamma_d^2$. Hence $\sem{e}\gamma\phi(\omicron\rho)=\sem{e}\gamma\phi(\omicron_1\rho)=\return(\monadic{v},\gamma',\rho)$ where $\monadic{v}=\top$ and $\gamma'=\gamma^1$. As $\gen{d_1}\supseteq\gen{d_0}$ implies $d_1\subtype d_0$, also $d_0\ne\publicdom$. As $\gen{d_1}\supseteq\gen{s_0}$ implies $s_0=\prestg$, this means that $\top$ is $q$-exact in circuit and $\monadic{v}_d\coincirc{q}\top$ vacuously. We also have $\gamma'_d\coincirc{\Gamma}\gamma'$ because of $\gamma_d^1\coincirc{q}\gamma_d^2$, $\gamma_d^1\coincirc{q}\gamma^1$ and Lemma~\ref{semantics:coinlemma}. Finally, $\gamma'$ is $\Gamma$-exact in circuit by the above.
\item The case $\monadic{v}_d^1=\return\fls$ is analogous.
\end{itemize}

\item Let $e=\forexpr{x}{e_1}{e_2}{e_3}$. Then
\[
\begin{array}{l}
\Gamma\vd e_1:\qualty{\uintty}{\prestg}{d_0}\eff D_1\mbox{,}\\
\Gamma\vd e_2:\qualty{\uintty}{\prestg}{d_0}\eff D_2\mbox{,}\\
(x:\qualty{\uintty}{\prestg}{d_0}),\Gamma\vd e_3:\qualty{t_1}{s_1}{d_1}\eff D_3\mbox{,}\\
\gen{d_0}\supseteq\gen{s_1}\cup\gen{d_1}\cup D_3\mbox{,}\\
t_0=\listty{\qualty{t_1}{s_1}{d_1}}\mbox{,}\quad s_0=\prestg\mbox{,}
\end{array}
\]
and
\[
\begin{array}{l}
\semd{e_1}\gamma_d\phi=\return(\monadic{v}_d^1,\gamma_d^1,\omicron_1)\mbox{,}\\
\semd{e_2}\gamma_d^1\phi=\return(\monadic{v}_d^2,\gamma_d^2,\omicron_2)\mbox{.}
\end{array}
\]
Denote $q'=(\qualty{\uintty}{\prestg}{d_0})$ and $q_1=(\qualty{t_1}{s_1}{d_1})$. Note that $\monadic{v}_d^1,\monadic{v}_d^2$ are $q'$-exact in $\proverdom$ by Theorem~\ref{semantics:localthm}. Hence $\monadic{v}_d^1=\return i_1$ and $\monadic{v}_d^2=\return i_2$ where $i_1,i_2\in\NN$. Denoting $n=\max(0,i_2-i_1)$, we obtain
\[
\begin{array}{l}
\semd{e_3}((x,\return i_1),\gamma_d^2)\phi=\return(\monadic{v}_d^3,\gamma_d^3,\omicron_3)\mbox{,}\\
\semd{e_3}([x\mapsto\return(i_1+1)]\gamma_d^3)\phi=\return(\monadic{v}_d^4,\gamma_d^4,\omicron_4)\mbox{,}\\
\dotfill\\
\semd{e_3}([x\mapsto\return(i_1+n-1)]\gamma_d^{n+1})\phi=\return(\monadic{v}_d^{n+2},\gamma_d^{n+2},\omicron_d^{n+2})\mbox{,}\\
\monadic{v}_d=\return(\monadic{v}_d^3,\ldots,\monadic{v}_d^{n+2})\mbox{,}\quad\gamma'_d=\tail\gamma_d^{n+2}\mbox{,}\quad\omicron=\omicron_1\ldots\omicron_{n+2}\mbox{.}
\end{array}
\]
By the induction hypothesis about~$e_1$, $\sem{e_1}\gamma\phi(\omicron_1\ldots\omicron_{n+2}\rho)=\return(\monadic{v}^1,\gamma^1,\omicron_2\ldots\omicron_{n+2}\rho)$ where $\monadic{v}^1$ is $q'$-exact and $\gamma^1$ is $\Gamma$-exact in circuit, $\monadic{v}_d^1\coincirc{q'}\monadic{v}^1$ and $\gamma_d^1\coincirc{\Gamma}\gamma^1$. Hence by the induction hypothesis about~$e_2$, $\sem{e_2}\gamma^1\phi(\omicron_2\ldots\omicron_{n+2}\rho)=\return(\monadic{v}^2,\gamma^2,\omicron_3\ldots\omicron_{n+2}\rho)$ where $\monadic{v}^2$ is $q'$-exact and $\gamma^2$ is $\Gamma$-exact in circuit, $\monadic{v}_d^2\coincirc{q'}\monadic{v}^2$ and $\gamma_d^2\coincirc{\Gamma}\gamma^2$. If $d_0=\publicdom$ then $\monadic{v}_d^1\coincirc{q'}\monadic{v}^1$ and $\monadic{v}_d^2\coincirc{q'}\monadic{v}^2$ imply $\monadic{v}^1=\return i_1$ and $\monadic{v}^2=\return i_2$. As $\gamma^2$ is $\Gamma$-exact and $\return i_1$ is $q'$-exact in circuit, the updated environment $\left((x,\return i_1),\gamma^2\right)$ is $((x:q'),\Gamma)$-exact in circuit. Moreover, $\gamma_d^2\coincirc{\Gamma}\gamma^2$ and $\return i_1\coincirc{q'}\return i_1$ together give $\left((x,\return i_1),\gamma_d^2\right)\coincirc{\Gamma'}\left((x,\return i_1),\gamma^2\right)$ where $\Gamma'=((x:q'),\Gamma)$. Hence the induction hypothesis about~$e_3$ applies and gives $\sem{e_3}\gamma^2\phi(\omicron_3\ldots\omicron_{n+2}\rho)=\return(\monadic{v}^3,\gamma^3,\omicron_4\ldots\omicron_{n+2}\rho)$ where $\monadic{v}^3$ is $q_1$-exact and $\gamma_3$ is $\Gamma'$-exact in circuit, $\monadic{v}_d^3\coincirc{q_1}\monadic{v}^3$ and $\gamma_d^3\coincirc{\Gamma'}\gamma^3$. Replacing $i_1$ by $i_1+1$ does not violate the required properties, so we analogously obtain $\sem{e_3}\gamma^{k-1}\phi(\omicron_k\ldots\omicron_{n+2}\rho)=\return(\monadic{v}^k,\gamma^k,\omicron_{k+1}\ldots\omicron_{n+2}\rho)$ where $\monadic{v}^k$ is $q_1$-exact in circuit and $\monadic{v}_d^k\coincirc{q_1}\monadic{v}^k$ for all $k=3,4,\ldots,n+2$, $\gamma^{n+2}$ is $\Gamma'$-exact in circuit and $\gamma_d^{n+2}\coincirc{\Gamma'}\gamma^{n+2}$. Obviously the latter implies $\tail\gamma^{n+2}$ being $\Gamma$-exact in circuit and $\tail\gamma_d^{n+2}\coincirc{\Gamma}\tail\gamma^{n+2}$. Thus we obtain $\sem{e}\gamma\phi(\omicron\rho)=\return(\monadic{v},\gamma',\rho)$ where $\monadic{v}=\return(\monadic{v}^3,\ldots,\monadic{v}^{n+2})$ and $\gamma'=\gamma^{n+2}$. By clause~2 of Definitions~\ref{semantics:exactdef} and~\ref{semantics:coincidentdef}, $\return(\monadic{v}^3,\ldots,\monadic{v}^{n+2})$ is $q$-exact in circuit and $\return(\monadic{v}_d^3,\ldots,\monadic{v}_d^{n+2})\coincirc{q}\return(\monadic{v}^3,\ldots,\monadic{v}^{n+2})$. The desired claim follows in this case. If $d_0\ne\publicdom$ then $\monadic{v}^1,\monadic{v}^2$ being $q'$-exact in circuit implies $\monadic{v}^1=\monadic{v}^2=\top$. As $\gen{d_0}\supseteq D_3$ implies $\publicdom\notin D_3$, Theorems~\ref{semantics:outputthm} and~\ref{semantics:effcircthm} imply $\omicron_3=\epsilon$, $\left((x,\return i_1),\gamma_d^2\right)\coincirc{\Gamma'}\gamma_d^3$ and also $\omicron_k=\epsilon$, $[x\mapsto(i_1+k-3)]\gamma_d^{k-1}\coincirc{\Gamma'}\gamma_d^k$ for every $k=4,\ldots,n+2$. Obviously this implies $\gamma_d^2\coincirc{\Gamma}\tail\gamma_d^3$ and $\tail\gamma_d^{k-1}\coincirc{\Gamma}\tail\gamma_d^k$ for every $k=4,\ldots,n+2$. By Lemma~\ref{semantics:coinlemma}, $\gamma_d^2\coincirc{\Gamma}\tail\gamma_{n+2}$ which is the same as $\gamma_d^2\coincirc{\Gamma}\gamma'_d$. We obtain $\sem{e}\gamma\phi(\omicron\rho)=\sem{e}\gamma\phi(\omicron_1\omicron_2\rho)=\return(\monadic{v},\gamma',\rho)$ where $\monadic{v}=\top$ and $\gamma'=\gamma^2$. As $d_0\ne\publicdom$ and $s_0=\prestg$, $\top$ is $q$-exact in circuit and $\monadic{v}_d\coincirc{q}\top$ vacuously. By Lemma~\ref{semantics:coinlemma}, $\gamma_d^2\coincirc{\Gamma}\gamma'_d$ and $\gamma_d^2\coincirc{\Gamma}\gamma^2$ together imply $\gamma'_d\coincirc{\Gamma}\gamma'$. By the above, we also get $\gamma'$ being $\Gamma$-exact in circuit. The desired claim follows.

\item Let $e=\wireexpr{e_1}$. Then
\[
\begin{array}{l}
\Gamma\vd e_1:\qualty{t_0}{\prestg}{d_0}\eff D_1\mbox{,}\\
s_0=\poststg\mbox{,}\quad\mbox{$t_0$ is $\uintmodty{\mbox{\lstinline+N+}}$ or $\boolmodty{\mbox{\lstinline+N+}}$}\mbox{,}
\end{array}
\]
and
\[
\begin{array}{l}
\semd{e_1}\gamma\phi=\return(\monadic{v}_d^1,\gamma_d^1,\omicron_1)\mbox{,}\\
\monadic{v}_d=\monadic{v}_d^1\mbox{,}\quad\gamma'_d=\gamma_d^1\mbox{,}\quad\omicron=\lam{d'}{\branching{(\omicron_1)_{d'}\monadic{v}_d^1&\mbox{if $d'$ is the domain of~$e_1$}\\(\omicron_1)_{d'}&\mbox{otherwise}}}\mbox{.}
\end{array}
\]
Here, $(\omicron_1)_{d'}$ denotes the member of the pair of streams $\omicron_1$ that corresponds to domain~$d'$ (for $d'\in\set{\proverdom,\verifierdom}$). Denote $q'=(\qualty{t_0}{\prestg}{d_0})$. By the induction hypothesis about~$e_1$,
\[
\sem{e_1}\gamma\phi(\omicron\rho)=\return(\monadic{v}^1,\gamma^1,\lam{d'}{\branching{\monadic{v}_d^1\rho&\mbox{if $d'$ is the domain of~$e_1$}\\\rho&\mbox{otherwise}}})
\]
where $\monadic{v}^1$ is $q'$-exact and $\gamma^1$ is $\Gamma$-exact in circuit, $\monadic{v}_d^1\coincirc{q'}\monadic{v}^1$ and $\gamma_d^1\coincirc{\Gamma}\gamma^1$. Hence $\sem{e}\gamma\phi(\omicron\rho)=\return(\monadic{v},\gamma',\rho)$ where $\monadic{v}=\branching{\monadic{v}^1&\mbox{if $d_0=\publicdom$}\\\monadic{v}_d^1&\mbox{otherwise}}$ and $\gamma'=\gamma^1$. Obviously we have $\gamma'$ being $\Gamma$-exact in circuit and $\gamma'_d\coincirc{\Gamma}\gamma'$. By Lemma~\ref{semantics:coinlemma}, $\monadic{v}_d\coincirc{q}\monadic{v}$. By $\monadic{v}_d^1$ being $q'$-exact in $\proverdom$, we know that $\monadic{v}_d^1=\return v_d^1$ where $v_d^1\in t_0$. As $s_0=\poststg$, $\return v_d^1$ is $q$-exact in circuit. Moreover, as $\monadic{v}^1$ is $q'$-exact in circuit, if $d_0=\publicdom$ then $\monadic{v}^1$ is $q$-exact in circuit. The desired claim follows.

\item Let $e=\castexpr{e_1}{q_0}$. Let the type of $e_1$ be $q_1$, i.e.,
\[
\begin{array}{l}
\Gamma\vd e_1:\qualty{t_0}{s_1}{d_1}\eff D\mbox{,}\\
s_1\subtype s_0\mbox{,}\quad d_1\subtype d_0\mbox{,}\quad\gen{d_0}\supseteq\gen{t_0}\mbox{.}
\end{array}
\]
We have
\[
\begin{array}{l}
\semd{e_1}\gamma_d\phi=\return(\monadic{v}_d^1,\gamma_d^1,\omicron_1)\mbox{,}\\
\monadic{v}_d=\monadic{v}_d^1\mbox{,}\quad\gamma'_d=\gamma_d^1\mbox{,}\quad\omicron=\omicron_1\mbox{.}
\end{array}
\]
By the induction hypothesis, $\sem{e_1}\gamma\phi(\omicron_1\rho)=\return(\monadic{v}^1,\gamma^1,\rho)$ where $\monadic{v}^1$ is $q_1$-exact and $\gamma^1$ is $\Gamma$-exact in circuit, $\monadic{v}_d^1\coincirc{q_1}\monadic{v}^1$ and $\gamma_d^1\coincirc{\Gamma}\gamma^1$. Hence $\sem{e}\gamma\phi(\omicron\rho)=\return(\monadic{v},\gamma',\rho)$ where $\monadic{v}=\branching{\monadic{v}^1&\mbox{if $s_0=\poststg$ or $d_0=\publicdom$}\\\top&\mbox{otherwise}}$ and $\gamma'=\gamma^1$. Consider two cases:
\begin{itemize}
\item If $s_0=\poststg$ or $d_0=\publicdom$ then $\monadic{v}=\monadic{v}^1$. As $s_1\subtype s_0$ and $d_1\subtype d_0$, we have $s_1=\poststg$ or $d_1=\publicdom$. Hence $\monadic{v}^1$ being $q_1$-exact in circuit and $\monadic{v}_d^1\coincirc{q_1}\monadic{v}^1$ imply $\monadic{v}^1$ being $q_0$-exact in circuit and $\monadic{v}_d^1\coincirc{q_0}\monadic{v}^1$, respectively. The desired claim follows.
\item If $s_0=\prestg$ and $d_0\ne\publicdom$ then $\monadic{v}=\top$. We have $\top$ being $q$-exact in circuit and $\monadic{v}_d\coincirc{q}\top$ vacuously. The desired claim follows.
\end{itemize}

\item Let $e=(\assignexpr{e_1}{e_2})$. Then
\[
\begin{array}{l}
\Gamma\vd e_1:\qualty{t_1}{s_1}{d_1}\mbox{,}\\
\Gamma\vd e_2:\qualty{t_1}{s_1}{d_1}\mbox{,}\\
t_0=\unitty\mbox{,}\quad s_0=\prestg\mbox{,}\quad d_0=\publicdom\mbox{.}
\end{array}
\]
Let $e_1=\loadexpr{\loadexpr{\loadexpr{x}{y_1}}{y_2}\ldots}{y_n}$ and denote $\monadic{a}_d=\lookup{\gamma_d}{x}$, $\monadic{a}=\lookup{\gamma}{x}$. Then
\[
\begin{array}{l}
\semd{y_1}\gamma_d\phi=\return(\monadic{i}_d^1,\gamma_d^1,\omicron_1)\mbox{,}\\
\semd{y_2}\gamma_d^1\phi=\return(\monadic{i}_d^2,\gamma_d^2,\omicron_2)\mbox{,}\\
\dotfill\\
\semd{y_n}\gamma_d^{n-1}\phi=\return(\monadic{i}_d^n,\gamma_d^n,\omicron_n)\mbox{,}\\
\semd{e_2}\gamma_d^n\phi=\return(\monadic{r}_d,\gamma_d^{n+1},\omicron_{n+1})\mbox{,}\\
\monadic{v}_d=\return\singleton\mbox{,}\quad\gamma'_d=[x\mapsto\update(\monadic{a}_d,\monadic{i}_d^1\ldots\monadic{i}_d^n,\monadic{r}_d)]\gamma_d^{n+1}\mbox{,}\quad\omicron=\omicron_1\ldots\omicron_{n+1}\mbox{.}
\end{array}
\]
Denote $q_1=(\qualty{t_1}{s_1}{d_1})$, $q'_i=\qualty{\uintty}{\prestg}{d'_i}$ for every $i=1,\ldots,n$, and
\[
q'=\qualty{\listty{\ldots\listty{\qualty{\listty{\qualty{\listty{\qualty{t_1}{s_1}{d_1}}}{\prestg}{d'_n}}}{\prestg}{d'_{n-1}}}\ldots}}{\prestg}{d'_1}\mbox{.}
\] 
By Lemma~\ref{typesystem:lhslemma}, $\Gamma\vd y_i:q'_i$ for each $i=1,\ldots,n$ and $\Gamma\vd x:q'$. Hence $\monadic{a}_d$ is $q'$-exact in $\proverdom$ by exactness of~$\gamma$. By the induction hypothesis about all~$y_k$, $\sem{y_k}\gamma^{k-1}\phi(\omicron_k\ldots\omicron_{n+1}\rho)=\return(\monadic{i}^k,\gamma^k,\omicron_{k+1}\ldots\omicron_{n+1}\rho)$ where $\monadic{i}^k$ is $q'_k$-exact and $\gamma^k$ is $\Gamma$-exact in circuit, $\monadic{i}_d^k\coincirc{q'_k}\monadic{i}^k$ and $\gamma_d^k\coincirc{\Gamma}\gamma^k$. The induction hypothesis about~$e_2$ implies that $\sem{e_2}\gamma^k\phi(\omicron_{n+1}\rho)=\return(\monadic{r},\gamma^{n+1},\rho)$ where $\monadic{r}$ is $q_1$-exact and $\gamma^{n+1}$ is $\Gamma$-exact in circuit, $\monadic{r}_d\coincirc{q_1}\monadic{r}$ and $\gamma_d^{n+1}\coincirc{\Gamma}\gamma^{n+1}$. By Lemma~\ref{semantics:updexactcoinlemma}, $\update(\monadic{a},\monadic{i}^1\ldots\monadic{i}^n,\monadic{r})$ is a well-defined monadic value which is $q'$-exact in circuit, whereby $\update(\monadic{a}_d,\monadic{i}_d^1\ldots\monadic{i}_d^n,\monadic{r}_d)\coincirc{q'}\update(\monadic{a},\monadic{i}^1\ldots\monadic{i}^n,\monadic{r})$. Hence $\sem{e}\gamma\phi(\omicron\rho)=\return(\monadic{v},\gamma',\rho)$ where $\monadic{v}=\return\singleton$ and $\gamma'=[x\mapsto\update(\monadic{a},\monadic{i}^1\ldots\monadic{i}^n,\monadic{r})]\gamma^{n+1}$. Thereby, $[x\mapsto\update(\monadic{a},\monadic{i}^1\ldots\monadic{i}^n,\monadic{r})]\gamma^{n+1}$ is $\Gamma$-exact in circuit by the above and
\[
[x\mapsto\update(\monadic{a}_d,\monadic{i}_d^1\ldots\monadic{i}_d^n,\monadic{r}_d)]\gamma_d^{n+1}\coincirc{\Gamma}[x\mapsto\update(\monadic{a},\monadic{i}^1\ldots\monadic{i}^n,\monadic{r})]\gamma^{n+1}\mbox{.}
\]
As $\return\singleton$ is $q$-exact in circuit and $\return\singleton\coincirc{q}\return\singleton$, we are done.

\item Let $e=\loadexpr{e_1}{e_2}$. Then
\[
\begin{array}{l}
\Gamma\vd e_1:\qualty{\listty{q}}{s_1}{d_1}\mbox{,}\\
\Gamma\vd e_2:\qualty{\uintty}{s_1}{d_1}\mbox{,}
\end{array}
\]
and
\[
\begin{array}{l}
\semd{e_1}\gamma_d\phi=\return(\monadic{a}_d,\gamma_d^1,\omicron_1)\mbox{,}\\
\semd{e_2}\gamma_d^1\phi=\return(\monadic{i}_d,\gamma_d^2,\omicron_2)\mbox{,}\\
\monadic{v}_d=\mcomp{a_d\gets\monadic{a}_d\hstop i_d\gets\monadic{i}_d\hstop (a_d)_{i_d}}\mbox{,}\quad\gamma'_d=\gamma_d^2\mbox{,}\quad\omicron=\omicron_1\omicron_2\mbox{.}
\end{array}
\]
Denote $q'=(\qualty{\listty{q}}{s_1}{d_1})$ and $q_1=(\qualty{\uintty}{s_1}{d_1})$. By $\monadic{a}_d,\monadic{i}_d$ being $q'$-exact and $q_1$-exact, respectively, in $\proverdom$, we have $\monadic{a}=\return(\monadic{a}_d^1,\ldots,\monadic{a}_d^n)$, where all $\monadic{a}_d^k$ are $q$-exact in $\proverdom$, and $\monadic{i}_d=\return i$ where $i\in\NN$. Thus $\monadic{v}_d=\monadic{a}_d^i$ (whereby $i\leq n$ as the result is well-defined). By the induction hypothesis about $e_1$, $\sem{e_1}\gamma\phi(\omicron_1\omicron_2\rho)=\return(\monadic{a},\gamma^1,\omicron_2\rho)$ where $\monadic{a}$ is $q'$-exact and $\gamma^1$ is $\Gamma$-exact in circuit, $\monadic{a}_d\coincirc{q'}\monadic{a}$ and $\gamma_d^1\coincirc{\Gamma}\gamma^1$. Now by the induction hypothesis about $e_2$, $\sem{e_2}\gamma^1\phi(\omicron_2\rho)=\return(\monadic{i},\gamma^2,\rho)$ where $\monadic{i}$ is $q_1$-exact and $\gamma^2$ is $\Gamma$-exact in circuit, $\monadic{i}_d\coincirc{q_1}\monadic{i}$ and $\gamma_d^2\coincirc{\Gamma}\gamma^2$. Hence $\sem{e}\gamma\phi(\omicron\rho)=\return(\monadic{v},\gamma',\rho)$ where
\[
\monadic{v}=\mcomp{a\gets\monadic{a}\hstop i\gets\monadic{i}\hstop a_i}\mbox{,}\quad\gamma'=\gamma^2\mbox{.}
\]
Consider two cases:
\begin{itemize}
\item If $d_1=\publicdom$ then, by $\monadic{a}_d\coincirc{q'}\monadic{a}$, $\monadic{i}_d\coincirc{q_1}\monadic{i}$ and exactness in circuit, we get $\monadic{a}=\return(\monadic{a}^1,\ldots,\monadic{a}^n)$ where $\monadic{a}_d^k\coincirc{q}\monadic{a}^k$ for every $k=1,\ldots,n$ and $\monadic{i}=\return i$, whereby $\monadic{a}^1,\ldots,\monadic{a}^n$ are $q$-exact in circuit. Hence $\monadic{v}=\monadic{a}^i$ exists, is $q$-exact in circuit and satisfies $\monadic{v}_d\coincirc{q}\monadic{v}$.
\item If $d_1\ne\publicdom$ then $\monadic{a}=\top$ and $\monadic{i}=\top$ by $s_1=\prestg$ and exactness in circuit. Hence $\monadic{v}=\top$. As $\Gamma$ is well-structured, $\qualty{\listty{q}}{s_1}{d_1}$ is well-structured by Theorem~\ref{typesystem:thm}. Thus $\gen{d_1}\supseteq\gen{d_0}\cup\gen{s_0}$, implying that $s_0=\prestg$ and $d_0\ne\publicdom$. Hence $\top$ is $q$-exact in circuit and $\monadic{v}_d\coincirc{q}\monadic{v}$ vacuously.
\end{itemize}
In both cases, we obtain $\gamma'$ being $\Gamma$-exact in circuit and $\gamma'_d\coincirc{\Gamma}\gamma'$ by the above.

\item Let $e=\stmtcomp{\letexpr{x}{e_1}}{e_2}$. Then
\[
\begin{array}{l}
\Gamma\vd e_1:q_1\eff D_1\mbox{,}\\
(x:q_1),\Gamma\vd e_2:q\eff D_2\mbox{,}
\end{array}
\]
and
\[
\begin{array}{l}
\semd{e_1}\gamma_d\phi=\return(\monadic{v}_d^1,\gamma_d^1,\omicron_1)\mbox{,}\\
\semd{e_2}((x,\monadic{v}_d^1),\gamma_d^1)\phi=\return(\monadic{v}_d^2,\gamma_d^2,\omicron_2)\mbox{,}\\
\monadic{v}_d=\monadic{v}_d^2\mbox{,}\quad\gamma'_d=\tail\gamma_d^2\mbox{,}\quad\omicron=\omicron_1\omicron_2\mbox{.}
\end{array}
\]
By the induction hypothesis about~$e_1$, $\sem{e_1}\gamma\phi(\omicron_1\omicron_2\rho)=\return(\monadic{v}^1,\gamma^1,\omicron_2\rho)$ where $\monadic{v}^1$ is \mbox{$q_1$-exact} and $\gamma^1$ is $\Gamma$-exact in circuit, $\monadic{v}_d^1\coincirc{q_1}\monadic{v}^1$ and $\gamma_d^1\coincirc{\Gamma}\gamma^1$. Denote $\Gamma'=((x:q_1),\Gamma)$; then $\left((x,\monadic{v}^1),\gamma^1\right)$ is $\Gamma'$-exact in circuit and $\left((x,\monadic{v}_d^1),\gamma_d^1\right)\coincirc{\Gamma'}\left((x,\monadic{v}^1),\gamma^1\right)$ by the above. Hence by the induction hypothesis about~$e_2$, $\sem{e_2}((x,\monadic{v}^1),\gamma^1)\phi(\omicron_2\rho)=\return(\monadic{v}^2,\gamma^2,\rho)$ where $\monadic{v}^2$ is $q$-exact and $\gamma^2$ is $\Gamma'$-exact in circuit, $\monadic{v}_d^2\coincirc{q}\monadic{v}^2$ and $\gamma_d^2\coincirc{\Gamma'}\gamma^2$. Thus $\sem{e}\gamma\phi(\omicron\rho)=\return(\monadic{v},\gamma',\rho)$ where $\monadic{v}=\monadic{v}^2$ and $\gamma'=\tail\gamma^2$. Here $\monadic{v}$ is $q$-exact in circuit and $\monadic{v}_d\coincirc{q}\monadic{v}$ by the above. Obviously $\tail\gamma^2$ is $\Gamma$-exact in circuit and $\tail\gamma_d^2\coincirc{\Gamma}\tail\gamma^2$, hence the other desired claims also follow.

\item Let $e=\stmtcomp{e_1}{e_2}$. Then
\[
\begin{array}{l}
\Gamma\vd e_1:q_1\eff D_1\mbox{,}\\
\Gamma\vd e_2:q\eff D_2\mbox{,}
\end{array}
\]
and
\[
\begin{array}{l}
\semd{e_1}\gamma_d\phi=\return(\monadic{v}_d^1,\gamma_d^1,\omicron_1)\mbox{,}\\
\semd{e_2}\gamma_d^1\phi=\return(\monadic{v}_d^2,\gamma_d^2,\omicron_2)\mbox{,}\\
\monadic{v}_d=\monadic{v}_d^2\mbox{,}\quad\gamma'_d=\gamma_d^2\mbox{,}\quad\omicron=\omicron_1\omicron_2\mbox{.}
\end{array}
\]
By the induction hypothesis about~$e_1$, $\sem{e_1}\gamma\phi(\omicron_1\omicron_2\rho)=\return(\monadic{v}^1,\gamma^1,\omicron_2\rho)$ where $\monadic{v}^1$ is \mbox{$q_1$-exact} and $\gamma^1$ is $\Gamma$-exact in circuit, $\monadic{v}_d^1\coincirc{q_1}\monadic{v}^1$ and $\gamma_d^1\coincirc{\Gamma}\gamma^1$. By the induction hypothesis about~$e_2$, $\sem{e_2}\gamma^1\phi(\omicron_2\rho)=\return(\monadic{v}^2,\gamma^2,\rho)$ where $\monadic{v}^2$ is $q$-exact and $\gamma^2$ is $\Gamma$-exact in circuit, $\monadic{v}_d^2\coincirc{q}\monadic{v}^2$ and $\gamma_d^2\coincirc{\Gamma}\gamma^2$. Thus $\sem{e}\gamma\phi(\omicron\rho)=\return(\monadic{v},\gamma',\rho)$ where $\monadic{v}=\monadic{v}^2$ and $\gamma'=\gamma^2$. The desired claims follow directly.
\end{itemize}
\end{proof}

\section{Proofs of Theorems~\ref{thm:can_compile}--\ref{thm:correct_compilation}}\label{compilation-proofs}

\begin{lemma}\label{compilation:alltoplemma}
Let $u\in U$ and let $q$ be a qualified type such that $\allpre_d(q)$ for $d=\publicdom$. If $\allpure(u)$ is $q$-exact in $\publicdom$ then $\allpuretop(u)$ is $q$-exact.
\end{lemma}

\begin{proof}
Let $q=(\qualty{t_0}{s_0}{d_0})$. By assumptions, $s_0=\prestg$ and $d_0=\publicdom$. By definition, $\allpure(u)=\return v$ for some~$v$. We proceed by induction on the structure of~$t_0$.
\begin{itemize}
\item Let $t_0$ be a primitive type. Then, by exactness, $v\in t_0$. This is possible only if $u=v$. Hence $\allpuretop(u)=(\return v,\top)$. This pair is $q$-exact by Definition~\ref{compilation:exact}.
\item Let $t_0=\listty{q'}$ for some qualified type~$q'$. By $\allpre_d(q)$, we must have $\allpre_d(q')$. By exactness, $v=(\monadic{v}_1,\ldots,\monadic{v}_n)$ where $n\in\NN$ and $\monadic{v}_i$ is $q'$-exact in $\publicdom$ for every $i=1,\ldots,n$. By definition of $\allpure$, all $\monadic{v}_i$ are of the form $\allpure(u_i)$ for some $u_i\in U$ such that $u=(u_1,\ldots,u_n)$. By the induction hypothesis, $\allpuretop(u_i)$ are $q'$-exact. Thus $\allpuretop(u)=(\return(\allpuretop(u_1),\ldots,\allpuretop(u_n)),\top)$ is $q'$-exact.
\end{itemize}
\end{proof}

\begin{lemma}\label{compilation:getlemma}
If $v\in U$ and $\pi\in(\NN^*)^2$ then $\allpuretop(v)\bullet\pi=\allpure(v)$.
\end{lemma}

\begin{proof}
We proceed by induction on the structure of~$v$. If $v$ is of a primitive type then
\[
\allpuretop(v)\bullet\pi
=(\return v,\top)\bullet\pi
=\return v
=\allpure(v)\mbox{.}
\]
If $v=(v_1,\ldots,v_n)$ then
\[
\begin{array}{lcl}
\allpuretop(v)\bullet\pi
&=&(\return(\allpuretop(v_1),\ldots,\allpuretop(v_n)),\top)\bullet\pi\\
&=&\return(\allpuretop(v_1)\bullet\pi,\ldots,\allpuretop(v_n)\bullet\pi)\mbox{.}
\end{array}
\]
On the other hand,
\[
\allpure(v)
=\return(\allpure(v_1),\ldots,\allpure(v_n))\mbox{.}
\]
The desired claim follows by the induction hypothesis.
\end{proof}

\begin{lemma}\label{compilation:updexactlemma}
Let $d_1,\ldots,d_n$ be domains and $q$ be a qualified type. Moreover, let $q_k=(\qualty{\uintty}{\prestg}{d_k})$ for every $k=1,\ldots,n$ and  $q'=(\qualty{\listty{\ldots\listty{\qualty{\listty{\qualty{\listty{q}}{\prestg}{d_n}}}{\prestg}{d_{n-1}}}\ldots}}{\prestg}{d_1})$. Assume $(\monadic{v},\monadic{c})$ being $q'$-exact, let $(\monadic{i}_k,\top)$ be $q_k$-exact for every $k=1,\ldots,n$, and let $(\monadic{v}',\monadic{c}')$ be $q$-exact. Then $\updatec((\monadic{v},\monadic{c}),\monadic{i}_1\ldots\monadic{i}_n,(\monadic{v}',\monadic{c}'))$ (assuming that it is well defined) is $q'$-exact.
\end{lemma}

\begin{proof}
We proceed by induction on~$n$. If $n=0$ then $q'=q$ and $\updatec((\monadic{v},\monadic{c}),\monadic{i}_1\ldots\monadic{i}_n,(\monadic{v}',\monadic{c}'))=(\monadic{v}',\monadic{c}')$ which is $q$-exact by assumption. Let now $n>0$ and assume that the claim holds for $n-1$. Denote $q'_1=(\qualty{\listty{\ldots\listty{\qualty{\listty{\qualty{\listty{q}}{\prestg}{d_n}}}{\prestg}{d_{n-1}}}\ldots}}{\prestg}{d_2})$. In order to show that $\updatec((\monadic{v},\monadic{c}),\monadic{i}_1\ldots\monadic{i}_n,(\monadic{v}',\monadic{c}'))$ is $q'$-exact, we have to prove the implications in Definition~\ref{compilation:exact}:
\begin{itemize}
\item If $d_1=\publicdom$ then, by exactness, $\monadic{v}=\return((\monadic{v}_1,\monadic{c}_1),\ldots,(\monadic{v}_{l_1},\monadic{c}_{l_1}))$ where all pairs $(\monadic{v}_k,\monadic{c}_k)$ ($k=1,\ldots,l_1$) are $q'_1$-exact and $\monadic{i}_1=\return i_1$ with $i_1\in\NN$. Then $\updatec((\monadic{v},\monadic{c}),\monadic{i}_1\ldots\monadic{i}_n,(\monadic{v}',\monadic{c}'))=(\monadic{v}'',\top)$ where $\monadic{v}''=\mcomp{a\gets\monadic{v}\hstop i_1\gets\monadic{i}_1\hstop\return([i_1\mapsto\updatec(a_{i_1},\monadic{i}_2\ldots\monadic{i}_n,(\monadic{v}',\monadic{c}'))]a)}$; as the latter is well-defined, $i_1\leq l_1$. As $\updatec((\monadic{v}_{i_1},\monadic{c}_{i_1}),\monadic{i}_2\ldots\monadic{i}_n,(\monadic{v}',\monadic{c}'))$ is $q'_1$-exact by the induction hypothesis, all components of $v''$ are $q'_1$-exact as required.
\item As the qualified type $q'$ has stage $\prestg$, $\top$ as the circuit part of the compound value corresponds to the definition.
\item Finally, suppose that $d_1\ne\publicdom$. By $q'$-exactness of $(\monadic{v},\monadic{c})$, we have $\monadic{v}=\top$. Hence $\updatec((\monadic{v},\monadic{c}),\monadic{i}_1\ldots\monadic{i}_n,(\monadic{v}',\monadic{c}'))=(\top,\top)$, matching the definition.
\end{itemize}
This establishes the required claim.
\end{proof}

\begin{lemma}\label{compilation:updlemma}
Let $d_1,\ldots,d_n$ be domains and $q$ be a qualified type. Denote $q_k=(\qualty{\uintty}{\prestg}{d_k})$ for every $k=1,\ldots,n$ and  $q'=(\qualty{\listty{\ldots\listty{\qualty{\listty{\qualty{\listty{q}}{\prestg}{d_n}}}{\prestg}{d_{n-1}}}\ldots}}{\prestg}{d_1})$. Let $\pi\in\left(\NN^*\right)^2$ be a pair of circuit input sequences (for $\proverdom$ and $\verifierdom$). Let $(\monadic{a},\monadic{c}_0)$ be $q'$-exact and $\monadic{a}'$ be $q'$-exact in circuit; for all $k=1,\ldots,n$, let $(\monadic{i}_k,\top)$ be $q_k$-exact and $\monadic{i}'_k$ be $q_k$-exact in circuit; let $(\monadic{v},\monadic{c}_1)$ be $q$-exact and $\monadic{v}'$ be \mbox{$q$-exact} in circuit. Moreover, assume that $(\monadic{a},\monadic{c}_0)\bullet\pi\sim\monadic{a}'$, for every $k=1,\ldots,n$ we have $(\monadic{i},\top)\bullet\pi\sim\monadic{i}'_k$, and $(\monadic{v},\monadic{c}_1)\bullet\pi\sim\monadic{v}'$. If $\updatec((\monadic{a},\monadic{c}_0),\monadic{i}_1\ldots\monadic{i}_n,(\monadic{v},\monadic{c}_1))$ is well-defined then $\update(\monadic{a}',\monadic{i}'_1\ldots\monadic{i}'_n,\monadic{v}')$ is well-defined, too, whereby $\updatec((\monadic{a},\monadic{c}_0),\monadic{i}_1\ldots\monadic{i}_n,(\monadic{v},\monadic{c}_1))\bullet\pi\sim\update(\monadic{a}',\monadic{i}'_1\ldots\monadic{i}'_n,\monadic{v}')$.
\end{lemma}

\begin{proof}
We proceed by induction on~$n$. If $n=0$ then $q'=q$, and $\updatec((\monadic{a},\monadic{c}_0),\monadic{i}_1\ldots\monadic{i}_n,(\monadic{v},\monadic{c}_1))=(\monadic{v},\monadic{c}_1)$, $\update(\monadic{a}',\monadic{i}'_1\ldots\monadic{i}'_n,\monadic{v})=\monadic{v}'$. By assumption, $(\monadic{v},\monadic{c}_1)\bullet\pi\sim\monadic{v}'$. Let now $n>0$ and assume that the claim holds for $n-1$. Denoting $q'_1=(\qualty{\listty{\ldots\listty{\qualty{\listty{\qualty{\listty{q}}{\prestg}{d_n}}}{\prestg}{d_{n-1}}}\ldots}}{\prestg}{d_2})$, consider two cases:
\begin{itemize}
\item Assume that $d_1=\publicdom$. By the exactness assumptions, $\monadic{a}=\return((\monadic{a}_1,\monadic{c}'_1),\ldots,(\monadic{a}_{l_1},\monadic{c}'_{l_1}))$ and $\monadic{a}'=\return(\monadic{a}'_1,\ldots,\monadic{a}'_{l_1})$ with $(\monadic{a}_k,\monadic{c}'_k)$ being $q'_1$-exact and $a'_k$ being $q'_1$-exact in circuit, whereby $(\monadic{a},\monadic{c}_0)\bullet\pi\sim\monadic{a}'$ implies $(\monadic{a}_k,\monadic{c}'_k)\bullet\pi\sim\monadic{a}'_k$ for every $k=1,\ldots,n$. Similarly, we must have $\monadic{i}_1=\monadic{i}'_1=\return i_1$ with $i_1\in\NN$. As $\updatec((\monadic{a},\monadic{c}_0),\monadic{i}_1\ldots\monadic{i}_n,(\monadic{v},\monadic{c}_1))$ is well-defined, so is $\updatec((\monadic{a}_{i_1},\monadic{c}'_{i_1}),\monadic{i}_2\ldots\monadic{i}_n,(\monadic{v},\monadic{c}_1))$ and $i_1\leq l_1$. By the induction hypothesis, $\update(\monadic{a}_{i_1},\monadic{i}'_2\ldots\monadic{i}'_n,\monadic{v}')$ is well-defined and $\updatec((\monadic{a}_{i_1},\monadic{c}'_{i_1}),\monadic{i}_2\ldots\monadic{i}_n,(\monadic{v},\monadic{c}_1))\bullet\pi\sim\update(\monadic{a}_{i_1},\monadic{i}'_2\ldots\monadic{i}'_n,\monadic{v}')$. All this implies that $\update(\monadic{a}',\monadic{i}'_1\ldots\monadic{i}'_n,\monadic{v}')$ is also well-defined. 
Using the definitions of $\updatec$ and $\update$, we obtain $\updatec((\monadic{a},\monadic{c}_0),\monadic{i}_1\ldots\monadic{i}_n,(\monadic{v},\monadic{c}_1))\bullet\pi\sim\update(\monadic{a}',\monadic{i}'_1\ldots\monadic{i}'_n,\monadic{v}')$ as desired.
\item Now assume that $d_1\ne\publicdom$. Then $(\monadic{a},\monadic{c}_0)=(\top,\top)$ and $\monadic{a}'=\top$ by exactness and exactness in circuit. Hence also $\updatec((\monadic{a},\monadic{c}_0),\monadic{i}_1\ldots\monadic{i}_n,(\monadic{v},\monadic{c}_1))=(\top,\top)$ and $\update(\monadic{a}',\monadic{i}'_1\ldots\monadic{i}'_n,\monadic{v}')=\top$ which is well-defined. The desired claim follows since $(\top,\top)\bullet\pi=\top$.
\end{itemize}
\end{proof}

\begin{theorem}[Theorem~\ref{thm:can_compile}]
Let $\Gamma \vd e : q\eff D$ with well-structured~$\Gamma$. Let $\gamma\in\Envc$ be $\Gamma$-exact. Let $\phi\in\In^3$ such that, for all subexpressions of~$e$ of the form $\ofty{\getexpr{d}{k}}{q'}$ where $d=\publicdom$, the value $\allpure(\phi_{d}(k))$ is $q'$-exact in~$d$. Let $\nu\in\NN^2$. Unless an array lookup fails due to an index being out of bounds, we have $\semc{e}\gamma\phi\nu=\return((\monadic{v},\monadic{c}),\gamma',\omicron,\nu')$ for some $\monadic{v}$, $\monadic{c}$, $\gamma'$, $\omicron$ and $\nu'$, whereby $(\monadic{v},\monadic{c})$ is $q$-exact and $\gamma'$ is $\Gamma$-exact.
\end{theorem}

\begin{proof}
Let $q=(\qualty{t_0}{s_0}{d_0})$. We proceed by induction on the structure of~$e$. Assume that no array lookup fails due to an index being out of bounds.

\begin{itemize}
\item If $e=\epsilon$ then $t_0=\unitty$, $s_0=\prestg$, $d_0=\publicdom$. We obtain $\semc{e}\gamma\phi\nu=\return((\monadic{v},\monadic{c}),\gamma',\omicron,\nu')$ where
\[
(\monadic{v},\monadic{c})=\return\singleton\mbox{,}\quad\gamma'=\gamma\mbox{,}\quad\omicron=\epsilon\mbox{,}\quad\nu'=\nu\mbox{.}
\]
The pair $(\monadic{v},\monadic{c})$ is $q$-exact since $\singleton\in\unitty$. Finally, $\gamma'$ is $\Gamma$-exact by assumption about~$\gamma$.

\item If $e=\overline{n}$ where $n\in\NN$ then $t_0=\uintmodty{\mbox{\lstinline+N+}}$. We obtain $\semc{e}\gamma\phi\nu=\return((\monadic{v},\monadic{c}),\gamma',\omicron,\nu')$ where
\[
\begin{array}{l}
(\monadic{v},\monadic{c})=(\branching{\return n&\mbox{if $d_0=\publicdom$}\\\top&\mbox{otherwise}},\branching{\return(\mknode{\mkconst(n)})&\mbox{if $s_0=\poststg$}\\\top&\mbox{otherwise}})\mbox{,}\\\gamma'=\gamma\mbox{,}\quad\omicron=\epsilon\mbox{,}\quad\nu'=\nu\mbox{.}
\end{array}
\]
The pair $(\monadic{v},\monadic{c})$ is $q$-exact since $|\mknode{\mkconst(n)}|=\return n=\return(\mknode{\mkconst(n)}(\pi))$. Finally, $\gamma'$ is $\Gamma$-exact by assumption about~$\gamma$.

\item If $e=\overline{b}$ where $b\in\BB$ then $t_0=\boolmodty{\mbox{\lstinline+N+}}$. We obtain $\semc{e}\gamma\phi\nu=\return((\monadic{v},\monadic{c}),\gamma',\omicron,\nu')$ where
\[
\begin{array}{l}
(\monadic{v},\monadic{c})=(\branching{\return b&\mbox{if $d_0=\publicdom$}\\\top&\mbox{otherwise}},\branching{\return(\mknode{\mkconst(|b|)})&\mbox{if $s_0=\poststg$}\\\top&\mbox{otherwise}})\mbox{,}\\\gamma'=\gamma\mbox{,}\quad\omicron=\epsilon\mbox{,}\quad\nu'=\nu\mbox{.}
\end{array}
\]
The pair $(\monadic{v},\monadic{c})$ is $q$-exact since $|\mknode{\mkconst(|b|)}|=\return|b|=\return(\mknode{\mkconst(|b|)}(\pi))$. Finally, $\gamma'$ is $\Gamma$-exact by assumption about~$\gamma$.

\item If $e=x$ then $\lookup{\Gamma}{x}=q$. We obtain $\semc{e}\gamma\phi\nu=\return((\monadic{v},\monadic{c}),\gamma',\omicron,\nu')$ where
\[
(\monadic{v},\monadic{c})=\lookup{\gamma}{x}\mbox{,}\quad\gamma'=\gamma\mbox{,}\quad\omicron=\epsilon\mbox{,}\quad\nu'=\nu\mbox{.}
\]
By assumption, $\gamma$ is $\Gamma$-exact, whence $\lookup{\gamma}{x}$ is $q$-exact and $\gamma'$ is $\Gamma$-exact as needed.

\item If $e=\addexpr{e_1}{e_2}$ then
\[
\begin{array}{l}
\Gamma\vd e_1:q\mbox{,}\\
\Gamma\vd e_2:q\mbox{.}
\end{array}
\]
By the induction hypothesis about $e_1$, we have $\semc{e_1}\gamma\phi\nu=\return((\monadic{v}_1,\monadic{c}_1),\gamma_1,\omicron_1,\nu_1)$ where $(\monadic{v}_1,\monadic{c}_1)$ is $q$-exact and $\gamma_1$ is $\Gamma$-exact. Hence by the induction hypothesis about $e_2$, we have $\semc{e_2}\gamma_1\phi\nu_1=\return((\monadic{v}_2,\monadic{c}_2),\gamma_2,\omicron_2,\nu_2)$ where $(\monadic{v}_2,\monadic{c}_2)$ is $q$-exact and $\gamma_2$ is $\Gamma$-exact. We obtain $\semc{e}\gamma\phi\nu=\return((\monadic{v},\monadic{c}),\gamma',\omicron,\nu')$ where
\[
\begin{array}{l}
\monadic{v}=\mcomp{i_1\gets\monadic{v}_1\hstop i_2\gets\monadic{v}_2\hstop\return(i_1+i_2)}\mbox{,}\\
\monadic{c}=\mcomp{c_1\gets\monadic{c}_1\hstop c_2\gets\monadic{c}_2\hstop\return(\mknode{\mkoper(+)}(c_1,c_2))}\mbox{,}\\
\gamma'=\gamma_2\mbox{,}\quad\omicron=\omicron_1\omicron_2\mbox{,}\quad\nu'=\nu_2\mbox{.}
\end{array}
\]
The desired claim about~$\gamma'$ holds because of $\gamma_2$ being $\Gamma$-exact. For proving that $(\monadic{v},\monadic{c})$ is $q$-exact, we have to establish all implications in Definition~\ref{compilation:exact}.
\begin{itemize}
\item If $s_0=\poststg$ then $\monadic{c}_1=\return c_1$ and $\monadic{c}_2=\return c_2$ by exactness. Hence we obtain $\monadic{c}=\return(\mknode{\mkoper(+)}(c_1,c_2))$. Suppose that $\monadic{v}=\return n$, $n\in\NN$. Then $\monadic{v}_1=\return i_1$, $\monadic{v}_2=\return i_2$, $n=i_1+i_2$. Moreover, $|c_1|=\return i_1=\return(c_1(\pi))$ and $|c_2|=\return i_2=\return(c_2(\pi))$. This implies $|c|=\return(i_1+i_2)=\return n$ and $|c|=\return(c_1(\pi)+c_2(\pi))=\return(\mknode{\mkoper(+)}(c_1,c_2)(\pi))$ as desired.
\item If $d_0=\publicdom$ then $\monadic{v}_1=\return i_1$ and $\monadic{v}_2=\return i_2$ where $i_1,i_2\in\NN$. Hence $\monadic{v}=\return(i_1+i_2)$, $i_1+i_2\in\NN$ as desired.
\item If $s_0=\prestg$ then $\monadic{c}_1=\monadic{c}_2=\top$. Hence $\monadic{c}=\top$.
\item If $d_0\ne\publicdom$ then $\monadic{v}_1=\monadic{v}_2=\top$. Hence $\monadic{v}=\top$.
\end{itemize}

\item If $e=\assertexpr{e_1}$ then
\[
\begin{array}{l}
\Gamma\vd e_1:\qualty{\boolmodty{\mbox{\lstinline+N+}}}{\poststg}{d_1}\mbox{,}\\
t_0=\unitty\mbox{,}\quad s_0=\prestg\mbox{,}\quad d_0=\publicdom\mbox{.}
\end{array}
\]
Denote $q'=(\qualty{\boolmodty{\mbox{\lstinline+N+}}}{\poststg}{d_1})$. By the induction hypothesis about $e_1$, we have $\semc{e_1}\gamma\phi\nu=\return((\monadic{v}_1,\monadic{c}_1),\gamma_1,\omicron_1,\nu_1)$ where $(\monadic{v}_1,\monadic{c}_1)$ is $q'$-exact and $\gamma_1$ is $\Gamma$-exact. The former guarantees $\monadic{c}_1=\return c_1$ for some $c_1\in T$. Therefore we have $\semc{e}\gamma\phi\nu=\return((\monadic{v},\monadic{c}),\gamma',\omicron,\nu')$ where
\[
\begin{array}{l}
(\monadic{v},\monadic{c})=(\return\singleton,\top)\mbox{,}\quad\gamma'=\gamma_1\mbox{,}\quad\omicron=\omicron_1c_1\mbox{,}\quad\nu'=\nu_1\mbox{.}
\end{array}
\]
The desired claim about~$\gamma'$ holds because of $\gamma_1$ being $\Gamma$-exact. The pair $(\monadic{v},\monadic{c})$ is $q$-exact since $\singleton\in\unitty$.

\item If $e=\getexpr{d}{k}$ then $s_0=\prestg$ and $d_0=d$. We obtain $\semc{e}\gamma\phi\nu=\return((\monadic{v},\monadic{c}),\gamma',\omicron,\nu')$ where
\[
(\monadic{v},\monadic{c})=\branching{\allpuretop(\phi_d(k))&\mbox{if $d=\publicdom$}\\(\top,\top)&\mbox{otherwise}}\mbox{,}\quad\gamma'=\gamma\mbox{,}\quad\omicron=\epsilon\mbox{,}\quad\nu'=\nu\mbox{.}
\]
If $d=\publicdom$ then, by assumption, $\allpure(\phi_d(k))$ is $q$-exact in $\publicdom$, which implies $\allpuretop(\phi_d(k))$ being $q$-exact by Lemma~\ref{compilation:alltoplemma}. Hence if $d=\publicdom$ then $(\monadic{v},\monadic{c})$ is $q$-exact. Otherwise, $(\top,\top)$ is $q$-exact by Definition~\ref{compilation:exact}. Finally, $\gamma'$ is $\Gamma$-exact by the assumption about~$\gamma$.

\item If $e=\ifexpr{e_1}{e_2}{e_3}$ then
\[
\begin{array}{l}
\Gamma\vd e_1:\qualty{\boolmodty{\mbox{\lstinline+N+}}}{\prestg}{d_1}\mbox{,}\\
\Gamma\vd e_2:q\mbox{,}\\
\Gamma\vd e_3:q\mbox{,}\\
\gen{d_1}\supseteq\gen{s_0}\cup\gen{d_0}\mbox{.}
\end{array}
\]
Denote $q'=\qualty{\boolmodty{\mbox{\lstinline+N+}}}{\prestg}{d_1}$. By the induction hypothesis about~$e_1$, we have $\semc{e_1}\gamma\phi\nu=\return((\monadic{v}_1,\monadic{c}_1),\gamma_1,\omicron_1,\nu_1)$ where $(\monadic{v}_1,\monadic{c}_1)$ is $q'$-exact and $\gamma_1$ is $\Gamma$-exact. Consider three cases:
\begin{itemize}
\item Suppose $\monadic{v}_1=\return\tru$. By the induction hypothesis about~$e_2$, we have $\semc{e_2}\gamma_1\phi\nu_1=\return((\monadic{v}_2,\monadic{c}_2),\gamma_2,\omicron_2,\nu_2)$ where $(\monadic{v}_2,\monadic{c}_2)$ is $q$-exact and $\gamma_2$ is $\Gamma$-exact. We obtain $\semc{e}\gamma\phi\nu=\return((\monadic{v},\monadic{c}),\gamma',\omicron,\nu')$ where 
\[
\begin{array}{l}
(\monadic{v},\monadic{c})=(\monadic{v}_2,\monadic{c}_2)\mbox{,}\quad\gamma'=\gamma_2\mbox{,}\quad\omicron=\omicron_1\omicron_2\mbox{,}\quad\nu'=\nu_2\mbox{.}
\end{array}
\]
The desired claim follows.
\item Suppose $\monadic{v}_1=\return\fls$. By the induction hypothesis about~$e_3$, we have $\semc{e_3}\gamma_1\phi\nu_1=\return((\monadic{v}_3,\monadic{c}_3),\gamma_3,\omicron_3,\nu_3)$ where $(\monadic{v}_3,\monadic{c}_3)$ is $q$-exact and $\gamma_3$ is $\Gamma$-exact. We obtain $\semc{e}\gamma\phi\nu=\return((\monadic{v},\monadic{c}),\gamma',\omicron,\nu')$ where 
\[
\begin{array}{l}
(\monadic{v},\monadic{c})=(\monadic{v}_3,\monadic{c}_3)\mbox{,}\quad\gamma'=\gamma_3\mbox{,}\quad\omicron=\omicron_1\omicron_3\mbox{,}\quad\nu'=\nu_3\mbox{.}
\end{array}
\]
The desired claim follows.
\item If $\monadic{v}_1\ne\return b$ for $b\in\BB$ then we obtain $\semc{e}\gamma\phi\nu=\return((\monadic{v},\monadic{c}),\gamma',\omicron,\nu')$ where
\[
(\monadic{v},\monadic{c})=(\top,\top)\mbox{,}\quad\gamma'=\gamma_1\mbox{,}\quad\omicron=\omicron_1\mbox{,}\quad\nu'=\nu_1\mbox{.}
\]
By $\gamma_1$ being $\Gamma$-exact, we obtain the desired result about~$\gamma'$. Moreover, $q'$-exactness of $(\monadic{v}_1,\monadic{c}_1)$ implies $d_1\ne\publicdom$. Now $\gen{d_1}\supseteq\gen{s_0}\cup\gen{d_0}$ implies $d_0\ne\publicdom$ and $s_0=\prestg$. Consequently, $(\top,\top)$ is $q$-exact as desired.
\end{itemize}

\item If $e=\forexpr{x}{e_1}{e_2}{e_3}$ then
\[
\begin{array}{l}
\Gamma\vd e_1:\qualty{\uintty}{\prestg}{d_0}\mbox{,}\\
\Gamma\vd e_2:\qualty{\uintty}{\prestg}{d_0}\mbox{,}\\
(x:\qualty{\uintty}{\prestg}{d_0}),\Gamma\vd e_3:q_1\mbox{,}\\
t_0=\listty{q_1}\mbox{,}\quad s_0=\prestg\mbox{.}
\end{array}
\]
Denote $q'=(\qualty{\uintty}{\prestg}{d_0})$. By the induction hypothesis about $e_1$, we have $\semc{e_1}\gamma\phi\nu=\return((\monadic{v}_1,\monadic{c}_1),\gamma_1,\omicron_1,\nu_1)$ where $(\monadic{v}_1,\monadic{c}_1)$ is $q'$-exact and $\gamma_1$ is $\Gamma$-exact. Hence by the induction hypothesis about $e_2$, we have $\semc{e_2}\gamma_1\phi\nu_1=\return((\monadic{v}_2,\monadic{c}_2),\gamma_2,\omicron_2,\nu_2)$ where $(\monadic{v}_2,\monadic{c}_2)$ is $q'$-exact and $\gamma_2$ is $\Gamma$-exact. Note that $q'$-exactness implies $\monadic{c}_1=\monadic{c}_2=\top$. Consider two cases:
\begin{itemize}
\item Suppose $\monadic{v}_1=\return i_1$ and $\monadic{v}_2=\return i_2$. By $q'$-exactness, $d_0=\publicdom$. Denote $n=\max(0,i_2-i_1)$. All pairs of the form $(\return(i_1+k),\top)$ are $q'$-exact as $(\return i_1,\top)$ is \mbox{$q'$-exact} by the above. Hence $\left((x,(\return i_1,\top)),\gamma_2\right)$ is $\Gamma'$-exact where $\Gamma'=((x:q'),\Gamma)$. By the induction hypothesis about~$e_3$, we have $\semc{e_3}((x,(\return i_1,\top)),\gamma_2)\phi\nu_2=\return((\monadic{v}_3,\monadic{c}_3),\gamma_3,\omicron_3,\nu_3)$ where $(\monadic{v}_3,\monadic{c}_3)$ is $q_1$-exact and $\gamma_3$ is $\Gamma'$-exact. As the updated environment $[x\mapsto(\return(i_1+1),\top)]\gamma_3$ is also $\Gamma'$-exact, the induction hypothesis also implies $\semc{e_3}([x\mapsto(\return(i_1+1),\top)]\gamma_3)\phi\nu_3=\return((\monadic{v}_4,\monadic{c}_4),\gamma_4,\omicron_4,\nu_4)$ where $(\monadic{v}_4,\monadic{c}_4)$ being $q_1$-exact and $\gamma_4$ being $\Gamma'$-exact. Analogously, for all $k=5,\ldots,n+2$, we obtain $\semc{e_3}([x\mapsto(\return(i_1+k-3),\top)]\gamma_{k-1})\phi\nu_{k-1}=\return((\monadic{v}_k,\monadic{c}_k),\gamma_k,\omicron_k,\nu_k)$ where $(\monadic{v}_k,\monadic{c}_k)$ is $q_1$-exact and $\gamma_k$ is $\Gamma'$-exact. We obtain $\semc{e}\gamma\phi\nu=\return((\monadic{v},\monadic{c}),\gamma',\omicron,\nu')$ where
\[
\begin{array}{l}
(\monadic{v},\monadic{c})=(\return((\monadic{v}_3,\monadic{c}_3),\ldots,(\monadic{v}_{n+2},\monadic{c}_{n+2})),\top)\mbox{,}\\\gamma'=\tail\gamma_{n+2}\mbox{,}\quad\omicron=\omicron_1\ldots\omicron_{n+2}\mbox{,}\quad\nu'=\nu_{n+2}\mbox{.}
\end{array}
\]
Hence $(\monadic{v},\monadic{c})$ is $q$-exact since $d_0=\publicdom$ and $s_0=\prestg$. As $\gamma_{n+2}$ is $\Gamma'$-exact, $\tail\gamma_{n+2}$ is $\Gamma$-exact, implying also the desired claim about~$\gamma'$.
\item If $\monadic{v}_1\ne\return n$ or $\monadic{v}_2\ne\return n$ for $n\in\NN$ then we obtain $\semc{e}\gamma\phi\nu=\return((\monadic{v},\monadic{c}),\gamma',\omicron,\nu')$ where
\[
(\monadic{v},\monadic{c})=(\top,\top)\mbox{,}\quad\gamma'=\gamma_2\mbox{,}\quad\omicron=\omicron_1\omicron_2\mbox{,}\quad\nu'=\nu_2\mbox{.}
\]
By $\gamma_2$ being $\Gamma$-exact, we obtain the desired result about~$\gamma'$. Moreover, $q'$-exactness of $(\monadic{v}_i,\monadic{c}_i)$ for $i=1,2$ implies $d_0\ne\publicdom$. As $s_0=\prestg$, this implies $(\top,\top)$ being $q$-exact.
\end{itemize}

\item If $e=\wireexpr{e_1}$ then
\[
\begin{array}{l}
\Gamma\vd e_1:\qualty{t_0}{\poststg}{d_0}\mbox{,}\\
t_0\in\left\lbrace\uintmodty{\mbox{\lstinline{N}}},\boolmodty{\mbox{\lstinline{N}}}\right\rbrace\mbox{.}
\end{array}
\]
Denote $q'=\qualty{t_0}{\poststg}{d_0}$. By the induction hypothesis about $e_1$, we have $\semc{e_1}\gamma\phi\nu=\return((\monadic{v}_1,\monadic{c}_1),\gamma_1,\omicron_1,\nu_1)$ where $(\monadic{v}_1,\monadic{c}_1)$ is $q'$-exact and $\gamma_1$ is $\Gamma$-exact. We obtain $\semc{e}\gamma\phi\nu=\return((\monadic{v},\monadic{c}),\gamma',\omicron,\nu')$ where
\[
\begin{array}{l}
(\monadic{v},\monadic{c})=(\monadic{v}_1,\branching{\return(\mknode{\mkconst(n)})&\mbox{if $\monadic{v}_1=\return n$, $n\in\NN$}\\\return(\mknode{\mkconst(|b|)})&\mbox{if $\monadic{v}_1=\return b$, $b\in\BB$}\\\return(\mknode{\mkinput_d({(\nu_1)}_d)})&\mbox{if $\monadic{v}_1=\top$}})\mbox{,}\\\gamma'=\gamma_1\mbox{,}\quad\omicron=\omicron_1\mbox{,}\quad\nu'=\lam{d'}{\branching{(\nu_1)_{d'}+1&\mbox{if $d'=d$}\\(\nu_1)_{d'}&\mbox{otherwise}}}\mbox{.}
\end{array}
\]
To establish that $(\monadic{v},\monadic{c})$ is $q$-exact, note that in all cases $\monadic{c}=\return c$ for some circuit~$c$. Furthermore, consider two cases:
\begin{itemize}
\item Suppose that $\monadic{v}=\return n$ for some $n\in\NN$. Since $\monadic{v}=\monadic{v}_1$, we must have $c=\mknode{\mkconst(n)}$. Hence $|c|=\return n=\return(c(\pi))$.
\item Suppose that $\monadic{v}=\return b$ for some $b\in\BB$. Since $\monadic{v}=\monadic{v}_1$, we must have $c=\mknode{\mkconst(|b|)}$. Hence $|c|=\return|b|=\return(c(\pi))$.
\end{itemize}
The clauses of Definition~\ref{compilation:exact} about $\monadic{v}$ hold by $\monadic{v}=\monadic{v}_1$ and $(\monadic{v}_1,\monadic{c}_1)$ being $q'$-exact. Finally, $\gamma'$ is $\Gamma$-exact since $\gamma'=\gamma_1$ and $\gamma_1$ is $\Gamma$-exact.

\item If $e=(\castexpr{e_1}{q})$ then
\[
\begin{array}{l}
\Gamma\vd e_1:\qualty{t_0}{s_1}{d_1}\mbox{,}\\
s_1\subtype s_0\mbox{,}\quad d_1\subtype d_0\mbox{.}
\end{array}
\]
Denote $q_1=\qualty{t_0}{s_1}{d_1}$. The induction hypothesis implies $\semc{e_1}\gamma\phi\nu=\return((\monadic{v}_1,\monadic{c}_1),\gamma_1,\omicron_1,\nu_1)$ where $(\monadic{v}_1,\monadic{c}_1)$ is $q_1$-exact and $\gamma_1$ is $\Gamma$-exact. We obtain $\semc{e}\gamma\phi\nu=\return((\monadic{v},\monadic{c}),\gamma',\omicron,\nu')$ where
\[
(\monadic{v},\monadic{c})=(\branching{\monadic{v}_1&\mbox{if $d_0=\publicdom$}\\\top&\mbox{otherwise}},\branching{\monadic{c}_1&\mbox{if $s_0=\poststg$}\\\top&\mbox{otherwise}})\mbox{,}\quad\gamma'=\gamma_1\mbox{,}\quad\omicron=\omicron_1\mbox{,}\quad\nu'=\nu_1\mbox{.}
\]
To establish that $(\monadic{v},\monadic{c})$ is $q$-exact, consider cases:
\begin{itemize}
\item Suppose that $s_0=\poststg$. Then $\monadic{c}=\monadic{c}_1$. By $s_1\subtype s_0$, we also have $s_1=\poststg$. By $(\monadic{v}_1,\monadic{c}_1)$ being $q_1$-exact, $\monadic{c}=\return c$ for some circuit~$c$. If $\monadic{v}=\return v$ for some~$v$ then $d_0=\publicdom$ and $\monadic{v}_1=\monadic{v}$. By $d_1\subtype d_0$, also $d_1=\publicdom$. Consequently, $q_0=q_1$, whence $(\monadic{v},\monadic{c})$ is $q$-exact because $(\monadic{v}_1,\monadic{c}_1)$ is $q_1$-exact.
\item Suppose that $d_0=\publicdom$ and $t_0$ is a primitive type. Then $\monadic{v}=\monadic{v}_1$. As $d_1\subtype d_0$ implies $d_1=\publicdom$, we have $\monadic{v}_1=\return v$ with $v\in t_0$ by $(\monadic{v}_1,\monadic{c}_1)$ being $q_1$-exact.
\item Suppose that $d_0=\publicdom$ and $t_0=\listty{q'}$. Then $\monadic{v}=\monadic{v}_1$. As $d_1\subtype d_0$ implies $d_1=\publicdom$ and $(\monadic{v}_1,\monadic{c}_1)$ is $q_1$-exact, we have $\monadic{v}_1=\return((\monadic{v}'_1,\monadic{c}'_1),\ldots,(\monadic{v}'_l,\monadic{c}'_l))$ where all $(\monadic{v}'_i,\monadic{c}'_i)$ are $q'$-exact.
\item Suppose that $d_0\ne\publicdom$. Then $\monadic{v}=\top$ by definition.
\end{itemize}
Hence $(\monadic{v},\monadic{c})$ is $q$-exact indeed. Finally, $\gamma'$ is $\Gamma$-exact since $\gamma'=\gamma_1$ and $\gamma_1$ is $\Gamma$-exact.

\item If $e=(\assignexpr{e_1}{e_2})$ then $q=\qualty{\unitty}{\prestg}{\publicdom}$. Let $e_1=\loadexpr{\loadexpr{\loadexpr{x}{y_1}}{y_2}\ldots}{y_n}$ where $x$ is a variable. By Lemma~\ref{typesystem:lhslemma}, there exist domains $d_1,\ldots,d_n$ such that
\[
\begin{array}{l}
\Gamma\vd x:q'\mbox{,}\\
\Gamma\vd y_i:q_i\mbox{ for every $i=1,\ldots,n$,}\\
\Gamma\vd e_1:q_{n+1}\mbox{,}\\
\Gamma\vd e_2:q_{n+1}\mbox{,}
\end{array}
\]
where $q_i=\qualty{\uintty}{\prestg}{d_i}$ for each $i=1,\ldots,n$ and
\[
q'=\qualty{\listty{\ldots\listty{\qualty{\listty{\qualty{\listty{q_{n+1}}}{\prestg}{d_n}}}{\prestg}{d_{n-1}}}\ldots}}{\prestg}{d_1}\mbox{.}
\]
By the induction hypothesis, we have
\[
\begin{array}{l}
\semc{y_1}\gamma\phi\nu=\return((\monadic{v}_1,\monadic{c}_1),\gamma_1,\omicron_1,\nu_1)\mbox{,}\\
\semc{y_2}\gamma_1\phi\nu_1=\return((\monadic{v}_2,\monadic{c}_2),\gamma_2,\omicron_2,\nu_2)\mbox{,}\\
\dotfill\\
\semc{y_n}\gamma_{n-1}\phi\nu_{n-1}=\return((\monadic{v}_n,\monadic{c}_n),\gamma_n,\omicron_n,\nu_n)\mbox{,}\\
\semc{e_2}\gamma_n\phi\nu_n=\return((\monadic{v}_{n+1},\monadic{c}_{n+1}),\gamma_{n+1},\omicron_{n+1},\nu_{n+1})\mbox{,}
\end{array}
\]
where $(\monadic{v}_i,\monadic{c}_i)$ is $q_i$-exact and $\gamma_i$ is $\Gamma$-exact for every $i=1,\ldots,n+1$. Hence we obtain $\semc{e}\gamma\phi\nu=\return((\monadic{v},\monadic{c}),\gamma',\omicron,\nu')$ where
\[
\begin{array}{l}
(\monadic{v},\monadic{c})=(\return\singleton,\top)\mbox{,}\\\gamma'=[x\mapsto\updatec(\lookup{\gamma}{x},\monadic{v}_1\ldots\monadic{v}_n,(\monadic{v}_{n+1},\monadic{c}_{n+1}))]\gamma_{n+1}\mbox{,}\\\omicron=\omicron_1\ldots\omicron_{n+1}\mbox{,}\\\nu'=\nu_{n+1}\mbox{.}
\end{array}
\]
Obviously $(\return\singleton,\top)$ is $q$-exact. By Lemma~\ref{compilation:updexactlemma}, $\updatec(\lookup{\gamma}{x},\monadic{v}_1\ldots\monadic{v}_n,(\monadic{v}_{n+1},\monadic{c}_{n+1}))$ is $q'$-exact. Hence also $\gamma'$ is $\Gamma$-exact.

\item If $e=\loadexpr{e_1}{e_2}$ then
\[
\begin{array}{l}
\Gamma\vd e_1:\qualty{\listty{q}}{\prestg}{d_1}\mbox{,}\\
\Gamma\vd e_2:\qualty{\uintty}{\prestg}{d_1}\mbox{.}
\end{array}
\]
Denote $q_1=\qualty{\listty{q}}{\prestg}{d_1}$ and $q'=\qualty{\uintty}{\prestg}{d_1}$. By the induction hypothesis about $e_1$, we have $\semc{e_1}\gamma\phi\nu=\return((\monadic{v}_1,\monadic{c}_1),\gamma_1,\omicron_1,\nu_1)$ where $(\monadic{v}_1,\monadic{c}_1)$ is $q_1$-exact and $\gamma_1$ is $\Gamma$-exact. Hence by the induction hypothesis about~$e_2$, we have $\semc{e_2}\gamma_1\phi\nu_1=\return((\monadic{v}_2,\monadic{c}_2),\gamma_2,\omicron_2,\nu_2)$ where $(\monadic{v}_2,\monadic{c}_2)$ is $q'$-exact and $\gamma_2$ is $\Gamma$-exact. Let $\monadic{r}=\mcomp{a\gets\monadic{v}_1\hstop i\gets\monadic{v}_2\hstop\return a_i}$; then $\semc{e}\gamma\phi\nu=\return((\monadic{v},\monadic{c}),\gamma',\omicron,\nu')$ where
\[
\begin{array}{l}
(\monadic{v},\monadic{c})=\branching{(\monadic{v}',\monadic{c}')&\mbox{if $\monadic{r}=\return(\monadic{v}',\monadic{c}')$}\\(\top,\top)&\mbox{otherwise}}\mbox{,}\\
\gamma'=\gamma_2\mbox{,}\quad\omicron=\omicron_1\omicron_2\mbox{,}\quad\nu'=\nu_2\mbox{,}
\end{array}
\]
Consider two cases:
\begin{itemize}
\item If $d_1=\publicdom$ then, by exactness, $\monadic{v}_1=\return((\monadic{v}'_1,\monadic{c}'_1),\ldots,(\monadic{v}'_l,\monadic{c}'_l))$ where $(\monadic{v}'_k,\monadic{c}'_k)$ is $q$-exact for every $k=1,\ldots,l$, and $\monadic{v}_2=\return i$ for some $i\in\NN$. Therefore $\monadic{r}=\return(\monadic{v}'_i,\monadic{c}'_i)$ and $(\monadic{v},\monadic{c})=(\monadic{v}'_i,\monadic{c}'_i)$. The latter pair is $q$-exact by the above.
\item If $d_1\ne\publicdom$ then $d_0\ne\publicdom$ and $s_0=\prestg$ because of $q_1$ being well-structured by Theorem~\ref{typesystem:thm}. By exactness, $\monadic{v}_1=\top$. Therefore also $\monadic{r}=\top$ and $(\monadic{v},\monadic{c})=(\top,\top)$. The latter pair is $q$-exact because of $d_0\ne\publicdom$ and $s_0=\prestg$.
\end{itemize}
Finally, $\gamma'$ is $\Gamma$-exact since $\gamma'=\gamma_2$ and $\gamma_2$ is $\Gamma$-exact.

\item If $e=\stmtcomp{\letexpr{x}{e_1}}{e_2}$ then
\[
\begin{array}{l}
\Gamma\vd e_1:\qualty{t_1}{s_1}{d_1}\mbox{,}\\
(x:\qualty{t_1}{s_1}{d_1}),\Gamma\vd e_2:\qualty{t_0}{s_0}{d_0}\mbox{.}
\end{array}
\]
Denote $q_1=(\qualty{t_1}{s_1}{d_1})$. By the induction hypothesis, $\semc{e_1}\gamma\phi\nu=\return((\monadic{v}_1,\monadic{c}_1),\gamma_1,\omicron_1,\nu_1)$ where $(\monadic{v}_1,\monadic{c}_1)$ is $q_1$-exact and $\gamma_1$ is $\Gamma$-exact. This implies $\left((x,(\monadic{v}_1,\monadic{c}_1)),\gamma_1\right)$ being $\Gamma'$-exact where $\Gamma'=((x:q_1),\Gamma)$. Hence by the induction hypothesis, $\semc{e_2}((x,(\monadic{v}_1,\monadic{c}_1)),\gamma_1)\phi\nu_1=\return((\monadic{v}_2,\monadic{c}_2),\gamma_2,\omicron_2,\nu_2)$ where $(\monadic{v}_2,\monadic{c}_2)$ is $q$-exact and $\gamma_2$ is $\Gamma'$-exact. Obviously $\tail\gamma_2$ is \mbox{$\Gamma$-exact}. We obtain $\semc{e}\gamma\phi\nu=\return((\monadic{v},\monadic{c}),\gamma',\omicron,\nu')$ where
\[
(\monadic{v},\monadic{c})=(\monadic{v}_2,\monadic{c}_2)\mbox{,}\quad\gamma'=\tail\gamma_2\mbox{,}\quad\omicron=\omicron_1\omicron_2\mbox{,}\quad\nu'=\nu_2\mbox{.}
\]
The desired claim follows.

\item If $e=\stmtcomp{e_1}{e_2}$ then
\[
\begin{array}{l}
\Gamma\vd e_1:\qualty{t_1}{s_1}{d_1}\mbox{,}\\
\Gamma\vd e_2:\qualty{t_0}{s_0}{d_0}\mbox{.}
\end{array}
\]
Denote $q_1=(\qualty{t_1}{s_1}{d_1})$. By the induction hypothesis, $\semc{e_1}\gamma\phi\nu=\return((\monadic{v}_1,\monadic{c}_1),\gamma_1,\omicron_1,\nu_1)$ where $(\monadic{v}_1,\monadic{c}_1)$ is $q_1$-exact and $\gamma_1$ is $\Gamma$-exact. Hence by the induction hypothesis, $\semc{e_2}\gamma_1\phi\nu_1=\return((\monadic{v}_2,\monadic{c}_2),\gamma_2,\omicron_2,\nu_2)$ where $(\monadic{v}_2,\monadic{c}_2)$ is $q$-exact and $\gamma_2$ is $\Gamma$-exact. We obtain $\semc{e}\gamma\phi\nu=\return((\monadic{v},\monadic{c}),\gamma',\omicron,\nu')$ where
\[
(\monadic{v},\monadic{c})=(\monadic{v}_2,\monadic{c}_2)\mbox{,}\quad\gamma'=\gamma_2\mbox{,}\quad\omicron=\omicron_1\omicron_2\mbox{,}\quad\nu'=\nu_2\mbox{.}
\]
The desired claim follows.

\end{itemize}

\end{proof}

\begin{theorem}[Theorem \ref{thm:correct_compilation}]
Let $\Gamma \vd e:\qualty{t}{s}{d}\eff D$ with well-structured~$\Gamma$. Let $\gamma\in\Envc$ be $\Gamma$-exact, $\phi\in\In^3$ be a triple of type correct input dictionaries, and $\nu\in\NN^2$. Let $\semc{e}\gamma\phi\nu=\return((\monadic{v},\monadic{c}),\tilde{\gamma},\tilde{\omicron},\tilde{\nu})$. Let $\pi\in(\NN^*)^2$, $\gamma'\in\Env$, $\omicron'\in\Out^2$ be such that $\gamma'$ is $\Gamma$-exact in circuit, $\gamma\bullet\pi\sim\gamma'$ and $\nu\bullet\pi\sim\omicron'$. Then:
\begin{enumerate}
\item If $\sem{e}\gamma'\phi\omicron'=\return(\monadic{v}',\tilde{\gamma}',\tilde{\omicron}')$ then $\tilde{\omicron}$ accepts $\pi$; moreover, $\tilde{\gamma}\bullet\pi\sim\tilde{\gamma}'$ and $\tilde{\nu}\bullet\pi\sim\tilde{\omicron}'$, and $d=\publicdom$ or $s=\poststg$ implies $(\monadic{v},\monadic{c})\bullet\pi\sim\monadic{v}'$;
\item If $\sem{e}\gamma'\phi\omicron'$ fails (i.e., it is not of the form $\return(\monadic{v}',\tilde{\gamma}',\tilde{\omicron}')$),  while all circuits that arise during compilation need only inputs in~$\pi$, then $\tilde{\omicron}$~{}does not accept~$\pi$.
\end{enumerate}
\end{theorem}

\begin{proof}
We proceed by induction on the structure of~$e$.
\begin{itemize}
\item If $e=\epsilon$ then
\[
t=\unitty\mbox{,}\quad s=\prestg\mbox{,}\quad d=\publicdom\mbox{.}
\]
Furthermore, $\semc{e}\gamma\phi\nu=\return((\monadic{v},\monadic{c}),\tilde{\gamma},\tilde{\omicron},\tilde{\nu})$ where
\[
(\monadic{v},\monadic{c})=(\return\singleton,\top)\mbox{,}\quad\tilde{\gamma}=\gamma\mbox{,}\quad\tilde{\omicron}=\epsilon\mbox{,}\quad\tilde{\nu}=\nu\mbox{,}
\]
and $\sem{e}\gamma'\phi\omicron'=\return(\monadic{v}',\tilde{\gamma}',\tilde{\omicron}')$ where
\[
\monadic{v}'=\return\singleton\mbox{,}\quad\tilde{\gamma}'=\gamma'\mbox{,}\quad\tilde{\omicron}'=\omicron'\mbox{.}
\]
The empty output circuit vacuously accepts~$\pi$. As $\tilde{\gamma}=\gamma$ and $\tilde{\gamma}'=\gamma'$, the claim $\tilde{\gamma}\bullet\pi\sim\tilde{\gamma}'$ follows from assumptions. As $\tilde{\nu}=\nu$ and $\tilde{\omicron}'=\omicron'$, the claim $\tilde{\nu}\bullet\pi\sim\tilde{\omicron}'$ also follows from assumptions. Finally,
\[
(\monadic{v},\monadic{c})\bullet\pi=(\return\singleton,\top)\bullet\pi=\return\singleton=\monadic{v}'\mbox{,}
\]
implying $(\monadic{v},\monadic{c})\bullet\pi\sim\monadic{v}'$.

\item If $e=\overline{n}$ where $n\in\NN$ then
\[
t=\uintmodty{\mbox{\lstinline+N+}}\mbox{.}
\]
Furthermore, $\semc{e}\gamma\phi\nu=\return((\monadic{v},\monadic{c}),\tilde{\gamma},\tilde{\omicron},\tilde{\nu})$ where
\[
\begin{array}{l}
(\monadic{v},\monadic{c})=(\branching{\return n&\mbox{if $d=\publicdom$}\\\top&\mbox{otherwise}},\branching{\return(\mknode{\mkconst(n)})&\mbox{if $s=\poststg$}\\\top&\mbox{otherwise}})\mbox{,}\\\tilde{\gamma}=\gamma\mbox{,}\quad\tilde{\omicron}=\epsilon\mbox{,}\quad\tilde{\nu}=\nu\mbox{,}
\end{array}
\]
and $\sem{e}\gamma'\phi\omicron'=\return(\monadic{v}',\tilde{\gamma}',\tilde{\omicron}')$ where
\[
\monadic{v}'=\branching{\return n&\mbox{if $s=\poststg$ or $d=\publicdom$}\\\top&\mbox{otherwise}}\mbox{,}\quad\tilde{\gamma}'=\gamma'\mbox{,}\quad\tilde{\omicron}'=\omicron'\mbox{.}
\]
The empty output circuit vacuously accepts~$\pi$. As $\tilde{\gamma}=\gamma$ and $\tilde{\gamma}'=\gamma'$, the claim $\tilde{\gamma}\bullet\pi\sim\tilde{\gamma}'$ follows from assumptions. As $\tilde{\nu}=\nu$ and $\tilde{\omicron}'=\omicron'$, the claim $\tilde{\nu}\bullet\pi\sim\tilde{\omicron}'$ also follows from assumptions. Finally, if $d=\publicdom$ then
\[
(\monadic{v},\monadic{c})\bullet\pi=(\return n,\monadic{c})\bullet\pi=\return n=\monadic{v}'\mbox{,}
\]
and if $d\ne\publicdom$ but $s=\poststg$ then
\[
(\monadic{v},\monadic{c})\bullet\pi=(\top,\return(\mknode{\mkconst(n)}))\bullet\pi=\return(\mknode{\mkconst(n)}(\pi))=\return n=\monadic{v}'\mbox{,}
\]
implying $(\monadic{v},\monadic{c})\bullet\pi\sim\monadic{v}'$.

\item If $e=\overline{b}$ where $b\in\BB$ then
\[
t=\boolmodty{\mbox{\lstinline+N+}}\mbox{.}
\]
Furthermore, $\semc{e}\gamma\phi\nu=\return((\monadic{v},\monadic{c}),\tilde{\gamma},\tilde{\omicron},\tilde{\nu})$ where
\[
\begin{array}{l}
(\monadic{v},\monadic{c})=(\branching{\return b&\mbox{if $d=\publicdom$}\\\top&\mbox{otherwise}},\branching{\return(\mknode{\mkconst(|b|)})&\mbox{if $s=\poststg$}\\\top&\mbox{otherwise}})\mbox{,}\\\tilde{\gamma}=\gamma\mbox{,}\quad\tilde{\omicron}=\epsilon\mbox{,}\quad\tilde{\nu}=\nu\mbox{,}
\end{array}
\]
and $\sem{e}\gamma'\phi\omicron'=\return(\monadic{v}',\tilde{\gamma}',\tilde{\omicron}')$ where
\[
\monadic{v}'=\branching{\return b&\mbox{if $s=\poststg$ or $d=\publicdom$}\\\top&\mbox{otherwise}}\mbox{,}\quad\tilde{\gamma}'=\gamma'\mbox{,}\quad\tilde{\omicron}'=\omicron'\mbox{.}
\]
The empty output circuit vacuously accepts~$\pi$. As $\tilde{\gamma}=\gamma$ and $\tilde{\gamma}'=\gamma'$, the claim $\tilde{\gamma}\bullet\pi\sim\tilde{\gamma}'$ follows from assumptions. As $\tilde{\nu}=\nu$ and $\tilde{\omicron}'=\omicron'$, the claim $\tilde{\nu}\bullet\pi\sim\tilde{\omicron}'$ also follows from assumptions. Finally, if $d=\publicdom$ then
\[
(\monadic{v},\monadic{c})\bullet\pi=(\return b,\monadic{c})\bullet\pi=\return|b|\sim\return b=\monadic{v}'\mbox{,}
\]
and if $d\ne\publicdom$ but $s=\poststg$ then
\[
\begin{array}{lcl}
(\monadic{v},\monadic{c})\bullet\pi
&=&(\top,\return(\mknode{\mkconst(|b|)}))\bullet\pi\\
&=&\return(\mknode{\mkconst(|b|)}(\pi))=\return |b|\sim\return b=\monadic{v}'\mbox{.}
\end{array}
\]

\item If $e=x$ then $\semc{e}\gamma\phi\nu=\return((\monadic{v},\monadic{c}),\tilde{\gamma},\tilde{\omicron},\tilde{\nu})$ where
\[
(\monadic{v},\monadic{c})=\lookup{\gamma}{x}\mbox{,}\quad\tilde{\gamma}=\gamma\mbox{,}\quad\tilde{\omicron}=\epsilon\mbox{,}\quad\tilde{\nu}=\nu\mbox{,}
\]
and $\sem{e}\gamma'\phi\omicron'=\return(\monadic{v}',\tilde{\gamma}',\tilde{\omicron}')$ where
\[
\monadic{v}'=\lookup{\gamma'}{x}\mbox{,}\quad\tilde{\gamma}'=\gamma'\mbox{,}\quad\tilde{\omicron}'=\omicron'\mbox{.}
\]
The empty output circuit vacuously accepts~$\pi$. As $\tilde{\gamma}=\gamma$ and $\tilde{\gamma}'=\gamma'$, the claim $\tilde{\gamma}\bullet\pi\sim\tilde{\gamma}'$ follows from assumptions. As $\tilde{\nu}=\nu$ and $\tilde{\omicron}'=\omicron'$, the claim $\tilde{\nu}\bullet\pi\sim\tilde{\omicron}'$ also follows from assumptions. Finally,
\[
(\monadic{v},\monadic{c})\bullet\pi=(\lookup{\gamma}{x})\bullet\pi=\lookup{(\gamma\bullet\pi)}{x}\sim\gamma' x=\monadic{v}'\mbox{.}
\]

\item Let $e=\addexpr{e_1}{e_2}$. Then
\[
\begin{array}{l}
\Gamma\vd e_1:q\mbox{,}\\
\Gamma\vd e_2:q\mbox{,}\\
t=\uintmodty{\mbox{\lstinline+N+}}\mbox{.}
\end{array}
\]
Furthermore,
\[
\begin{array}{l}
\semc{e_1}\gamma\phi\nu=\return((\monadic{v}_1,\monadic{c}_1),\gamma_1,\omicron_1,\nu_1)\mbox{,}\\
\semc{e_2}\gamma_1\phi\nu_1=\return((\monadic{v}_2,\monadic{c}_2),\gamma_2,\omicron_2,\nu_2)\mbox{,}
\end{array}
\]
and $\semc{e}\gamma\phi\nu=\return((\monadic{v},\monadic{c}),\tilde{\gamma},\tilde{\omicron},\tilde{\nu})$ where
\[
\begin{array}{l}
\monadic{v}=\mcomp{i_1\gets\monadic{v}_1\hstop i_2\gets\monadic{v}_2\hstop\return(i_1+i_2)}\mbox{,}\\
\monadic{c}=\mcomp{c_1\gets\monadic{c}_1\hstop c_2\gets\monadic{c}_2\hstop\return(\mknode{\mkoper(+)}(c_1,c_2))}\mbox{,}\\
\tilde{\gamma}=\gamma_2\mbox{,}\quad\tilde{\omicron}=\omicron_1\omicron_2\mbox{,}\quad\tilde{\nu}=\nu_2\mbox{.}
\end{array}
\]

Firstly, suppose that $\sem{e}\gamma'\phi\omicron'=\return(\monadic{v}',\tilde{\gamma}',\tilde{\omicron}')$. Then
\[
\begin{array}{l}
\sem{e_1}\gamma'\phi\omicron'=\return(\monadic{v}'_1,\gamma'_1,\omicron'_1)\mbox{,}\\
\sem{e_2}\gamma'_1\phi\omicron'_1=\return(\monadic{v}'_2,\gamma'_2,\omicron'_2)\mbox{,}\\
\monadic{v}'=\mcomp{v_1\gets\monadic{v}'_1\hstop v_2\gets\monadic{v}'_2\hstop\return(v_1+v_2)}\mbox{,}\quad\tilde{\gamma}'=\gamma'_2\mbox{,}\quad\tilde{\omicron}'=\omicron'_2\mbox{.}
\end{array}
\]
By the induction hypothesis, $\omicron_1$ accepts~$\pi$, whereby $\gamma_1\bullet\pi\sim\gamma'_1$, $\nu_1\bullet\pi\sim\omicron'_1$, and $d=\publicdom$ or $s=\poststg$ implies $(\monadic{v}_1,\monadic{c}_1)\bullet\pi\sim\monadic{v}'_1$. By the induction hypothesis again, $\omicron_2$ accepts~$\pi$, whereby $\gamma_2\bullet\pi\sim\gamma'_2$, $\nu_2\bullet\pi\sim\omicron'_2$, and $d=\publicdom$ or $s=\poststg$ implies $(\monadic{v}_2,\monadic{c}_2)\bullet\pi\sim\monadic{v}'_2$. Hence $\tilde{\omicron}$ accepts~$\pi$, whereby $\tilde{\gamma}\bullet\pi\sim\tilde{\gamma}'$ and $\tilde{\nu}\bullet\pi\sim\tilde{\omicron}'$. Suppose that $d=\publicdom$ or $s=\poststg$. Consider two cases:
\begin{itemize}
\item If $d=\publicdom$ then, by exactness, $\monadic{v}_1=\return v_1$, $\monadic{v}_2=\return v_2$ and $\monadic{v}=\return(v_1+v_2)$ where $v_1,v_2\in\NN$. Thus $(\monadic{v}_1,\monadic{c}_1)\bullet\pi=\return v_1$, $(\monadic{v}_2,\monadic{c}_2)\bullet\pi=\return v_2$ and $(\monadic{v},\monadic{c})\bullet\pi=\return(v_1+v_2)$. Hence $\return v_1\sim\monadic{v}'_1$ and $\return v_2\sim\monadic{v}'_2$, implying $\monadic{v}'_1=\return v_1$ and $\monadic{v}'_2=\return v_2$. Therefore $\monadic{v}'=\return(v_1+v_2)$, leading to $(\monadic{v},\monadic{c})\bullet\pi\sim\monadic{v}'$ as needed.
\item If $s=\poststg$ and $d_0\ne\publicdom$ then, by exactness, $\monadic{c}_1=\return c_1$, $\monadic{c}_2=\return c_2$ and $\monadic{c}=\mcomp{c_1\gets\monadic{c}_1\hstop c_2\gets\monadic{c}_2\hstop\return(\mknode{\mkoper(+)}(c_1,c_2))}$. Thus $(\monadic{v}_1,\monadic{c}_1)\bullet\pi=\return(c_1(\pi))$, $(\monadic{v}_2,\monadic{c}_2)\bullet\pi=\return(c_2(\pi))$ and $(\monadic{v},\monadic{c})\bullet\pi=\return(c_1(\pi)+c_2(\pi))$. Let $c_1(\pi)=v_1$ and $c_2(\pi)=v_2$; then $\return v_1\sim\monadic{v}'_1$ and $\return v_2\sim\monadic{v}'_2$, implying $\monadic{v}'_1=\return v_1$ and $\monadic{v}'_2=\return v_2$. Therefore $\monadic{v}'=\return(v_1+v_2)$, leading to $(\monadic{v},\monadic{c})\bullet\pi\sim\monadic{v}'$ as needed.
\end{itemize}

Conversely, assume that $\tilde{\omicron}$ accepts~$\pi$. Then both $\omicron_1$ and $\omicron_2$ accept~$\pi$. By the induction hypothesis, $\sem{e_1}\gamma'\phi\omicron'=\return(\monadic{v}'_1,\gamma'_1,\omicron'_1)$ where $\gamma_1\bullet\pi\sim\gamma'_1$ and $\nu_1\bullet\pi\sim\omicron'_1$. By the induction hypothesis, $\sem{e_2}\gamma'_1\phi\omicron'_1=\return(\monadic{v}'_2,\gamma'_2,\omicron'_2)$. Therefore $\sem{e}\gamma'\phi\omicron'$ does not fail.

\item Let $e=\assertexpr{e_1}$. Then
\[
\begin{array}{l}
\Gamma\vd e_1:\qualty{\boolmodty{\mbox{\lstinline+N+}}}{\poststg}{d'}\mbox{,}\\
t=\unitty\mbox{,}\quad s=\prestg\mbox{,}\quad d=\publicdom\mbox{.}
\end{array}
\]
Furthermore,
\[
\begin{array}{l}
\semc{e_1}\gamma\phi\nu=\return((\monadic{v}_1,\monadic{c}_1),\gamma_1,\omicron_1,\nu_1)
\end{array}
\]
with $\monadic{c}_1=\return c_1$, and $\semc{e}\gamma\phi\nu=\return((\monadic{v},\monadic{c}),\tilde{\gamma},\tilde{\omicron},\tilde{\nu})$ where
\[
\begin{array}{l}
(\monadic{v},\monadic{c})=(\return\singleton,\top)\mbox{,}\\
\tilde{\gamma}=\gamma_1\mbox{,}\quad\tilde{\omicron}=\omicron_1c_1\mbox{,}\quad\tilde{\nu}=\nu_1\mbox{.}
\end{array}
\]

Firstly, suppose that $\sem{e}\gamma'\phi\omicron'=\return(\monadic{v}',\tilde{\gamma}',\tilde{\omicron}')$. Then
\[
\begin{array}{l}
\sem{e_1}\gamma'\phi\omicron'=\return(\monadic{v}'_1,\gamma'_1,\omicron'_1)\mbox{,}\\
\monadic{v}'_1\ne\return\fls\mbox{,}\quad\monadic{v}'=\return\singleton\mbox{,}\quad\tilde{\gamma}'=\gamma'_1\mbox{,}\quad\tilde{\omicron}'=\omicron'_1\mbox{.}
\end{array}
\]
By the induction hypothesis, $\omicron_1$ accepts~$\pi$, whereby $\gamma_1\bullet\pi\sim\gamma'_1$, $\nu_1\bullet\pi\sim\omicron'_1$, and $(\monadic{v}_1,\monadic{c}_1)\bullet\pi\sim\monadic{v}'_1$. If $d'=\publicdom$ then, by exactness, $\monadic{v}_1=\return v_1$ where $v_1\in\BB$ and $|c_1|=\return|v_1|=\return(c_1(\pi))$. Hence $(\monadic{v}_1,\monadic{c}_1)\bullet\pi=\return|v_1|=\return(c_1(\pi))$. If $d'\ne\publicdom$ then, by exactness, $\monadic{v}_1=\top$. Therefore $(\monadic{v}_1,\monadic{c}_1)\bullet\pi=\return(c_1(\pi))$ again. By $\return(c_1(\pi))\sim\monadic{v}'_1$, we obtain $\monadic{v}'_1=\return v'_1$ where $|v'_1|=c_1(\pi)$. As $\monadic{v}'_1\ne\return\fls$, we must have $v'_1=\tru$, whence $c_1(\pi)=0$. Altogether, we proved that $\tilde{\omicron}$ accepts~$\pi$. We have also obtained that $\tilde{\gamma}\bullet\pi\sim\tilde{\gamma}'$ and $\tilde{\nu}\bullet\pi\sim\tilde{\omicron}'$. In addition, $(\monadic{v},\monadic{c})\bullet\pi=\return\singleton=\monadic{v}'$.

Conversely, assume that $\tilde{\omicron}$ accepts~$\pi$. Then $\omicron_1$ accepts~$\pi$ and $c_1(\pi)=0$. By the induction hypothesis, $\sem{e_1}\gamma'\phi\omicron'=\return(\monadic{v}'_1,\gamma'_1,\omicron'_1)$ where $\gamma_1\bullet\pi\sim\gamma'_1$, $\nu_1\bullet\pi\sim\omicron'_1$ and $(\monadic{v}_1,\monadic{c}_1)\bullet\pi\sim\monadic{v}'_1$. If $d'=\publicdom$ then, by exactness, $\monadic{v}_1=\return v_1$ where $v_1\in\BB$ and $|c_1|=\return|v_1|=\return(c_1(\pi))$. Hence $(\monadic{v}_1,\monadic{c}_1)\bullet\pi=\return|v_1|=\return(c_1(\pi))$. If $d'\ne\publicdom$ then, by exactness, $\monadic{v}_1=\top$. Therefore $(\monadic{v}_1,\monadic{c}_1)\bullet\pi=\return(c_1(\pi))$ again. As $c_1(\pi)\sim\monadic{v}'_1$ and $c_1(\pi)=0$, we must have $\monadic{v}'_1=\return\tru$. Consequently, $\sem{e}\gamma'\phi\omicron'=\return(\return\singleton,\gamma'_1,\omicron'_1)$. The desired result follows.

\item If $e=\getexpr{d'}{k}$ where $d'$ is a domain and $k$ is an input key then
\[
\allpre_{d'}(q)\mbox{,}
\]
which implies $d=d'$ and $s=\prestg$. Furthermore, $\semc{e}\gamma\phi\nu=\return((\monadic{v},\monadic{c}),\tilde{\gamma},\tilde{\omicron},\tilde{\nu})$ where
\[
\begin{array}{l}
(\monadic{v},\monadic{c})=\branching{\allpuretop(\phi_{d'}(k))&\mbox{if $d'=\publicdom$}\\(\top,\top)&\mbox{otherwise}}\mbox{,}\\\tilde{\gamma}=\gamma\mbox{,}\quad\tilde{\omicron}=\epsilon\mbox{,}\quad\tilde{\nu}=\nu\mbox{,}
\end{array}
\]
and $\sem{e}\gamma'\phi\omicron'=\return(\monadic{v}',\tilde{\gamma}',\tilde{\omicron}')$ where
\[
\monadic{v}'=\branching{\allpure(\phi_{d'}(k))&\mbox{if $d'=\publicdom$}\\\top&\mbox{otherwise}}\mbox{,}\quad\tilde{\gamma}'=\gamma'\mbox{,}\quad\tilde{\omicron}'=\omicron'\mbox{.}
\]
The empty output circuit vacuously accepts~$\pi$. As $\tilde{\gamma}=\gamma$ and $\tilde{\gamma}'=\gamma'$, the claim $\tilde{\gamma}\bullet\pi\sim\tilde{\gamma}'$ follows from assumptions. As $\tilde{\nu}=\nu$ and $\tilde{\omicron}'=\omicron'$, the claim $\tilde{\nu}\bullet\pi\sim\tilde{\omicron}'$ also follows from assumptions. Finally, $\allpuretop_{d'}(\phi_{d'}(k))\bullet\pi=\allpure_{d'}(\phi_{d'}(k))$ by Lemma~\ref{compilation:getlemma}, whence $(\monadic{v},\monadic{c})\bullet\pi\sim\monadic{v}'$ in the case $d=\publicdom$.

\item Let $e=\ifexpr{e_1}{e_2}{e_3}$. Then 
\[
\begin{array}{l}
\Gamma\vd e_1:\qualty{\boolmodty{\mbox{\lstinline+N+}}}{\prestg}{d'}\mbox{,}\\
\Gamma\vd e_2:q\mbox{,}\\
\Gamma\vd e_3:q\mbox{.}
\end{array}
\]
Furthermore,
\[
\begin{array}{l}
\semc{e_1}\gamma\phi\nu=\return((\monadic{v}_1,\monadic{c}_1),\gamma_1,\omicron_1,\nu_1)\mbox{.}
\end{array}
\]

Firstly, suppose that $\sem{e}\gamma'\phi\omicron'=\return(\monadic{v}',\tilde{\gamma}',\tilde{\omicron}')$. Then, for some $\monadic{v}'_1$, $\gamma'_1$, $\omicron'_1$,
\[
\begin{array}{l}
\sem{e_1}\gamma'\phi\omicron'=\return(\monadic{v}'_1,\gamma'_1,\omicron'_1)\mbox{.}
\end{array}
\]
By the induction hypothesis, $\omicron_1$ accepts~$\pi$, whereby $\gamma_1\bullet\pi\sim\gamma'_1$, $\nu_1\bullet\pi\sim\omicron'_1$, and $d'=\publicdom$ implies $(\monadic{v}_1,\monadic{c}_1)\bullet\pi\sim\monadic{v}'_1$. Consider three cases:
\begin{itemize}
\item Let $\monadic{v}_1=\return\tru$. Then
\[
\semc{e_2}\gamma_1\phi\nu_1=\return((\monadic{v}_2,\monadic{c}_2),\gamma_2,\omicron_2,\nu_2)\mbox{,}
\]
and $\semc{e}\gamma\phi\nu=\return((\monadic{v},\monadic{c}),\tilde{\gamma},\tilde{\omicron},\tilde{\nu})$ where
\[
\begin{array}{l}
(\monadic{v},\monadic{c})=(\monadic{v}_2,\monadic{c}_2)\mbox{,}\quad\tilde{\gamma}=\gamma_2\mbox{,}\quad\tilde{\omicron}=\omicron_1\omicron_2\mbox{,}\quad\tilde{\nu}=\nu_2\mbox{.}
\end{array}
\]
If $d'\ne\publicdom$ then, by exactness, $\monadic{v}_1=\top$. Thus $d'=\publicdom$ must hold. Hence $(\monadic{v}_1,\monadic{c}_1)\bullet\pi\sim\monadic{v}'_1$ which implies $\monadic{v}'_1=\return\tru$. Consequently,
\[
\begin{array}{l}
\sem{e_2}\gamma'_1\phi\omicron'_1=\return(\monadic{v}'_2,\gamma'_2,\omicron'_2)\mbox{,}\\
\monadic{v}'=\monadic{v}'_2\mbox{,}\quad\tilde{\gamma}'=\gamma'_2\mbox{,}\quad\tilde{\omicron}'=\omicron'_2\mbox{.}
\end{array}
\]
By the induction hypothesis, $\omicron_2$ accepts~$\pi$, whereby $\gamma_2\bullet\pi\sim\gamma'_2$, $\nu_2\bullet\pi\sim\omicron'_2$, and $d=\publicdom$ or $s=\poststg$ implies $(\monadic{v}_2,\monadic{c}_2)\bullet\pi\sim\monadic{v}'_2$. Hence $\omicron_1\omicron_2$ accepts~$\pi$, whereby $\tilde{\gamma}\bullet\pi\sim\tilde{\gamma}'$, $\tilde{\nu}\bullet\pi\sim\tilde{\omicron}'$, and $d=\publicdom$ or $s=\poststg$ implies $(\monadic{v},\monadic{c})\bullet\pi\sim\monadic{v}'$.

\item The case $\monadic{v}_1=\return\fls$ is similar to the previous case.

\item If $\monadic{v}_1\ne\return b$ for $b\in\BB$ then, by exactness, $d'\ne\publicdom$ and $\monadic{v}_1=\top$. Hence $\semc{e}\gamma\phi\nu=\return((\monadic{v},\monadic{c}),\tilde{\gamma},\tilde{\omicron},\tilde{\nu})$ where
\[
\begin{array}{l}
(\monadic{v},\monadic{c})=(\top,\top)\mbox{,}\quad\tilde{\gamma}=\gamma_1\mbox{,}\quad\tilde{\omicron}=\omicron_1\mbox{,}\quad\tilde{\nu}=\nu_1\mbox{.}
\end{array}
\]
As $d'\ne\publicdom$ and $e_1$ has stage $\prestg$, we must have $\monadic{v}'_1=\top$ by exactness in circuit. Consequently,
\[
\monadic{v}'=\top\mbox{,}\quad\tilde{\gamma}'=\gamma'_1\mbox{,}\quad\tilde{\omicron}'=\omicron'_1\mbox{.}
\]
We obtain all the desired claims.

\end{itemize}

Conversely, suppose that $\tilde{\omicron}$ accepts~$\pi$. In all possible cases, this implies that $\omicron_1$ accepts~$\pi$. By the induction hypothesis, $\sem{e_1}\gamma'\phi\omicron'=\return(\monadic{v}'_1,\gamma'_1,\omicron'_1)$, whereby $\gamma_1\bullet\pi\sim\gamma'_1$, $\nu_1\bullet\pi\sim\omicron'_1$, and $d'=\publicdom$ implies $(\monadic{v}_1,\monadic{c}_1)\bullet\pi\sim\monadic{v}'_1$. Consider three cases:
\begin{itemize}
\item Let $\monadic{v}_1=\return\tru$. Then
\[
\semc{e_2}\gamma_1\phi\nu_1=\return((\monadic{v}_2,\monadic{c}_2),\gamma_2,\omicron_2,\nu_2)\mbox{,}
\]
and $\semc{e}\gamma\phi\nu=\return((\monadic{v},\monadic{c}),\tilde{\gamma},\tilde{\omicron},\tilde{\nu})$ where
\[
\begin{array}{l}
(\monadic{v},\monadic{c})=(\monadic{v}_2,\monadic{c}_2)\mbox{,}\quad\tilde{\gamma}=\gamma_2\mbox{,}\quad\tilde{\omicron}=\omicron_1\omicron_2\mbox{,}\quad\tilde{\nu}=\nu_2\mbox{.}
\end{array}
\]
If $d'\ne\publicdom$ then, by exactness, $\monadic{v}_1=\top$. Thus $d'=\publicdom$ must hold. Hence $(\monadic{v}_1,\monadic{c}_1)\bullet\pi\sim\monadic{v}'_1$ which implies $\monadic{v}'_1=\return\tru$. The assumption that $\tilde{\omicron}$ accepts~$\pi$ implies that $\omicron_2$ accepts~$\pi$. By the induction hypothesis, $\sem{e_2}\gamma'_1\phi\omicron'_1=\return(\monadic{v}'_2,\gamma'_2,\omicron'_2)$, whereby $\gamma_2\bullet\pi\sim\gamma'_2$, $\nu_2\bullet\pi\sim\omicron'_2$, and $d=\publicdom$ or $s=\poststg$ implies $(\monadic{v}_2,\monadic{c}_2)\bullet\pi\sim\monadic{v}'_2$. We obtain $\sem{e}\gamma'\phi\omicron'=\return(\monadic{v}'_2,\gamma'_2,\omicron'_2)$ which implies the desired result. 

\item The case $\monadic{v}_1=\return\fls$ is similar to the previous case.

\item If $\monadic{v}_1\ne\return b$ for $b\in\BB$ then, by exactness, $d'\ne\publicdom$ and $\monadic{v}_1=\top$. As $d'\ne\publicdom$ and $e_1$ has stage $\prestg$, we must have $\monadic{v}'_1=\top$ by exactness in circuit. Hence $\sem{e}\gamma'\phi\omicron'=\return(\top,\gamma'_1,\omicron'_1)$ which implies the desired result.

\end{itemize}

\item Let $e=\forexpr{x}{e_1}{e_2}{e_3}$. Then 
\[
\begin{array}{l}
\Gamma\vd e_1:\qualty{\uintty}{\prestg}{d}\mbox{,}\\
\Gamma\vd e_2:\qualty{\uintty}{\prestg}{d}\mbox{,}\\
(x:\qualty{\uintty}{\prestg}{d}),\Gamma\vd e_3:\qualty{t_1}{s_1}{d_1}\mbox{,}\\
t=\listty{\qualty{t_1}{s_1}{d_1}}\mbox{,}\quad s=\prestg\mbox{.}
\end{array}
\]
Furthermore,
\[
\begin{array}{l}
\semc{e_1}\gamma\phi\nu=\return((\monadic{v}_1,\monadic{c}_1),\gamma_1,\omicron_1,\nu_1)\mbox{,}\\
\semc{e_2}\gamma_1\phi\nu_1=\return((\monadic{v}_2,\monadic{c}_2),\gamma_2,\omicron_2,\nu_2)\mbox{.}
\end{array}
\]

Firstly, suppose that $\sem{e}\gamma'\phi\omicron'=\return(\monadic{v}',\tilde{\gamma}',\tilde{\omicron}')$. Then
\[
\begin{array}{l}
\sem{e_1}\gamma'\phi\omicron'=\return(\monadic{v}'_1,\gamma'_1,\omicron'_1)\mbox{,}\\
\sem{e_2}\gamma'_1\phi\omicron'_1=\return(\monadic{v}'_2,\gamma'_2,\omicron'_2)\mbox{.}
\end{array}
\]
By the induction hypothesis, $\omicron_1$ accepts~$\pi$, whereby $\gamma_1\bullet\pi\sim\gamma'_1$, $\nu_1\bullet\pi\sim\omicron'_1$, and $d=\publicdom$ implies $(\monadic{v}_1,\monadic{c}_1)\bullet\pi\sim\monadic{v}'_1$. Again by the induction hypothesis, $\omicron_2$ accepts~$\pi$, whereby $\gamma_2\bullet\pi\sim\gamma'_2$, $\nu_2\bullet\pi\sim\omicron'_2$, and $d=\publicdom$ implies $(\monadic{v}_2,\monadic{c}_2)\bullet\pi\sim\monadic{v}'_2$. Consider two cases:
\begin{itemize}
\item Let $\monadic{v}_1=\return i_1$ and $\monadic{v}_2=\return i_2$, where $i_1,i_2\in\NN$. Denoting $n=\max(0,i_2-i_1)$, we obtain
\[
\begin{array}{l}
\semc{e_3}\left((x,(\return i_1,\top)),\gamma_2\right)\phi\nu_2=\return((\monadic{v}_3,\monadic{c}_3),\gamma_3,\omicron_3,\nu_3)\mbox{,}\\
\semc{e_3}([x\mapsto(\return(i_1+1),\top)]\gamma_3)\phi\nu_3=\return((\monadic{v}_4,\monadic{c}_4),\gamma_4,\omicron_4,\nu_4)\mbox{,}\\
\dotfill\\
\semc{e_3}([x\mapsto(\return(i_1+n-1),\top)]\gamma_{n+1})\phi\nu_{n+1}=\return((\monadic{v}_{n+2},\monadic{c}_{n+2}),\gamma_{n+2},\omicron_{n+2},\nu_{n+2})\mbox{,}\\
\end{array}
\]
and $\semc{e}\gamma\phi\nu=\return((\monadic{v},\monadic{c}),\tilde{\gamma},\tilde{\omicron},\tilde{\nu})$ where
\[
\begin{array}{l}
(\monadic{v},\monadic{c})=(\return((\monadic{v}_3,\monadic{c}_3),\ldots,(\monadic{v}_{n+2},\monadic{c}_{n+2})),\top)\mbox{,}\quad\tilde{\gamma}=\tail\gamma_{n+2}\mbox{,}\quad\tilde{\omicron}=\omicron_1\ldots\omicron_{n+2}\mbox{,}\quad\tilde{\nu}=\nu_{n+2}\mbox{.}
\end{array}
\]
If $d\ne\publicdom$ then, by exactness, $\monadic{v}_1=\monadic{v}_2=\top$. Thus $d=\publicdom$ must hold. Hence $(\monadic{v}_1,\monadic{c}_1)\bullet\pi\sim\monadic{v}'_1$, $(\monadic{v}_2,\monadic{c}_2)\bullet\pi\sim\monadic{v}'_2$ which imply $\monadic{v}'_1=\return i_1$, $\monadic{v}'_2=\return i_2$. Consequently,
\[
\begin{array}{l}
\sem{e_3}\left((x,\return i_1),\gamma'_2\right)\phi\omicron'_2=\return(\monadic{v}'_3,\gamma'_3,\omicron'_3)\mbox{,}\\
\sem{e_3}([x\mapsto\return(i_1+1)]\gamma'_3)\phi\omicron'_3=\return(\monadic{v}'_4,\gamma'_4,\omicron'_4)\mbox{,}\\
\dotfill\\
\sem{e_3}([x\mapsto\return(i_1+n-1)]\gamma'_{n+1})\phi\omicron'_{n+1}=\return(\monadic{v}'_{n+2},\gamma'_{n+2},\omicron'_{n+2})\mbox{,}\\
\monadic{v}'=\return(\monadic{v}'_3,\ldots,\monadic{v}'_{n+2})\mbox{,}\quad\tilde{\gamma}'=\tail\gamma'_{n+2}\mbox{,}\quad\tilde{\omicron}'=\omicron'_{n+2}\mbox{.}
\end{array}
\]
As $\gamma_2\bullet\pi\sim\gamma'_2$ and $(\return i_1,\top)\bullet\pi\sim\return i_1$, we have $\left((x,(\return i_1,\top)),\gamma_2\right)\bullet\pi\sim\left((x,\return i_1),\gamma'_2\right)$. Hence by the induction hypothesis, $\omicron_3$ accepts~$\pi$, whereby $\gamma_3\bullet\pi\sim\gamma'_3$, $\nu_3\bullet\pi\sim\omicron'_3$, and $d_1=\publicdom$ or $s_1=\poststg$ implies $(\monadic{v}_3,\monadic{c}_3)\bullet\pi\sim\monadic{v}'_3$. Similarly, we get $\omicron_k$ accepting~$\pi$, $\gamma_k\bullet\pi\sim\gamma'_k$, $\nu_k\bullet\pi\sim\omicron'_k$, and $d_1=\publicdom$ or $s_1=\poststg$ implying $(\monadic{v}_k,\monadic{c}_k)\bullet\pi\sim\monadic{v}'_k$ for every $k=4,\ldots,n+2$. Obviously also $\tail\gamma_{n+2}\bullet\pi\sim\tail\gamma'_{n+2}$. Altogether, we have proved that $\omicron_1\ldots\omicron_{n+2}$ accepts~$\pi$, whereby $\tilde{\gamma}\bullet\pi\sim\tilde{\gamma}'$ and $\tilde{\nu}\bullet\pi\sim\tilde{\omicron}'$. Assume $d=\publicdom$. If $d_1=\publicdom$ or $s_1=\poststg$ then
\[
(\monadic{v},\monadic{c})\bullet\pi
=\return((\monadic{v}_3,\monadic{c}_3)\bullet\pi,\ldots,(\monadic{v}_{n+2},\monadic{c}_{n+2})\bullet\pi)
\sim\return(\monadic{v}'_3,\ldots,\monadic{v}'_{n+2})
=\monadic{v}'
\]
as required. If $d_1\ne\publicdom$ and $s_1=\prestg$ then, by exactness and exactness in circuit, $(\monadic{v}_k,\monadic{c}_k)=(\top,\top)$ and $\monadic{v}'_k=\top$. Hence $(\monadic{v}_k,\monadic{c}_k)\bullet\pi\sim\monadic{v}'_k$ in this case, too, whence all the desired claims follow.

\item If $\monadic{v}_1\ne\return i_1$ or $\monadic{v}_2\ne\return i_2$ for $i_1,i_2\in\NN$ then, by exactness, $d\ne\publicdom$ and $\monadic{v}_1=\monadic{v}_2=\top$. Hence $\semc{e}\gamma\phi\nu=\return((\monadic{v},\monadic{c}),\tilde{\gamma},\tilde{\omicron},\tilde{\nu})$ where
\[
\begin{array}{l}
(\monadic{v},\monadic{c})=(\top,\top)\mbox{,}\quad\tilde{\gamma}=\gamma_2\mbox{,}\quad\tilde{\omicron}=\omicron_1\omicron_2\mbox{,}\quad\tilde{\nu}=\nu_2\mbox{.}
\end{array}
\]
As $d\ne\publicdom$ and $e_1,e_2$ have stage $\prestg$, we must have $\monadic{v}'_1=\monadic{v}'_2=\top$ by exactness in circuit. Consequently,
\[
\monadic{v}'=\top\mbox{,}\quad\tilde{\gamma}'=\gamma'_2\mbox{,}\quad\tilde{\omicron}'=\omicron'_2\mbox{.}
\]
We obtain all the desired claims.

\end{itemize}

Conversely, suppose that $\tilde{\omicron}$ accepts~$\pi$. In all possible cases, this implies that $\omicron_1$ and $\omicron_2$ accept~$\pi$. By the induction hypothesis, $\sem{e_1}\gamma'\phi\omicron'=\return(\monadic{v}'_1,\gamma'_1,\omicron'_1)$, whereby $\gamma_1\bullet\pi\sim\gamma'_1$, $\nu_1\bullet\pi\sim\omicron'_1$, and $d=\publicdom$ implies $(\monadic{v}_1,\monadic{c}_1)\bullet\pi\sim\monadic{v}'_1$. By the induction hypothesis, $\sem{e_2}\gamma'_1\phi\omicron'_1=\return(\monadic{v}'_2,\gamma'_2,\omicron'_2)$, whereby $\gamma_2\bullet\pi\sim\gamma'_2$, $\nu_2\bullet\pi\sim\omicron'_2$, and $d=\publicdom$ implies $(\monadic{v}_2,\monadic{c}_2)\bullet\pi\sim\monadic{v}'_2$. Consider two cases:
\begin{itemize}
\item Let $\monadic{v}_1=\return i_1$ and $\monadic{v}_2=\return i_2$ where $i_1,i_2\in\NN$. Denoting $n=\max(0,i_2-i_1)$, we obtain
\[
\begin{array}{l}
\semc{e_3}\left((x,(\return i_1,\top)),\gamma_2\right)\phi\nu_2=\return((\monadic{v}_3,\monadic{c}_3),\gamma_3,\omicron_3,\nu_3)\mbox{,}\\
\semc{e_3}([x\mapsto(\return(i_1+1),\top)]\gamma_3)\phi\nu_3=\return((\monadic{v}_4,\monadic{c}_4),\gamma_4,\omicron_4,\nu_4)\mbox{,}\\
\dotfill\\
\semc{e_3}([x\mapsto(\return(i_1+n-1),\top)]\gamma_{n+1})\phi\nu_{n+1}=\return((\monadic{v}_{n+2},\monadic{c}_{n+2}),\gamma_{n+2},\omicron_{n+2},\nu_{n+2})\mbox{,}\\
\end{array}
\]
and $\semc{e}\gamma\phi\nu=\return((\monadic{v},\monadic{c}),\tilde{\gamma},\tilde{\omicron},\tilde{\nu})$ where
\[
\begin{array}{l}
(\monadic{v},\monadic{c})=(\return((\monadic{v}_3,\monadic{c}_3),\ldots,(\monadic{v}_{n+2},\monadic{c}_{n+2})),\top)\mbox{,}\\\tilde{\gamma}=\tail\gamma_{n+2}\mbox{,}\quad\tilde{\omicron}=\omicron_1\ldots\omicron_{n+2}\mbox{,}\quad\tilde{\nu}=\nu_{n+2}\mbox{.}
\end{array}
\]
If $d\ne\publicdom$ then, by exactness, $\monadic{v}_1=\monadic{v}_2=\top$. Thus $d=\publicdom$ must hold. Hence $(\monadic{v}_1,\monadic{c}_1)\bullet\pi\sim\monadic{v}'_1$ and $(\monadic{v}_2,\monadic{c}_2)\bullet\pi\sim\monadic{v}'_2$, which imply $\monadic{v}'_1=\return i_1$, $\monadic{v}'_2=\return i_2$. The assumption that $\tilde{\omicron}$ accepts~$\pi$ implies that $\omicron_3,\ldots,\omicron_{n+2}$ all accept~$\pi$. As $\gamma_2\bullet\pi\sim\gamma'_2$ and $(\return i_1,\top)\bullet\pi\sim\return i_1$, we have $\left((x,(\return i_1,\top)),\gamma_2\right)\bullet\pi\sim\left((x,\return i_1),\gamma'_2\right)$. By the induction hypothesis, $\sem{e_3}\left((x,\return i_1),\gamma'_2\right)\phi\omicron'_2=\return(\monadic{v}'_3,\gamma'_3,\omicron'_3)$, whereby $\gamma_3\bullet\pi\sim\gamma'_3$, $\nu_3\bullet\pi\sim\omicron'_3$, and $d_1=\publicdom$ or $s_1=\poststg$ implies $(\monadic{v}_3,\monadic{c}_3)\bullet\pi\sim\monadic{v}'_3$. Similarly, we get $\sem{e_k}([x\mapsto\return(i_1+k-3)]\gamma'_{k-1})\phi\omicron'_{k-1}=\return(\monadic{v}'_k,\gamma'_k,\omicron'_k)$ where $\gamma_k\bullet\pi\sim\gamma'_k$, $\nu_k\bullet\pi\sim\omicron'_k$, and $d_1=\publicdom$ or $s_1=\poststg$ implies $(\monadic{v}_k,\monadic{c}_k)\bullet\pi\sim\monadic{v}'_k$ for every $k=4,\ldots,n+2$. We obtain $\sem{e}\gamma'\phi\omicron'=\return(\return(\monadic{v}'_3,\ldots,\monadic{v}'_{n+2}),\gamma'_{n+2},\omicron'_{n+2})$ which implies the desired result.

\item If $\monadic{v}_1\ne\return i_1$ or $\monadic{v}_1\ne\return i_2$ for $i_1,i_2\in\NN$ then, by exactness, $d\ne\publicdom$ and $\monadic{v}_1=\monadic{v}_2=\top$. As $d\ne\publicdom$ and $e_1,e_2$ have stage $\prestg$, we must have $\monadic{v}'_1=\monadic{v}'_2=\top$ by exactness in circuit. Hence $\sem{e}\gamma'\phi\omicron'=\return(\top,\gamma'_2,\omicron'_2)$ which implies the desired result.

\end{itemize}

\item Let $e=\wireexpr{e_1}$. Then
\[
\begin{array}{l}
\Gamma\vd e_1:\qualty{t}{\prestg}{d}\mbox{,}\\
t\in\set{\uintmodty{\mbox{\lstinline+N+}},\boolmodty{\mbox{\lstinline+N+}}}\mbox{,}\quad s=\poststg\mbox{.}
\end{array}
\]
Furthermore,
\[
\begin{array}{l}
\semc{e_1}\gamma\phi\nu=\return((\monadic{v}_1,\monadic{c}_1),\gamma_1,\omicron_1,\nu_1)\mbox{,}\\
\end{array}
\]
and $\semc{e}\gamma\phi\nu=\return((\monadic{v},\monadic{c}),\tilde{\gamma},\tilde{\omicron},\tilde{\nu})$ where
\[
\begin{array}{l}
(\monadic{v},\monadic{c})=(\monadic{v}_1,\branching{\return(\mknode{\mkconst(n)})&\mbox{if $\monadic{v}_1=\return n$, $n\in\NN$}\\\return(\mknode{\mkconst(|b|)})&\mbox{if $\monadic{v}_1=\return b$, $b\in\BB$}\\\return(\mknode{\mkinput_d(\left(\nu_1\right)_d)})&\mbox{if $\monadic{v}_1=\top$}})\mbox{,}\\
\tilde{\gamma}=\gamma_1\mbox{,}\quad\tilde{\omicron}=\omicron_1\mbox{,}\quad\tilde{\nu}=\lam{d'}{\branching{\left(\nu_1\right)_{d'}+1&\mbox{if $d'=d$}\\\left(\nu_1\right)_{d'}&\mbox{otherwise}}}\mbox{.}
\end{array}
\]

Firstly, suppose that $\sem{e}\gamma'\phi\omicron'=\return(\monadic{v}',\tilde{\gamma}',\tilde{\omicron}')$. Then
\[
\begin{array}{l}
\sem{e_1}\gamma'\phi\omicron'=\return(\monadic{v}'_1,\gamma'_1,\omicron'_1)\mbox{,}\\
\monadic{v}'=\branching{\monadic{v}'_1&\mbox{if $d=\publicdom$}\\\head\left(\omicron'_1\right)_d&\mbox{otherwise}}\mbox{,}\\\tilde{\gamma}'=\gamma'_1\mbox{,}\quad\tilde{\omicron}'=\lam{d'}{\branching{\tail\left(\omicron'_1\right)_{d'}&\mbox{if $d'=d$}\\\left(\omicron'_1\right)_{d'}&\mbox{otherwise}}}\mbox{.}
\end{array}
\]
By the induction hypothesis, $\omicron_1$ accepts~$\pi$, whereby $\gamma_1\bullet\pi\sim\gamma'_1$, $\nu_1\bullet\pi\sim\omicron'_1$, and $d=\publicdom$ implies $(\monadic{v}_1,\monadic{c}_1)\bullet\pi\sim\monadic{v}'_1$. Hence $\tilde{\omicron}$ accepts~$\pi$ and $\tilde{\gamma}\bullet\pi\sim\tilde{\gamma}'$. We also get $\tilde{\nu}\bullet\pi\sim\tilde{\omicron}'$ since dropping one more item from the input list results in the tail of the result that would be obtained otherwise. To complete, we have to prove $(\monadic{v},\monadic{c})\bullet\pi\sim\monadic{v}'$. Consider two cases:
\begin{itemize}
\item If $d=\publicdom$ then $(\monadic{v}_1,\monadic{c}_1)\bullet\pi\sim\monadic{v}'_1$. By exactness and exactness in circuit, $\monadic{v}_1=\return v_1$ and $\monadic{v}'_1=\return v'_1$ where $v_1,v'_1\in t$. Thus $(\monadic{v},\monadic{c})\bullet\pi=(\monadic{v}_1,\monadic{c}_1)\bullet\pi\in\set{\return v_1,\return|v_1|}$ where the outcome depends on whether $t=\uintmodty{\mbox{\lstinline+N+}}$ or $t=\boolmodty{\mbox{\lstinline+N+}}$. In both cases, $\monadic{v}'_1=\return v_1$, leading to $(\monadic{v},\monadic{c})\bullet\pi\sim\monadic{v}'$ as needed.
\item If $d\ne\publicdom$ then, by exactness, $\monadic{v}_1=\top$. Hence
\[
\begin{array}{lcl}
(\monadic{v},\monadic{c})\bullet\pi
&=&(\top,\return(\mknode{\mkinput_d(\left(\nu_1\right)_d)}))\bullet\pi\\
&=&\return(\mknode{\mkinput_d(\left(\nu_1\right)_d)}(\pi))\\
&=&\return(\pi_d(\left(\nu_1\right)_d))\mbox{.}
\end{array}
\]
From $\nu_1\bullet\pi\sim\omicron'_1$, we get $\map\return\left(\drop\left(\nu_1\right)_d\pi_d\right)\sim\left(\omicron'_1\right)_d$, and taking the first elements in the lists occurring in the latter statement gives $\return(\pi_d(\left(\nu_1\right)_d))\sim\head\left(\omicron'_1\right)_d$ as needed.
\end{itemize}

Conversely, assume that $\tilde{\omicron}$ accepts~$\pi$. Then $\omicron_1$ accepts~$\pi$. By the induction hypothesis, $\sem{e_1}\gamma'\phi\omicron'=\return(\monadic{v}'_1,\gamma'_1,\omicron'_1)$ where $\gamma_1\bullet\pi\sim\gamma'_1$ and $\nu_1\bullet\pi\sim\omicron'_1$. Then $\sem{e}\gamma'\phi\omicron'=\return(\monadic{v}',\tilde{\gamma},\tilde{\omicron}')$ where $\monadic{v}'=\branching{\monadic{v}'_1&\mbox{if $d=\publicdom$}\\\head\left(\omicron'_1\right)_d&\mbox{otherwise}}$. The latter is well-defined if $d=\publicdom$; if $d\ne\publicdom$ then it is well defined, provided that $\left(\omicron'_1\right)_d$ is non-empty. As $\nu_1\bullet\pi\sim\omicron'_1$, this condition is equivalent to $(\nu_1)_d$ being smaller than the length of $\pi_d$. Since $(\monadic{v},\monadic{c})=(\top,\return(\mknode{\mkinput_d((\nu_1)_d)}))$ and $(\monadic{v},\monadic{c})\bullet\pi$ is well-defined, $(\nu_1)_d$ must indeed be smaller than the length of~$\pi_d$. The desired claim follows.

\item Let $e=\castexpr{e_1}{q}$. Then
\[
\begin{array}{l}
\Gamma\vd e_1:\qualty{t}{s'}{d'}\mbox{,}\\
s'\subtype s\mbox{,}\quad d'\subtype d\mbox{.}
\end{array}
\]
Furthermore,
\[
\begin{array}{l}
\semc{e_1}\gamma\phi\nu=\return((\monadic{v}_1,\monadic{c}_1),\gamma_1,\omicron_1,\nu_1)\mbox{,}\\
\end{array}
\]
and $\semc{e}\gamma\phi\nu=\return((\monadic{v},\monadic{c}),\tilde{\gamma},\tilde{\omicron},\tilde{\nu})$ where
\[
\begin{array}{l}
(\monadic{v},\monadic{c})=(\branching{\monadic{v}_1&\mbox{if $d=\publicdom$}\\\top&\mbox{otherwise}},\branching{\monadic{c}_1&\mbox{if $s=\poststg$}\\\top&\mbox{otherwise}})\mbox{,}\\
\tilde{\gamma}=\gamma_1\mbox{,}\quad\tilde{\omicron}=\omicron_1\mbox{,}\quad\tilde{\nu}=\nu_1\mbox{.}
\end{array}
\]

Firstly, suppose that $\sem{e}\gamma'\phi\omicron'=\return(\monadic{v}',\tilde{\gamma}',\tilde{\omicron}')$. Then
\[
\begin{array}{l}
\sem{e_1}\gamma'\phi\omicron'=\return(\monadic{v}'_1,\gamma'_1,\omicron'_1)\mbox{,}\\
\monadic{v}'=\branching{\monadic{v}'_1&\mbox{if $d=\publicdom$ or $s=\poststg$}\\\top&\mbox{otherwise}}\mbox{,}\\\tilde{\gamma}'=\gamma'_1\mbox{,}\quad\tilde{\omicron}'=\omicron'_1\mbox{.}
\end{array}
\]
By the induction hypothesis, $\omicron_1$ accepts~$\pi$, whereby $\gamma_1\bullet\pi\sim\gamma'_1$, $\nu_1\bullet\pi\sim\omicron'_1$, and $d'=\publicdom$ or $s'=\poststg$ implies $(\monadic{v}_1,\monadic{c}_1)\bullet\pi\sim\monadic{v}'_1$. Hence $\tilde{\omicron}$ accepts~$\pi$, whereby $\tilde{\gamma}\bullet\pi\sim\tilde{\gamma}'$ and $\tilde{\nu}\bullet\pi\sim\tilde{\omicron}'$. Suppose that  $d=\publicdom$ or $s=\poststg$. Then $d'=\publicdom$ or $s'=\poststg$ since $d'\subtype d$ and $s'\subtype s$. Thus $(\monadic{v}_1,\monadic{c}_1)\bullet\pi\sim\monadic{v}'_1$ and $\monadic{v}'=\monadic{v}'_1$. Consider two cases:
\begin{itemize}
\item If $d=\publicdom$ then $(\monadic{v},\monadic{c})\bullet\pi=(\monadic{v},\monadic{c}_1)\bullet\pi=(\monadic{v}_1,\monadic{c}_1)\bullet\pi$ and the desired claim follows.
\item If $d\ne\publicdom$ and $s=\poststg$ then $s'=\poststg$. If $d'\ne\publicdom$ then $\monadic{v}_1=\top$ by exactness, whence $(\monadic{v},\monadic{c})=(\monadic{v}_1,\monadic{c}_1)$ and the desired claim follows. Suppose now that $d'=\publicdom$. By types being well-structured, $t=\uintmodty{\mbox{\lstinline+N+}}$ or $t=\boolmodty{\mbox{\lstinline+N+}}$, whence exactness implies $\monadic{v}_1=\return v_1$ where $v_1\in\NN$ or $v_1\in\BB$. Thus $(\monadic{v}_1,\monadic{c}_1)\bullet\pi=(\return v_1,\monadic{c}_1)\bullet\pi\sim\return v_1$. As by exactness, $\monadic{c}_1=\return c_1$ and $\return(c_1(\pi))\sim\return v_1$, we obtain
\[
(\monadic{v},\monadic{c})\bullet\pi=(\top,\monadic{c}_1)\bullet\pi=\return(c_1(\pi))\sim\return v_1\mbox{.}
\]
But $(\return v_1,\return c_1)\bullet\pi\sim\return v_1$, too, while we also know that $(\return v_1,\return c_1)\bullet\pi\sim\monadic{v}'_1$. Hence $\return v_1\sim\monadic{v}'_1$ and the desired claim follows.
\end{itemize}

Conversely, assume that $\tilde{\omicron}$ accepts~$\pi$. Then $\omicron_1$ accepts~$\pi$. By the induction hypothesis, $\sem{e_1}\gamma'\phi\omicron'=\return(\monadic{v}'_1,\gamma'_1,\omicron'_1)$ for some $\monadic{v}'_1,\gamma'_1,\omicron'_1$. Hence $\sem{e}\gamma'\phi\omicron'$ does not fail.

\item Let $e=\assignexpr{e_1}{e_2}$ where $e_1=\loadexpr{\loadexpr{\loadexpr{x}{y_1}}{y_2}\ldots}{y_n}$. By Lemma~\ref{typesystem:lhslemma},
\[
\begin{array}{l}
\Gamma\vd x:q'\mbox{,}\\
\Gamma\vd y_k:q_k\mbox{ for every $k=1,\ldots,n$,}\\
\Gamma\vd e_1:q_{n+1}\mbox{,}\\
\Gamma\vd e_2:q_{n+1}\mbox{,}\\
t=\unitty\mbox{,}\quad s=\prestg\mbox{,}\quad d=\publicdom\mbox{,}
\end{array}
\]
where $q_k=\qualty{\uintty}{\prestg}{d_k}$ for each $k=1,\ldots,n$, and
\[
q'=\qualty{\listty{\ldots\listty{\qualty{\listty{\qualty{\listty{q_{n+1}}}{\prestg}{d_n}}}{\prestg}{d_{n-1}}}\ldots}}{\prestg}{d_1}\mbox{.}
\]
Denote $q_{n+1}=(\qualty{t_{n+1}}{s_{n+1}}{d_{n+1}})$. Furthermore,
\[
\begin{array}{l}
\semc{y_1}\gamma\phi\nu=\return((\monadic{i}_1,\monadic{c}_1),\gamma_1,\omicron_1,\nu_1)\mbox{,}\\
\semc{y_2}\gamma_1\phi\nu_1=\return((\monadic{i}_2,\monadic{c}_2),\gamma_2,\omicron_2,\nu_2)\mbox{,}\\
\dotfill\\
\semc{y_n}\gamma_{n-1}\phi\nu_{n-1}=\return((\monadic{i}_n,\monadic{c}_n),\gamma_n,\omicron_n,\nu_n)\mbox{,}\\
\sem{e_2}\gamma_n\phi\nu_n=\return((\monadic{v}_1,\monadic{c}_{n+1}),\gamma_{n+1},\omicron_{n+1},\nu_{n+1})\mbox{,}
\end{array}
\]
and $\semc{e}\gamma\phi\nu=\return((\monadic{v},\monadic{c}),\tilde{\gamma},\tilde{\omicron},\tilde{\nu})$ where
\[
\begin{array}{l}
(\monadic{v},\monadic{c})=(\return\singleton,\top)\mbox{,}\\
\tilde{\gamma}=[x\mapsto\updatec(\gamma(x),\monadic{i}_1\ldots\monadic{i}_n,(\monadic{v}_1,\monadic{c}_{n+1}))]\gamma_{n+1}\mbox{,}\quad\tilde{\omicron}=\omicron_1\ldots\omicron_{n+1}\mbox{,}\quad\tilde{\nu}=\nu_{n+1}\mbox{.}
\end{array}
\]

Firstly, suppose that $\sem{e}\gamma'\phi\omicron'=\return(\monadic{v}',\tilde{\gamma}',\tilde{\omicron}')$. Then
\[
\begin{array}{l}
\sem{y_1}\gamma'\phi\omicron'=\return(\monadic{i}'_1,\gamma'_1,\omicron'_1)\mbox{,}\\
\sem{y_2}\gamma'_1\phi\omicron'_1=\return(\monadic{i}'_2,\gamma'_2,\omicron'_2)\mbox{,}\\
\dotfill\\
\sem{y_n}\gamma'_{n-1}\phi\omicron'_{n-1}=\return(\monadic{i}'_n,\gamma'_n,\omicron'_n)\mbox{,}\\
\sem{e_2}\gamma'_n\phi\omicron'_n=\return(\monadic{v}'_1,\gamma'_{n+1},\omicron'_{n+1})\mbox{,}\\
\monadic{v}'=\return\singleton\mbox{,}\quad\tilde{\gamma}'=[x\mapsto\update(\gamma'(x),\monadic{i}'_1\ldots\monadic{i}'_n,\monadic{v}'_1)]\gamma'_{n+1}\mbox{,}\quad\tilde{\omicron}'=\omicron'_{n+1}\mbox{.}
\end{array}
\]
By the induction hypothesis, $\omicron_1$ accepts~$\pi$, whereby $\gamma_1\bullet\pi\sim\gamma'_1$, $\nu_1\bullet\pi\sim\omicron'_1$, and $d_1=\publicdom$ implies $(\monadic{i}_1,\monadic{c}_1)\bullet\pi\sim\monadic{i}'_1$. Similarly, we see that $\omicron_k$ accepts~$\pi$, whereby $\gamma_k\bullet\pi\sim\gamma'_k$, $\nu_k\bullet\pi\sim\omicron'_k$, and $d_k=\publicdom$ implies $(\monadic{i}_k,\monadic{c}_k)\bullet\pi\sim\monadic{i}'_k$ for every $k=2,\ldots,n$. Furthermore, the induction hypothesis implies $\omicron_{n+1}$ accepting~$\pi$, whereby $\gamma_{n+1}\bullet\pi\sim\gamma'_1$, $\nu_{n+1}\bullet\pi\sim\omicron'_{n+1}$, and $d_{n+1}=\publicdom$ or $s_{n+1}=\poststg$ implying $(\monadic{v}_1,\monadic{c}_{n+1})\bullet\pi\sim\monadic{v}'_1$. Hence $\tilde{\omicron}$ accepts~$\pi$ and $\tilde{\nu}\bullet\pi\sim\tilde{\omicron}'$. We also have $(\monadic{v},\monadic{c})\bullet\pi\sim\monadic{v}'$ since $(\return\singleton,\top)\bullet\pi=\return\singleton$. Finally, note that $d_k\ne\publicdom$ for any $k=1,\ldots,n$ implies $(\monadic{i}_k,\monadic{c}_k)=(\top,\top)$ and $\monadic{i}'_k=\top$ by exactness and exactness in circuit, respectively. As $(\top,\top)\bullet\pi=\top$, we have $(\monadic{i}_k,\monadic{c}_k)\bullet\pi\sim\monadic{i}'_k$ for every $k=1,\ldots,n$. In addition, if $d_{n+1}\ne\publicdom$ and $s_{n+1}=\prestg$ then $(\monadic{v}_1,\monadic{c}_{n+1})=(\top,\top)$ and $\monadic{v}'_1=\top$ by exactness and exactness in circuit, respectively, whence $(\monadic{v}_1,\monadic{c}_{n+1})\bullet\pi\sim\monadic{v}'_1$ in all cases. By Lemma~\ref{compilation:updlemma}, $\updatec(\gamma(x),\monadic{i}_1\ldots\monadic{i}_n,(\monadic{v}_1,\monadic{c}_{n+1}))\bullet\pi\sim\update(\gamma'(x),\monadic{i}'_1\ldots\monadic{i}'_n,\monadic{v}'_1)$. Consequently, $[x\mapsto\updatec(\gamma(x),\monadic{i}_1\ldots\monadic{i}_n,(\monadic{v}_1,\monadic{c}_{n+1}))]\gamma_{n+1}\bullet\pi\sim[x\mapsto\update(\gamma'(x),\monadic{i}'_1\ldots\monadic{i}'_n,\monadic{v}'_1)]\gamma'_{n+1}$ and the desired claim follows.

Conversely, assume that $\tilde{\omicron}$ accepts~$\pi$. Then $\omicron_1,\ldots,\omicron_{n+1}$ all accept~$\pi$. By the induction hypothesis, $\sem{y_1}\gamma'\phi\omicron'=\return(\monadic{i}'_1,\gamma'_1,\omicron'_1)$ where $\gamma_1\bullet\pi\sim\gamma'_1$, $\nu_1\bullet\pi\sim\omicron'_1$ and $d_1=\publicdom$ implies $(\monadic{i}_1,\monadic{c}_1)\bullet\pi\sim\monadic{i}'_1$. Similarly, we see for every $k=2,\ldots,n$ that $\sem{y_k}\gamma'_{k-1}\phi\omicron'_{k-1}=\return(\monadic{i}'_k,\gamma'_k,\omicron'_k)$ where $\gamma_k\bullet\pi\sim\gamma'_k$, $\nu_k\bullet\pi\sim\omicron'_k$ and $d_k=\publicdom$ implies $(\monadic{i}_k,\monadic{c}_k)\bullet\pi\sim\monadic{i}'_k$. Furthermore, $\sem{e_2}\gamma'_n\phi\omicron'_n=\return(\monadic{v}'_1,\gamma'_{n+1},\omicron'_{n+1})$ where $d_{n+1}=\publicdom$ or $s_{n+1}=\poststg$ implies $(\monadic{v}_1,\monadic{c}_{n+1})\bullet\pi\sim\monadic{v}'_1$. Similarly to the first part of the proof, we see that $(\monadic{i}_k,\monadic{c}_k)\bullet\pi\sim\monadic{i}'_k$ and $(\monadic{v}_1,\monadic{c}_{n+1})\bullet\pi\sim\monadic{v}'_1$ in all cases. By Lemma~\ref{compilation:updlemma}, $\update(\gamma'(x),\monadic{i}'_1\ldots\monadic{i}'_n,\monadic{v}'_1)$ is well-defined. Therefore $\sem{e}\gamma'\phi\omicron'$ does not fail.

\item Let $e=\loadexpr{e_1}{e_2}$. Then
\[
\begin{array}{l}
\Gamma\vd e_1:\qualty{\listty{q}}{\prestg}{d_1}\mbox{,}\\
\Gamma\vd e_2:\qualty{\uintty}{\prestg}{d_1}\mbox{.}
\end{array}
\]
Furthermore,
\[
\begin{array}{l}
\semc{e_1}\gamma\phi\nu=\return((\monadic{a},\monadic{c}_1),\gamma_1,\omicron_1,\nu_1)\mbox{,}\\
\semc{e_2}\gamma_1\phi\nu_1=\return((\monadic{i},\monadic{c}_2),\gamma_2,\omicron_2,\nu_2)\mbox{,}
\end{array}
\]
and $\semc{e}\gamma\phi\nu=\return((\monadic{v},\monadic{c}),\tilde{\gamma},\tilde{\omicron},\tilde{\nu})$ where
\[
\begin{array}{l}
(\monadic{v},\monadic{c})=\branching{(\monadic{v}_1,\monadic{c}_3)&\mbox{if $\mcomp{a\gets\monadic{a}\hstop i\gets\monadic{i}\hstop\return a_i}=\return(\monadic{v}_1,\monadic{c}_3)$}\\(\top,\top)&\mbox{otherwise}}\mbox{,}\\
\tilde{\gamma}=\gamma_2\mbox{,}\quad\tilde{\omicron}=\omicron_1\omicron_2\mbox{,}\quad\tilde{\nu}=\nu_2\mbox{.}
\end{array}
\]

Firstly, suppose that $\sem{e}\gamma'\phi\omicron'=\return(\monadic{v}',\tilde{\gamma}',\tilde{\omicron}')$. Then
\[
\begin{array}{l}
\sem{e_1}\gamma'\phi\omicron'=\return(\monadic{a}',\gamma'_1,\omicron'_1)\mbox{,}\\
\sem{e_2}\gamma'_1\phi\omicron'_1=\return(\monadic{i}',\gamma'_2,\omicron'_2)\mbox{,}\\
\monadic{v}'=\mcomp{a\gets\monadic{a}'\hstop i\gets\monadic{i}'\hstop a_i}\mbox{,}\quad\tilde{\gamma}'=\gamma'_2\mbox{,}\quad\tilde{\omicron}'=\omicron'_2\mbox{.}
\end{array}
\]
By the induction hypothesis, $\omicron_1$ accepts~$\pi$, whereby $\gamma_1\bullet\pi\sim\gamma'_1$, $\nu_1\bullet\pi\sim\omicron'_1$, and $d_1=\publicdom$ implies $(\monadic{a},\monadic{c}_1)\bullet\pi\sim\monadic{a}'$. By the induction hypothesis again, $\omicron_2$ accepts~$\pi$, whereby $\gamma_2\bullet\pi\sim\gamma'_2$, $\nu_2\bullet\pi\sim\omicron'_2$, and $d_1=\publicdom$ implies $(\monadic{i},\monadic{c}_2)\bullet\pi\sim\monadic{i}'$. Hence $\tilde{\omicron}$ accepts~$\pi$, whereby $\tilde{\gamma}\bullet\pi\sim\tilde{\gamma}'$ and $\tilde{\nu}\bullet\pi\sim\tilde{\omicron}'$. Suppose that $d=\publicdom$ or $s=\poststg$. Then $d_1=\publicdom$ because types are well-structured. Thus $(\monadic{a},\monadic{c}_1)\bullet\pi\sim\monadic{a}'$ and $(\monadic{i},\monadic{c}_2)\bullet\pi\sim\monadic{i}'$. By exactness and exactness in circuit, $\monadic{a}=\return((\monadic{a}_1,\monadic{c}'_1),\ldots,(\monadic{a}_l,\monadic{c}'_l))$ and $\monadic{a}'=\return(\monadic{a}'_1,\ldots,\monadic{a}'_{l'})$. By $(\monadic{a},\monadic{i})\bullet\pi\sim\monadic{a}'$, we obtain $l=l'$ and $(\monadic{a}_k,\monadic{c}'_k)\bullet\pi\sim\monadic{a}'_k$ for every $k=1,\ldots,l$. For similar reasons, we obtain $\monadic{i}=\monadic{i}'=\return i$. Thus $(\monadic{v},\monadic{c})=(\monadic{a}_i,\monadic{c}'_i)$ and $\monadic{v}'=\monadic{a}'_i$. The desired result follows by $(\monadic{a}_i,\monadic{c}'_i)\bullet\pi\sim\monadic{a}'_i$.

Conversely, assume that $\tilde{\omicron}$ accepts~$\pi$. Then both $\omicron_1$ and $\omicron_2$ accept~$\pi$. By the induction hypothesis, $\sem{e_1}\gamma'\phi\omicron'=\return(\monadic{a}',\gamma'_1,\omicron'_1)$ where $\gamma_1\bullet\pi\sim\gamma'_1$, $\nu_1\bullet\pi\sim\omicron'_1$ and $d_1=\publicdom$ implies $(\monadic{a},\monadic{c}_1)\bullet\pi\sim\monadic{a}'$. By the induction hypothesis, $\sem{e_2}\gamma'_1\phi\omicron'_1=\return(\monadic{i}',\gamma'_2,\omicron'_2)$, whereby $d_1=\publicdom$ implies $(\monadic{i},\monadic{c}_2)\bullet\pi\sim\monadic{i}'$. Consider two cases:
\begin{itemize}
\item If $d=\publicdom$ or $s=\poststg$ then, like in the first half of the proof, we obtain $\monadic{a}=\return((\monadic{a}_1,\monadic{c}'_1),\ldots,(\monadic{a}_l,\monadic{c}'_l))$, $\monadic{a}'=\return(\monadic{a}'_1,\ldots,\monadic{a}'_l)$ and $\monadic{i}=\monadic{i}'=\return i$, whereby $i\leq l$ since $(\monadic{v},\monadic{c})$ is well-defined. Hence $\sem{e}\gamma'\phi\omicron'$ is of the form $\return(\monadic{v}',\tilde{\gamma},\tilde{\omicron}')$, implying the desired result.
\item If $d\ne\publicdom$ and $s=\prestg$ then, by exactness in circuit, $\monadic{a}'=\top$ and $\monadic{i}'=\top$, whence $\sem{e}\gamma'\phi\omicron'=\return(\top,\gamma'_2,\omicron'_2)$, implying the desired result.
\end{itemize}

\item Let $e=\stmtcomp{\letexpr{x}{e_1}}{e_2}$. Then
\[
\begin{array}{l}
\Gamma\vd e_1:\qualty{t_1}{s_1}{d_1}\mbox{,}\\
(x:\qualty{t_1}{s_1}{d_1}),\Gamma\vd e_2:q\mbox{.}
\end{array}
\]
Furthermore,
\[
\begin{array}{l}
\semc{e_1}\gamma\phi\nu=\return((\monadic{v}_1,\monadic{c}_1),\gamma_1,\omicron_1,\nu_1)\mbox{,}\\
\semc{e_2}((x,(\monadic{v}_1,\monadic{c}_1)),\gamma_1)\phi\nu_1=\return((\monadic{v}_2,\monadic{c}_2),\gamma_2,\omicron_2,\nu_2)\mbox{,}
\end{array}
\]
and $\semc{e}\gamma\phi\nu=\return((\monadic{v},\monadic{c}),\tilde{\gamma},\tilde{\omicron},\tilde{\nu})$ where
\[
\begin{array}{l}
(\monadic{v},\monadic{c})=(\monadic{v}_2,\monadic{c}_2)\mbox{,}\quad\tilde{\gamma}=\tail\gamma_2\mbox{,}\quad\tilde{\omicron}=\omicron_1\omicron_2\mbox{,}\quad\tilde{\nu}=\nu_2\mbox{.}
\end{array}
\]

Firstly, suppose that $\sem{e}\gamma'\phi\omicron'=\return(\monadic{v}',\tilde{\gamma}',\tilde{\omicron}')$. Then
\[
\begin{array}{l}
\sem{e_1}\gamma'\phi\omicron'=\return(\monadic{v}'_1,\gamma'_1,\omicron'_1)\mbox{,}\\
\sem{e_2}((x,\monadic{v}'_1),\gamma'_1)\phi\omicron'_1=\return(\monadic{v}'_2,\gamma'_2,\omicron'_2)\mbox{,}\\
\monadic{v}'=\monadic{v}'_2\mbox{,}\quad\tilde{\gamma}'=\tail\gamma'_2\mbox{,}\quad\tilde{\omicron}'=\omicron'_2\mbox{.}
\end{array}
\]
By the induction hypothesis, $\omicron_1$ accepts~$\pi$, whereby $\gamma_1\bullet\pi\sim\gamma'_1$, $\nu_1\bullet\pi\sim\omicron'_1$, and $d_1=\publicdom$ or $s_1=\poststg$ implies $(\monadic{v}_1,\monadic{c}_1)\bullet\pi\sim\monadic{v}'_1$. If $d_1\ne\publicdom$ and $s_1=\prestg$ then $(\monadic{v}_1,\monadic{c}_1)=(\top,\top)$ and $\monadic{v}'_1=\top$ by exactness and exactness in circuit, respectively. Hence $(\monadic{v}_1,\monadic{c}_1)\bullet\pi\sim\monadic{v}'_1$ in all cases, establishing $((x,(\monadic{v}_1,\monadic{c}_1)),\gamma_1)\bullet\pi\sim((x,\monadic{v}'_1),\gamma'_1)$. By the induction hypothesis again, $\omicron_2$ accepts~$\pi$, whereby $\gamma_2\bullet\pi\sim\gamma'_2$, $\nu_2\bullet\pi\sim\omicron'_2$, and $d=\publicdom$ or $s=\poststg$ implies $(\monadic{v}_2,\monadic{c}_2)\bullet\pi\sim\monadic{v}'_2$. Hence $\tilde{\omicron}$ accepts~$\pi$, whereby $\tilde{\nu}\bullet\pi\sim\tilde{\omicron}'$ and $d=\publicdom$ or $s=\poststg$ implies $(\monadic{v},\monadic{c})\bullet\pi\sim\monadic{v}'$. We obviously obtain also $\tilde{\gamma}\bullet\pi\sim\tilde{\gamma}'$.

Conversely, assume that $\tilde{\omicron}$ accepts~$\pi$. Then both $\omicron_1$ and $\omicron_2$ accept~$\pi$. By the induction hypothesis, $\sem{e_1}\gamma'\phi\omicron'=\return(\monadic{v}'_1,\gamma'_1,\omicron'_1)$ where $\gamma_1\bullet\pi\sim\gamma'_1$, $\nu_1\bullet\pi\sim\omicron'_1$, and $d_1=\publicdom$ or $s_1=\poststg$ implies $(\monadic{v}_1,\monadic{c}_1)\bullet\pi\sim\monadic{v}'_1$. Like in the proof of the first half, we obtain $((x,(\monadic{v}_1,\monadic{c}_1)),\gamma_1)\bullet\pi\sim((x,\monadic{v}'_1),\gamma'_1)$. By the induction hypothesis, $\sem{e_2}((x,\monadic{v}'_1),\gamma'_1)\phi\omicron'_1=\return(\monadic{v}'_2,\gamma'_2,\omicron'_2)$. Therefore $\sem{e}\gamma'\phi\omicron'=\return(\monadic{v}'_2,\tail\gamma'_2,\omicron'_2)$ and the desired result follows.

\item Let $e=\stmtcomp{e_1}{e_2}$. Then
\[
\begin{array}{l}
\Gamma\vd e_1:\qualty{t_1}{s_1}{d_1}\mbox{,}\\
\Gamma\vd e_2:q\mbox{.}
\end{array}
\]
Furthermore,
\[
\begin{array}{l}
\semc{e_1}\gamma\phi\nu=\return((\monadic{v}_1,\monadic{c}_1),\gamma_1,\omicron_1,\nu_1)\mbox{,}\\
\semc{e_2}\gamma_1\phi\nu_1=\return((\monadic{v}_2,\monadic{c}_2),\gamma_2,\omicron_2,\nu_2)\mbox{,}
\end{array}
\]
and $\semc{e}\gamma\phi\nu=\return((\monadic{v},\monadic{c}),\tilde{\gamma},\tilde{\omicron},\tilde{\nu})$ where
\[
\begin{array}{l}
(\monadic{v},\monadic{c})=(\monadic{v}_2,\monadic{c}_2)\mbox{,}\quad\tilde{\gamma}=\gamma_2\mbox{,}\quad\tilde{\omicron}=\omicron_1\omicron_2\mbox{,}\quad\tilde{\nu}=\nu_2\mbox{.}
\end{array}
\]

Firstly, suppose that $\sem{e}\gamma'\phi\omicron'=\return(\monadic{v}',\tilde{\gamma}',\tilde{\omicron}')$. Then
\[
\begin{array}{l}
\sem{e_1}\gamma'\phi\omicron'=\return(\monadic{v}'_1,\gamma'_1,\omicron'_1)\mbox{,}\\
\sem{e_2}\gamma'_1\phi\omicron'_1=\return(\monadic{v}'_2,\gamma'_2,\omicron'_2)\mbox{,}\\
\monadic{v}'=\monadic{v}'_2\mbox{,}\quad\tilde{\gamma}'=\gamma'_2\mbox{,}\quad\tilde{\omicron}'=\omicron'_2\mbox{.}
\end{array}
\]
By the induction hypothesis, $\omicron_1$ accepts~$\pi$, whereby $\gamma_1\bullet\pi\sim\gamma'_1$, $\nu_1\bullet\pi\sim\omicron'_1$, and $d_1=\publicdom$ or $s_1=\poststg$ implies $(\monadic{v}_1,\monadic{c}_1)\bullet\pi\sim\monadic{v}'_1$. By the induction hypothesis again, $\omicron_2$ accepts~$\pi$, whereby $\gamma_2\bullet\pi\sim\gamma'_2$, $\nu_2\bullet\pi\sim\omicron'_2$, and $d=\publicdom$ or $s=\poststg$ implies $(\monadic{v}_2,\monadic{c}_2)\bullet\pi\sim\monadic{v}'_2$. Hence $\tilde{\omicron}$ accepts~$\pi$, whereby $\tilde{\gamma}\bullet\pi\sim\tilde{\gamma}'$, $\tilde{\nu}\bullet\pi\sim\tilde{\omicron}'$, and $d=\publicdom$ or $s=\poststg$ implies $(\monadic{v},\monadic{c})\bullet\pi\sim\monadic{v}'$.

Conversely, assume that $\tilde{\omicron}$ accepts~$\pi$. Then both $\omicron_1$ and $\omicron_2$ accept~$\pi$. By the induction hypothesis, $\sem{e_1}\gamma'\phi\omicron'=\return(\monadic{v}'_1,\gamma'_1,\omicron'_1)$ where $\gamma_1\bullet\pi\sim\gamma'_1$ and $\nu_1\bullet\pi\sim\omicron'_1$. By the induction hypothesis, $\sem{e_2}\gamma'_1\phi\omicron'_1=\return(\monadic{v}'_2,\gamma'_2,\omicron'_2)$. Therefore $\sem{e}\gamma'\phi\omicron'=\return(\monadic{v}'_2,\gamma'_2,\omicron'_2)$, implying the desired result.

\end{itemize}
\end{proof}



\end{document}